%% file: main.tex
\documentclass[11pt]{article}

\newif\ifcomments
\commentstrue

\input{header.tex}

\title{Dynamic Deterministic Constant-Approximate Distance Oracles with $n^{\epsilon}$ Worst-Case Update Time}

\author{Bernhard Haeupler\thanks{\texttt{bernhard.haeupler@inf.ethz.ch}}\\ETH Zurich \& CMU \and Yaowei Long\thanks{\texttt{yaoweil@umich.edu}}\\University of Michigan \and Thatchaphol Saranurak\thanks{\texttt{thsa@umich.edu}}\\University of Michigan}

\date{}

\begin{document}

\maketitle

\pagenumbering{gobble}
\input{0-abstract}

\newpage

\tableofcontents

\newpage

\pagenumbering{arabic}

\input{1-intro}

\input{2-overview}
\input{3-preliminary}
\input{4-local_flow}

\input{13-path_reporting}
\input{5-dynamic_ED}

\input{6-dense_dynamicED}

\input{7-vertex_sparsifier}

\input{8-expander_hierarchy}

\input{10-low_distance_oracle}

\input{9-hop_reducing_emulator}

\input{11-distance_oracle}
\input{12-spasifier_extension}
\input{14-multicommodity_flow}

\section*{Acknowledgments}
Haeupler is partially funded by the European Research Council (ERC) under the European Union's Horizon 2020 research and innovation program (ERC grant agreement 949272).
Saranurak is supported by NSF grant CCF-2238138.
This work was done in part while Long and Saranurak were visiting the Simons Institute for the Theory of Computing during the Data Structures and Optimization for Fast Algorithms problem.

\appendix
\input{tables}

\input{main.bbl}

\end{document}

%% file: header.tex
\usepackage{fullpage}
\usepackage[utf8]{inputenc}

\usepackage{graphicx}
\usepackage{amsmath,amsthm,amssymb}
\usepackage{algorithm}
\usepackage{algpseudocode}
\usepackage{multirow}
\usepackage{makecell}

\usepackage{thm-restate,color,xspace}
\usepackage{comment}
\usepackage{thmtools}
\usepackage{xcolor}
\usepackage{nameref}
\usepackage{array}
\usepackage{dsfont}

\usepackage{tikz}
\usetikzlibrary{shapes.geometric, arrows, positioning, calc}

\definecolor{ForestGreen}{rgb}{0.1333,0.5451,0.1333}
\definecolor{DarkRed}{rgb}{0.65,0,0}
\definecolor{Red}{rgb}{1,0,0}
\usepackage[linktocpage=true,
pagebackref=true,colorlinks,
linkcolor=DarkRed,citecolor=ForestGreen,
bookmarks,bookmarksopen,bookmarksnumbered]
{hyperref}
\usepackage{cleveref}
\usepackage[T1]{fontenc}

\declaretheorem[numberwithin=section]{theorem}
\declaretheorem[numberlike=theorem,name=Theorem]{thm}
\declaretheorem[numberlike=theorem]{lemma}

\declaretheorem[numberlike=theorem]{corollary}

\declaretheorem[numberlike=theorem]{claim}

\declaretheorem[numberlike=theorem]{observation}

\declaretheorem[numberlike=theorem,style=definition]{definition}
\declaretheorem[numberlike=theorem,style=definition]{remark}
\declaretheorem[numberlike=theorem,name=Definition,style=definition]{defn}
\declaretheorem[numberlike=theorem,name=Open Question]{question}

\newcommand{\rmnum}[1]{\romannumeral #1}

\global\long\def\Ball{\mathrm{Ball}}

\global\long\def\polylog{\mathrm{polylog}}

\global\long\def\iin{{\rm in}}
\global\long\def\out{{\rm out}}

\global\long\def\poly{{\rm poly}}

\global\long\def\cov{{\rm cov}}

\global\long\def\sep{{\rm sep}}
\global\long\def\diam{{\rm diam}}
\global\long\def\wit{{\rm wit}}
\global\long\def\emb{{\rm emb}}
\global\long\def\prune{{\rm prune}}
\global\long\def\new{{\rm new}}

\global\long\def\init{{\rm init}}
\global\long\def\iint{{\rm int}}

\newcommand{\wtilde}{\widetilde}

\newcommand{\eps}{\epsilon}
\global\long\def\poly{\mathrm{poly}}

\global\long\def\spars{\mathrm{spars}}
\global\long\def\dist{\mathrm{dist}}
\global\long\def\cond{\mathrm{cond}}
\global\long\def\supp{\mathrm{supp}}
\global\long\def\core{\mathrm{core}}

\global\long\def\part{\mathrm{part}}
\global\long\def\src{\mathrm{src}}
\global\long\def\sink{\mathrm{sink}}
\global\long\def\Core{\mathrm{Core}}

\global\long\def\supp{\mathrm{supp}}
\global\long\def\rou{\mathrm{rt}}
\global\long\def\rt{\mathrm{rt}}

\global\long\def\ED{\mathrm{ED}}
\global\long\def\sat{\mathrm{sat}}
\global\long\def\PC{\mathrm{NC}}

\global\long\def\size{\mathrm{size}}

\global\long\def\CM{\mathrm{CM}}
\global\long\def\low{\mathrm{low}}
\global\long\def\up{\mathrm{up}}
\global\long\def\heavy{\mathrm{heavy}}
\global\long\def\dom{\mathrm{dom}}
\global\long\def\rcs{\mathrm{rcs}}
\global\long\def\stk{\mathrm{stk}}
\global\long\def\query{\mathrm{LowDO}}

\global\long\def\dynED{\mathrm{dynED}}
\global\long\def\trace{\mathrm{trace}}
\global\long\def\drop{\mathrm{drop}}
\global\long\def\ite{\mathrm{item}}
\global\long\def\virtual{\mathrm{vtul}}
\global\long\def\Member{\textsc{Membership}}

\global\long\def\spars{\mathrm{spars}}
\global\long\def\sep{\mathrm{sep}}
\global\long\def\cond{\mathrm{cond}}
\global\long\def\val{\mathrm{val}}
\global\long\def\vertex{\mathrm{vtx}}
\global\long\def\path{\mathrm{path}}
\global\long\def\blocker{\mathrm{blocker}}
\global\long\def\val{\mathrm{val}}
\global\long\def\size{\mathrm{size}}
\global\long\def\insLM{\mathrm{insLM}}
\global\long\def\delete{\mathrm{delete}}
\global\long\def\tmn{\mathrm{tmn}}
\global\long\def\Gdiam{\mathrm{RealDiam}}

\global\long\def\EH{\mathrm{EH}}
\global\long\def\NC{\mathrm{NC}}

\global\long\def\core{\mathrm{core}}
\global\long\def\recourse{\mathrm{Rcs}}

\global\long\def\DO{\mathrm{DO}}
\global\long\def\rr{\mathrm{emb}}
\global\long\def\wit{\mathrm{wit}}
\global\long\def\aff{\mathrm{aff}}
\global\long\def\edge{\mathrm{edge}}
\global\long\def\del{\mathrm{del}}
\global\long\def\sp{\mathrm{sp}}
\global\long\def\local{\mathrm{local}}

\global\long\def\fresh{\mathrm{fresh}}
\global\long\def\final{\mathrm{final}}
\global\long\def\report{\mathrm{report}}

\newcommand{\ignore}[1]{}

%% file: 0-abstract.tex
\begin{abstract}
We present a new distance oracle in the fully dynamic setting: given a weighted undirected graph $G=(V,E)$ with $n$ vertices undergoing both edge insertions and deletions, and an arbitrary parameter $\epsilon\in[1/\log^{c} n, 1]$ where $c>0$ is a small constant, we can deterministically maintain a data structure with $O(n^{\epsilon})$ worst-case update time that, given any pair of vertices $(u,v)$, returns a $2^{\poly(1/\epsilon)}$-approximate distance between $u$ and $v$ in $\poly(1/\epsilon)\log\log n$ query time.

Our algorithm significantly advances the state-of-the-art in two aspects, both for fully dynamic algorithms and even decremental algorithms.
First, no existing algorithm with  \emph{worst-case update time} guarantees a $o(n)$-approximation while also achieving an $n^{2-\Omega(1)}$ update and $n^{o(1)}$ query time, while our algorithm offers a constant $O_\eps(1)$-approximation with $O(n^\epsilon)$ update time and $O_{\epsilon}(\log \log n)$ query time. 
Second, even if amortized update time is allowed, it is the first \emph{deterministic} constant-approximation algorithm with $n^{1-\Omega(1)}$ update and query time. The best result in this direction is the recent deterministic distance oracle by Chuzhoy and Zhang \cite{chuzhoy2023new} which achieves an approximation of $(\log\log n)^{2^{O(1/\epsilon^{3})}}$ with amortized update time of $O(n^{\epsilon})$ and query time of $2^{\poly(1/\epsilon)}\log n\log\log n$.

We obtain the result by dynamizing tools related to \emph{%
length-constrained expanders} 
\cite{HRG22,HaeuplerHT2023length,haeupler2022cut}. Our technique completely bypasses the 40-year-old Even-Shiloach tree, which has remained the most pervasive tool in the area but is inherently amortized.

\end{abstract}

%% file: 1-intro.tex
\section{Introduction}

The \emph{dynamic distance oracle} problem (or the \emph{dynamic all-pairs shortest paths} problem) is one of the most well-studied dynamic graph problems since the beginning of the field (see \Cref{tab:exact,tab:small approx,tab:large approx}). In this problem, given an undirected weighted graph $G=(V,E)$ with $n$ vertices undergoing edge insertions and deletions, the goal is to build a data structure that given a vertex pair $(u,v)$, quickly returns an (approximate) distance between $u$ and $v$. The time required to update the data structure for each edge update is called \emph{update time}. The data structure has $t$ \emph{worst-case} update time if every update needs at most $t$ time. On the other hand, it has $t$ \emph{amortized} update time if it only guarantees that after $U$ updates for all large enough $U$, the total time taken is at most $t\times U$. The time to answer each distance query is called \emph{query time}. If the algorithms handle only edge insertions or deletions, then they are \emph{incremental} or \emph{decremental}, respectively. If they handle both, then they are \emph{fully dynamic}. 

One of the overarching goals in this long line of research is to obtain, for every integer $k\ge2$, a fully dynamic deterministic $(2k-1)$-approximation algorithm with $O(n^{1/k})$ worst-case update time and $n^{o(1)}$ query time, which is an optimal trade-off assuming the 3SUM conjecture \cite{abboud2023stronger,jin2023removing}. In 2018, Chechik \cite{chechik}, improving upon \cite{rodittyZ2,BernsteinR11}, nearly achieved this goal \emph{quantitatively}; she gave a decremental randomized $((2+\epsilon)k-1)$-approximation algorithm with $n^{1/k+o(1)}$ amortized update time and $O(\log\log n)$ query time. A similar trade-off (i.e., constant approximation with arbitrarily small polynomial update time) was also recently shown in the fully dynamic setting. Improving upon \cite{abraham2014fully,forster2021dynamic}, Forster et al\@.~\cite{forster2023bootstrapping} gave a fully dynamic $(1/\epsilon)^{O(1/\epsilon)}$-approximation algorithm with $O(n^{\epsilon})$ amortized update time and $O(n^{\epsilon/8})$ query time. All these results, however, have two important \emph{qualitative} drawbacks:
\begin{enumerate}
\item they are randomized and also assume an oblivious adversary\footnote{An \emph{oblivious} adversary fixes the update sequence from the beginning and reveals the updates to the data structure one by one. In particular, the updates are independent from the random choice of the data structure.}, and
\item they only guarantee amortized update time.
\end{enumerate}
Fixing either of the two drawbacks remains a big challenge for fully dynamic distance oracles and even decremental ones. In contrast, in the \emph{incremental} setting, both drawbacks were addressed by Chen et al.~\cite{chen2020fast} who gave deterministic algorithms with worst-case update time and the following trade-off: $(1/\epsilon)^{O(1/\epsilon)}$-approximation using $n^{\epsilon}$ update and query time, and $(2k-1)$-approximation using $O(m^{1/2}n^{1/k})$ update and query time. Unfortunately, their techniques inherently cannot work in the decremental setting, which is more challenging and implies more applications \cite{Madry10_stoc,chuzhoy2021decremental,bernstein2022deterministic}. Below, we discuss the progress to fixing both drawbacks in decremental and fully dynamic distance oracles.

\paragraph{Derandomization.}

Derandomization and, more generally, removing the oblivious adversary assumption are major research programs on dynamic algorithms (see, e.g., \cite{nanongkai2017dynamic,fast-vertex-sparsest,Wajc20,beimel2022dynamic}) because they enable data structures to be used as a subroutine inside static algorithms and have led to many exciting applications \cite{chuzhoy2021decremental,bernstein2022deterministic,chen2022maximum,abboud2022breaking}.%

Towards derandomizing Chechik's distance oracle, Chuzhoy \cite{chuzhoy2021decremental} showed a deterministic decremental $(\log n)^{2^{O(1/\epsilon)}}$-approximate distance oracle with $O(n^{\epsilon})$ amortized update time and $O(\log n\log\log n)$ query time, and applied it to obtain faster multi-commodity flow and multicut algorithms.\footnote{Similar trade-off was independently shown in \cite{bernstein2022deterministic}.} Recently, Chuzhoy and Zhang \cite{chuzhoy2023new} improved the approximation factor to $(\log\log n)^{2^{O(1/\epsilon^{3})}}$ and also make the algorithm fully dynamic. Can one improve the approximation further to a constant? Unfortunately, the fastest deterministic constant-approximation algorithm still takes a large update time of at least $n^{1+o(1)}$ \cite{bernstein2022deterministic}, but it gives $(1+\epsilon)$-approximation. Hence, we ask:
\begin{question}
\label{question:deterministic}Is there a \emph{deterministic} fully dynamic (or even decremental) distance oracle with constant approximation and  $n^{1-\Omega(1)}$ update and query time?
\end{question}

\paragraph{Deamortization.}

Deamortization is another major research program in dynamic algorithms (e.g.~\cite{thorup2005worst,KapronKM13,NanongkaiSW17,bernstein2021deamortization}) since worst-case update time guarantees are required by many real-world applications and also some reductions \cite{blikstad2023fast}. For dynamic distance oracles, the deamortization challenge highlights one of the biggest gap of our understanding. While decremental $(2k-1)$-approximation algorithms with $n^{1/k+o(1)}$ amortized update time and $O(\log\log n)$ query time are known \cite{chechik,lkacki2022near}, there is no non-trivial algorithm with $n^{2-\Omega(1)}$ worst-case update time and $n^{o(1)}$ query time!%

{} 

There are two explanations behind the lack of progress. First, a conditional lower bound \cite[Corollary 5.10]{BrandNS19} implies that, any decremental or incremental distance oracle with $(5/3-\epsilon)$-approximation cannot have $n^{2-\Omega(1)}$ worst-case update time and $n^{o(1)}$ query time. However, this does not explain the lack of progress in larger approximation regimes. The deeper explanation is that \emph{all} known fully dynamic and decremental algorithms with $n^{2-\Omega(1)}$ update time and $n^{o(1)}$ query time share one common technique that is inherently amortized, called the \emph{Even-Shiloach tree} \cite{EvenS,king1999fully}.\footnote{The Even-Shiloach tree is actually a technique for maintaining single-source short paths, but previous dynamic distance oracle algorithms still need to use it as a subroutine to maintain ``balls'' around vertices.} This technique is already 40 year old, yet no alternative techniques were developed to avoid amortization until now. 
\begin{question}
\label{question:worst-case}Is there a fully dynamic (or even decremental) distance oracle with $o(n)$-approximation, $n^{2-\Omega(1)}$ \emph{worst-case} update time, and $n^{o(1)}$ query time? More specifically, can we avoid the Even-Shioach tree?
\end{question}

\paragraph{Our Result.}

We give a strong affirmative answer to both Questions \ref{question:deterministic} and \ref{question:worst-case}.

\begin{restatable}{theorem}{maintheorem}\label{thm:main}
\label{thm:FullyDynamicDistanceOracle}
Let $G$ be an $n$-vertex fully dynamic undirected graph with size $m_{0}$ initially and polynomially-bounded positive integral edge lengths over all updates. For some small constant $c>0$, given a parameter $1/\log^{c} n\leq \epsilon\leq 1$, there is a $2^{\poly(1/\epsilon)}$-approximate deterministic fully dynamic distance oracle with $O(m_{0}^{1+\epsilon})$ initialization time, $O(n^{\epsilon})$ worst-case update time and $O(\log\log n/\epsilon^{4})$ query time. 

\end{restatable}

Our distance oracle can also output an approximate shortest path $P$ in additional $O(|P|)$ time, or an approximate shortest \emph{simple} path $P_{\sp}$ in additional $O(|P_{\sp}|\cdot n^{\epsilon})$ time. See \Cref{thm:MainDetailed} for a detailed version of \Cref{thm:main}. 

Up to a constant factor in the approximation, \Cref{thm:main} achieves the optimal trade-off \cite{abboud2023stronger,jin2023removing} and simultaneously guarantees desirable qualitative properties including being deterministic, fully dynamic, and having worst-case update time. \Cref{tab:high level} gives a high-level comparison between our distance oracle with the start-of-the-art.%

\begin{table}
\begin{centering}
\begin{tabular}{|c|c|c|c|c|}
\hline 
 & Approximation & Fully dynamic? & Deterministic? & Worst-case?\tabularnewline
\hline 
\hline 
\cite{chechik} & $O(1)$ & \textcolor{red}{decremental} & \textcolor{red}{no} & \textcolor{red}{no}\tabularnewline
\hline 
\cite{chen2020fast} & $O(1)$ & \textcolor{red}{incremental} & yes & yes\tabularnewline
\hline 
\cite{forster2023bootstrapping} & $O(1)$ & fully & \textcolor{red}{no} & \textcolor{red}{no}\tabularnewline
\hline 
\cite{chuzhoy2023new} & \textbf{\textcolor{red}{$\omega(1)$}} & fully & yes & \textcolor{red}{no}\tabularnewline
\hline 
\makecell{\cite{KMP23}\\\footnotesize{Independent}} & \textbf{\textcolor{red}{$\omega(1)$}} & fully & yes & yes\tabularnewline
\hline
\textbf{Ours} & $O(1)$ & fully & yes & yes\tabularnewline
\hline 
\end{tabular}
\par\end{centering}
\caption{\label{tab:high level}Comparison between our distance oracle with the state-of-the-art distance oracles when $n^{\epsilon}$ update time is allowed. \cite{KMP23} is an independent work as discussed in \Cref{sect:IndependentWork}.}
\end{table}

\subsection{Our Techniques}

The technique we use to prove \Cref{thm:main} completely bypasses the Even-Shiloach tree. Instead we dynamize tools based on \emph{length-constrained expanders}, which were recently introduced by Haeupler, Räcke, and Ghaffari~\cite{HRG22}. At a high level, length-constrained expanders are a generalization of expanders that offer the same advantages (e.g., robustness under updates, low congestion routing) but in addition to approximating flows and cuts also provide approximations and control over distances. Thanks to the recent expansion of this theory to constant-step expanders~\cite{haeupler2022cut,HaeuplerHT2023length} and efficient algorithms for them \cite{HaeuplerHT2023length} constant-step length-constrained expanders can even achieve constant-fractor approximations and go beyond the polylogarithmic flow-cut gap inherent in regular expanders and their decompositions.

Now, we highlight our novel technical components, which may be of independent interest, by placing them in the context of previous works. A high-level discussion on how we construct and use these components can be found in \Cref{sect:overview}.

Before going into details, we first introduce the \emph{online-batch dynamic setting}. Unlike the fully dynamic setting in which we receive a sequence of updates, in the online-batch dynamic setting, the updates come in the form of several online batches. The algorithm needs to handle the batches one by one, and the update time for each batch should be roughly proportional (with a small overhead, e.g. $n^{\epsilon}$) to the size of this batch \emph{in the worst case}.\footnote{Generally, dynamic algorithms with \emph{amortized} update time do not work in the online-batch dynamic setting.} A standard reduction to obtain fully dynamic algorithms (with worst-case update time) from online-batch algorithms was explicitly shown in \cite{NanongkaiSW17,JS22}. Therefore, we will aim at online-batch dynamic distance oracles and most components introduced below are online-batch dynamic. Refer to \Cref{sect:OnelineBatch} for preliminaries and notations of the online-batch dynamic setting, and see \Cref{sect:Reduction} for the formal statement of the reduction.

\paragraph{Localized Length-Constrained Flows.} Local flow algorithms refer to algorithms for solving flow problems with running times that are independent of the graph size. 
Many such algorithms with impressive applications have been developed, such as the local versions of blocking flow \cite{orecchia2014flow, nanongkai2019breaking,hua2023maintaining}, push-relabel \cite{henzinger2020local,SW19}, and Ford-Fulkerson \cite{chechik2017faster,forster2020computing}. In \Cref{sect:LocalLengthConstrainedFlow}, we extend local flow techniques to the new context of length-constrained flow by localizing the length-constrained flow algorithms from \cite{haeupler2023maximum}. \Cref{thm:IntroLocalFlow} is the statement of our local length-constrained flow algorithm.\footnote{An $h$-length flow has only flow paths with length at most $h$, and a fractional moving cut is an edge-weight function s.t. any $s$-$t$ path has weight at least $1$. In fact, $h$-length flows and fractional moving cuts are solution to the primal and dual LPs formalizing the $h$-length maxflow problem.} See \Cref{thm:LocalLengthConstrainedFlow} for a detailed version. 

\begin{theorem}
\label{thm:IntroLocalFlow}
Let $G$ be a directed graph with positive integral lengths $\ell$, capacities $U$, source vertex $s$, sink vertex $t$, further satisfying that for each vertex $v\in V(G)\setminus\{s,t\}$,
\begin{itemize}
\item its sink capacity $\nabla(v) := U(v,t)$ is at least its out-degree $\deg^{+}_{G}(v)$, and
\item its sink edge $(v,t)$ has length $\ell(v,t) = 1$.
\end{itemize}
Given parameters $0<\delta<1, h\geq 1$, there is an algorithm that compute a feasible pair of $h$-length flow and fractional moving cut that is $(2+\delta)$-approximate. The running time is $\poly(h,1/\delta,\log n)\cdot\sum_{v\in V(G)\setminus\{s,t\}}\Delta(v)$,
where $\Delta(v):= U(s,v)$ is the source capacity of $v$.
\end{theorem}

\paragraph{Dynamic Router and Dynamic Length-Constrained Expander Decomposition.} The expander pruning algorithm by \cite{SW19} efficiently maintains a large expander subgraph of a decremental expander, enabling numerous dynamic algorithms to leverage the power of expanders (e.g., \cite{NanongkaiSW17,JS22}). In \Cref{sect:DynamicRouter}, we design dynamic router (some concrete expander with constant diameter) which supports online-batch pruning operations. In \Cref{sect:DynamicCertifiedED}, we will implement expander pruning for length-constrained expanders by embedding dynamic routers into a length-constrained expander as its expansion certificate. Further, using our local length-constrained flow algorithm to compute a \emph{cutmatch} in local time between the pruned graph and the remaining expander, part of the pruned graph can be matched back to the expander and the rest will be separated by adding a new cut. We note that a cutmatch algorithm with non-local running time was shown in \cite{haeupler2023maximum} using their non-local length-constrained maxflow algorithm. These two steps lead to our online-batch dynamic algorithm, \Cref{thm:MainDynED}, for maintaining the length-constrained version of expander decomposition.\footnote{The integral moving cut in \Cref{thm:MainDynED} is a non-negative integral weight function on $G$ (slightly different from the fractional moving cut in \Cref{thm:IntroLocalFlow}), and its size is the total weight on all edges. The graph $G-C$ is the graph $G$ with different edge lengths $\ell_{G-C} = \ell_{G} + C$. An $h$-length unit demand includes only demand pairs between vertices with distance at most $h$.} We note that a static version of length-constrained expander decomposition is shown in \cite{HaeuplerHT2023length}, and our dynamic algorithm will use it as a subroutine.

\begin{theorem}
\label{thm:MainDynED}
Let $G$ be a graph with positive integral edge lengths undergoing $t$ batched updates $\pi^{(1)},\pi^{(2)},...,\pi^{(t)}$ of edge insertions and deletions. For any $\eps \in [1/\log^{\Omega(1)}(n),1]$, $h\geq 1$ and $\phi\leq 1$, there is a deterministic algorithm that maintains an integral moving cut $C$ satisfying the following.
\begin{itemize}
\item Initially, $C$ has size  $O(\phi\cdot h\cdot |G^{(0)}|)$ and the size increase after batch $i$ is $O(\phi\cdot h\cdot |\pi^{(i)}|)$.
\item The graph $G$ and moving cut $C$ can be characterized in the flow view: any $h$-length unit demand on $G-C$ can be routed on $G$ with length $h\cdot 2^{\poly(1/\epsilon)}$ and congestion $O(n^{\epsilon}/\phi)$.
\end{itemize}
The initialization time is $|G^{(0)}|\cdot \poly(h)\cdot n^{\epsilon}$ and the update time for $\pi^{(i)}$ is $|\pi^{(i)}|\cdot \poly(h)\cdot n^{\epsilon}/\phi$.
\end{theorem}

For better understanding, if we put \Cref{thm:MainDynED} in the context of classic expanders, the cut $C$ will be a subset of edges with size at most a $\phi$ fraction of $G$ and previous updates, and the flow characterization will become that any unit demand $D$ on $G\setminus C$ (with demand pairs between connected vertices) can be routed on $G$ with roughly $1/\phi$ congestion.
\Cref{thm:DynamicED} is a detailed version of \Cref{thm:MainDynED}, where, besides the cut $C$, we will also maintain an explicit structure for routing the demand.

\paragraph{Dynamic Length-Constrained Expander Hierarchy.} Expander decomposition becomes even more powerful when we build them on top of each other. We consider two structures that repeatedly perform expander decomposition in a bottom-up manner. The first structure, the dynamic expander hierarchy in \cite{GRST21}, can approximate various cut-flow values in the dynamic setting but is not useful for distance-based problems. The second structure, the length-constrained expander hierarchy in \cite{HRG22}, encodes both cut and distance information simultaneously but does not work in the dynamic setting. In \Cref{sect:ExpanderHierarchy}, we aim to combine the best of both worlds by developing a hierarchy that works in the dynamic setting and approximates distances, although it does not handle cut-flow values. This limitation is due to technical reasons that we believe could be overcome. \Cref{thm:ExpanderHierarchy}  is a formal statement of our dynamic hierarchy. To efficiently maintain our hierarchy, we introduce dynamic vertex sparsifiers for distances in \Cref{sect:DynSparsifier} and generalize this result in \Cref{sect:ExtendedSparsifier}, which may be of interest and will be discussed later.

\paragraph{Length-Reducing Emulator.} Roughly speaking, an $h$-length-reducing emulator $Q$ of a graph $G$ is a graph with unit edge lengths such that $\dist_{Q}(u,v)\approx\dist_{G}(u,v)/h$ for all $u,v$. This is a key object in the literature because many dynamic distance oracles \cite{chechik,chuzhoy2021decremental,bernstein2022deterministic,lkacki2022near,chuzhoy2023new} can be viewed as maintaining dynamic length-reducing emulators on top of each other, so that they only need to handle small distance values. Our final distance oracle also follows this standard approach. Unfortunately, known dynamic constructions of length-reducing emulators only guarantee amortized update time (as they use Even-Shiloach trees) and do not directly work in the online-batch dynamic setting. In \Cref{sect:DynHopEmu}, we present a new construction, based on our expander hierarchy, that overcomes both of these limitations. See \Cref{thm:HopReducingEmulator} for a detailed version of \Cref{thm:IntroEmulator}.

\begin{theorem}
\label{thm:IntroEmulator}
Let $G$ be an $n$-vertex dynamic graph with positive integral edge lengths under $t$ batched updates $\pi^{(1)},\pi^{(2)},...,\pi^{(t)}$ of edge insertions and deletions. Given parameters $\epsilon\in [1/\log^{\Omega(1)} n,1]$ and $h\geq 1$, there is a deterministic algorithm that maintains a graph $Q$ with unit edge lengths s.t. for all vertices $u,v\in V(G)$, we have
\begin{itemize}
\item $\dist_{G}(u,v)\leq 2^{\poly(1/\epsilon)}\cdot h\cdot \dist_{Q}(u,v)$, and
\item if $\dist_{G}(u,v)\leq h$, then $\dist_{Q}(u,v)\leq \poly(1/\epsilon)$.
\end{itemize}
The size of $Q$ is at most $O(n^{1+\epsilon})$ at all time. The initialization time is $|G^{(0)}|\cdot \poly(h)\cdot n^{\epsilon}$ and the update time for $\pi^{(i)}$ is $|\pi^{(i)}|\cdot \poly(h)\cdot n^{\epsilon}$.
\end{theorem}

\subsection{Applications and Future Work}

\paragraph{Multicommodity Maxflow with Unit Vertex Capacity.}
In the maximum multicommodity flow problem with unit vertex capacity, the input is a graph $G$ and $k$ demand pairs $(s_{1},t_{1}),\dots,(s_{k},t_{k})$. The goal is to send a maximum amount of flow between any demand pairs such that the total amount of flow traversing through any vertex is at most one. This problem can be solved exactly by an LP solver. Faster approximation algorithms can be obtained via applying dynamic shortest path algorithms to the multiplicative weight update framework as first observed by Madry \cite{Madry10_stoc}. The state of the art of this problems currently corresponds precisely to the one of dynamic shortest path algorithms that support shortest paths queries (not just distance queries).

More precisely, the dynamic $(1+\epsilon)$-approximate single-source shortest paths algorithms of \cite{APSP-old,bernstein2022deterministic} implies a $(1+\epsilon)$-approximation algorithm with $\ensuremath{km^{1+o(1)}}$ time, which is slow for big $k$. The dynamic all-pairs shortest paths algorithm algorithms of \cite{chuzhoy2021decremental,chuzhoy2023new} implies $(\log\log n)^{2^{O(1/\epsilon^{3})}}$-approximation algorithms with $O((m+k)n^{\epsilon})$ time. 
\Cref{thm:main} immediately implies the first constant-approximation algorithm with $O((m+k)n^{\epsilon})$ time. 

\begin{theorem}
Let $G$ be an vertex-unit-capacitated undirected graph with $m$ edges, and $k$ source-sink pairs $\{(s_{j},t_{j})\mid 1\leq j\leq k\}$. Given $\epsilon\in [1/\log^{\Omega(1)} n,1]$, there is a deterministic algorithm computing a $2^{\poly(1/\epsilon)}$-approximate multicommodity maxflow in $O((m+k)n^{\epsilon})$ time.
\end{theorem}

Very recently, \cite{HHLRS2023emu} showed a $2^{\mathrm{poly}(1/\epsilon)}$-approximation $O((m+k)n^{\epsilon})$-time algorithm for various version of multicommodity flow (e.g.~both maximum and concurrent version, with capacities). But their technique only handle edge capacities and does not work with vertex capacity like ours.

\paragraph{Vertex Sparsifiers for Distances.}

For any graph $G$, a vertex sparsifier for distances with respect to a terminal set $T$ is a small graph $H$ containing close to $|T|$ vertices such that, for any terminals $u,v\in T$, $\dist_{H}(u,v) \approx \dist_{G}(u,v)$.

Previous dynamic algorithms for this problem either work on only incremental graphs \cite{chen2020fast,forster2023deterministic} or assume oblivious adversary \cite{forster2023bootstrapping}.
Recently, they also were used in the breakthrough results for computing maximum flow in almost-linear time \cite{chen2022maximum,van2023deterministic}. However, these dynamic algorithms assume very specific structure in the update sequence. 
We resolved these issues by showing a deterministic fully dynamic algorithm for maintaining a vertex sparsifier of size $O(|T|n^\eps)$ with $O(1)$-approximation and $O(n^\epsilon)$ amortized update time.

\begin{theorem}
\label{thm:IntroSparsifiers}
\label{thm:vertex sparsifier}
Let $G$ be an $n$-vertex fully dynamic graph with edge lengths and a fully dynamic terminal set $T\subseteq V(G)$. For any $\epsilon\in [1/\log^{\Omega(1)} n,1]$, there is a deterministic algorithm that processes a sequence of edge insertions/deletions and terminal insertions/deletions, and maintains a graph $H$ of size $|H|\leq O(|T|\cdot n^{\epsilon})$ using $O(n^{\epsilon})$ amortized update time, such that for all time and all terminals $u,v \in T$,
$\dist_{G}(u,v)\leq \dist_{H}(u,v)\leq 2^{\poly(1/\epsilon)}\cdot\dist_{G}(u,v).$
\end{theorem}
We can also obtain worst-case update time by maintaining $2^{O(1/\epsilon)}$ candidates sparsifiers instead of one explicit sparsifier. See \Cref{thm:FullyDynamicSparsifier} for the formal statement.
Very recently, \cite{KMP23} independently showed a fully dynamic deterministic vertex sparsifier of size $|T|n^{o(1)}$ with $n^{o(1)}$ approximation and worst-case update time.

\paragraph{Future Work: Dynamic Min-Cost Flows.}
\cite{HHLRS2023emu} very recently introduced a stronger version of length-reducing emulators that captures both flow and distance information, and used it to develop very fast constant-approximation algorithms for (min-cost) multi-commodity flow problems. We believe it is possible to enhance our dynamic expander hierarchy and emulators to capture flow information (as done by \cite{HHLRS2023emu} but in the static setting) and obtain dynamic low-step flow-emulators for unit-capacity graphs as well. This would be a powerful tool for dynamic min-cost multi-commodity flow algorithms and other applications.%

\subsection{An Independent Work}
\label{sect:IndependentWork}
An independent work by Kyng, Meierhans and Probst Gutenberg \cite{KMP23}\footnote{Their paper and the first version of our paper were both finished slightly before the submission deadline of STOC 2024 (Nov 13, 2023).} also obtains deterministic fully dynamic approximate distance oracles with worst-case update time using completely different techniques. Precisely, their oracle has $\exp(O(\log^{6/7} n\log\log n)) = n^{o(1)}$ approximation, $\exp(O(\log^{20/21} n\log\log n)) = n^{o(1)}$ worst-case update time and $O(\log^{2} n)$ query time. For comparison, our distance oracle has $\exp(\poly(1/\epsilon))$ approximation, $n^{\epsilon} = \exp(\eps \log n)$ worst-case update time and $O((1/\epsilon)^4  \cdot \log\log n)$ query time, for any parameter $\epsilon\in[1/\log^{c}n,1]$ where $c>0$ is a small constant. By setting $\eps = 1/\log^{c}n$, we can achieve similar upper bounds with $\exp(\log^{1-\Omega(1)}n) = n^{o(1)}$ overhead, but we can further achieve $O(1)$ approximation by setting $\eps$ to be an arbitrarily small constant at the expense of a small polynomial overhead in update time.
We note that the \cite{KMP23} distance oracle supports path-reporting queries, but the oracle in the first version of our manuscript does not. We note that our current version, which is finished after \cite{KMP23} was published, includes results about path reporting.

Additionally, both papers build deterministic fully dynamic distance-preserving vertex sparsifiers that support insertions and deletions of edges and terminal vertices on the way to their distance oracle results 
(see \cite[Theorem 1.1]{KMP23} and \Cref{thm:IntroSparsifiers}). 
The sparsifier in \cite{KMP23} has $\exp(O(\log^{20/21}n\log\log n)) = n^{o(1)}$ approximation, $\exp(O(\log^{20/21}n\log\log n)) = n^{o(1)}$ \emph{worst-case} update time and size $|A|\cdot\exp(O(\log^{20/21}n\log\log n)) = |A|\cdot n^{o(1)}$, where $A$ denotes the set of terminals. For comparison, our sparsifier has $\exp(\poly(1/\epsilon))$ approximation, $O(n^{\epsilon})$ \emph{amortized} update time and size $O(|A|\cdot n^{\epsilon})$, for any parameter $\epsilon\in[1/\log^{c}n, 1]$ where $c>0$ is a sufficiently small constant. We note that our dynamic sparsifer algorithm can guarantee $O(n^{\epsilon})$ \emph{worst-case} update time when it maintains $2^{O(1/\epsilon)}$ candidate sparsifiers instead of only one explicit sparsifier.

\paragraph{Organziation}
We give an overview of the paper in the next section and preliminaries in \Cref{sect:prelim}. The dependency between the remaining sections is shown in \Cref{figure:dependency}.
\input{dependency_DAG}

%% file: dependency_DAG.tex
\begin{figure}[ht]
\centering

\begin{tikzpicture}[auto]
\tikzset{
    mynode/.style={rectangle, draw=red, thick, fill=red!10, text width=5cm, text centered, rounded corners, minimum height=1cm},
    arrow/.style={->, >=latex', shorten >=2pt, thick},
    line/.style={-, thick}
}

\node[mynode] (Section 4) at (0,0) {\Cref{sect:LocalLengthConstrainedFlow}\\ Local Flow};
\node[mynode] (Section 5) at (6,0) {\Cref{sect:DynamicRouter}\\Dynamic Routers\\with Path Reporting};
\node[mynode] (Section 6) at (3,-2) {\Cref{sect:DynamicCertifiedED}\\Expander Decomposition\\ with Landmarks};
\node[mynode] (Section 7) at (9,-2) {\Cref{sect:DynDenseED}\\Expander Decomposition\\ with Landmarks and Density};
\node[mynode] (Section 8) at (9,-4) {\Cref{sect:DynSparsifier}\\Vertex Sparsifiers for Bounded Distance};
\node[mynode] (Section 9) at (3,-4) {\Cref{sect:ExpanderHierarchy}\\Expander Hierarchy};
\node[mynode] (Section 10) at (0,-6) {\Cref{sect:OracleShort}\\Distance Oracles for Bounded Distance};
\node[mynode] (Section 11) at (6,-6) {\Cref{sect:DynHopEmu}\\Length-reducing Emulators};
\node[mynode] (Section 12) at (3,-8) {\Cref{sect:Oracle}\\Distance Oracles};
\node[mynode] (Section 13) at (9,-8) {\Cref{sect:ExtendedSparsifier}\\Vertex Sparsifiers};
\node[mynode] (Section 14) at (3, -10)
{\Cref{sect:MultiFlow}\\Multicommodify Maxflows};

\draw[arrow] (Section 4) -- (Section 6);
\draw[arrow] (Section 5) -- (Section 6);
\draw[arrow] (Section 6) -- (Section 7);
\draw[arrow] (Section 7) -- (Section 8);
\draw[arrow] (Section 8) -- (Section 9);
\draw[arrow] (Section 6) -- (Section 9);
\draw[arrow] (Section 9) -- (Section 10);
\draw[arrow] (Section 9) -- (Section 11);
\draw[arrow] (Section 10) -- (Section 12);
\draw[arrow] (Section 11) -- (Section 12);
\draw[arrow] (Section 8) -- (Section 13);
\draw[arrow] (Section 11) -- (Section 13);
\draw[arrow] (Section 12) -- (Section 14);

\end{tikzpicture}
\caption{Dependencies between sections in this paper.\label{figure:dependency}}

\end{figure}
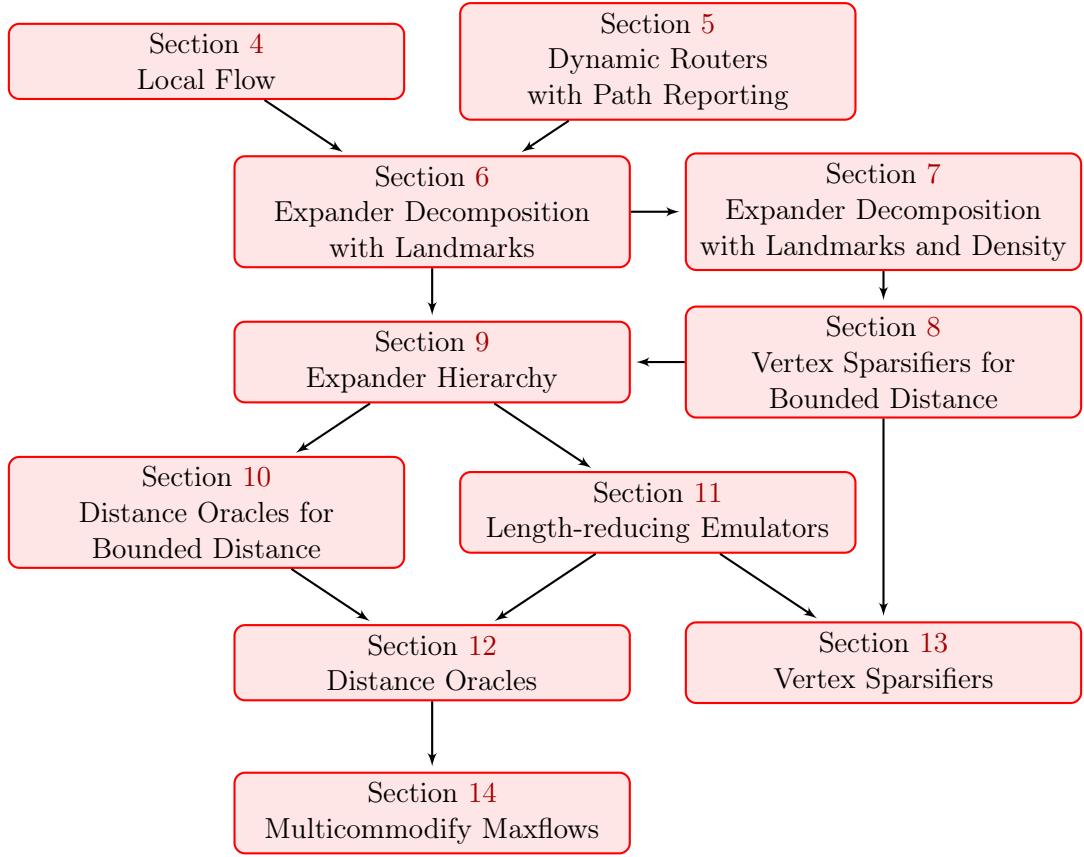

%% file: 2-overview.tex
\section{Overview}

\label{sect:overview}
In this overview, for simplicity, we assume that all input graphs have constant degree. Moreover, we consider the \emph{online-batch dynamic} setting, in which a dynamic algorithm only needs to handle a constant number of \emph{batched updates} $\pi^{(1)},\pi^{(2)},...,\pi^{(t)}$, instead of a long sequence of \emph{unit updates}. The update time to handle $\pi^{(i)}$ should be $|\pi^{(i)}|\cdot n^{O(\epsilon)}$ \emph{in the worst case}. Then, a standard reduction can transform an online-batch dynamic algorithm to a fully dynamic algorithm with \emph{worst-case} update time.

For better understanding, we want to clarify the difference between online-batch dynamic algorithms and fully dynamic algorithms with \emph{amortized} update time. The former is stronger than the latter in the sense that an online-batch dynamic algorithm can be transformed into a fully dynamic algorithm with amortized (even worst-case) update time, but generally a fully dynamic algorithm with amortized update time does not work in the online-batch dynamic setting. Precisely, for a fully dynamic algorithm has small amortized update time, it may have worst-case update time $\Omega(n)$ for one unit update (it is indeed the case for all previous dynamic distance oracles using Even-Shiloach trees), so when put it into the online-batch dynamic setting, its worst-case update time $\Omega(n)$ for a batch $\pi^{(i)}$ cannot be bounded by any function of $|\pi^{(i)}|$ (since $|\pi^{(i)}|$ can be extremely small, e.g. a constant).

Consider an online-batch dynamic algorithm with a dynamic object ${\cal O}_{\iin}$ under batched updates $\pi^{(1)},...,\pi^{(t)}$ as the input (generally ${\cal O}_{\iin}$ is a graph) and another object ${\cal O}_{\out}$ as the output. The \emph{recourse} of ${\cal O}_{\out}$ at time $i$, denoted by $\recourse({\cal O}^{(i-1)}_{\out} \to {\cal O}^{(i)}_{\out})$, is the size of the batched update (generated by the algorithm) to update ${\cal O}^{(i-1)}_{\out}$ to ${\cal O}^{(i)}_{\out}$ (${\cal O}^{(i-1)}_{\out}$ and ${\cal O}^{(i)}_{\out}$ are the ${\cal O}_{\out}$ right before and after the algorithm handles $\pi^{(i)}$). In our main body, the recourse $\recourse({\cal O}^{(i-1)}_{\out} \to {\cal O}^{(i)}_{\out})$ will be \emph{worst-case} bounded in the form $f(n)\cdot |\pi^{(i)}|$, and $f(n)$ is usually a small polynomial e.g. $n^{O(\epsilon^{4})}$, but in this overview, we will say the recourse (amortized inside \emph{each batched update}) of ${\cal O}_{\out}$ is $f(n)$ for simplicity. See \Cref{sect:OnelineBatch} for more discussions.

\subsection{From Length Constrained Expander Hierarchy to Distance Oracles}

As a warm-up, we start from discussing how a simplified version of the static expander hierarchy from \cite{HRG22} gives a structure that supports constant-approximate distance queries. Our high-level approach to obtain dynamic distance oracles is to dynamize this structure, and we will point out the key challenges we need to overcome at the end of this subsection.

\paragraph{Preliminaries on Length Constrained Expanders.}

We first give an over-simplified description (which may be incorrect) of some concepts related to length constrained expanders.
Roughly speaking, a graph $G$ is an \emph{$h$-length $\phi$-expander} for a terminal set $T\subseteq V(G)$ if, for any $h$-length unit demand $D$ (roughly speaking, $D$ is a collection of \emph{demand pairs} where each vertex $v\in V(G)$ appears in at most one pair and each pair $(u,v)\in D$ has $\dist_{G}(u,v)\leq h$), one can embed $D$ into $G$ with length $O(h)$ and congestion $O(1/\phi)$, where the embedding can be interpreted as a mapping from demand pairs to embedding paths on $G$ s.t. all embedding paths have length at most $O(h)$ and each edge in $G$ is used by at most $O(1/\phi)$ embedding paths. 
Furthermore, an $h$-length $\phi$-expander w.r.t. $T$ admits a \emph{certificate} $({\cal N}, {\cal R}, \Pi_{{\cal N}\to G})$ that includes the following.
\begin{itemize}
\item ${\cal N}$ is a pairwise cover of $T$ on $G$ with \emph{covering radius} $h_{\cov} = h$ and \emph{width} $\omega = O(1)$. Roughly speaking, ${\cal N}$ is a collection of vertex sets $S\subseteq T$ (called \emph{clusters}), which can be partitioned into $\omega$ many \emph{clusterings}, where each clustering is a collection of disjoint clusters. The covering radius means each pair of vertices $u,v\in V(G)$ with $\dist_{G}(u,v)\leq h_{\cov}$ has a cluster $S$ s.t. $u,v\in S$. We note that, the following two objects ${\cal R}$ and $\Pi_{{\cal R}\to G}$ will certify that each cluster $S\in {\cal N}$ has diameter $O(h)$ on $G$, i.e. $\dist_{G}(u,v)\leq O(h)$ for all $u,v\in S$.
\item ${\cal R} = \{R^{S}\mid S\in{\cal N}\}$ is a set of routers, one for each cluster $S\in {\cal N}$, where $V(R^{S}) = S$. A \emph{router} $R$ is a graph that can route any unit-demand on $V(R)$ with $O(1)$ length and $O(1)$ congestion. For better understanding, we can interpret a router as an $\Omega(1)$-classic expander with diameter $O(1)$.

\item $\Pi_{{\cal R}\to G}$ is an embedding of the union of routers $\bigcup_{R\in{\cal R}} R$ into $G$ with length $O(h)$ and congestion $O(1/\phi)$. The embedding here follows the same definition as above, except that it matches router edges to embedding paths. %
\end{itemize}

An \emph{integral moving cut} $C$ is an integral function on $E(G)$. Its size is $|C| = \sum_{e}C(e)$. The graph after cutting $C$ on $G$, denoted by $G-C$, has the same set of vertices and edges but different edge lengths $\ell_{G-C}$ s.t. $\ell_{G-C}(e) = \ell_{G}(e) + C(e)$ for each $e\in E(G)$. A \emph{$h$-length $\phi$-expander decomposition w.r.t. terminals $T$} is an integral moving cut $C$ s.t. $|C|\leq O(\phi\cdot h\cdot |T|)$ and $G-C$ is an $h$-length $\phi$-expander with respect to $T$. 

A \emph{certified $h$-length $\phi$-expander decomposition} (called $(h,\phi)$-certified-ED for short), denoted by $(C,{\cal N},{\cal R},\Pi_{{\cal R}\to G})$, is an object containing both the expander decomposition $C$ and a certificate of $G-C$, except that the embedding $\Pi_{{\cal R}\to G}$ has $h$-length embedding paths on $G$ rather than $G-C$. In other words, an $(h,\phi)$-certified-ED is only a weak form of expander decomposition, which guarantees some kind of flow characterization about $G$ and $C$ but does not certify that $G-C$ is an $h$-length $\phi$-expander. In particular, the routers ${\cal R}$ and embedding $\Pi_{{\cal R}\to G}$ will now guarantee that each cluster $S\in{\cal N}$ has diameter $O(h)$ on $G$ (instead of $G-C$), i.e. $\dist_{G}(u,v)\leq O(h)$ for each $u,v\in S$, so in this context, we say the \emph{real diameter} of ${\cal N}$ is $O(h)$.

\paragraph{The Expander Hierarchy and the Query Algorithm.} Let $G$ be the input graph with edge length function $\ell_{G}$. Suppose we already have a (static) \emph{$h$-length expander hierarchy} with $\bar{k} = O(1/\epsilon^{2})$ levels, which is defined in an algorithmic way as follows.

We let $h_{1} = 4 h$ and let $h_{k} = O(h_{k-1})$ with sufficiently large hidden constant factor for each level $2\leq k\leq \bar{k}$. Fix $\phi = 1/(n^{\epsilon^{2}}\cdot h_{\bar{k}})$. We will grow the hierarchy in a bottom-up manner. Let $T_{1} = V(G)$ be the first-level terminal set. Starting from the first level, for each level $k\geq 1$, we construct an $(h_{k},\phi)$-certified-ED $(C_{k},{\cal N}_{k},{\cal R}_{k},\Pi_{{\cal R}_{k}\to G})$ of $T_{k}$ on $G$, and then define the next-level terminals to be $T_{k+1} = V(\supp(C_{k}))$, which includes all vertices incident to edges $e$ with $C_{k}(e)>0$. The hierarchy stops growing when we reach a level $k$ with empty $T_{k+1}$ (which implies $C_{k}$ is zero at all edges), and this is the highest level, denoted by $\bar{k}$. Note that number of levels can be bounded by $\bar{k}\leq O(1/\epsilon^{2})$ because $|T_{k+1}|\leq O(|C_{k}|)\leq O(\phi\cdot h_{k}\cdot |T_{k}|)\leq O(|T_{k}|/n^{\epsilon^{2}})$.

Given two vertices $u,v\in V(G)$, the query algorithm should correctly declare either (\textsc{Close}) $\dist_{G}(u,v)\leq 2^{O(\bar{k})}\cdot h$ or (\textsc{Far}) $\dist_{G}(u,v)>h$. The algorithm is as follows. We climb up the hierarchy from the first level. Suppose at level $k$, we have two vertices $u_{k},v_{k}\in T_{k}$ as the input (the initial input satisfies $u,v\in T_{1} = V(G)$ at the first level). Then we will perform the following steps. 
\begin{enumerate}
\item If there exists a cluster $S\in{\cal N}_{k}$ s.t. $u_{k},v_{k}\in S$ (which will hold if $\dist_{G-C_{k}}(u_{k},v_{k})\leq h_{k}$ by the covering radius of ${\cal N}_{k}$), then we will output \textsc{Close}.
\item Otherwise, we have $\dist_{G-C_{k}}(u_{k},v_{k})>h_{k}$ by the covering radius of ${\cal N}_{k}$. If now $k=\bar{k}$ is the top level, we will output \textsc{Far}. 
\item Otherwise, we have $k\leq \bar{k}-1$. We pick an arbitrary $u_{k+1}\in T_{k+1}$ s.t. there exists a cluster $S\in {\cal N}_{k}$ s.t. $u_{k},u_{k+1}\in S$, and we pick $v_{k+1}$ similarly. If there is no such $u_{k+1}$ or $v_{k+1}$, we will output \textsc{Far}, otherwise proceed to the next level. 
\end{enumerate}

If the algorithm outputs \textsc{Close} in step 1 at some level $k^{\star}$, then indeed $\dist_{G}(u,v)\leq 2^{O(\bar{k})}\cdot h$ by the following reasons. At each level $1\leq k\leq k^{\star}-1$, we pick another $u_{k+1}$ (resp. $v_{k+1}$) that is quite close to $u_{k}$ (resp. $v_{k}$). Precisely, we have $\dist_{G}(u_{k},u_{k+1}),\dist_{G}(v_{k},v_{k+1})\leq O(h_{k})$ by the real diameter of ${\cal N}_{k}$. Moreover, level $k^{\star}$ guarantees that $\dist_{G}(u_{k^{\star}},v_{k^\star})\leq O(h_{k^{\star}})$, i.e. it is at most the real diameter of ${\cal N}_{k^{\star}}$. Therefore, $\dist_{G}(u,v)\leq O(h_{k^{\star}}) + 2\sum_{1\leq k\leq k^{\star}-1}O(h_{k}) \leq 2^{O(\bar{k})}$.

Suppose the initial input $u,v$ has $\dist_{G}(u,v)\leq h$. We claim that the algorithm must end with \textsc{Close} by the following reasons. Let $k$ be the current level. If the algorithm ends at step 1, it is good. Otherwise, we assume $\dist_{G}(u_{k},v_{k})\leq h_{k}/4$, which can be shown by induction as we will see (it holds at level $1$ trivially). Then, we know the algorithm will not output \textsc{Far} at step 2, because when $k=\bar{k}$, we have $C_{k}$ is a zero function and $\dist_{G-C_{k}}(u_{k},v_{k}) = \dist_{G}(u_{k},v_{k}) \leq h_{k}/4$, contradicting the statement of step 2. About step 3, consider the shortest $u_{k}$-to-$v_{k}$ path $P$ on $G$. The key observation is that $\ell_{G}(P)\leq h_{k}/4$ but $\ell_{G-C_{k}}(P)>h_{k}$. Hence, there must exist $u_{k+1}\in T_{k+1}\cap P$ s.t. $\dist_{G-C_{k}}(u_{k},u_{k+1})\leq h_{k}/4$.\footnote{For example, $u_{k+1}$ can be the $P$-vertex closest to $u_{k}$ and incident to some edge $e\in P$ with $C_{k}(e)\geq 0$. Then $u_{k+1}$ satisfies that $u_{k+1}\in V(\supp(C_{k}))\cap P = T_{k+1}\cap P$ and $\dist_{G-C_{k}}(u_{k},u_{k+1}) = \dist_{G}(u_{k},u_{k+1})\leq \ell_{G}(P)\leq h_{k}/4$.} We can find such $u_{k+1}$ successfully because ${\cal N}_{k}$ has covering radius $h_{k}$. The same argument holds for $v_{k+1}$, so the algorithm will not end at step 3. Finally, to complete the induction, we have $\dist_{G}(u_{k+1},v_{k+1})\leq \dist_{G}(u_{k+1},u_{k}) + \dist_{G}(u_{k},v_{k}) + \dist_{G}(v_{k},v_{k+1})\leq O(h_{k}) + h_{k}/4 + O(h_{k})\leq h_{k+1}/4$, where the first inequality is by the real diameter of ${\cal N}_{k}$ and the second inequality is from the definition of $h_{k+1}$.

\paragraph{Obstacles to Dynamizing the Hierarchy.} This warm-up scenario motivates us to dynamize the length-constrained expander hierarchy to obtain dynamic constant-approximate distance oracles. However, there are some obstacles along the way. 

The primary task is to develop a \emph{dynamic certified-ED} algorithm. However, having dynamic certified-ED algorithms is not sufficient to make the hierarchy dynamic because of the following inherent obstacle. To bound the number of levels in the hierarchy, we need that at each level $k$, the size of next-level terminal set will drop by a factor $n^{\epsilon^{2}}$. Hence, intuitively, when $|T_{k+1}|$ becomes too large compared to $|T_{k}|/n^{\epsilon^{2}}$ (this will happen when the current update batch $\pi^{(i)}$ is too large, say $|\pi^{(i)}|$ reaches $|T_{k}|$, and quite a lot of new terminals are added into $T_{k}$ by the dynamic certified-ED algorithm), we have to reinitialize the dynamic certified-EDs on this level and above. However, it is unaffordable to use roughly $O(|G|)$ time to initialize it in a batched update with size $|T_{k}|$, because $|T_{k}|$ can be much smaller than $|G|$.

Besides, our dynamic certified-ED algorithm has a technical limitation. That is, we can only maintain certified-ED w.r.t. the \emph{whole vertex set} $V(G)$ instead of a subset of terminals $T$ (we will see the reason in the next subsection).

To overcome both obstacles, we develop a \emph{dynamic vertex sparsifier} (preserving pairwise distance between terminals) algorithm. Before constructing the certified-ED w.r.t. $T_{k+1}$ at the next level, we first construct a vertex sparsifier $H_{k+1}$ of $T_{k+1}$ on $H_{k}$, s.t. $|H_{k+1}|$ is roughly the same as $|T_{k+1}|$. At the next level $k+1$, we maintain a certified-ED on $H_{k+1}$ w.r.t. $V(H_{k+1})\supseteq T_{k+1}$. Intuitively, doing this will keep reducing the size of the working graph $H_{k}$ when we climb up, and now the reinitialization at level $k$ takes only $O(|H_{k}|)$ time instead of $O(|G|)$.

\paragraph{Obstacles to Going Beyond Bounded-Distance.} The last issue is that the batched update time of our dynamic certified-EDs and sparsifiers will depend on $\poly(h)$. If we want $n^{O(\epsilon)}$ update time, it means we can only obtain \emph{dynamic distance oracles for bounded distances} that support query on $u,v$ when the real distance between $u,v$ is bounded by $h=n^{\epsilon}$.

To bypass this, we develop a \emph{dynamic length-reducing emulator} algorithm. Roughly speaking, an $h$-length-reducing emulator $Q$ of $G$ is a unit-edge-length graph with distance metric $\dist_{Q}(\cdot,\cdot)\approx \dist_{G}(\cdot,\cdot)/h$ for $h=n^{\epsilon}$. We can construct an emulator by exploiting the expander hierarchy again. The key intuition is that, for two vertices $u,v$ with $\dist_{G}(u,v)\leq h$, we only need to jump up level by level to reach two vertices $u_{k^{\star}},v_{k^{\star}}$ which are inside some cluster $S\in{\cal N}_{k^{\star}}$, as discussed above. Furthermore, each jump from $u_{k}$ to $u_{k+1}$ (or $v_{k}$ to $v_{k+1}$) is also jumping inside a cluster $S\in{\cal N}_{k}$. Therefore, for each cluster $S\in{\cal N}_{k}$, we construct a star graph $Q^{\star}_{S}$ with vertices $S\cap (T_{k}\cup T_{k+1})$ as leaves, and define the emulator to be $Q = \bigcup_{S\in{\cal N}_{k},1\leq k\leq \bar{k}} Q^{\star}_{S}$. Then each vertices $u,v$ with $\dist_{G}(u,v)\leq h$ will have $\dist_{Q}(u,v) \leq 2\cdot\bar{k}$, i.e. the distance drops by roughly an $h$ factor.

Providing the dynamic emulators, a standard \emph{stacking} technique will construct a chain of graphs $G_{1},G_{2},...,G_{\bar{x}}$ s.t. $G_{x}$ is an emulator of $G_{x-1}$, and $\bar{x} = \log_{h}\poly(n) = O(1/\epsilon)$ is sufficiently large (so the last graph $G_{\bar{x}}$ has maximum distance at most $h = n^{\epsilon}$). Then, by maintaining a bounded-distance oracle on each of $G_{x}$, we can query a pair $(u,v)$ with general distance by looking at the appropriate $G_{x}$.

\subsection{Dynamic Certified Expander Decomposition}
As a prerequisite step towards our dynamic expander hierarchy, we need to discuss how to maintain a certified $h$-length $\phi$-expander decomposition $(C,{\cal N},{\cal R},\Pi_{{\cal R}\to G})$ on $G$ (w.r.t. $V(G)$). 
One of the tools we develop is a \emph{dynamic router} algorithm in \Cref{sect:DynamicRouter}, which supports matching insertion and batched edge deletion to the router, where the former will prune some vertices out (to ensure the remaining graph is still a router) and the latter will simply add the matching (the graph is still a router after adding the matching). %

In the dynamic certified-ED algorithm, suppose that we want to apply a batched edge deletion update $\pi^{(i)}$ to $G$. One of the key ideas is to exploit the low-congestion embedding $\Pi_{{\cal R}\to G}$. Precisely, removing edges in $\pi^{(i)}$ will only destroy at most $O(|\pi^{(i)}|/\phi)$ many embedding paths, because $\Pi_{{\cal R}\to G}$ has congestion $O(1/\phi)$. We then update our dynamic routers in ${\cal R}$ with a batched edge deletion including all the router edges corresponding to destroyed embedding paths. The dynamic router algorithm will generate a set of pruned vertices $V_{\prune}\subseteq V(G)$ with size roughly $O(|\pi^{(i)}|/\phi)$. One can easily verify that after removing $V_{\prune}$, the remaining $(C',{\cal N}',{\cal R}',\Pi_{{\cal R}'\to G})$ is a certified-ED w.r.t. $V(G)\setminus V_{\prune}$.

The next step is to add $V_{\prune}$ back to $(C',{\cal N}',{\cal R}',\Pi_{{\cal R}'\to G})$. Let $V' = V(G)\setminus V_{\prune}$. In this step we will exploit a subroutine computing 
\emph{cutmatch} between $V_{\prune}$ and $V'$, and we design a \emph{local cutmatch} algorithm with local running time $O(|V_{\prune}|)$. Roughly speaking, the cutmatch between $V_{\prune}$ and $V'$ will partition $V_{\prune},V'$ into matching parts $V_{\prune,M}, V'_{M}$ and unmatched parts $V_{\prune,U},V'_{U}$, s.t. vertices in $V_{\prune,M}$ are matched to $V'_{M}$ with length $h/3$ and congestion $1/\phi$, but $V_{\prune,U}$ is $h/3$-far from $V'_{U}$ by adding an additional cut $C_{\CM}$ with size at most $O(\phi\cdot h\cdot |V_{\prune,M}|) = O(h\cdot |\pi^{(i)}|)$. Then intuitively, we can simply add $V_{\prune,M}$ back to $(C',{\cal N}',{\cal R}',\Pi_{{\cal R}'\to G})$, which leads to slightly worse quality parameters $h$ and $\phi$ (because we only need to handle constant many batched updates, so we are fine with the loss). For $V_{\prune,U}$, we can take care of it by initializing a new certified-ED on the \emph{local graph} $G[V_{\prune}\cup V'_{M}]$, because $V_{\prune,U}$ is kind of far from $V'_{U}$. Lastly we just need to compose the two certified-EDs of these two parts.

Dealing with batched edge insertion update is much more simpler, because we can assign cut value $h$ to each new edge. Intuitively, this will block all new edges, so the certified-ED after adding this cut still have the same quality $h$ and $\phi$ on the new graph. 

Recall that we mentioned the embedding $\Pi_{{\cal R}\to G}$ in a certified-ED is on $G$ rather than $G-C$, and now we have an explanation: we will add new cuts during the maintenance. Furthermore, We cannot maintain certified-EDs for a subset of terminals because the local cutmatch algorithm requires $V_{\prune}\cup V' = V(G)$. The update time of the dynamic certified-ED algorithm is roughly $|\pi^{(i)}|\cdot \poly(h)/\phi$. The dependency on $\poly(h)$ is because the initialization algorithm of certified-ED and the local cutmatch algorithm have running time depending on $\poly(h)$.

\paragraph{Local Length-Constrained Maxflow.}

In more details, to obtain our local cutmatch algorithm, actually we localize the length-constrained approximate maxflow algorithm in \cite{haeupler2023maximum}, which may be of independent interest. A key step is that we give a black-box localization of a key subroutine called $h$-length lightest path blockers in \cite{haeupler2023maximum}. Then, we can obtain the localized algorithms of length-constrained maximum flow and cutmatch because the previous reductions to lightest path blockers in \cite{haeupler2023maximum} are easy to localize.

\paragraph{Reduce the Recourse Using Landmarks.} By our argument above, handling each batched update $\pi^{(i)}$ will add a new cut with size roughly $O(h\cdot |\pi^{(i)}|)$. However, this factor $h$ in the size of the new cut is a trouble by the following reason. Recall our construction of the expander hierarchy. If a cut $C_{k}$ has its size increased by $O(h\cdot |\pi^{(i)}|)$, then in the worst case the number of new terminals added into $T_{k}$ can be roughly $h\cdot |\pi^{(i)}|$ (in other words, the terminal set $T_{k}$ defined via cut $C_{k}$ will have recourse $h$). Then unavoidably, when we dynamize the hierarchy using dynamic vertex sparsifiers, this factor $h$ will aggregate multiplicatively when we climb up, and the sparsifier at the highest level may have recourse $h^{O(1/\epsilon^{2})} = (n^{\epsilon})^{O(1/\epsilon^{2})} = \Omega(n)$, which makes it impossible to achieve fast update time. In fact, this factor $h$ in the recourse will bring the same issue when we use the dynamic hierarchy to design the dynamic emulator, because we will work on $O(1/\epsilon)$ levels of emulators at last.

To solve this issue, we introduce the notion of \emph{landmarks}. Roughly speaking, a landmark set $L$ with distortion $\sigma$ of an integral moving cut $C$ on $G$ is a subset of $V(G)$ such that for each vertex $v\in V(\supp(C))$ (called a $C$-vertex), there is a landmark $w\in L$ with $\dist_{G-C}(u,w)\leq \sigma$. When maintaining an $(h,\phi)$-certified-ED, we will also maintain the landmark set $L$ of $C$ on $G$ with distortion $\sigma \leq h/n^{\poly\epsilon}$. In other words, the representation of the certified-ED now becomes $(C,L,{\cal N},{\cal R},\Pi_{{\cal R}\to G})$. Given $C$, we will be able to initialize a landmark set $L$ with size $|L|\leq O(|C|/\sigma)\approx O(|C|/h)$, and furthermore, when the dynamic certified-ED algorithm adds a new cut $C_{\new}$, we can also extend the landmark set to keep it valid, by adding $O(|C_{\new}|/\sigma)\approx O(|C_{\new}|/h)$ new landmarks. We are not going to introduce the initialization and maintenance of landmark sets here, but it is actually non-trivial. 

Providing the landmark set, the way to remove the $h$ factor in the recourse of sparsifers in the hierarchy is simply using the landmarks $L_{k}$ as terminals $T_{k}$ instead of $C_{k}$-vertices (i.e. set $T_{k} = L_{k}$). This will not affect the correctness of the query algorithm since each $C_{k}$-vertex is quite close to a landmark on $G-C$. We will discuss this again in details after we introducing dynamic vertex sparsifiers.

\subsection{Dynamic Vertex Sparsifiers for Bounded Distances}

Next, we discuss our dynamic vertex sparsifier algorithm. Let $T$ be the terminal set. We say a graph $H$ with edge length $\ell_{H}$ is a $(\alpha_{\low},\alpha_{\up},h)$-sparsifier if (1) $T\subseteq V(H)$, (2) each pair $u,v\in T$ has $\dist_{H}(u,v)\cdot\alpha_{\low}\geq \dist_{G}(u,v)$, and (3) each pair $u,v\in T$ s.t. $\dist_{G}(u,v)\leq h$ has $\dist_{H}(u,v)\leq \alpha_{\up}\cdot \dist_{G}(u,v)$. That is, $H$ approximately preserves pairwise distances between all terminals whose distance in $G$ is at most $h$.

We first show how to construct a sparsifier $H$ statically, given a collection of $(h_{j},\phi)$-certified-EDs $(C_{j},L_{j},{\cal N}_{j},{\cal R}_{j},\Pi_{{\cal R}_{j}\to G})$ with $h_{j} = 2^{j}$, where $j$ has range from $1$ to $\bar{j} = \log h = O(\log n)$. Let $\bar{T} = T\cup \bigcup_{j}V(\supp(C_{j}))$ be the union of original terminals and $C_{j}$-vertices (vertices that incident to some edge $e$ with $C_{j}(e)>0$) of all certified-EDs ($\bar{T}$ is called the \emph{extended terminals}). Then for each pairwise cover ${\cal N}_{j}$ and each cluster $S\in{\cal N}_{j}$, we construct a star graph $H^{\star}_{S}$ with an artificial vertex $v_{S}$ as the center and vertices in $\bar{T}\cap S$ as leaves, and we assign length $h_{j}$ to all edges in $H^{\star}_{S}$ Then we define $H = \bigcup_{j}\supp(C_{j})\cup\bigcup_{S\in{\cal N}_{j},1\leq j\leq \bar{j}}H^{\star}_{S}$.

Now we show the correctness. First the lower bound side holds for $\alpha_{\low} = O(1)$. The shortest $u$-$v$ path on $H$ can be decomposed into subpaths, each of which is either a $2$-hop path inside some $H^{\star}_{S}$ or a $1$-hop path with the single edge inside some $\supp(C_{j})$. For each subpath $P'\subseteq H^{\star}_{S}$ for $S\in{\cal N}_{j}$ (with endpoints $u',v'$), we have $\ell_{H^{\star}_{S}}(P') = 2\cdot h_{j}$ and $\ell_{G}(u',v')\leq O(h_{j})$ (the diameter of ${\cal N}_{j}$ is $O(h_{j})$). 
The upper bound side holds for $\alpha_{\up} = O(1)$. We can decompose the shortest $u$-$v$ path on $G$ into (1) cut edges in $\bigcup_{j}\supp(C_{j})$, and (2) maximal subpaths $P'$ s.t. $P'$ has no cut edges. For each such $P'$ (with endpoints $u',v'$), we take the minimum $j^{\star}$ s.t. $h_{j^{\star}}\geq \ell_{G}(P')$ (so $\ell_{G}(P')\geq h_{j^{\star}}/2$), and then we have $\ell_{H}(u',v')\leq 2\cdot h_{j^{\star}}\leq 4\cdot \ell_{G}(P')$ because $\dist_{G-C_{j^{\star}}}(u',v')\leq \ell_{G-C_{j^{\star}}}(u',v') = \ell_{G}(u',v')$, which is at most the covering radius of ${\cal N}_{j^{\star}}$. 

However, this construction still has the recourse issue when we try to make it dynamic. First, the recourse will depend on $h$, because when a new cut $C_{\new}$ with size $h\cdot |\pi^{(i)}|$ is added by the dynamic certified-ED, we need to add $\supp(C_{\new})$ into $H$, and also let the star graphs include $C_{\new}$-vertices. As mentioned before, our solution is to use the landmarks. Roughly speaking, we define the extended terminals to be $\bar{T} = T\cup\bigcup_{j}L_{j}$. Then $H$ is similarly the union of stars on $\bar{T}$ and some original graph edges (but not exactly $\bigcup_{j}\supp(C_{j})$). However, proving the bounds on the approximation (i.e. $\alpha_{\low}$ and $\alpha_{\up}$) becomes more complicated. We will not explain it here and see \Cref{thm:CoversToEmulators} for details.

\paragraph{Reduce the Recourse Further Using Certified-EDs with Density.}

We are not done yet because the recourse of the sparsifier still depends on $1/\phi$ (we will explain the reason in a moment), which is not allowed because our hierarchy has $\bar{k} = O(\log_{1/\phi} n)$ levels and the each $(k+1)$-level sparsifier is built upon the $k$-level sparsifier. Hence, the $1/\phi$ factor in recourse will blow up to $n$ after all levels, which is too expensive.
This dependency on $1/\phi$ is from the dynamic certified-EDs when handling a batched edge deletion $\pi^{(i)}$. Recall that our algorithm first removes $V_{\prune}$ with size $|V_{\prune}| = O(|\pi^{(i)}|/\phi)$ from the old certified-ED and then add it back. This will bring $O(|\pi^{(i)}|/\phi)$ recourse to the pairwise cover, and so our sparsifier (even the one based on landmarks).

To fix this, we will exploit \emph{certified-EDs with density}, which can be interpreted as a weighted version of certified-EDs. The key observation is that, the recourse of the sparsifier will depend on the recourse of the pairwise covers \emph{restricted on} $\bar{T}$. That is, the sparsifier does not care how many times the vertices outside $\bar{T}$ are removed from or added to the pairwise covers. Therefore, we will assign each extended terminal in $\bar{T}$ a weight $1/\phi$ called its \emph{density} (each vertex not in $\bar{T}$ has density $1$), and map each unit of density to a router vertex, so each extended terminal $v$ in some cluster $S$ now corresponds to $1/\phi$ router vertices in $R^{S}$. Intuitively, we need to remove $v$ from $S$ only if all these $1/\phi$ router vertices are pruned, so $O(|\pi^{(i)}|/\phi)$ many pruned router vertices will only cause $O(|\pi^{(i)}|)$ many extended terminal removals from the pairwise covers. The reason why we have the flexibility to assign density $1/\phi$ to vertices in $\bar{T}$ is that the size of $\bar{T}$ is roughly $|V(G)|/\phi$ and the total density over $V(G)$ is still $O(|V(G)|)$.

We further mention a small technical issue induced by the density. Concretely, because our landmark sets are incremental during maintenance, when a set $L_{\new}$ of new landmarks is added, these vertices in $L_{\new}$ become new extended terminals, and we also need to raise the density of them from $1$ to $1/\phi$. Raising density can be done using cutmatch, but unavoidably this will generate an other set $L'_{\new}$ of new landmarks. To solve this circular situation, we can slightly adjust the parameters of the cutmatch, so that the number of new landmarks decrease by a $n^{\poly\epsilon}$ factor each time, and this will only cause $\poly(1/\epsilon)$ many sub-batched updates.

\subsection{Putting it all Together} We conclude this overview by putting everything together with a brief recourse and update time analysis. For simplicity, in this analysis we bound the recourse and update time amortized inside each batched update, which implies worst-case bounds of total recourse and update time for each batched update. The update time is easy to analyse because it just adds up, but the recourse will propagate multiplicatively if we run dynamic algorithms on top of each other, so we will focus more on the recourse. We will choose $\epsilon\in[1/\log^{c}n, 1]$ for some sufficiently small constant $c$, to make sure that $(\log n)^{\poly(1/\epsilon)}$ is always smaller than $n^{\poly(\epsilon)}$. %

The dynamic certified-ED algorithm (\Cref{thm:DynamicED}) has update time $\poly(h)\cdot n^{O(\epsilon)}/\phi$, and it will bring recourse $n^{O(\epsilon^{3})}/\phi$ to the pairwise cover ${\cal N}$, $2^{O(1/\epsilon)}\cdot h$ to the integral moving cut $C$, and $O(1)$ to the landmark set $L$. For each batched update, the recourse of ${\cal N}$ is the number of vertex insertions/deletions summing over all clusters, the recourse of $C$ is its size increment, and the recourse of $L$ is the number of vertex insertions/deletions.

The dynamic vertex sparsifiers algorithm for bounded distances (\Cref{thm:NonHopReducingEmulator}) has update time $\poly(h)\cdot n^{O(\epsilon)}/\phi^{2}$, and the recourse of the sparsifier $H$ is $n^{O(\epsilon^{4})}$. The recourse of $H$ is roughly the recourse of $L$ plus the recourse of ${\cal N}$ from the dynamic certified-ED algorithm \emph{with density}. Recall that introducing density will reduce the recourse of ${\cal N}$ by a $1/\phi$ factor. The difference between the overheads $n^{O(\epsilon^{3})}$ and $n^{O(\epsilon^{4})}$ is due to some technical details, which is not important because such overheads can be adjusted to be even smaller, e.g. $n^{O(\epsilon^{100})}$, by choosing appropriate tradeoffs for the based subroutines\footnote{These based subroutines include \Cref{thm:NeighborhoodCovers}, \Cref{thm:WitnessedED}, \Cref{thm:RouterRouting}, \Cref{thm:cutmatch} and \Cref{thm:Router}.}. 

The dynamic length-constrained expander hierarchy (\Cref{thm:ExpanderHierarchy}) has $\bar{k} = O(\log_{1/\phi} n) = O(1/\epsilon^{2})$ levels by setting $\phi = n^{\epsilon^{2}}$. At each level $k$, we maintain a certified-ED $(C_{k},L_{k},{\cal N}_{k},{\cal R}_{k},\Pi_{{\cal R}_{k}\to H_{k}})$ (with landmarks $L_{k}$ but without density) on sparsifer $H_{k}$ (at the first level $H_{1} = G$), and then maintain a sparsifier $H_{k+1}$ (for the next level) of $L_{k}$ on $H_{k}$. The recourse of each $H_{k}$ is at most $n^{O(\epsilon^{4}\cdot\bar{k})} = n^{O(\epsilon^{2})}$. The recourse of each ${\cal N}_{k}$ is at most $n^{O(\epsilon^{2})}\cdot n^{O(\epsilon^{3})}/\phi = n^{O(\epsilon^{2})}$ (i.e. the recourse of $H_{k}$ times the recourse of ${\cal N}$ from dynamic certified-EDs). The update time is $\poly(h)\cdot n^{O(\epsilon)}/\phi^{2} = \poly(h)\cdot n^{O(\epsilon)}$.

The dynamic length-reducing emulator (\Cref{thm:HopReducingEmulator}) is based on the dynamic hierarchy. From the construction, the recourse of the emulator $Q$ is proportional to the total recourse of all pairwise covers in the hierarchy, which is $n^{O(\epsilon^{2})}/\epsilon^{2} = n^{O(\epsilon^{2})}$. The update time is $\poly(h)\cdot n^{O(\epsilon)}$.

We will stack dynamic emulator on top of each other and the number of levels is $\bar{x} = O(\log_{h}n) = O(1/\epsilon)$ by setting $h = n^{\epsilon}$. The recourse of each emulator is at most $n^{O(\epsilon^{2}\cdot\bar{x})} = n^{O(\epsilon)}$. The update time is $\poly(h)\cdot n^{O(\epsilon)} = n^{O(\epsilon)}$.

Finally, our online-batch dynamic distance oracle is just a collection of dynamic bounded-distance oracles (based on the expander hierarchy) for different length scales $h=2,4,8,...,n^{\epsilon}$ on these emulators. The final update time is $n^{O(\epsilon)}$.

%% file: 3-preliminary.tex
\section{Preliminaries}
\label{sect:prelim}

\subsection{Standard Graph Notations}

In this paper, we consider \textit{undirected} graphs\footnote{The only exception is \Cref{sect:PathBlockers,sect:LCMaxflow}, where we work on directed graphs with notations given at the beginning of \Cref{sect:LocalLengthConstrainedFlow}.}, say $G$. We use $V(G)$ to denote its vertices and $E(G)$ to denote its edges. Let $|G| = |V(G)| + |E(G)|$ denote the size of the graph $G$. If edges in $G$ have length, we denote this by a function $\ell_{G}:E(G)\to \mathbb{N}_{>0}$, where the length $\ell_{G}(e)$ of each edge is a positive and polynomially bounded integer. Let $\dist_{G}(\cdot,\cdot)$ denote the distance metric of $G$. Each path $P$ on $G$ has length $\ell_{G}(P)$, and a path $P$ is \textit{$h$-length} if $\ell_{G}(P)\leq h$. For two vertices $u,v\in V(G)$, we say $u$ and $v$ are \textit{$h$-separated} if $\dist_{G}(u,v)> h$. More generally, for two subset of vertices $S,T\subseteq V(G)$, $S$ and $T$ are $h$-separated on $G$ if for each $u\in S$ and $v\in T$, $u$ and $v$ are $h$-separated. We say $G$ is a \textit{fully dynamic graph} if it undergoes a sequence of unit updates including edge insertions, edge deletions, isolated vertex insertions and isolated vertex deletions.%

\subsection{The Online-Batch Dynamic Setting}
\label{sect:OnelineBatch}

Our final fully dynamic approximate distance oracle in \Cref{thm:main} is for fully dynamic graphs. However, for most of this paper, we will consider \textit{online-batch dynamic graphs} and design an online-batch dynamic distance oracle. Lastly, in \Cref{sect:Reduction}, we will turn the online-batch dynamic oracle into a fully dynamic one by a standard reduction.

Formally speaking, an online-batch dynamic graph $G$ will undergo $t$ \textit{batched updates} $\pi^{(1)},...,\pi^{(t)}$, each of which is a set of \textit{unit updates} including edge insertions, edge deletions, isolated vertex insertions and insolated vertex deletions. Let $G^{(0)}$ be the initial snapshot of $G$ and for each $1\leq i\leq t$, let $G^{(i)}$ be the snapshot of $G$ right after batched update $\pi^{(i)}$. In other words, $G^{(i)}$ is the graph from applying all unit updates in the batched update $\pi^{(i)}$ to $G^{(i-1)}$. More generally, for all online-batch dynamic object (e.g. graphs, sets, data structures), we use the superscript $(0)$ to denote its initial snapshot and ${(i)}$ to denote its snapshot right after the $i$-th batched update. We will define the batched updates for the object when an online-batch dynamic algorithm takes it as input. We emphasize that, when we write a statement of an online-batch dynamic object without the superscript, it means this statement holds at all time $0\leq i\leq t$.

We will analysis the \emph{recourse} of an object in the scenario that it is maintained by an online-batch dynamic algorithm. Formally, consider an online-batch dynamic algorithm with a dynamic object ${\cal O}_{\iin}$ under batched updates $\pi^{(1)},...,\pi^{(t)}$ as the input (generally ${\cal O}_{\iin}$ is a graph) and another object ${\cal O}_{\out}$ as the output. The \emph{recourse} of ${\cal O}_{\out}$ at time $i$, denoted by $\recourse({\cal O}^{(i-1)}_{\out} \to {\cal O}^{(i)}_{\out})$, is the size of the batched update (generated by the algorithm) to update ${\cal O}^{(i-1)}_{\out}$ to ${\cal O}^{(i)}_{\out}$. In our analysis, $\recourse({\cal O}^{(i-1)}_{\out} \to {\cal O}^{(i)}_{\out})$ will be \emph{worst-case} bounded in the form $f(n)\cdot |\pi^{(i)}|$, and $f(n)$ is usually a small polynomial e.g. $n^{O(\epsilon^{4})}$. In fact, this paper only includes recourse analysis for two type of objects: graphs and pairwise covers. The definition of batched updates on graphs is clear as mentioned above, where the batched updates on pairwise covers will be defined right below \Cref{def:PairwiseCover}.

An \textit{online-batch data structure} ${\cal D}$ for an input graph $G$ is such that at each time $0\leq i\leq t$, ${\cal D}^{(i)}$ is for the graph $G^{(i)}$. The running time to initialize ${\cal D}^{(0)}$ is called \textit{the initialization time}. The time to update ${\cal D}^{(i-1)}$ to ${\cal D}^{(i)}$ is called \textit{the update time for $\pi^{(i)}$}. 

For all online-batch dynamic algorithms in this paper, we allow \emph{mixed} input batched updates. That is, each of them may contain more than one type of unit updates. However, when we design and analyse the algorithm, we can assume without loss of generality that all input batched updates are \emph{pure}, i.e. it contains only one type of unit updates. This is because for a mixed batched update, we can simply substitute it with a constant number of pure batched updates. This will only increase the number of batched updates from $t$ to $O(t)$, and will not change our bounds asymptotically. 

\subsection{Global Parameters}
\label{sect:GlobalParameters}

\paragraph{The Input Graph Size $n$.} We use $n$ to denote the maximum number of vertices in the input dynamic graph of \Cref{thm:main} over all updates. We note that all the graphs in this paper will have number of vertices and edges polynomial in $n$.

\paragraph{The Global Tradeoff Parameter $\epsilon$.} We will fix $\epsilon$ to be the one in the input of \Cref{thm:main} throughout the paper, called the \textit{global tradeoff parameter}. Note that we require $\epsilon$ satisfies $1/\log^{c} n\leq \epsilon\leq 1$ for some sufficiently small constant $c>0$, which implies 
\[
(\log n)^{\poly(1/\epsilon)}\leq n^{\poly(\epsilon)},
\]
so we can hide $(\log n)^{\poly(1/\epsilon)}$ factors inside $n^{\poly(\epsilon)}$ in the analysis. We note that we set the lower bound $1/\log^{c}n$ for $\epsilon$ just to make the above inequality holds. In fact, most algorithms in this paper could set tradeoff parameters locally, but we present them with tradeoff parameters controlled by $\epsilon$ just to avoid clutter. %

\paragraph{The Length and Congestion Parameters $h$ and $\phi$.} Most subroutines in this paper will receive parameters $h$ and $\phi$ as input, and we call $h$ the \textit{length parameter} and $\phi$ the \textit{congestion parameter}. The parameters $h$ and $\phi$ of each subroutine are given \textit{locally} from the input. When we use these subroutines as building blocks of our main result \Cref{thm:main}, the subroutines may receive different parameters $h$ and $\phi$, but they are always bounded by $h=n^{O(\epsilon)}$ and $\phi = n^{\Theta(\epsilon^{2})}$.

\paragraph{The Global Time Parameter $t$.} We use $t$ to denote the number of batched update of all online-batch dynamic algorithms in this paper, called the \textit{global time parameter}. When we transform online-batch dynamic algorithms to fully dynamic algorithms, we will set $t = \Theta(1/\epsilon)$. The only exception is \Cref{thm:Router}, whose time parameter $t$ is local and can be larger than $\Theta(1/\epsilon)$.

\paragraph{Additional Global Parameters.} There will be additional global parameters defined in the following sections. We list them in \Cref{table:GlobalParameters} for references. We use $\lambda$ to denote all additional global parameters at most $2^{\poly(1/\epsilon)}$, and use $\kappa$ denote the others (which are at most $n^{O(\epsilon)}$).

\input{table_global_parameters}

\subsection{Pairwise Covers}
\label{sect:PairwiseCovers}

\begin{definition}[Node-Weighting]
Given a graph $G$, a \textit{node-weighting} is a function $A:V(G)\to \mathbb{N}$. Without ambiguity, $A$ also refers to a set of \emph{virtual nodes} s.t. each virtual node in $A$ is owned by a vertex $v\in V(G)$ and each vertex $v\in V(G)$ has exactly $A(v)$ virtual nodes. For each vertex $v\in V(G)$, $A(v)$ also refers to the set of virtual nodes owned by $v$. 

Furthermore, a node-weighting $A$ is \textit{positive} if $A(v)\geq 1$ for each $v\in V(G)$. For an arbitrary vertex set $W\subseteq V(G)$, we use $\mathds{1}(W)$ to denote the node-weighting $A$ with $A(v) = 1$ for each $v\in W$ and $A(v) = 0$ for each $v\notin W$. In fact, the vertex set $W$ and virtual node set $\mathds{1}(W)$ are equivalent and we may use them interchangeably.
\label{def:NodeWeighting}
\end{definition}

For two non-negative functions $g_{1}$ and $g_{2}$ (e.g. node-weightings), we use $g_{1}\uplus g_{2}$ to denote the function with domain $\dom(g_{1}\uplus g_{2}) = \dom(g_{1})\cup \dom(g_{2})$ s.t.
\[
(g_{1}\uplus g_{2})(v) = \left\{
\begin{aligned}
&g_{1}(v),\text{ if }v\in \dom(g_{1})\setminus \dom(g_{2})\\
&g_{2}(v),\text{ if }v\in \dom(g_{2})\setminus \dom(g_{1})\\
&g_{1}(v) + g_{2}(v),\text{ if }v\in \dom(g_{1})\cap\dom(g_{2}).
\end{aligned}
\right.
\]

\begin{definition}[Pairwise cover]
Given a graph $G$ with node-weighting $A$, an \emph{pairwise cover} ${\cal N}$ of $A$ on $G$ is a collection of \emph{clusterings}, where each clustering ${\cal S}\in{\cal N}$ is a collection of disjoint \emph{clusters}, and each cluster $S\in{\cal S}$ is a subset of virtual nodes in $A$, i.e. $S\subseteq A$. For simplicity, we also write $S\in{\cal N}$ as a short form of $S\in{\cal S}\in{\cal N}$. A pairwise cover ${\cal N}$ has covering radius $h_{\cov}$, separation $h_{\sep}$, diameter $h_{\diam}$ and width $\omega$ (called \emph{quality parameters}), if ${\cal N}$ satisfies the following properties.
\begin{itemize}
\item (covering radius $h_{\cov}$) For each pair of virtual node $u,v\in A$ with $\dist_{G}(u,v)\leq h_{\cov}$, there exists a cluster $S\in{\cal N}$ s.t. $u,v\in S$.
\item (Separation $h_{\sep}$) For each clustering ${\cal S}\in{\cal N}$, the distance between any two clusters $S,S'\in{\cal S}$ is larger than $h_{\sep}$. Precisely, for each $u\in S$ and $v\in S'$, $\dist_{G}(u,v)> h_{\sep}$.
\item (Diameter $h_{\diam}$) For each cluster $S\in{\cal N}$, its diameter on $G$ is at most $h_{\diam}$. Precisely, for each $u,v\in S$, $\dist_{G}(u,v)\leq h_{\diam}$.
\item (Width $\omega$) The number of clusterings ${\cal S}$ in ${\cal N}$ is at most $\omega$.
\end{itemize}
We require that ${\cal N}$ should cover all virtual nodes in $A$, i.e. $\bigcup_{S\in{\cal N}} S = A$. For better understanding, if some $u\in A$ has $\dist_{G}(u,v)>h_{\cov}$ for all other $v\in A$, then there exists a singleton cluster $S=\{u\}\in {\cal N}$.
The \textit{size} of ${\cal N}$, denoted by $\size({\cal N})$, is the total size of clusters $S\in{\cal N}$, i.e. $\size({\cal N}) = \sum_{S\in{\cal N}}|S|$.

In particular, if ${\cal N}$ is a pairwise cover of node-weighting $\mathds{1}(W)$ on $G$ for some vertex set $W\subseteq V(G)$, equivalently, we say ${\cal N}$ is a pairwise cover of \emph{vertex set} $W$ on $G$. For an arbitrary node-weighting $A^{\star}$, the \textit{restriction} of ${\cal N}$ on $A^{\star}$, denoted by ${\cal N}_{\mid A^{\star}}$, is the pairwise cover obtained by restricting every cluster $S\in{\cal N}$ on $A^{\star}$ (i.e. substituting $S$ with $S\cap A^{\star}$). 

\label{def:PairwiseCover}
\end{definition}

When a pairwise cover ${\cal N}$ is updated to $\wtilde{\cal N}$ by some algorithm, the \emph{recourse from ${\cal N}$ to $\wtilde{\cal N}$}, denoted by $\recourse({\cal N}\to\wtilde{\cal N})$, is the number of virtual node insertions and deletions\footnote{A virtual node insertion can add a virtual node into an arbitrary cluster or create a singleton cluster (i.e. a cluster with only one virtual node), while a virtual node deletion can remove a virtual node from an arbitrary cluster or remove a singleton cluster.} in this update. Furthermore, given an arbitrary node-weighting $A^{\star}$, the \emph{recourse from ${\cal N}$ to $\wtilde{\cal N}$ restricted on $A^{\star}$}, denoted by $\recourse_{\mid A^{\star}}({\cal N}\to\wtilde{\cal N})$, is the number of virtual node insertions and deletions involving virtual nodes in $A^{\star}$. In other words, $\recourse_{\mid A^{\star}}({\cal N}\to\wtilde{\cal N}) = \recourse({\cal N}_{\mid A^{\star}}\to\wtilde{\cal N}_{\mid A^{\star}})$.

Before going into \Cref{def:DistributedNC}, we define the notions of \emph{neighborhood/ball}. Let $G$ be a graph with a node-weighting $A$. For each virtual node $v\in A$ and a length parameter $h$, we let $\Ball_{G,A}(v,h) = \{u\in A\mid \dist_{G}(u,v)\leq h\}$ denote the \emph{$h$-neighborhood of $v$ on $G$ w.r.t. $A.$}.

\begin{definition}[Distributed Neighborhood Covers]
Given a graph $G$ with node-weighting $A$, a pairwise cover ${\cal N}$ of $A$ on $G$ is further a \textit{$b$-distributed neighborhood cover} for covering radius $h_{\cov}$ ($b$ is called the \textit{distributed factor}) if for each virtual node $v\in A$, there exists a set ${\cal S}_{v}\subseteq {\cal N}$ of at most $b$ s.t. $v$ is inside every $S\in{\cal S}_{v}$ and $v$'s $h_{\cov}$-neighborhood $\Ball_{G,A}(v,h_{\cov})$ is contained by $\bigcup_{S\in{\cal S}_{v}} S$. We call such ${\cal S}_{v}$ a \textit{ball cover} of $v$, and ${\cal S}_{v}$ is stored explicitly. A $1$-distributed neighborhood cover is called a \textit{neighborhood cover}.
\label{def:DistributedNC}
\end{definition}

\begin{observation}
Let ${\cal N}$ be a pairwise cover of some node-weighting $A$ on some graph $G$. For any subgraph $\wtilde{G}\subseteq G$ s.t. $V(\wtilde{G}) = V(G)$ and sub-node-weighting $\wtilde{A}\subseteq A$, the restriction $\wtilde{\cal N}$ of ${\cal N}$ on $\wtilde{A}$ has the same covering radius, separation, and width as those of ${\cal N}$. If ${\cal N}$ is further a $b$-distributed neighborhood cover for some covering radius $h_{\cov}$, then $\wtilde{\cal N}$ has the same distributed factor $b$ for this covering radius $h_{\cov}$.
\label{ob:CoversOnSubgraph}
\end{observation}
\begin{proof}
The covering radius is the same because for each $u,v\in \wtilde{S}\in\wtilde{\cal N}$, $\dist_{G}(u,v)\leq \dist_{\wtilde{G}}(u,v)$, which implies $u,v\in S\in{\cal N}$. The separation is the same because for each clustering ${\cal S}'\in \wtilde{\cal N}'$, any pair of virtual nodes $u,v$ in different clusters of ${\cal S}'$ has $\dist_{\wtilde{G}}(u,v)\geq \dist_{G}(u,v)$ at least the separation of ${\cal N}$. The width is trivally the same. ${\cal N}$ has the same distributed factor for $h_{\cov}$ because for each $v\in\wtilde{A}$, $\Ball_{\wtilde{G},\wtilde{A}}(v,h_{\cov})\subseteq \wtilde{A}\cap \Ball_{G,A}(v,h_{\cov})$.
\end{proof}

\begin{theorem}[Algorithmic Neighborhood Covers, \cite{ABCP98}]
Let $G$ be a graph. Given parameters $h$ and $\beta$, there is an algorithm that computes a neighborhood cover ${\cal N}$ of all vertices $V(G)$ on $G$ with covering radius $h_{\cov} = h$, diameter $h_{\diam} = O(\beta\cdot h)$ and width $\omega = O(\beta\cdot n^{1/\beta})$. In particular, for each vertex $v\in V(G)$, its ball cover $S_{v}\in {\cal N}$ covering $\Ball_{G,V(G)}(v,h)$ will be specified. The running time is $O(|G|\cdot \beta^{2}\cdot n^{1/\beta})$.
\label{thm:ClassicNC}
\end{theorem}

\begin{theorem}[Neighborhood Covers with Separation]
Given a graph $G$ with a node-weighting $A$ and length parameter $h$, there is an algorithm that computes a neighborhood cover of $A$ on $G$ with covering radius $h_{\cov} = h$, separation $h_{\sep} = h$, diameter $h_{\diam}=\lambda_{\NC,\diam}\cdot h_{cov}$ and width $\omega = \kappa_{\NC,\omega}$, where $\kappa_{\NC,\omega} = n^{O(\epsilon^{4})}$ and $\lambda_{\NC,\diam} = O(1/\epsilon^{4})$. Furthermore, if a cluster $S\in{\cal N}$ includes a virtual node of some vertex $v$, then $S$ includes all virtual nodes of $v$.
The running time is $(|G| + |A|)\cdot n^{O(\epsilon)}$.
\label{thm:NeighborhoodCovers} 
\end{theorem}

\begin{proof}
We first construct a neighborhood cover ${\cal N}'$ of $V(G)$ on $G$ with $h'_{\cov} = 100\cdot h$, $h'_{\diam} = O(h/\epsilon^{4})$ and width $\omega' = n^{O(\epsilon^{4})}$ by applying \Cref{thm:ClassicNC} with parameter $h' = 100\cdot h$ and $\beta = 1/\epsilon^{4}$. This step takes $|G|\cdot n^{O(\epsilon)}$ time.

We do the following for each clustering ${\cal S}'\in{\cal N}'$. For each cluster $S'\in{\cal S}'$, let $\Core(S') = \{v\in V(G)\mid \text{$S'$ is the ball cover of $v$}\}$. Note that for each $v\in \Core(S')$, we have $\Ball_{G,V(G)}(v,h')\subseteq S$ by the definition of ball covers. We define a new cluster $S = \bigcup_{v\in \Core(S')}\Ball_{G,V(G)}(v,h)$, which is the union of $h$-neighborhood of the core of $S'$. Obviously $S\subseteq S'$. Such new clusters $S$ of all old clusters $S'$ in this clustering can be computed by running single-source shortest path algorithm (with all cores together as the source), which takes $\wtilde{O}(|G|)$ time. Let ${\cal S}$ collect all new clusters $S$ when processing the old clustering ${\cal S}'$. Then we can define a new neighborhood cover ${\cal N}$ of $V(G)$ by taking all the new clustering ${\cal S}$.

Obviously, ${\cal N}$ has $h_{\cov} = h$, $h_{\diam} = h'_{\diam}$ and $\omega = \omega'$. The ball cover of each $v\in V(G)$ is just the new cluster corresponding to $v$'s old ball cover. It remains to show the separation of ${\cal N}$ is $h_{\sep} = h$. Consider two clusters $S_{1},S_{2}$ in some clustering ${\cal S}\in{\cal N}$. Let $S'_{1},S'_{2}\in{\cal S'}$ be the old clusters corresponding to $S_{1}$ and $S_{2}$. Assume for the contradiction that, there exists $u\in S_{1}$ and $v\in S_{2}$ s.t. $\dist_{G}(u,v)\leq h$. Note that $v\in \Ball_{G,V(G)}(v',h)$ for some $v'\in \Core(S'_{2})$ by the construction of $S_{2}$. Therefore, $u\in \Ball_{G,V(G)}(v',2h)\subseteq \Ball_{G,V(G)}(v',h')\subseteq S'_{2}$, a contradiction because $u\in S_{1}\subseteq S'_{1}$ and $S'_{1}$ is disjoint with $S'_{2}$.

Finally, to obtain the neighborhood cover of $A$, just simply replace each $S\in{\cal N}$ with $S\cap A$.
\end{proof}

\subsection{Routers}

\begin{definition}[Demand]
Given a graph $G$, a \textit{demand} $D:\binom{V}{2}\to \mathbb{R}_{\geq 0}$ assigns a nonnegative value to each unordered pair of vertices $(u,v)$. For convenience, both $D(u,v)$ and $D(v,u)$ denote the demand between $u$ and $v$. Each pair $d=(u,v)\in\supp(D)$ is called a \textit{demand pair}. \begin{itemize}
\item For a node-weighting $A$ on $G$, $D$ is \textit{$A$-respecting} if for each vertex $v$, $\sum_{u\in V(G)}D(v,u)\leq A(v)$ holds. In particular, We say $D$ is a $\mathds{1}(V(G))$-respecting demand if $\sum_{u\in V(G)} D(u,v)\leq 1$ for each $v\in V(G)$. 
\item For a length parameter $h$, $D$ is \textit{$h$-length} if each pair of vertices $u,v\in V(G)$ s.t. $D(u,v)>0$ has $\dist_{G}(u,v)\leq h$.
\end{itemize}
A demand $D$ is integral if all demand pairs $(u,v)$ have integral $D(u,v)$.
\end{definition}

\begin{definition}[Multicommodity Flows/Routing]
Given a graph $G$, a \textit{(multicommodity) flow} on $G$ is a function $f$ that assigns each simple path $P$ on $G$ a nonnegative real number $f(P)\geq 0$. a simple path $P$ is a \emph{flow path} of $f$ if $f(P)>0$. We let $\path(f)$ be the set collecting all flow paths of $f$, and let $|\path(f)|$ denote the number of flow paths, called \emph{path count}. The flow $f$ is integral if all $P\in\path(f)$ has integral value $f(P)$. 
\begin{itemize}
\item The flow $f$ has \textit{congestion} $\gamma$ if $\sum_{P\ni e}f(P)\leq \gamma$ for each edge $e\in E(G)$. 
\item The flow $f$ has \textit{length} $h$ if all flow paths $P\in\path(f)$ has length $\ell_{G}(P)\leq h$.
\end{itemize}
The \textit{demand routed by $f$}, deonted by $D_{f}$, is such that $D_{f}(u,v) = \sum_{\text{$(u,v)$-paths $P\in\path(f)$}}f(P)$ for each pair of $u,v\in V(G)$. For a demand $D$, a \textit{routing} of $D$ is a flow $f$ with $D_{f} = D$, and we say $D$ can be routed on $G$ with congestion $\gamma$ and length $h$ if there exists a routing $f$ of $D$ with congestion $\gamma$ and length $h$. 
\end{definition}

\begin{definition}
[Routers] An $(h_{\rt},\gamma_{rt})$-router $R$ for a node-weighting $A_{R}$ is a unit-edge-length multi-graph s.t. any $A_{R}$-respecting demand can be routed on $R$ with length 
$h_{\rt}$ and congestion $\gamma_{\rt}$. When $A_{R} = \mathds{1}(V(R))$, we simply say $R$ is an $(h_{\rt},\gamma_{\rt})$-router.
\label{def:router}
\end{definition}

Given a routable demand on a router, one can efficiently route it.  

\begin{thm}
[Routing on a Router \cite{HHLRS2023emu}] Let $R$ be a $(h_{\rt},\gamma_{\rt})$-router for a node-weighting $A_{R}$. Given an $A_{R}$-respecting demand $D$, there is an algorithm that computes an integral routing of $D$ with congestion $\gamma_{\rt}\cdot \kappa_{\rr,\gamma}$ and length $h_{\rt}\cdot\lambda_{\mathrm{\rr,h}}$, where $\lambda_{\rr,h}=2^{\poly(1/\epsilon)}$ and $\kappa_{\rr,\gamma}=n^{O(\epsilon^{4})}$. The running time is $(|R|+|A|)\cdot \poly(h_{\rt})\cdot n^{O(\epsilon)}$.
\label{thm:RouterRouting}
\end{thm}

A stronger capacitated version of the above theorem is proven in \cite{HHLRS2023emu}. The unit-capacity version here can be proven in a simpler way by applying the length-constrained cut-matching game from \cite{haeupler2022cut} in a recursive manner as in \cite{CGLNPS20} and \cite{chang2020deterministic}.

\subsection{Length-Constrained Expanders}

We will introduce some concepts related to length-constrained expanders.

\begin{definition}[Integral Moving Cut]
Given a graph $G$, an \textit{integral moving cut} is a function $C:E(G)\to \mathbb{N}$ which assigns an integral length increase $C(e)$ to each edge $e\in E(G)$. It has size $|C| = \sum_{e}C(e)$. A vertex $v\in V(G)$ is a \textit{$C$-vertex} if $v$ is incident to some edge $e\in \supp(C)$. Furthermore, $C$ is an \textit{$h$-length integral moving cut} if $0\leq C(e)\leq h$ holds for each edge $e\in E(G)$.

We use $G-C$ to denote the graph with the same set of vertices and edges as $G$, but it has edge length $\ell_{G-C}(e) = \ell_{G}(e) + C(e)$ for each edge $e$.
\label{def:IntegralMovingCut}
\end{definition}

We note that in the literature of length-constrained expanders, e.g. \cite{HRG22,HHLRS2023emu,HaeuplerHT2023length}, an $h$-length moving cut is usually an edge weight function with weights in $\{0,1/h,2/h,...,(h-1)/h,1\}$, and the length of $G-C$ is $\ell_{G-C} = \ell_{G} + h\cdot C$. In \Cref{def:IntegralMovingCut}, we just scale up the moving cut to be an integral function, because integral moving cuts are easier to handle in the dynamic setting. The scale up will also change the definition of length-constrained expanders below, but actually it is equivalent to that in the literature.

\begin{definition}[Length-Constrained Expanders]

Let $G$ be a graph with node-weighting $A$. For a length parameter $h$, a length slack factor $s$, an $(hs)$-length integral moving cut $C$ and an $A$-respecting $h$-length demand $D$ on $G$, the \textit{sparsity of $C$ with respect to $D$} is
\[
\spars_{(h,s)}(C,D) = \frac{|C|/(hs)}{\sep_{hs}(C,D)},
\]
where 
\[
\sep_{hs}(C,D) = \sum_{(u,v):\dist_{G-C}(h,s)>hs} D(u,v)
\]
denotes the total $D$-demands of vertex pairs whose distance is increased by a factor $s$ after adding the moving cut, i.e. from at most $h$ to larger than $hs$.

The \textit{$(h,s)$-length conductance of $G$ w.r.t. $A$} is
\[
\cond_{(h,s)}(A) = \min_{\substack{\text{$h$-length $A$-respecting}\\ \text{demand $D$}}}\ \min_{\text{integral moving cut $C$}}\spars_{(h,s)}(C,D).
\]
Note that, in the above expression, the moving cut $C$ minimizes $\spars_{(h,s)}(C,D)$ must be $(hs)$-length even if the minimizing operation considers integral moving cut without length constraint.

We say $G$ is an \textit{$(h,s)$-length $\phi$-expander} w.r.t. $A$ if $\cond_{(h,s)}(A)\geq \phi$. 

\end{definition}

\begin{lemma}[Flow Characterization of Length-Constrained Expanders \cite{HRG22}]
If $G$ is an $(h,s)$-length $\phi$-expander w.r.t. $A$, then every $h$-length $A$-respecting demand $D$ can be routed on $G$ with congestion $O(\log n/\phi)$ and length at most $sh$. 

For a graph $G$ and node-weighting $A$, if every $h$-length $A$-respecting demand can be routed on $G$ with congestion $1/(2\phi)$ and length $sh/2$, then $G$ is an $(h,s)$-length $\phi$-expander w.r.t. $A$.
\label{lemma:HopExpanderFlowView}
\end{lemma}

Based on the flow characterization of length-constrained expanders, the following object called \emph{expansion witness} can certify that a graph is an length-constrained expander, as shown in \Cref{lemma:WitnessToFlow}. Before that, we first define the notion of \emph{embedding}. Intuitively, an embedding is just an integral and discretized interpretation of routing.

\begin{definition}[Embedding]
Let $G$ be a graph with edge lengths. Let $G_{D}$ be a unit-edge-length multi-graph with a mapping $\tau:V(G_{D})\to V(G)$. An embedding of $G_{D}$ into $G$, denoted by $\Pi_{G_{D}\to G}$, is a mapping from $E(G_{D})$ to simple paths (called \emph{embedding paths}) on $G$, s.t. each edge $e=(u,v)\in E(G_{D})$ is mapped to a simple path $\Pi_{G_{D}\to G}(e)$ connecting vertices $\tau(u),\tau(v)\in V(G)$. 
\begin{itemize}
\item The embedding $\Pi_{G_{D}\to G}$ has length $h$ if all embedding paths have length at most $h$ on $G$. 
\item The embedding $\Pi_{G_{D}\to G}$ has congestion $\gamma$ if each edge $e\in E(G)$ is contained by at most $\gamma$ embedding paths.
\end{itemize}
We say $\Pi_{G_{D}\to G}$ is an $(h,\gamma)$-embedding if it has length $h$ and congestion $\gamma$.

\end{definition}

\begin{defn}
[Expansion Witness]Let $G$ be a graph. Let $A$ be a node weighting. An \emph{$(h,h_{{\cal R}},\gamma_{{\cal R}},h_{\Pi},\gamma_{\Pi})$-witness} of $A$ in $G$ is a tuple $({\cal N},\mathcal{R},\Pi_{{\cal R}\to G})$. 
\begin{itemize}
\item ${\cal N}$ is a neighborhood cover of $\supp(A)$ on $G$ with covering radius $h$.
\item $\mathcal{R} = \{R^{S}\mid S\in{\cal N}\}$ is a collection of routers. For each cluster $S\in{\cal N}$, its corresponding router $R^{S}\in{\cal R}$ is a $(h_{{\cal R}},\gamma_{{\cal R}})$-router $R^{S}\in{\cal R}$ with $V(R^{S}) = S$ for node-weighting $A_{S} = A_{\mid S}$ (the restriction of $A$ on $S$). 
\item $\Pi_{{\cal R}\to G}$ is an $(h_{\Pi},\gamma_{\Pi})$-embedding of the union of all routers (i.e. $\bigcup_{R\in{\cal R}}R$) to $G$. Without ambiguity, ${\cal R}$ also refers to $\bigcup_{R\in{\cal R}}R$. Note that the embedding uses the mapping $\tau:V({\cal R})\to V(G)$ which is naturally defined by the fact that $V({\cal R})\subseteq V(G)$.

\end{itemize}
\end{defn}

\begin{lemma}
[Expansion Witness Certifies Expansion]Let $G$ be a graph with node-weighting $A$. Suppose that there exists a \emph{$(h,h_{{\cal R}},\gamma_{\cal R},h_{\Pi},\gamma_{\Pi})$-witness} $({\cal N},\mathcal{R},\Pi_{{\cal R}\to G})$ of $A$ in $G$. Then, for an arbitrary $h$-length $A$-respecting demand $D$, it can be routed on $G$ with length $h_{{\cal R}}\cdot h_{\Pi}$ and congestion $\gamma_{\cal R}\cdot \gamma_{\Pi}$.
\label{lemma:WitnessToFlow}
\end{lemma}
\begin{proof}

We assign demand pairs of $D$ to clusters in ${\cal N}$ as follows. For each demand pair $(u,v)\in \supp(D)$, we assign it to an arbitrary cluster $S\in{\cal N}$ s.t. $u,v\in S$. Note that $S$ must exist because $D$ is $h$-length on $G$ and ${\cal N}$ has covering radius $h$. 

For each cluster $S\in{\cal N}$, let $D_{S}:\binom{S}{2}\to\mathbb{R}_{\geq 0}$ be the total demand of pairs assigned to $S$. Consider the router $R\in{\cal R}$ corresponding to $S$. Observe that $D_{S}$ is a $A_{S}$-respecting demand, so it can be routed on $R$ with length $h_{{\cal R}}$ and congestion $\gamma_{\cal R}$, by \Cref{def:router}. Finally, because all routers can be embedded into $G$ simultaneously with length $h_{\Pi}$ and congestion $\gamma_{\Pi}$, $D = \sum_{S\in{\cal N}}D_{S}$ can be routed on $G$ with length $h_{{\cal R}}\cdot h_{\Pi}$ and congestion $\gamma_{\cal R}\cdot \gamma_{\Pi}$.
\end{proof}

\begin{definition}[Witnessed Expander Decomposition]
Let $G$ be a graph with node-weighting $A$. A \emph{$(h,h_{\cal R},\gamma_{\cal R},h_{\Pi},\gamma_{\Pi})$-witnessed expander decomposition} ($(h,h_{\cal R},\gamma_{\cal R},h_{\Pi},\gamma_{\Pi})$-witnessed-ED for short) of $A$ on $G$ includes an integral moving cut $C$ and an $(h,h_{\cal R},\gamma_{\cal R},h_{\Pi},\gamma_{\Pi})$-witness $({\cal N}, {\cal R}, \Pi_{{\cal R}\to G-C})$ of $A$ on $G-C$.
\label{def:WitnessedED}
\end{definition}

\begin{thm}
[Algorithmic Witnessed Expander Decomposition \cite{HaeuplerHT2023length}]Let $G$ be a graph with edge lengths and node-weighting $A$. Given parameters $h$ and $\phi$, there exists a deterministic algorithm that computes an $(h,h_{{\cal R}},\gamma_{\cal R},h_{\Pi},\gamma_{\Pi})$-witnessed-ED $(C,{\cal N},\mathcal{R},\Pi_{{\cal R}\to G-C})$ of $A$ with the following additional guarantees.
\begin{itemize}
\item $|C|\leq \kappa_{\ED,C}\cdot h\cdot\phi\cdot|A|$,
\item ${\cal N}$ has width $\omega = \kappa_{\PC,\omega}$,
\item The total number of edges in all routers is $|E({\cal R})|\leq |E(G)|\cdot n^{O(\epsilon)}$,
\item $h_{{\cal R}} = \lambda_{\ED,h}$, $h_{\Pi}=\lambda_{\ED,h}\cdot h$, $\gamma_{\cal R} = 1$ and $\gamma_{\Pi}=1/\phi$,
\end{itemize}
where $\kappa_{\ED,C} = n^{O(\epsilon^{4})}$ and $\lambda_{\ED,h} = 2^{\poly(1/\epsilon)}$ are newly defined global parameters, and the global parameter $\kappa_{\PC,\omega} = n^{O(\epsilon^{4})}$ is the same as that in \Cref{thm:NeighborhoodCovers}. The running time is $|G|\cdot\poly(h)\cdot n^{O(\epsilon)}$.
\label{thm:WitnessedED}
\end{thm}

\begin{theorem}[Expander Routing with Witness]
Let $G$ be a graph with node-weighting $A$ and an $(h,h_{\cal R},\gamma_{\cal R},h_{\Pi},\gamma_{\Pi})$-witness $({\cal N}, {\cal R}, \Pi_{{\cal R}\to G})$ of $A$ on $G$. Let $\omega$ be the width of ${\cal N}$. Given an $h$-length $A$-respecting demand $D$, there is an algorithm that computes an integral routing of $D$ on $G$ with length $\lambda_{\rr,h}\cdot h_{\cal R}\cdot h_{\Pi}$ and congestion $\kappa_{\rr,\gamma}\cdot \gamma_{{\cal R}}\cdot \gamma_{\Pi}$, where $\lambda_{\emb,h} = 2^{\poly(1/\epsilon)}$ and $\kappa_{\emb,\gamma} = n^{O(\epsilon^{4})}$ are global parameters from \Cref{thm:RouterRouting}. The running time is $(\sum_{R\in {\cal R}}|R| + \omega\cdot |A|)\cdot\poly(h_{\cal R})\cdot h_{\Pi}\cdot n^{O(\epsilon)}$.\label{thm:WitnessExpanderRouting}
\end{theorem}
\begin{proof}
We use a similar strategy as the proof of \Cref{lemma:WitnessToFlow}. We assign each demand pair $(u,v)\in\supp(D)$ to an arbitrary cluster $S\in{\cal N}$ s.t. $u,v\in S$. Note that $S$ must exist because $D$ is $h$-length on $G$ and ${\cal N}$ has covering radius $h$. This step takes $\wtilde{O}(|\supp(D)|)\leq \wtilde{O}(|A|)$ time.

For each cluster $S\in{\cal N}$, let $D_{S}$ be the total demand of pairs assigned to $S$. Let $R\in{\cal R}$ be the router corresponding to $S$. Because $D_{S}$ is a $A_{S}$-respecting demand, by applying 
\Cref{thm:RouterRouting} we can compute an integral routing $f_{S}$ of $D_{S}$ on $R$ with congestion $\kappa_{\rr,\gamma}\cdot \gamma_{\cal R}$ and length $\lambda_{\rr,h}\cdot h_{\cal R}$. Doing this for all clusters takes $\sum_{S\in{\cal N}}(|R^{S}| + |A_{S}|)\cdot \poly(h_{\cal R})\cdot n^{O(\epsilon)} = (\sum_{R\in{\cal R}}|R| + \omega\cdot |A|)\cdot\poly(h_{\cal R})\cdot n^{O(\epsilon)}$ time.

Lastly, for each $D_{S}$ and its routing $f_{S}$ on $R^{S}$, we can extend $f_{S}$ to be a routing on $G$. Precisely, for each path $P_{R} \in \path(f_{S})$ on $R^{S}$, we substitute each edge $e\in P_{R}$ with the embedding path $\Pi_{{\cal R}\to G}(e)$. The routing $f$ of $D$ on $G$ is then $f = \bigcup_{S\in{\cal N}} f_{S}$. It obviously has congestion $\kappa_{\rr,\gamma}\cdot \gamma_{\cal R}\cdot \gamma_{\Pi}$ and length $\lambda_{\rr,h}\cdot h_{\cal R}\cdot h_{\Pi}$ on $G$. This step takes time $O(|\path(f)|\cdot \lambda_{\rr,h}\cdot h_{\cal R}\cdot h_{\Pi}) = |A|\cdot h_{\cal R}\cdot h_{\Pi}\cdot 2^{\poly(1/\epsilon)}$ time.
\end{proof}

\subsection{Landmarks}

In this subsection, we introduce a new concept of \emph{landmarks}, which is for reducing the recourse in our dynamic algorithms.

\begin{definition}[Landmark Sets]
Given a graph $G$ and an integral moving cut $C$, a vertex set $L\subseteq V(G)$ is a \textit{landmark set} of $C$ with distortion $\sigma$ if
\begin{itemize}
\item for each $e=(u,v)\in E(G)$ with $C(e)>0$ and $\ell_{G-C}(e)>\sigma$, both endpoints $u,v\in L$;
\item for each $e=(u,v)\in E(G)$ with $C(e)>0$ and $\ell_{G-C}(e)\leq \sigma$, there exists $w_{u},w_{v}\in L$ s.t. $\dist_{G-C}(u,w_{u})\leq \sigma$ and $\dist_{G-C}(v,w_{v})\leq \sigma$.
\end{itemize}

\label{def:Landmark}
\end{definition}

\begin{lemma}[Union of Landmark sets]
Given a graph $G$ and integral moving cuts $C_{1}$ and $C_{2}$. Let $L_{1}$ be a landmark set of $C_{1}$ on $G$ with distortion $\sigma_{1}$, and let $L_{2}$ be a landmark set of $C_{2}$ on $G-C_{1}$ with distortion $\sigma_{2}$. Then $L = L_{1}\cup L_{2}$ is a landmark set of $C = C_{1} + C_{2}$ with distortion $\sigma = \sigma_{1} + \sigma_{2}$.
\label{lemma:LandmarkUnion}
\end{lemma}
\begin{proof}

\underline{Case 1.} For edges $e=(u,v)\in \supp(C_{2})$. If $\ell_{G-C_{1}-C_{2}}(e)>\sigma$, then $\ell_{G-C_{1}-C_{2}}(e)>\sigma_{2}$ and both $u,v\in L_{2}$. If $\ell_{G-C_{1}-C_{2}}(e)\leq \sigma$, there exists a landmark $w_{u}$ of $u$ s.t. $\dist_{G-C_{1}-C_{2}}(u,w_{u})\leq \sigma_{2}\leq \sigma$ (the same holds for $v$). 

\underline{Case 2.} Consider each edge $e=(u,v)\in \supp(C_{1})\setminus\supp(C_{2})$. Note that $L_{1}$ is a landmark set of $C_{1}$ on $G$. If $\ell_{G-C_{1}-C_{2}}(e)>\sigma$, then $\ell_{G-C_{1}}(e)>\sigma\geq \sigma_{1}$ (because $C_{2}(e) = 0$) and both $u,v\in L_{1}$. From now we assume $\ell_{G-C_{1}-C_{2}}(e)\leq \sigma$, and in what follows, we will show that there is a landmark $w_{u}\in L$ of $u$ s.t. $\dist_{G-C_{1}-C_{2}}(w_{u},u)\leq \sigma$. The same argument works for $v$.

First, there exists a landmark $w_{1,u}\in L_{1}$ s.t. $\dist_{G-C_{1}}(w_{1,u},u)\leq \sigma_{1}$. Consider the shortest $w_{1,u}$-$u$ path $P_{1}$ in $G-C_{1}$. If all edge $e\in P_{1}$ has $C_{2}(e) = 0$, then $\dist_{G-C_{1}-C_{2}}(w_{1,u},u) \leq \sigma_{1}\leq \sigma$ and we can take $w_{u} = w_{1,u}$ as our desired landmark. Otherwise, let $u'\in P_{1}$ be the vertex closest to $u$ and incident to an edge $e\in P_{1}$ with $C_{2}(e) > 0$. Observe that $\dist_{G-C_{1}-C_{2}}(u,u') = \dist_{G-C_{1}}(u,u')\leq \ell_{G-C_{1}}(P_{1}) = \sigma_{1}$, because all edges $e$ on the subpath of $P'$ from $u$ to $u'$ have $C_{2}(e) = 0$. Also, by the property of $L_{2}$, there exists $w_{2,u}\in L_{2}$ s.t. $\dist_{G-C_{1}-C_{2}}(u',w_{2,u})\leq \sigma_{2}$. Then $\dist_{G-C_{1}-C_{2}}(u,w_{2,u})\leq \dist_{G-C_{1}-C_{2}}(u,u') +\dist_{G-C_{1}-C_{2}}(u',w_{2,u})\leq \sigma_{1} + \sigma_{2} = \sigma$, so we can let $w_{u}$ be $w_{2,u}$ as the desired landmark.

\end{proof}

\begin{lemma}[Update Landmark Sets Under Edge Deletions] Given a graph $G$ and an integral moving cut $C$. Let $L$ be a landmark set of $C$ on $G$ with distortion $\sigma$. Given a batch of edge deletions $F\subseteq E(G)$, $\wtilde{L} = L\cup V(F)$ is a landmark set of $\wtilde{C} = C\setminus F$ on $\wtilde{G} = G\setminus F$ with distortion $\sigma$.%
\label{lemma:LandmarkEdgeDel}
\end{lemma}

\begin{proof}

For each edge $e=(u,v)\in \supp(\wtilde{C}) = \supp(C)\setminus F$, if $\ell_{\wtilde{G}-\wtilde{C}}(e) > \sigma$, then $\ell_{G-C}(e) = \ell_{\wtilde{G}-\wtilde{C}}(e)>\sigma$ and $u,v$ are already in $L$.

For each $C$-vertex $v$, there exists $w\in L$ s.t. $\dist_{G-C}(v,w)\leq \sigma$, and let $P$ be the $v$-$w$ shortest path in $G-C$. If $P$ is disjoint with $F$, then $\dist_{\wtilde{G}-\wtilde{C}}(v,w) = \dist_{G-C}(v,w)$ and $w$ is still a valid landmark of $v$. Otherwise, $F$ is on $P$ and let $w'$ be the vertex in $P\cap V(F)$ closest to $v$. Then $\dist_{\wtilde{G}-\wtilde{C}}(v,w')\leq \ell_{G-C}(v,w) = \sigma$ and we can take $w'$ as the new landmark of $v$.

\end{proof}

%% file: table_global_parameters.tex
\begin{table}

\scriptsize{
\begin{tabular}{|l|c|l|c|}
\hline 
Subroutines & Parameters & \makecell[c]{Expressions} & Values \tabularnewline
\hline 
\hline 
\multirow{2}{0.18\textwidth}{\Cref{thm:NeighborhoodCovers},\\Construction of Neighborhood Covers} & Diameter slack & $\lambda_{\NC,\diam} = O(1/\epsilon^{4})$ & $O(1/\epsilon^{4})$ \tabularnewline
\cline{2-4}
& Width & $\kappa_{\NC,\omega} = n^{O(\epsilon^{4})}$ & $n^{O(\epsilon^{4})}$ \tabularnewline
\hline
\multirow{2}{0.18\textwidth}{\Cref{thm:WitnessedED},\\Witnessed Expander Decomposition} & Cut size slack & $\kappa_{\ED,C} = n^{O(\epsilon^{4})}$ & $n^{O(\epsilon^{4})}$ \tabularnewline
\cline{2-4}
& Length slack & $\lambda_{\ED,h} = 2^{\poly(1/\epsilon)}$ & $2^{\poly(1/\epsilon)}$ \tabularnewline
\hline
\multirow{2}{0.18\textwidth}{\Cref{thm:WitnessExpanderRouting},\\Expander Routing with Witness} & Congestion slack & $\kappa_{\emb,\gamma} = n^{O(\epsilon^{4})}$ & $n^{O(\epsilon^{4})}$ \tabularnewline
\cline{2-4}
& Length slack & $\lambda_{\emb,h} = 2^{\poly(1/\epsilon)}$ & $2^{\poly(1/\epsilon)}$\tabularnewline
\hline

\multirow{3}{0.18\textwidth}{\Cref{thm:cutmatch}\\Cutmatch} & Distortion shrink & $\kappa_{\sigma} = n^{\epsilon^{4}}$ & $n^{O(\epsilon^{4})}$ \tabularnewline
\cline{2-4}
& Congestion slack & $\kappa_{\CM,\gamma} = \log^{4} n$ & $\log^{4}n$ \tabularnewline
\cline{2-4}
& Landmark size slack & $\kappa_{\CM,L} = \kappa_{\sigma}\cdot\log^{3} n$ & $n^{O(\epsilon^{4})}$ \tabularnewline
\hline
\multirow{3}{0.18\textwidth}{\Cref{thm:Router},\\Dynamic Routers} & Maximum Degree & $\kappa_{\rt,\deg} = n^{O(\epsilon^{4})}$ & $n^{O(\epsilon^{4})}$ \tabularnewline
\cline{2-4}
& \makecell{Pairwise Distance} & $\lambda_{\rt,h}(i) = 2^{O(1/\epsilon^{4})} + O(i)$ & function of $i$ \tabularnewline
\cline{2-4}
& \makecell{Pruned set size inflation} & $\lambda_{\rt,\prune}(i) = (1/\epsilon)^{O(1/\epsilon^{4})}\cdot 2^{O(i)}$ & function of $i$ \tabularnewline
\hline
\multirow{4}{0.18\textwidth}{\Cref{thm:InitLandmark,thm:InitCertifiedED}\\Initialization of Certified-EDs} & Cut size slack & $\kappa_{\init,C} = \kappa_{\ED,C}\cdot\lambda_{\PC,\diam}$ & $n^{O(\epsilon^{4})}$ \tabularnewline
\cline{2-4}
& Landmark size slack & $\kappa_{\init,L} = \lambda_{\PC,\diam}\cdot\kappa_{\PC,\omega}\cdot\kappa_{\init,C}\cdot\kappa_{\sigma}$ & $n^{O(\epsilon^{4})}$ \tabularnewline
\cline{2-4}
& \makecell{Embedding congestion slack} & $\kappa_{\init,\gamma} = \kappa_{\PC,\omega}\cdot\kappa_{\rou,\deg}\cdot \kappa_{\emb,\gamma}$ & $n^{O(\epsilon^{4})}$ \tabularnewline
\cline{2-4}
& \makecell{Embedding length slack} & $\lambda_{\init,h} = \lambda_{\emb,h}\cdot\lambda^{2}_{\ED,h}\cdot\lambda_{\PC,\diam}$ & $2^{\poly(1/\epsilon)}$ \tabularnewline
\hline
\multirow{2}{0.18\textwidth}{\Cref{thm:DynamicED}\\Dynamic Expander Decomposition} & \makecell{Length slack} & $\lambda_{\dynED,h} = 2^{O(t)}$ & $2^{O(1/\epsilon)}$ \tabularnewline
\cline{2-4} & \makecell{Embedding\\congestion slack} & \makecell[l]{
$\kappa_{\dynED,\gamma} = ((t\cdot \kappa_{\PC,\omega}\cdot\kappa_{\CM,\gamma}+\kappa_{\init,\gamma})\cdot$
\\$(\kappa_{\CM,L} + \kappa_{\init,L})\cdot \lambda_{\rt,\prune}(t))^{O(t)}$
} & $n^{O(\epsilon^{3})}$ \tabularnewline
\hline
\multirow{2}{0.18\textwidth}{\Cref{lemma:LandmarkClosure}\\Insert Landmarks to Node-Weighting} & Batch number inflation & $\lambda_{\insLM,t} = O(1/\epsilon^{4})$ & $O(1/\epsilon^{4})$ \tabularnewline
\cline{2-4}
& Congestion slack & $\lambda_{\insLM,\gamma} = n^{O(\epsilon^{4})}$ & $n^{O(\epsilon^{4})}$ \tabularnewline
\hline
\multirow{5}{0.18\textwidth}{\Cref{thm:NonHopReducingEmulator}\\Dynamic Vertex Sparsifier} & \makecell{Stretch\\lower bound side} & $\lambda_{H,\low} = \lambda_{\rt,h}(t\cdot\lambda_{\insLM,t})\cdot\lambda_{\init,h}\cdot 2^{O(t\cdot\lambda_{\insLM,t})}$ & $2^{\poly(1/\epsilon)}$ \tabularnewline
\cline{2-4}
& \makecell{Stretch\\upper bound side} & $\lambda_{H,\up} = O(1)$ & $O(1)$ \tabularnewline
\cline{2-4}
& \makecell{Size of the\\Certified-ED collections} & $\kappa_{j} = O(\log n)$ & $O(\log n)$ \tabularnewline
\cline{2-4}
& Size inflation & \makecell[l]{$\kappa_{H,\size} = \kappa_{j}^{3}\cdot 2^{O(t\cdot\lambda_{\insLM,t})}\cdot \lambda_{\rt,\prune}(t\cdot\lambda_{\insLM,t})\cdot$\\
$\kappa_{\PC,\omega}^{3}\cdot (\kappa_{\CM,L} +\kappa_{\init,L})\cdot(\kappa_{\CM,\gamma} + \kappa_{\init,\gamma})$} & $n^{O(\epsilon^{4})}$ \tabularnewline
\cline{2-4}
& Recourse & \makecell[l]{$\kappa_{H,\rcs} = \kappa_{j}^{3}\cdot 2^{O(t\cdot\lambda_{\insLM,t})}\cdot \lambda_{\rt,\prune}(t\cdot\lambda_{\insLM,t}) \cdot$\\ $\kappa_{\PC,\omega}^{4}\cdot (\kappa_{\CM,L} +\kappa_{\init,L})\cdot(\kappa_{\CM,\gamma} + \kappa_{\init,\gamma})\cdot \kappa_{\insLM,\gamma}$} & $n^{O(\epsilon^{4})}$ \tabularnewline
\hline
\multirow{2}{0.18\textwidth}{\Cref{thm:ExpanderHierarchy}\\Length Constrained Expander Hierarchy} & Number of levels & $\lambda_{k} = O(1/\epsilon^{2})$ & $O(1/\epsilon^{2})$ \tabularnewline
\cline{2-4} & \makecell{Recourse on\\pairwise covers} & \makecell[l]{$\kappa_{\EH,\PC,\rcs} = \lambda_{\rt,\prune}(t)\cdot t\cdot \kappa_{\PC,\omega}\cdot$\\
$\kappa_{\dynED,\gamma}\cdot 2^{O(\lambda_{k})}\cdot \kappa^{\lambda_{k}}_{H,\rcs}$}
& $n^{O(\epsilon^{2})}$ \tabularnewline
\hline
\makecell[l]{\Cref{thm:LowDistanceOracle}\\Low Distance Oracles} & Approximation & \makecell[l]{$\lambda_{\query,\alpha} = 2^{O(\lambda_{k})}\cdot(\lambda_{H,\low}\cdot\lambda_{H,\up})^{\lambda_{k}}\cdot$\\
$(\lambda_{\rt,h}(t)\cdot\lambda_{\init,h}\cdot\lambda_{\dynED,h})^{\lambda_{k}}$} & $2^{\poly(1/\epsilon)}$ \tabularnewline
\hline
\multirow{5}{0.18\textwidth}{\Cref{thm:HopReducingEmulator}\\Length-Reducing Emulator} & \makecell{Stretch\\lower bound side} & \makecell[l]{$\lambda_{Q,\low} = \lambda_{\rt,h}(t)\cdot\lambda_{\init,h}\cdot\lambda_{\dynED,h}\cdot$\\
$2^{O(\lambda_{k})}\cdot (\lambda_{H,\up}\cdot\lambda_{H,\low})^{\lambda_{k}}$} & $2^{\poly(1/\epsilon)}$ \tabularnewline
\cline{2-4} & \makecell{Stretch\\upper bound side} & $\lambda_{Q,\up} = O(\lambda_{k})$ & $O(1/\epsilon^{2})$ \tabularnewline
\cline{2-4} & Congestion parameter $\phi$ & $\kappa_{\phi} = n^{\epsilon^{2}}$ & $n^{O(\epsilon^{2})}$ \tabularnewline
\cline{2-4} & Size inflation & $\kappa_{Q,\size} = O(\lambda_{k}\cdot t\cdot \kappa_{\PC,\omega})$ & $n^{O(\epsilon^{4})}$ \tabularnewline  
\cline{2-4} & Recourse & $\kappa_{Q,\rcs} =\kappa_{\EH,\PC,\rcs}\cdot\kappa_{\phi}$ & $n^{O(\epsilon^{2})}$ \tabularnewline
\hline
\makecell[l]{\Cref{thm:Stacking}\\Stacking} & \makecell{Stretch, Size, Recourse} & \makecell[l]{$\lambda_{\stk,\low} = O(\lambda_{Q,\low}) = 2^{\poly(1/\epsilon)}$,\\
$\lambda_{\stk,\up} = O(\lambda_{Q,\up}) = O(1/\epsilon^{2})$,\\ $\kappa_{\stk,\size} = O(\kappa_{Q,\size}) = n^{O(\epsilon^{4})}$,\\ $\kappa_{\stk,\rcs} = O(\kappa_{Q,\rcs}) = n^{O(\epsilon^{2})}$} & \tabularnewline
\hline
\multirow{2}{0.18\textwidth}{\Cref{thm:BatchDynDistanceOracle}\\Distance Oracles} & Levels of stacking & $\lambda_{x} = O(1/\epsilon)$ & $O(1/\epsilon)$ \tabularnewline
\cline{2-4} & Approximation & $\lambda_{\DO,\alpha} = O(\kappa_{\query,\alpha})\cdot (\kappa_{\stk,\low}\cdot\kappa_{\stk,\up})^{\lambda_{x}}\cdot 2^{O(1/\epsilon^{4})}$ & $2^{\poly(1/\epsilon)}$ \tabularnewline
\hline

\end{tabular}
}

\caption{Additional Global Parameters}
\label{table:GlobalParameters}
\end{table}

%% file: 4-local_flow.tex
\section{Local Length-Constrained Flows}
\label{sect:LocalLengthConstrainedFlow}

In this section, we develop a local algorithm in \Cref{thm:LocalLengthConstrainedFlow} for the approximate length-constrained maxflow problem by generalizing the results in \cite{haeupler2023maximum}. With this local flow algorithm, our second result is a local cutmatch algorithm shown in \Cref{thm:cutmatch}, which is an important subroutine for our dynamic expander decomposition algorithms in \Cref{sect:DynamicCertifiedED}.

In the length-constrained maxflow problem, we work with a \textit{directed} graph $G$ with edge lengths $\ell$ and edge capacities $U$, where the lengths and capacities are \emph{positive} integers. During the algorithm, we may also assign \emph{weights} $w$ to edges in $G$, and the weights are non-negative real numbers. For each vertex $v\in V(G)$, we let $E^{+}(v)$ (resp. $E^{-}(v)$) denote the set of outgoing (resp. ingoing) edges of $v$. Let $\deg^{+}_{G}(v) = |E^{+}(v)|$ be the out-degree of $v$ in $G$. For a subset of vertices $B\subseteq V(G)$, we let $E^{+}(B) = \bigcup_{v\in B}E^{+}(v)$ and $E^{-}(B) = \bigcup_{v\in B}E^{-}(v)$. Also, we use $N^{+}(B) = \{v\mid \exists (u,v)\in E(G)\text{ s.t. }u\in B\text{ and }v\in V(G)\setminus B\}$ to denote the outgoing neighbors of $B$, and let $N^{+}[B] = N^{+}(B)\cup B$. Similarly, $N^{-}(B) = \{u\mid \exists (u,v)\in E(G)\text{ s.t. }u\in V(G)\setminus B\text{ and }v\in B\}$ are the ingoing neighbors of $B$. 

Furthermore, we will specify two vertices $s$ and $t$ be the source vertex and sink vertex. For each vertex $v\in V(G)\setminus\{s,t\}$, we let $\Delta(v) = U(s,v)$ be the \textit{source capacity} of $v$, where $U(s,v)$ is the capacity of the edge from $s$ to $v$ (if no such edge, $U(s,v) = 0$). Similarly, $\nabla(v) = U(v,t)$ denote the \textit{sink capacity} of $v$. Edges in $E^{+}(s)$ are called \textit{source edges} and edges in $E^{-}(t)$ are \textit{sink edges}.

\begin{definition}[Single-Commodity Length-Constrained Flows]
Given a graph $G$ with a source vertex $s$, a sink vertex $t$ and a length parameter $h$, let ${\cal P}_{h}(s,t)$ denote the set of simple $s$-to-$t$ paths $P$ with $\ell_{G}(P)\leq h$. An \textit{(feasible) $h$-length $s$-$t$ flow} $f$ on $G$ assigns a non-negative real number $f(P)$ to each $P\in{\cal P}_{h}(s,t)$ s.t. $f(e)\leq U(e)$ for each $e\in E(G)$, where $f(e) = \sum_{P\ni e} f(P)$ denote the total flow value going through $e$. The \textit{value} of $f$ is $\val(f) = \sum_{P\in{\cal P}_{h}(s,t)}$. The flow $f$ is \textit{integral} if $f(P)$ is an integer for each $P$. We use $\path(f) = \{P\in {\cal P}_{h}(s,t)\mid f(P)>0\}$ to denote the \textit{flow paths} of $f$ and let $|\path(f)|$ be the \textit{path count} of $f$.

\end{definition}

\begin{definition}[Fractional Moving Cut]
Given a graph $G$ with a source vertex $s$, a sink vertex $t$ and a length parameter $h$, a \textit{(feasible) $h$-length fractional moving cut} is a function $w:E(G)\to \mathbb{R}_{\geq 0}$ s.t. for each path $P\in{\cal P}_{h}(s,t)$, $w(P)\geq 1$. The size of $w$ is denoted by $|w| = \sum_{e\in E(G)}U(e)\cdot w(e)$.
\end{definition}

Given graph $G$ with source $s$, sink $t$ and length parameter $h$, the (exact) length-constrained maxflow problem asks a feasible $h$-length $s$-$t$ fractional flow $f$ with maximum value, while the (exact) length-constrained mincut problems asks a feasible $h$-length fractional moving cut $w$ with minimum size. These two problems are dual to each other as shown in \cite{haeupler2023maximum}. That is, the value of maximum $h$-length flow is the same as the size of minimum $h$-length moving cut. 

Naturally, the approximate length-constrained maxflow problem asks a pair $(f,w)$ of feasible $h$-length $s$-$t$ maxflow $f$ and minimum moving cut $w$. A solution $(f,w)$ is $\alpha$-approximation if $|w|\leq \alpha\cdot\val(f)$. Namely, the moving cut $w$ certifies the value of the length-constrained flow up to an $\alpha$ factor.

Our result on local approximate length-constrained maxflow problem is formally stated in \Cref{thm:LocalLengthConstrainedFlow}. The running time is local in the sense that if every vertex $v$ has sink capacity $\nabla(v) = U(v,t)\geq \Omega(\deg^{+}_{G}(v))$ at least a constant fraction of its out-degree, then the running time is proportional to the total source capacities, with overhead $\poly(h,1/\delta,\log n)$. This notion of ``local running time'' appeared before in e.g. \cite{SW19}. We point out that the condition $\nabla(v)\geq \Omega(\deg^{+}_{G}(v))$ for all $v\in V(G)$ is kind of crucial to obtain local running time depending on total source capacities instead of the original graph size. Intuitively, imagine a graph in which half of vertices have zero sink capacity. Then in the worst case, even finding a path from $s$ to $t$ may need to scan through all these zero-sink-capacitated vertices. The running time bound in \Cref{thm:LocalLengthConstrainedFlow} is analysed in a general setting that some vertices $v$ may have $\nabla(v)$ much smaller than $\deg^{+}_{G}(v)$, which brings an extra term $\max\{\deg^{+}_{G}(v)-c\cdot\nabla(v)\}$ about the total \textit{deficit}, which, roughly speaking, measures the total differences between $\deg^{+}_{G}(v)$ and $\nabla(v)$ over all vertices. Furthermore, the requirement that all sink edges have equal length is also unavoidable, because otherwise we can easily reduce the following toy problem to an $h$-length maxflow problem: given $n$ (different) integers, answer whether there is an integer at most $h$ or not, which is impossible to be solved in local time. 

In fact, \Cref{thm:LocalLengthConstrainedFlow} includes two solutions with different approximation ratio, one is $(\log n/\delta)$-approximate and the other one is $O(2+\delta)$-approximate. Our local cutmatch algorithm (i.e. \Cref{thm:cutmatch}) will use the $O(\log n)$-approximate and exploit the additional guarantee. We include the $(2+\delta)$-approximate result since it may be of independent interest.

\begin{theorem}[Local Approximate Length-Constrained Maxflow]
Let $G = (V(G),E(G))$ be a directed graph with lengths $\ell$, capacities $U$, source vertex $s$, sink vertex $t$ and parameters $\delta > 0, h\geq 1$, further satisfying that each sink edge $e\in E^{-}(t)$ has $\ell(e) = 1$. With access to the adjacency list of $G$, there is an algorithm that compute 
\begin{enumerate}
\item a feasible $h$-length flow and moving cut pair $(f,w)$ that is $O(\log n/\delta)$-approximate, further satisfying that, for some $\zeta = 1/(6\delta)+3$ and $\eta = \Theta(\delta^{2}/\log n)$, $f/\eta$ is an integral flow and $w(e) = (1+\delta/6)^{f(e)/(\eta\cdot U(e))}/m^{\zeta}$ for each $e\in E(G)$, 
\item a feasible $h$-length flow and moving cut pair $(f,w_{\min})$ that is $(2+\delta)$-approximate.
\end{enumerate}
The running time is
\[
\poly(h,1/\delta,\log n)\cdot\sum_{v\in V(G)\setminus\{s,t\}}(c\cdot\Delta(v) + \max\{\deg_{G}^{+}(v)-c\cdot\nabla(v),0\}),
\]
where $c\geq 1$ is a parameter arbitrarily chosen.

\label{thm:LocalLengthConstrainedFlow}
\end{theorem}

In \Cref{sect:CutMatch}, we will use the local length-constrained flow algorithm to compute \textit{cutmatches}, which was first introduced in \cite{HRG22} with applications on the first efficient length-constrained expander decomposition algorithm. The notion of cutmatch is defined on an \textit{undirected} \textit{uncapacitated} graph $G$ with edge lengths $\ell_{G}$. The formal definition is given in \Cref{thm:cutmatch}. Roughly speaking, cutmatch is made up with a matching $M_{\CM}$ and a moving cut $C_{\CM}$. The matching $M_{\CM}$ will partially matches two given node-weightings via low-congestion $h$-length paths on $G$, and the moving cut certifies that the matching can be hardly extended under the low-congestion requirement. We point out that an efficient (global) cutmatch algorithm has already been shown in \cite{haeupler2023maximum}. Here we design a cutmatch algorithm with local running time. Moreover, it is adapted to be compatible with the algorithms in \Cref{sect:DynamicCertifiedED}. For example, $C_{\CM}$ will be rounded to an integral moving cut, and we also output a landmark set $L_{\CM}$ (see \Cref{def:Landmark}) of $C_{\CM}$.

\begin{theorem}[Local Cutmatches]
Let $G$ be a graph with source node-weighting $A^{\src}$, and sink node-weighting $A^{\sink}$ s.t. $A^{\src} + A^{\sink} \geq \deg_{G} + \mathds{1}(V(G))$. Given parameters $h_{\CM}$ and $\phi_{\CM}$, there is an algorithm that computes the following:
\begin{itemize}
\item a partition of $A^{\src}$ into $A^{\src}_{M}, A^{\src}_{U}\subseteq A^{\src}$,
\item a partition of $A^{\sink}$ into $A^{\sink}_{M}, A^{\sink}_{U}\subseteq A^{\sink}$ s.t. $|A^{\src}_{M}|\leq |A^{\sink}_{M}|$ and $|A^{\sink}_{M}| = O(|A^{\src}_{M}|)$,
\item a set of $h_{\CM}$-length paths which embeds a matching $M_{\CM}$ between $A^{\src}_{M}$ and $A^{\sink}_{M}$ with size $|M_{\CM}| = |A^{\src}_{M}|$ and congestion $\gamma_{\CM}=O(\kappa_{\CM,\gamma}/\phi_{\CM})$ where $\kappa_{\CM,\gamma} = \log^{4} n$,
\item an integral moving cut $C_{\CM}$ with size $|C_{\CM}|\leq O(\phi_{\CM}\cdot h_{\CM}\cdot |A^{\src}|)$ s.t. $A^{\src}_{U}$ and $A^{\sink}_{U}$ are $h_{\CM}$-separated on $G-C_{\CM}$.
\item a landmark set $L_{\CM}$ of $C_{\CM}$ on $G$ with distortion $\sigma_{\CM} = h_{\CM}/\kappa_{\sigma}$ and size $|L_{\CM}|\leq \kappa_{\CM,L}\cdot \phi_{\CM}\cdot |A^{\src}|$, where $\kappa_{\sigma} = n^{\epsilon^{4}}$ and $\kappa_{\CM,L} = O(\kappa_{\sigma}\cdot \log^{3} n)$.
\end{itemize}
The running time is $\poly(h_{\CM})\cdot|A^{\src}|/\phi_{\CM}$.
\label{thm:cutmatch}
\end{theorem}

Here we make some further explanations on the terminology in \Cref{thm:cutmatch}. Recall that $A^{\src}$ and $A^{\sink}$ also refer to sets of virtual nodes (see \Cref{def:NodeWeighting}), so partitions on $A^{\src}$ and $A^{\sink}$ are well-defined. The matching $M_{\CM}$ between $A^{\src}_{M}$ and $A^{\sink}_{M}$ is actually an unweighted matching between virtual nodes in $A^{\src}_{M}$ and $A^{\sink}_{M}$. 

\begin{remark}
We want to emphasize that, although in \Cref{thm:cutmatch} and its proof we talk about node-weighting and virtual nodes, \Cref{thm:cutmatch} also holds when the input $A^{\src}$ and $A^{\sink}$ are general integral function of $V(G)$. In particular, in \Cref{thm:NWInsertWithDensity}, we will apply \Cref{thm:cutmatch} with \textit{density functions} and \textit{items} (analogous to node-weightings and virtual nodes) as input.
\label{remark:GeneralCutmatch}
\end{remark}

\subsection{Local $h$-Length Lightest Path Blockers}
\label{sect:PathBlockers}
In \cite{haeupler2023maximum}, it is shown that the approximate length-constrained maxflows can be reduced to computing lightest path blockers via a multiplicative weight update framework. In this subsection, we localize the path blockers subroutine by invoking the (global) path blockers subroutine in \cite{haeupler2023maximum} as a blackbox. 

\begin{definition}[Length-Constrained Lightest Path Blockers]
Let $G=(V(G),E(G))$ be a graph with lengths $\ell$, weights $\omega$ and capacities $U$. Given source vertex $s$, sink vertex $t$, and parameters $\gamma_{2}\geq \gamma_{1}\geq 1$, $h\geq 1$, $\lambda\leq d_{w}^{(h)}(s,t)$, where $d_{w}^{(h)}(s,t)$ denote the minimum weights among $h$-length $s$-$t$ paths, an $h$-length integral $s$-$t$ flow $f$ is an \textit{$h$-length $(\gamma_{1},\gamma_{2})$-lightest path blocker} if
\begin{itemize}
\item[(1)] every path $P\in\supp(f)$ has weight at most $\gamma_{2}\cdot\lambda$,
\item[(2)] for each $h$-length $s$-$t$ path $P'$ with weight at most $\gamma_{1}\cdot\lambda$, there exists $e\in P'$ s.t. $f(e) = U(e)$ (we say that $f$ \textit{blocks} $P'$ in this case).
\end{itemize}
\end{definition}

\begin{theorem}[Theorem 11.1, \cite{haeupler2023maximum}]
Given a directed graph $G=(V(G),E(G))$ with lengths $\ell$, weights $w$, capacities $U$, source vertex $s$, sink vertex $t$ and parameters $\delta>0$, $h\geq 1$ and $\lambda\leq d_{w}^{(h)}(s,t)$, there is an algorithm that compute an $h$-length $(1+\delta, 1+2\delta)$-lightest path blocker $f$.
The running time is $\poly(h,1/\delta,\log n)\cdot |G|$.
\label{thm:GlobalPathBlockers}
\end{theorem}

\begin{theorem}
Let $G=(V(G),E(G))$ be a directed graph with lengths $\ell$, weights $w$, capacities $U$, source vertex $s$, sink vertex $t$ and parameters $\delta>0$, $h\geq 1$ and $\lambda\leq d_{w}^{(h)}(s,t)$, further satisfying that, $(s,t)\notin E(G)$ and each sink edge $e\in E^{-}(t)$ has $w(e) = 0$ and $\ell(e)=1$. With access to the adjacency list of $G$, there is an algorithm that compute an $h$-length $(1+\delta, 1+2\delta)$-lightest path blocker $f$. The running time is
\[
\poly(h,1/\delta,\log n)\cdot\sum_{v\in V(G)\setminus\{s,t\}}(c\cdot\Delta(v) + \max\{\deg_{G}^{+}(v)-c\cdot\nabla(v),0\}),
\]
where $c\geq 1$ is a parameter arbitrarily chosen.
\label{thm:PathBlockers}
\end{theorem}

\begin{proof}
\textbf{The Algorithm.} The algorithm is iterative. In the $i$-th iteration, we will construct some local graph $G_{i}$ and compute a blocker $f_{i}$ on $G_{i}$ using \Cref{thm:GlobalPathBlockers}. The final output is the accumulated flow $f = \sum_{i}f_{i}$.

Suppose we are now in the $i$-th iteration. The graph $G_{i}$ is constructed as follows. Let $\hat{f}_{i-1} = \sum_{i'\leq i-1} f_{i'}$ be the accumulated flow of the previous $i-1$ iterations. In particular, $\hat{f}_{0} = 0$ is the zero flow. Let $U_{i} = U-\hat{f}_{i-1}$ be the remaining capacity function. Let $V_{i-1,\sat}\subseteq \{v\in V(G)\setminus\{s,t\}\mid \hat{f}_{i-1}(v,t) = \nabla(v)\}$ be vertices reachable from $s$ via a BFS which only goes through non-saturated edges (i.e. $e$ s.t. $U_{i}(e) >0$) and sink-saturated vertices (i.e. $v$ s.t. $\hat{f}_{i-1}(v,t) = \nabla(v)$). 
Then the vertex set and edge set of $G_{i}$ are
\[
V(G_{i}) = \{s,t\}\cup N^{+}[\{s\}\cup V_{i-1,\sat}],
\]
and 
\[
E(G_{i}) = (\{(s,v),(v,t)\mid v\in V(G_{i})\setminus\{s,t\}\}\cap E(G))\cup E^{+}(V_{i-1,\sat}).
\]
For each edge $e\in E(G_{i})$, we define its capacity be $U_{i}(e) = U(e) - \hat{f}_{i-1}(e)$ and keep its weight and length unchanged, i.e. $w_{i}(e) = w(e)$ and $\ell_{i}(e) = \ell(e)$. 

After constructing $G_{i}$, we compute a blocker $f_{i}$ on $G_{i}$ using \Cref{thm:GlobalPathBlockers} with parameters $h_{i} = i+1$ and proceed to the next iteration. The whole algorithm will terminate right after the $h$-th iteration, and the output is $f = \hat{f}_{h} = \sum_{i\leq h}f_{i}$.

\

\noindent\textbf{Proof of Correctness.} First, the flow paths in the final flow $f$ indeed have weight at most $(1+2\delta)\cdot\lambda$ because \Cref{thm:GlobalPathBlockers} guarantees that flow paths in each $f_{i}$ has weight at most $(1+2\delta)\cdot\lambda$. Next we show that $f$ blocks all $h$-length path with weight at most $(1+\delta)\lambda$ by induction. Let the induction hypothesis be that, right after each iteration $i$, all $(i+1)$-length $s$-$t$ paths in $G$ with weight at most $(1+\delta)\cdot\lambda$ have been blocked by $\hat{f}_{i}$. 

\

\underline{The Base Case.} Consider the first iteration. Note that every $2$-length path $P$ in $G$ must in the form $s\to v\to t$, because $(s,t)\notin E(G)$ and every edge in $G$ has positive integral length. Because $v\in N^{+}(\{s\})$, $P$ is fully inside the local graph $G_{1}$. By \Cref{thm:GlobalPathBlockers}, the path blocker $\hat{f}_{1} = f_{1}$ will blocked all $2$-length paths in $G_{1}$ with weight at most $(1+\delta)\cdot\lambda$, so every $2$-length path in $G$ with weight at most $(1+\delta)\cdot\lambda$ will be blocked by $\hat{f}_{1}$.

\

\underline{An Inductive Step.} Now consider iteration $i$ with $i\geq 2$. Let $P = s\to v_{1}\to \cdots \to v_{k}\to t$ be an $(i+1)$-length $s$-$t$ path in $G$ with weight at most $(1+\delta)\cdot\lambda$. If $P$ has been blocked by $\hat{f}_{i-1}$ in $G$, the trivially it will blocked by $\hat{f}_{i}$ in $G$. Now suppose $P$ has not been blocked by $\hat{f}_{i-1}$.

\begin{claim}
If $P$ has not been blocked by $\hat{f}_{i-1}$, then $P\subseteq G_{i}$.
\label{claim:4.9}
\end{claim}
\begin{proof}
We assume that $k\geq 2$ without loss of generality. Because if $k=1$, then $P$ is in the form $s\to v\to t$ and it must be a path in $G_{i}$ because $v\in N^{+}(\{s\})$.

We first show that each $v_{k'}$ with $1\leq k'\leq k-1$ is in $V_{i-1,\sat}$. Assume this is not true for contradiction. Let $k'$ be the smallest s.t. $v_{k'}\notin V_{i-1,\sat}$, and we consider the path $P'=s\to v_{1}\to\cdots\to v_{k'}\to t$. We know that $P'$ has not been blocked by $\hat{f}_{i-1}$ because the edges $(s,v_{1}),(v_{1},v_{2}),...,(v_{k'-1},v_{k'})$ are not saturated (since $P$ is not blocked by $\hat{f}_{i-1}$) and $(v_{k'},t)$ is not saturated (since $v_{k'}\notin V_{i-1,\sat}$). However, we know that $P'$ is a $i$-length path because 
\[
\ell_{G}(P') = \ell_{G}(P)-\ell_{G}(v_{k'}\to\cdots\to v_{k}) - \ell_{G}(v_{k},t) + \ell_{G}(v_{k'},t),
\]
where $\ell_{G}(v_{k},t)=\ell_{G}(v_{k'},t)=1$ from the theorem statement and $\ell_{G}(v_{k'}\to\cdots\to v_{k})\geq 1$. The weight of $P'$ can be bounded by 
\[
w(P')=w(P)-w(v_{k'}\to\cdots\to v_{k}) - w(v_{k},t) + w(v_{k'},t)\leq (1+\delta)\cdot\lambda.
\]
because $w(P)\leq (1+\delta)\cdot\lambda$ and $w(v_{k'},t)=0$ from the theorem statement. Therefore, the existence of $P'$ violates the induction hypothesis.

Providing that $v_{k'}\in V_{i-1,\sat}$ for all $1\leq k'\leq k-1$, we know $P\subseteq G_{i}$ because edges $(v_{1},v_{2}),...,(v_{k-1},v_{k})\in E^{+}(V_{i-1,\sat})\subseteq E(G_{i})$, and $(s,v_{1}),(v_{k},t)\in E(G_{i})$ (since $v_{1},v_{k}\in V(G_{i})$ by definition).

\end{proof}

By \Cref{claim:4.9} and \Cref{thm:GlobalPathBlockers}, $f_{i}$ will blocked $P$ in $G_{i}$. Note that the capacity of $G_{i}$ is defined by $U_{i} = U - \hat{f}_{i-1}$, so $\hat{f}_{i} = \hat{f}_{i-1} + f_{i}$ will block $P$ in $G$.

\

\noindent\textbf{Time Analysis.} We now analyse the running time. We first bound the size of each $G_{i}$. Observe that $|V_{i-1,\sat}|\leq \deg^{+}_{G}(\{s\}\cup V_{i-1,\sat})$ by considering the BFS procedure constructing $V_{i-1,\sat}$. Thus, the number of vertices in $G_{i}$ is
\begin{align*}
|V(G_{i})|\leq 2+|V_{i-1,\sat}| + \deg^{+}_{G}(\{s\}\cup V_{i-1,\sat})\leq 2+2\cdot\deg^{+}_{G}(\{s\}\cup V_{i-1,\sat}).
\end{align*}
Also, the number of edges in $G_{i}$ is bounded by
\begin{align*}
|E(G_{i})|&\leq 2|V(G_{i})| + \deg^{+}_{G}(V_{i-1,\sat})\\
&\leq 4+3(\deg_{G}^{+}(s) + \deg^{+}_{G}(V_{i-1,\sat}))\\
&\leq O(\sum_{v\in V(G)\setminus\{s,t\}} \Delta(v)) + O(\sum_{v\in V_{i-1,\sat}}c\cdot \nabla(v)) + O(\sum_{v\in V_{i-1,\sat}}\max\{\deg_{G}^{+}(v)-c\cdot\nabla(v),0\})\\
&\leq O(\sum_{v\in V(G)\setminus\{s,t\}}(c\cdot \Delta(v) + \max\{\deg_{G}^{+}(v)-c\cdot\nabla(v),0\}))
\end{align*}
where the third inequality is by $\sum_{v\in V(G)\setminus\{s,t\}}\Delta(v)\geq \deg^{+}_{G}(s)$ because each source edge has capacity at least $1$ (the capacity is a positive integer), and the last inequality we use that $\sum_{v\in V_{i-1,\sat}} \nabla(v) \le \sum_{v\in V(G)\setminus\{s,t\}} \Delta(v)$ because the total capacity of saturated sinks is at most the total source.%

The running time $i$-th iteration is dominated by applying \Cref{thm:GlobalPathBlockers} on $G_{i}$, which is
\begin{align*}
&~~~~\poly(h,1/\delta,\log n)\cdot|G_{i}| \\
&= \poly(h,1/\delta,\log n)\cdot\sum_{v\in V(G)\setminus\{s,t\}}(c\cdot\Delta(v) + \max\{\deg^{+}_{G}(v)-c\cdot\nabla(v),0\}).
\end{align*}
Summing over $h$ iterations gives the desired bound of total running time.

\end{proof}

\begin{corollary}
Let $G=(V(G),E(G))$ be a directed graph with lengths $\ell$, weights $w$, capacities $U$, source vertex $s$, sink vertex $t$ and parameters $\delta>0$, $h\geq 1$ and $\lambda\leq d_{w}^{(h)}(s,t)$, further satisfying that each sink edge $e\in E^{-}(t)$ has $\ell(e) = 1$. 
For each vertex $v\in V(G)\setminus\{s,t\}$, let 
\begin{equation*}
\wtilde{\nabla}(v) = 
\left\{
\begin{aligned}
&0,\ \text{if $e=(v,t)$ does not exist or $w(v,t)>(1+\delta)\lambda$}\\
&\nabla(v),\ \text{if $w(v,t)\leq (1+\delta)\lambda$}
\end{aligned}
\right.
\end{equation*}
With access to the adjacency list of $G$, There is an algorithm that compute an $h$-length $(1+\delta, 2+3\delta)$-lightest path blocker $f$ with $|\path(f)|\leq\sum_{v\in V(G)\setminus\{s,t\}} \Delta(v)$.
The running time is
\[
\poly(h,1/\delta,\log n)\cdot\sum_{v\in V(G)\setminus\{s,t\}}(c\cdot\Delta(v) + \max\{\deg_{G}^{+}(v)-c\cdot\wtilde{\nabla}(v),0\}),
\]
where $c\geq 1$ is a parameter arbitrarily chosen.
\label{Coro:LocalPathBlockers}
\end{corollary}
\begin{proof}

We consider a graph $\wtilde{G}$ obtained by (1) removing all sink edges of $G$ with weight larger than $(1+\delta)\lambda$ and then (2) setting the weight of each remaining sink edge (with original weight at most $(1+\delta)\lambda$ in $G$) to be zero. Note that for each $v\in V(\wtilde{G})\setminus\{s,t\}$, its source capacity is $\wtilde{\Delta}(v) = \Delta(v)$ and its sink capacity is exactly $\wtilde{\nabla}(v)$ defined in the statement of this corollary. 

Now we apply \Cref{thm:PathBlockers} on $G'$ with the same parameters $\delta, h$ and $\lambda$. We will show that the output $f$ is an $h$-length $(1+\delta,2+3\delta)$-lightest path blocker of $G$. First, every path $\wtilde{P}\in\supp(f)$ has weight at most $(1+2\delta)\lambda$ on $\wtilde{G}$. Every edge on $\wtilde{P}$ has the same weight on $G$ and $\wtilde{G}$, except that the only sink edge of $\wtilde{P}$ has weight zero on $\wtilde{G}$ but at most $(1+\delta)\lambda$ on $G$. Hence $\wtilde{P}$ has weight at most $(1+2\delta)\lambda + (1+\delta)\lambda = (2+3\delta)\lambda$ on $G$. For the second property, consider an $h$-length $s$-$t$ path $P$ with weight at most $(1+\delta)\lambda$ on $G$. The path $P$ is totally inside $\wtilde{G}$ because we only remove those sink edges with weight larger than $(1+\delta)\lambda$ and they will not belong to $P$. Also, the weight of $P$ on $\wtilde{G}$ is still at most $(1+\delta)\lambda$, so $P$ is blocked by $f$ as desired.

The running time bound is straightforward by \Cref{thm:PathBlockers}, and the path count $|\path(f)|$ of $f$ is at most the total source capacities $\sum_{v\in V(G)\setminus\{s,t\}}\Delta(v)$ because $f$ is an integral flow.
\end{proof}

\subsection{Local Approximate Length-Constrained Maxflows: Proof of \Cref{thm:LocalLengthConstrainedFlow}}
\label{sect:LCMaxflow}

In this subsection, we will adapt the multiplicative weight update framework in \cite{haeupler2023maximum} and use our local path blockers algorithm to achieve local running time for the length-constrained max flow problem.

\begin{algorithm}[H]
\caption{Local Algorithms for Approximate Length-Constrained Maxflows and Moving Cuts}
\label{algo:LocalFlow}
\begin{algorithmic}[1]
\Require A directed graph $G=(V(G),E(G))$ with length $\ell$, capacities $U$, source vertex $s$, sink vertex $t$ and an parameters $\delta > 0$, $h\geq 1$.
\Ensure An $O(\log n/\delta)$-approximate pair $(f,w)$ and a $(2+O(\delta))$-approximate pair $(f,w_{\min})$.
\State Let $\delta_0 = \frac{\delta}{6}$, let $\zeta = \frac{1+2 \delta_0}{\delta_0} + 1$ and let $\eta = \frac{\delta_0}{(1 + \delta_0) \cdot \zeta} \cdot \frac{1}{\log m}$.
\State Initialize $w(e) = 1/m^{\zeta}$ for all edges $e$ in $G$.%
\State Initialize $\lambda = 1/m^{\zeta}$.
\State Initialize $w_{\min} = w/\lambda$.
\State Initialize $f$ to be zero flow.
\While{$\lambda < 1$}
\For{$\Theta(h\log_{1+\delta_{0}}(n)/\delta_{0})$ iterations}
\State Compute an $h$-length $(1+\delta_{0}, 2+3\delta_{0})$-lightest path blocker $\hat{f}$ in $G$ by applying \Cref{Coro:LocalPathBlockers}, with length $\ell$, weight $w$, capacity $U$, and parameters $h,\lambda$ as input.
\State Update $f\gets f + \eta\cdot\hat{f}$.
\State For every $e\in E(G)$ where $\hat{f}(e) >0$, $w(e)\gets (1+\delta_{0})^{\hat{f}(e)/U(e)}\cdot w(e)$.
\State If $|w|/\lambda<|w_{\min}|$, $w_{\min}\gets w/\lambda$.
\EndFor
\State $\lambda \gets (1+\delta_{0})\cdot\lambda$
\EndWhile
\State Return $(f,w)$ and $(f,w_{\min})$.
\end{algorithmic}
\end{algorithm}

\paragraph{The Algorithm.} The detailed algorithm is shown in \Cref{algo:LocalFlow}. Roughly speaking, the algorithm just iteratively compute a path blocker $\hat{f}$ on $G$ by applying \Cref{Coro:LocalPathBlockers}. The input functions are the original edge length function $\ell$, the edge weight function $w$ (i.e. the current moving cut) and the original capacity function $U$ as input functions. The input parameters are $h,\lambda$. After computing $\hat{f}$, we then update the flow solution $f$ additively and update the moving cut solution $w$ multiplicatively.

In what follows, we will discuss some implementation details to achieve local running time.
\begin{itemize}
\item The flow solution $f$ is stored in two representations. The first one is the \textit{flow path representation}, which just stores each flow path $P\in\path(f)$ and its flow value $f(P)$ explicitly. The second one is the \textit{edge function representation}, which stores all edges $e$ s.t. $f(e)>0$ and its flow value $f(e)$ in binary search tree. Note that the space to store $f$ in this way is obviously $O(h\cdot |\path(f)|)$, because each flow path has at most $h$ edges (recall that each flow path is $h$-length and each edge has length at least 1). By the same reason, each additive update to $f$ (line 8) just takes time $\wtilde{O}(h\cdot\path(\hat{f}))$.

\item The moving cut solution $w$ will be stored in a similar \textit{edge function representation}. we store all edges $e$ and its weight $w(e)$ s.t. $w(e)$ is large than the default value $1/m^{\zeta}$ into a binary search tree. Actually, edges $e$ with $w(e)>1/m^{\zeta}$ are exactly those edges $e$ with $f(e)>0$ by the simple invariant $w(e) = (1+\delta_{0})^{f(e)/(\eta\cdot U(e))}$, so the space to store $w$ is still $O(h\cdot|\path(f)|)$. The multiplicative update of $w$ (line 9) can be done in $\wtilde{O}(h\cdot|\path(\hat{f})|)$ time\footnote{Actually, a cleverer implementation is to use to the same binary search tree as the edge function representations of both $f$ and $w$.} with access to the edge function representation of $\hat{f}$.

\end{itemize}

\paragraph{Proof of Correctness.} Intuitively, $w$ is the cut solution at the end, while $w_{\min}$ is the best cut solution over the course of the algorithm. By adapting the proof in \cite{haeupler2023maximum}, we can show that $(f,w_{\min})$ is indeed a $(2+O(\delta))$-approximate solution. We omit the proof here because our local cutmatch algorithm does not use this solution. The approximation is $2+O(\delta)$ instead of $1+\delta$ in \cite{haeupler2023maximum} because our local path blockers (i.e. \Cref{Coro:LocalPathBlockers}) have quality worse than the path blockers in \cite{haeupler2023maximum}.

In the remaining proof, we will formally prove that $(f,w)$ is an $O(\log n/\delta)$-approximate solution. Our proof basically follows the presentation in \cite{haeupler2023maximum}. In particular, \Cref{claim:LocalFlowLambdaBounds} and \Cref{claim:LocalFlowFeasible} are almost identical to Lemmas 12.1 and 12.2 in \cite{haeupler2023maximum}, but \Cref{claim:LocalFlowApprox} is new.

We refer to the outer loop of \Cref{algo:LocalFlow} as one \emph{phase}, and refer to the inner loop as one \emph{iteration}. We will say $\lambda$ is unchanged in one iteration, and $\lambda$ is updated right before the end of one phase. Strictly speaking, for each phase, the end of this phase and the beginning of the next phase represent the same moment. However, for one iteration, the end of this iteration and the beginning of the next iteration may represent different moments, and $\lambda$ may be updated between them. 

\begin{claim}
At the beginning and the end of \Cref{algo:LocalFlow}, and at the end of each iteration, we have $\lambda\leq d^{(h)}_{w}(s,t)\leq (2+3\delta_{0})(1+\delta_{0})\lambda$.
\label{claim:LocalFlowLambdaBounds}
\end{claim}
\begin{proof}

The statement trivially holds at the very beginning of the algorithm.

First, we prove $\lambda\leq d^{(h)}_{w}(s,t)$. Note that $d^{(h)}_{w}(s,t)$ will only increase, so it suffices to show that at the beginning of each phase, $\lambda\leq d^{(h)}_{w}(s,t)$ holds. We will show this by induction. Consider a phase. At the beginning and the end of this phase, we let $\lambda_{1} = \lambda$ and $\lambda_{2}=(1+\delta_{0})\lambda$ denote the value of $\lambda$, and similarly, we let $d_{1}$ and $d_{2}$ denote the value of $d^{(h)}_{w}(s,t)$. To make the induction work, it suffices to show $\lambda_{2}\leq d_{2}$ given $\lambda_{1}\leq d_{1}$. Assume the opposite for contradiction. From $\lambda_{2}>d_{2}$, we know there exists a $h$-length $s$-$t$ path $P$ with weight at most $\lambda_{2}$ at the end of this phase, which means at the beginning of each of the $\Theta(h\log n/\delta_{0}^{2})$ iterations in this phase, $P$ has weight at most $\lambda_{2} = (1+\delta_{0})\lambda$. Then, in each of these iterations, because by the definition of path blockers, at least one edge $e\in P$ will be saturated by $\hat{f}$, which means the weight of $e$ increases by a factor $(1+\delta_{0})$. Because there are at most $h$ edges on $P$ (recall that edge lengths are positive integers), in this phase, there is an edge $e\in P$ has its weight increasing by a factor of $(1+\delta_{0})^{\Theta(h\log n/\delta_{0}^{2})/h}$ by an averaging argument. However, at the beginning of this phase, the weight of $e$ is at least $1/m^{\zeta}$ (since this is the initial weight), so at the end of this phase, $w(e)\geq (1/m^{1/\zeta})\cdot (1+\delta_{0})^{\Theta(h\log n/\delta_{0}^{2})/h}>m^{2}$ by setting the constant hidden in $\Theta$ to be large enough, and it contradicts that $P$ has weight at most $\lambda_{2}=(1+\delta_{0})\lambda\leq 1+\delta_{0}$ (the outter loop ensures $\lambda\leq 1$).

Next, we proof prove $d^{(h)}_{w}(s,t)\leq (2+3\delta_{0})(1+\delta_{0})\lambda$ at the end of each iteration. Again we use induction. Consider one iteration. Assume that $d^{(h)}_{w}(s,t)\leq (2+3\delta_{0})(1+\delta_{0})\lambda$ holds at the end of the previous iteration. Then this inequality holds at the beginning of this iteration, because $\lambda$ may keep unchanged or increase. If $d^{(h)}_{w}(s,t)>(2+3\delta_{0})\lambda$ at the beginning of this iteration, then actually we do nothing in this iteration because the path blocker $\hat{f}$ will be a zero flow, so $d^{(h)}_{w}(s,t)\leq (2+3\delta_{0})(1+\delta_{0})\lambda$ holds at the end of this iteration trivially. Otherwise, we have $d^{(h)}_{w}(s,t)\leq (2+3\delta_{0})\lambda$ at the beginning of this iteration. By our update rule of weights and the fact that $\hat{f}$ is a feasible flow, the weight of each edge will only increase by at most a $1+\delta_{0}$ factor in this iteration. Therefore, we also have $d^{(h)}_{w}(s,t)\leq (2+3\delta_{0})(1+\delta_{0})\lambda$ at the end of this iteration.

\end{proof}

\begin{claim}
The pair $(f,w)$ returned by \Cref{algo:LocalFlow} is feasible.
\label{claim:LocalFlowFeasible}
\end{claim}
\begin{proof}
The moving cut $w$ is feasible because we have $\lambda\geq 1$ and $\lambda \leq d^{(h)}_{w}(s,t)$ at the end of the algorithm (by \Cref{claim:LocalFlowLambdaBounds}). To see that the flow $f$ is feasible, we let $\hat{F} = \sum \hat{f} = f/\eta$ be the summation of path blockers without scaling. By our update rule, at the end, an edge $e$ will have weight $w(e) = (1+\delta_{0})^{\hat{F}/U(e)}/m^{\zeta}$. Furthermore, we know that $w(e)\leq (1+\delta_{0})(2+3\delta_{0})$ at the end because right before the last update of $w(e)$, $w(e)$ must be at most $(2+3\delta_{0})\lambda$ (otherwise, the path blocker $\hat{f}$ cannot assign non-zero flow on $e$) and $w(e)$ will increase by at most a $1+\delta_{0}$ factor in this update. Therefore, from the inequality 
\[
(1+\delta_{0})^{\hat{F}/U(e)}/m^{\zeta}\leq (1+\delta_{0})(2+3\delta_{0}),
\]
we have $\hat{F}(e)\leq U(e)\cdot(\log(1+\delta_{0})(2+3\delta_{0})+\log m/\delta_{0})/\log\zeta\leq U(e)\cdot(10\log m/\delta_{0}^{2}) \leq U(e)/\eta$. Lastly, we have $f(e) = \eta\cdot\hat{F}(e) \leq U(e)$ as desired.

\end{proof}

\begin{claim}
The pair $(f,w)$ returned by \Cref{algo:LocalFlow} satisfies $|w|\leq O(\log n/\delta)\cdot \val(f)$.
\label{claim:LocalFlowApprox}
\end{claim}
\begin{proof}
Again we use $\hat{F} = \sum \hat{f}$ to denote the accumulated path blockers without scaling. Because $f = \eta\cdot\hat{F}$ and $\eta = \Theta(\delta^{2}/\log m)$, it suffices to show that $|w|\leq O(\delta)\val(\hat{F})$.

We consider iterations one by one according to the time order. Fix an iteration, we then consider the flow paths in $\path(\hat{f})$ one by one in some certain order. Our update rule on weights is equivalent to the following procedure: consider the flow paths in $\hat{F}$ in the above order, and then for each path $P$, update the weights using $P$ (i.e. for each $e\in P$, multiply $w(e)$ by $(1+\delta_{0})^{\hat{f}(P)/U(e)}$).

Now we will show that, after updating the weights using some flow path $P$, the value of $|w|$ increases by at most $O(\delta \hat{f}(P))$ additively. To see this, note that at the beginning of this iteration, we have $w(P)\leq (2+3\delta_{0})\lambda$ guaranteed by the path blocker. Let $w_{P}$ denote the weight function right before the moment we update weights using $P$. We have $\sum_{e\in P}w_{P}(e)\leq (2+3\delta_{0})(1+\delta_{0})\lambda$ because each edge weight can increase by at most a $(1+\delta_{0})$ factor in this iteration. Therefore, updating the weights using $P$ will increase $|w|$ by at most
\begin{align*}
&~~~~\sum_{e\in P} (1+\delta_{0})^{\hat{f}(P)/U(e)}\cdot w_{P}(e)\cdot U(e) - \sum_{e\in P} w_{P}(e)\cdot U(e)\\
&\leq \sum_{e\in P} (1+\delta_{0}\hat{f}(P)/U(e))\cdot w_{P}(e)\cdot U(e) - \sum_{e\in P} w_{P}(e)\cdot U(e)\\
&= \sum_{e\in P} (\delta_{0}\hat{f}(P)/U(e))\cdot w_{P}(e)\cdot U(e)\\
&\leq \delta_{0}(2+3\delta_{0})(1+\delta_{0})\lambda\cdot \hat{f}(P)\\
&\leq O(\delta_{0}\hat{f}(P)).
\end{align*}

Therefore, the final $|w|$ is at most $\sum_{P\in\path(\hat{F})} O(\delta_{0})\hat{f}(P) = O(\delta)\val(\hat{F})$ as desired.
\end{proof}

\paragraph{Time Analysis.} By \Cref{Coro:LocalPathBlockers}, the running time for computing the blocker in one iteration is
\[
t_{\blocker} = \poly(h,1/\delta,\log n)\cdot\sum_{v\in V(G)\setminus\{s,t\}}(\Delta(v) + \max\{\deg_{G}^{+}(v)-\wtilde{\nabla}(v),0\}).
\]
Recall that for each $v\in V\setminus\{s,t\}$, $\wtilde{\nabla}(v) = \nabla(v)$ if currently $w(v,t)\leq (1+\delta_{0})\lambda$ (or there is no such edge $e=(v,t)$), and $\wtilde{\nabla}(v) = 0$ if $w(v,t)>(1+\delta_{0})\cdot\lambda$. Thus, for each $v$ with $w(v,t)> (1+\delta_{0}\lambda)$, we have 
\[
\max\{\deg_{G}^{+}(v)-\wtilde{\nabla}(v)\} = \deg_{G}^{+}(v)\leq \max\{\deg_{G}^{+}(v)-\nabla(v),0\} + \nabla(v).
\]
Therefore, we can bound $t_{\blocker}$ by
\[
t_{\blocker}\leq \poly(h,1/\delta,\log n)\cdot\left(\sum_{v\in V(G)\setminus\{s,t\}}(\Delta(v) + \max\{\deg_{G}^{+}(v)-\nabla(v),0\})+\sum_{v\text{ s.t. }w(v,t)>(1+\delta_{0})\lambda}\nabla(v)\right).
\]

\begin{observation}
For each edge $e$, if at some moment $w(e)>(1+\delta_{0})\lambda$ holds, then finally we have $f(e)\geq \eta\cdot U(e)$.
\label{ob:InLocalFlow}
\end{observation}
\begin{proof}
We have $w(e)>(1+\delta_{0})\lambda\geq (1+\delta_{0})/m^{\zeta}$ because $\lambda \geq 1/m^{\zeta}$. Combining the invariant $w(e) = (1+\delta_{0})^{f(e)/(\eta\cdot U(e))}/m^{\zeta}$, we immediately get $f(e)\geq \eta\cdot U(e)$.
\end{proof}

Next, we bound the total sink capacities of vertices $v$ s.t. $w(v,t)>(1+\delta_{0})\lambda$ by
\begin{align*}
\sum_{v\text{ s.t. }w(v,t)>(1+\delta_{0})\lambda}\nabla(v) &\leq \sum_{e=(v,t)\text{ s.t. }w(e)>(1+\delta_{0})\lambda} f(e)/\eta\\
&\leq \val(f)/\eta\\
&\leq \sum_{v\in V(G)\setminus\{s,t\}} \Delta(v)/\eta\\
&\leq \poly(1/\delta,\log n)\cdot \sum_{v\in V(G)\setminus\{s,t\}}\Delta(v),
\end{align*}
where the first inequality is by \Cref{ob:InLocalFlow}, the second inequality is because each flow path of $f$ will go through exactly one sink edge, the third inequality is because the value of $f$ is bounded by the total source capacities, and the last inequality is by $\eta = \Theta(\delta^{2}/\log m)$.

Lastly, we have
\[
t_{\blocker}\leq \poly(h,1/\delta,\log n)\cdot\sum_{v\in V(G)\setminus\{s,t\}}(\Delta(v) + \max\{\deg_{G}^{+}(v)-\nabla(v),0\}),
\]
and the final running time bound follows that the total number of iteration is at most $\poly(h,1/\delta,\log n)$.

\subsection{Local Cutmatch: Proof of \Cref{thm:cutmatch}}
\label{sect:CutMatch}

Providing the local length-constrained maxflow subroutine in \Cref{thm:LocalLengthConstrainedFlow}, we are ready to prove \Cref{thm:cutmatch}. In this proof, When applying \Cref{thm:LocalLengthConstrainedFlow} on some graph $H$, we will always set $\delta = 0.01$. Our proof will exploit that the flow-cut pair $(f,w)$ returned by \Cref{thm:LocalLengthConstrainedFlow} is $O(\log n)$-approximate with additional properties stated in \Cref{lemma:AdditionalProperties}.

\begin{lemma}
For $\delta = 0.01$, $\zeta = 1/(6\delta)+3$ and $\eta = O(\delta^{2}/\log n)$. The output $(f,w)$ given by \Cref{thm:LocalLengthConstrainedFlow} satisfies the following.
\begin{enumerate}
\item $f/\eta$ is an integral flow.
\item For each edge $e\in E(H)$ with $w(e)\geq 1/m$, there is $f(e)\geq \Omega_{\delta}(U(e))$.
\item For each edge $e\in E(H)$ with $w(e)\geq 1/m^{3}$, there is $f(e)>0$.
\end{enumerate}
\label{lemma:AdditionalProperties}
\end{lemma}
\begin{proof}
Property (1) has been stated in \Cref{thm:LocalLengthConstrainedFlow}. Properties (2) and (3) rely on the guarantee that, for each edge $e\in E(G)$, $w(e) = (1+\delta/6)^{f(e)/(\eta\cdot U(e))}/m^{\zeta}$. Regarding property (2), for each edge $e$ with $w(e)>1/m$, we have
\begin{align*}
(1+\delta/6)^{f(e)/(\eta\cdot U(e))}/m^{\zeta} &\geq 1/m\\
(1+\delta/6)^{f(e)/(\eta\cdot U(e))} &\geq m\\
\frac{f(e)}{\eta\cdot U(e)}&\geq\frac{\log m}{\log (1+\delta/6)}\\
f(e)&\geq \Omega(U(e)).
\end{align*}
Here the third inequality is because $\zeta> 3$. For property (3), the invariant and $\zeta > 3$ immediately imply each edge $e$ with $w(e)\geq 1/m^{3}> 1/m^{\zeta}$ has $f(e)>0$.
\end{proof}

\paragraph{The Algorithm.} We now show the local algorithm for computing cutmatch in \Cref{thm:cutmatch}. Recall that $M_{\CM}$ is a matching between virtual nodes in $A^{\src}$ and $A^{\sink}$.

We start with some notations. First, we define a directed auxiliary graph $H(A^{\src},A^{\sink})$, given source and sink node-weightings $A^{\src}$ and $A^{\sink}$.
\begin{enumerate}
\item Start from the unit-capacitated directed version of $G$ (i.e. substitute each undirected edge with two opposite directed edges with capacity $1$ and the same length).
\item Add a source vertex $s$ and a sink vertex $t$.
\item For each vertex $v^{\src}\in \supp(A^{\src})$, add a \textit{source edge} $(s,v^{\src})$ with length $\ell(s,v^{\src}) = 1$\footnote{We set the lengths of source and sink edges to be $1$ because the lengths should be positive integers.} and capacity $U(s,v^{\src}) = A^{\src}(v^{\src})$. For each vertex $v^{\sink}\in\supp(A^{\sink})$, add a \textit{sink edge} $(v^{\sink},t)$ with length $\ell(v^{\sink},t) = 1$ and capacity $U(v^{\sink},t) = A^{\sink}(v^{\sink})$. 
\end{enumerate}

We define some notations on \textit{matchings}. For a matching $M$ between two sets $V^{\src}$ and $V^{\sink}$ ($V^{\src}$ and $V^{\sink}$ may intersect), each edge $e\in E(M)$ will connect a \textit{source endpoint} in $V^{\src}$ and \textit{sink endpoint} in $V^{\sink}$, and its weight is $M(e)$ if $M$ is weighted. For each element $v\in V^{\src}$, we let $M^{\src}(v)$ denote the total weights of matching edges with $v$ as their source endpoint. Similarly, for each $v\in V^{\sink}$, $M^{\sink}(v)$ denote the total weights of edges with $v$ as their sink endpoint. Furthermore, we use $M^{\src}$ (resp. $M^{\sink}$) refer to the set of elements in $V^{\src}$ (resp. $V^{\sink}$) matched by $M$.

The cutmatch algorithm is iterative. We initialize the matching $M_{\CM}$ to be empty. In the $i$-th iteration, we first define the remaining source and sink node-weightings $A^{\src}_{i}$ and $A^{\sink}_{i}$. Concretely, for each vertex $v\in \supp(A^{\src})$, $A^{\src}_{i}(v) = A^{\src}(v) - M^{\src}_{\CM}(v)$\footnote{Strictly speaking, we slightly abuse the notations here because $M_{\CM}$ is between virtual nodes but $v$ is a vertex. Here $M^{\src}_{\CM}(v)$ is naturally the sum of $M^{\src}_{\CM}(v_{\virtual})$ over all $v$'s virtual nodes $v_{\virtual}$ in $A^{\src}$.}. Similarly, for each vertex $v\in \supp(A^{\sink})$, $A^{\sink}_{i}(v) = A^{\sink}(v) - M^{\sink}_{\CM}(v)$. We run the local flow algorithm in \Cref{thm:LocalLengthConstrainedFlow} on the graph $H_{i} =H(A^{\src}_{i},A^{\sink}_{i})$ with $\delta = 0.01$ and length parameter $h = h_{\CM}+2$, and let the output be $(f_{i},w_{i})$. Let $\phi = \phi_{\CM}/\log^{2} n$. We now consider two cases.

\paragraph{The first case $\val(f_{i})\geq \phi |A^{\src}_{i}|$.} In this case, we will compute a partial unweighted matching $M_{i}$ between virtual nodes of $A^{\src}_{i}$ and $A^{\sink}_{i}$. 

We first construct a fractional weighted matching $M'_{i}$ between vertices in $\supp(A^{\src}_{i})$ and $\supp(A^{\sink}_{i})$. By our construction of $H_{i}$, each flow path $P_{H}\in \path(f_{i})$ on $H_{i}$ corresponds to a path $P$ on $G$ by removing the source and sink vertices $s$ and $t$ and the source and sink edges on $P$. Note that $P$ now connects two vertices $v^{\src}$ and $v^{\sink}$ on $G$. We add a matching edge with $v^{\src}$ and $v^{\sink}$ as source and sink endpoints. The weight of this edge is $f_{i}(P_{H})$. Furthermore, we \textit{simplify} $M'_{i}$ to a simple graph. That is, for each pair of vertices $v^{\src}\in \supp(A^{\src}_{i})$ and $v^{\sink}\in\supp(A^{\sink}_{i})$, if currently there are multiple matching edges connecting $v^{\src}$ and $v^{\sink}$ in $M'_{i}$, we merge them into a single matching edge with weight the total weight of these edges. 

From the construction of $M'_{i}$, for each vertex $v^{\src}\in \supp(A^{\src}_{i})$, $M_{i}'^{,\src}(v^{\src})$ is at most the capacity of the source edge $(s,v^{\src})$, i.e. $A^{\src}_{i}(v^{\src})$. Because $A^{\src}_{i}(v^{\src})$ is integral, we further have $\lceil M_{i}'^{,\src}(v^{\src})\rceil\leq A^{\src}_{i}(v^{\src})$. Similarly, each vertex $v^{\sink}\in\supp(A^{\sink}_{i})$ has $\lceil M_{i}'^{,\sink}(v^{\sink})\rceil\leq A^{\sink}_{i}(v^{\sink})$. The size of $M'_{i}$ is bounded by $|M'_{i}|\leq \val(f_{i})\leq \phi|A^{\src}_{i}|$.

The next step is to round $M'_{i}$ into an integral weighted matching $M''_{i}$ between vertices in $\supp(A^{\src}_{i})$ and $\supp(A^{\sink}_{i})$ by applying \Cref{lemma:MatchingRounding}. 

\begin{lemma}[\cite{KP15}]
Given a fractional matching $M$ between two sets $V^{\src}$ and $V^{\sink}$, there is an algorithm that rounds it to an integral matching $M_{\iint}$ s.t. 
\begin{itemize}
\item $E(M_{\iint})\subseteq E(M)$, $|M_{\iint}|=\lceil |M| \rceil$;
\item each edge $e\in E(M_{\iint})$ has integral weight $M_{\iint}(e)$ s.t. $\lfloor M(e)\rfloor\leq M_{\iint}(e)\leq \lceil M(e) \rceil$
\item each $v\in V^{\src}$ has $\lfloor M^{\src}(v)\rfloor\leq M_{\iint}^{\src}(v)\leq \lceil M^{\src}(v) \rceil$ and each $v\in V^{\sink}$ has $\lfloor M^{\sink}(v)\rfloor\leq M_{\iint}^{\sink}(v)\leq \lceil M^{\sink}(v) \rceil$. 
\end{itemize}
The running time is $\tilde{O}(|E(M)|)$.
\label{lemma:MatchingRounding}
\end{lemma}

Finally, we \textit{refine} $M''_{i}$ to be our unweighted matching $M_{i}$ between virtual nodes in $A^{\src}_{i}$ and $A^{\sink}_{i}$. That is, we first transform $M''_{i}$ into an unweighted matching by splitting each edge $e$ into $M''_{i}(e)$ many unweighted copies. Then for each vertex $v\in \supp(A^{\src}_{i})$, we assign each edge with $v$ as its source endpoint to a distinct virtual node in $A^{\src}_{i}(v)$. Note that because  $M_{i}''^{,\src}(v)\leq \lceil M_{i}'^{,\src}(v)\rceil\leq A^{\src}_{i}(v)$, there are enough virtual nodes in $A^{\src}_{i}(v)$ to receive matching edges of $v$. We do the similar things for each $v\in \supp(A^{\sink}_{i})$.

We add $M_{i}$ to $M_{\CM}$ and proceed to the next iteration. The property of $M_{i}$ summarized in \Cref{lemma:PropertyOfMi}

\begin{lemma}
The matching $M_{i}$ can be embedded into $G$ with congestion $O(1/\eta)$ and length $h_{\CM}$. Furthermore, $|M_{i}|\geq \lceil\phi\cdot|A^{\src}_{i}|\rceil$.
\label{lemma:PropertyOfMi}
\end{lemma}
\begin{proof}
We can bound the size of $M_{i}$ by $|M_{i}| = |M''_{i}| = \lceil M'_{i}\rceil\geq \lceil\phi\cdot|A^{\src}_{i}|\rceil$.

For the embedding, we consider the flow $\hat{f}_{i} = (1/\eta)\cdot f_{i}$, where $\eta = \Theta(1/\log m)$ is the parameter from \Cref{lemma:AdditionalProperties}. By (1) in \Cref{lemma:AdditionalProperties}, $\hat{f}_{i}$ is an integral flow. Let $f_{i}^{G}$ (resp. $\hat{f}_{i}^{G}$) be the flow on $G$ by removing source edges and sink edges on flow paths of $f_{i}$ (resp. $\hat{f}_{i}$). 

Note that $f_{i}^{G}$ and $\hat{f}_{i}^{G}$ have congestion $2$ and $2/\eta$ respectively on $G$ because each undirected unweighted edge on $G$ corresponds to two directed edges with capacity $1$ on $H_{i}$ (from step 1 in the construction of $H(A^{\src}_{i}, A^{\sink}_{i})$). Also, $f_{i}^{G}$ and $\hat{f}_{i}^{G}$ are $h_{\CM}$-length on $G$ because $f_{i}$ and $\hat{f}_{i}$ are $(h_{\CM}+2)$-length on $H_{i}$ and the source and sink edges removed have length $1$.

Now we construct an embedding of $M_{i}$ into $G$. Because $M_{i}$ is a refinement of $M''_{i}$, it is equivalent to consider $M''_{i}$. For each matching edges $e=(v^{\src},v^{\sink})\in M''_{i}$, we have $M''_{i}(e)\leq \lceil M'_{i}(e)\rceil$. Furthermore, because $M'_{i}(e) = \sum_{P\in \path(f_{i}^{G})\text{ connecting $v^{\src}$ and $v^{\sink}$}} f_{i}^{G}(P)$ and all $f^{G}_{i}(P)$ are multiples of $\eta$, we have $M''_{i}(e)\leq \lceil 
M'_{i}(e) \rceil\leq M'_{i}(e)/\eta\leq \sum_{P\in \path(\hat{f}_{i}^{G})\text{ connecting $v^{\src}$ and $v^{\sink}$}} \hat{f}_{i}^{G}(P)$. Hence we can obtain an embedding $\Pi_{M_{i}\to G}$ of $M_{i}$ into $G$ by decomposing each flow path $P$ of $\hat{f}^{G}_{i}$ into $\hat{f}^{G}_{i}(P)$ many unweighted copies. The congestion and length of $\Pi_{M_{i}\to G}$ are at most the congestion and length of $\hat{f}_{i}^{G}$ respectively.

\end{proof}

\paragraph{The second case $\val(f_{i})<\phi|A^{\src}_{i}|$.}

This will be the last iteration of the algorithm. We will compute the last unweighted matching $M_{i}$ between virtual nodes in $A^{\src}_{i}$ and $A^{\sink}_{i}$, and finalize $M_{\CM}$ by adding $M_{i}$. Also, we will compute the integral moving cut $C_{\CM}$ and the partitions $(A^{\src}_{M}, A^{\src}_{U})$ of $A^{\src}$ and $(A^{\sink}_{M}, A^{\sink}_{U})$ of $A^{\sink}$.

A natural idea to let $A^{\src}_{M}$ and $A^{\sink}_{M}$ be the virtual nodes matched by $M_{\CM}$, and then compute $C_{\CM}$ to $h_{\CM}$-separate $A^{\src}_{i}$ and $A^{\sink}_{i}$ by rounding the fractional moving cut $h_{\CM}\cdot w_{i}$. Roughly speaking, each $(h_{\CM}+2)$-length path on $H_{i}$ (which corresponding to an $h_{\CM}$-length path on $G$ between $A^{\src}_{i}$ and $A^{\sink}_{i}$) will have total $w_{i}$-weight at least $1$, so multiplying $w_{i}$ by $h_{\CM}$ will bring $h_{\CM}$ length increasing to this path, which gives the desired separation. However, a big issue of this idea is that, $w_{i}$ may assign large cut value to the source and sink edges of $H_{i}$ but these edges do not exist in $G$ (so we call them \textit{dummy edges}). This means some $h_{\CM}$-length path on $G$ between $A^{\src}_{i}$ and $A^{\sink}_{i}$ may not get enough length increase. This is the reason why we compute the last matching $M_{i}$. By further matching $A^{\src}_{i}$ and $A^{\sink}_{i}$, we can somehow exclude all $(h_{\CM}+2)$-length paths with large cut value on dummy edges, where we heavily exploit the additional property (2) of \Cref{thm:LocalLengthConstrainedFlow}.

Concretely, we consider another auxiliary graph $\bar{H}_{i} = H(x\cdot A^{\src}_{i}, A^{\sink})$ for some large constant $x$. Let $\bar{U}_{i}$ be the capacity function of $\bar{H}_{i}$. We run the local flow algorithm from \Cref{thm:LocalLengthConstrainedFlow} on $\bar{H}_{i}$ with $\delta = 0.01$ and length parameter $h = h_{\CM} +2$ to get $(\bar{f}_{i},\bar{w}_{i})$. Recall that property (2) in \Cref{lemma:AdditionalProperties} shows that each edge with non-negligible $\bar{w}_{i}$-weight ($\bar{w}(e) \geq 1/m$) will be almost saturated $(\bar{f}(e)\geq \Omega(\bar{U}_{i}(e)))$.

To construct $M_{i}$, we let $V^{\src}_{\sat}$ collect all vertices $v^{\src}\in\supp(A^{\src}_{i})$ whose source edges have $\bar{w}_{i}(s,v^{\src})\geq 1/m$, and let $V^{\sink}_{\sat}$ collect all vertices $v^{\sink}\in \supp(A^{\sink}_{i})$ s.t. $\bar{w}_{i}(v^{\sink},t)\geq 1/m$. 
Similar to the first case, we transform $f_{i}$ to a fractional matching $M'_{i}$ between vertices in $\supp(A^{\src}_{i})$ and $\supp(A^{\sink}_{i})$. Each pair of vertices $v^{\src}\in \supp(A^{\src}_{i})$ and $v^{\sink}\in \supp(A^{\sink}_{i})$ has a matching edge with weight $\sum_{P\in\path(f_{i})\text{ connecting $v^{\src}$ and $v^{\sink}$}}f_{i}(P)$. Again, by applying \Cref{lemma:MatchingRounding}, we round $M'_{i}$ to an integral matching $M''_{i}$ between vertices in $\supp(A^{\src}_{i})$ and $\supp(A^{\sink}_{i})$. Next, we refine $M''_{i}$ to be a matching between virtual nodes (similar to the refinement operation in the first case), and let $M_{i}$ be a submatching of the refined $M''_{i}$ s.t. each vertex $v^{\src}\in \supp(A^{\src}_{i})$ has exactly $\min\{M''^{,\src}_{i}(v^{\src}),A^{\src}_{i}(v^{\src})\}$ many source virtual nodes matched by $M_{i}$.

\begin{claim}
For each vertex $v^{\src}\in V^{\src}_{\sat}$, we have $M_{i}^{\src}(v^{\src})= A^{\src}_{i}(v^{\src})$.
\label{claim:4.19}
\end{claim}
\begin{proof}
We first show $M_{i}''^{,\src}(v^{\src})\geq A^{\src}_{i}(v^{\src})$ by
\[
M_{i}''^{,\src}(v^{\src})\geq \lfloor M_{i}'^{,\src}(v^{\src}) \rfloor = \lfloor \bar{f}_{i}(s,v^{\src}) \rfloor \geq \Omega(\bar{U}_{i}(s,v^{\src})) = \Omega(x\cdot A^{\src}_{i}(v^{\src}))\geq A^{\src}_{i}(v^{\src}),
\]
where $M_{i}''^{,\src}(v^{\src})\geq \lfloor M_{i}'^{,\src}(v^{\src}) \rfloor$ is by \Cref{lemma:MatchingRounding}, $\lfloor M_{i}'^{,\src}(v^{\src}) \rfloor = \lfloor \bar{f}_{i}(s,v^{\src}) \rfloor$ is from the construction of $M'_{i}$, $\lfloor \bar{f}_{i}(s,v^{\src}) \rfloor \geq \Omega(\bar{U}_{i}(s,v^{\src}))$ is by $\bar{w}_{i}(s,v^{\src})$ and property (2) in \Cref{lemma:AdditionalProperties}, and $\Omega(x\cdot A^{\src}_{i}(v^{\src}))\geq A^{\src}_{i}(v^{\src})$ is because we choose a sufficiently large constant $x$.

Therefore, by our construction of $M_{i}$, we have $M^{\src}_{i}(v^{\src}) = \min\{M''^{,\src}(v^{\src}),A^{\src}_{i}(v^{\src})\} = A^{\src}_{i}(v^{\src})$.
\end{proof}

\paragraph{Finalizing $M_{\CM}$.} We add $M_{i}$ to $M_{\CM}$, and $M_{\CM}$ is now finalized. $M_{\CM}$ can be embedded into $G$ with congestion $\gamma_{\CM} = O(\log^{2} n/\phi) = O(\log^{4}n/\phi_{\CM})$ and length $h_{\CM}$. The reason is that the number of iterations is $O(\log n/\phi)$ by \Cref{lemma:IterationsOfCutmatch}, and the matching $M_{i}$ in each iteration can be embedded into $G$ with congestion $O(\log n)$ and length $h_{\CM}$ by \Cref{lemma:PropertyOfMi} (it also holds for the $M_{i}$ in the last iteration by the same proof of \Cref{lemma:PropertyOfMi}). Furthermore, this embedding $\Pi_{M_{\CM}\to G}$ can be computed explicitly following the proof of \Cref{lemma:PropertyOfMi}. 

\begin{lemma}
The number of iterations is $O(\log n/\phi)$.
\label{lemma:IterationsOfCutmatch}
\end{lemma}
\begin{proof}
By \Cref{lemma:PropertyOfMi}, each first-case $i$-th iteration will add a matching $M_{i}$ with size $|M_{i}|\geq \lceil\phi\cdot  |A^{\src}_{i}|\rceil$. Also, recall that the $i$-th iteration have 
\[
|A^{\src}_{i}| = |A^{\src}|-\sum_{1\leq i'\leq i-1}|M_{i'}| = |A^{\src}_{i-1}| - |M_{i-1}|\leq \lfloor 
(1-\phi)\cdot|A^{\src}_{i-1}| \rfloor.
\]
Thus the number of iterations is at most $\log_{1/(1-\phi)}\poly(n) = O(\log n/\phi)$.
\end{proof}

\

\noindent\textbf{Partitions of $A^{\src}$ and $A^{\sink}$.} 
Let $M^{\src}_{\CM}$ (resp. $M^{\sink}_{\CM}$) denote the set of all source (resp. sink) virtual nodes matched by $M_{\CM}$.
The set of source virtual nodes $A^{\src}$ is partitioned into 
\[
A^{\src}_{M} = M_{\CM}^{\src}\text{ and }A^{\src}_{U} = A^{\src}\setminus M_{\CM}^{\src}.
\]
Slightly differently, $A^{\sink}$ is partitioned into
\[
A^{\sink}_{M} = M_{\CM}^{\sink}\cup \bigcup_{v^{\sink}\in V^{\sink}_{\sat}} A^{\sink}_{i}(v^{\sink})\text{ and }A^{\sink}_{U} = A^{\sink}\setminus A^{\sink}_{M}.
\]
That is, $A^{\sink}_{M}$ is the union of matched sink virtual node and all remaining sink virtual nodes of vertices in $V^{\sink}_{\sat}$.

To see the correctness, we have $|A^{\src}_{M}|\leq |A^{\sink}_{M}|$ because the matching $M_{\CM}$ match each virtual node in $A^{\src}_{M}$ to a distinct virtual node in $A^{\sink}_{M}$. Next we show $|A^{\sink}_{M}|\leq O(|A^{\src}_{M}|)$. Because $|M^{\sink}_{\CM}| = |M^{\src}_{\CM}| = |A^{\src}_{M}|$, it suffices to argue that $\sum_{v^{\sink}\in V^{\sink}_{\sat}} A^{\sink}_{i}(v^{\sink})\leq O(|A^{\src}_{M}|)$. First, we have
\begin{align*}
\sum_{v^{\sink}\in V^{\sink}_{\sat}} A^{\sink}_{i}(v^{\sink})= \sum_{v^{\sink}\in V^{\sink}_{\sat}} \bar{U}_{i}(v^{\sink},t)\leq \sum_{v^{\sink}\in V^{\sink}_{\sat}} O(\bar{f}_{i}(v^{\sink},t))\leq O(\val(\bar{f}_{i})),
\end{align*}
Note that for each $v^{\sink}\in V^{\sink}_{\sat}$, $\bar{U}_{i}(v^{\sink},t) \leq O(\bar{f}_{i}(v^{\sink},t))$ by $\bar{w}_{i}(v^{\sink},t)>1/m$ and property (2) in \Cref{lemma:AdditionalProperties}. We further have
\begin{align*}
\val(\bar{f}_{i}) &= |M'_{i}| \leq |M''_{i}|\\
&\leq \sum_{\substack{v^{\src}\in \supp(A^{\src}_{i})\text{ 
s.t. }\\M''^{,\src}_{i}(v^{\src})\leq A^{\src}_{i}(v^{\src})}} M''^{,\src}_{i}(v^{\src}) + \sum_{\substack{v^{\src}\in \supp(A^{\src}_{i})\text{ 
s.t. }\\M''^{,\src}_{i}(v^{\src})> A^{\src}_{i}(v^{\src})}} \bar{U}_{i}(s,v^{\src})\\
&\leq \sum_{\substack{v^{\src}\in \supp(A^{\src}_{i})\text{ 
s.t. }\\M''^{,\src}_{i}(v^{\src})\leq A^{\src}_{i}(v^{\src})}} M^{\src}_{i}(v^{\src}) + \sum_{\substack{v^{\src}\in \supp(A^{\src}_{i})\text{ 
s.t. }\\M''^{,\src}_{i}(v^{\src})> A^{\src}_{i}(v^{\src})}} O(M^{\src}_{i}(v^{\src}))\\
& = O(|M_{i}|)\leq O(|A^{\src}_{M}|).
\end{align*}
In the second line, we exploit that $M''^{,\src}_{i}(v^{\src})\leq \bar{U}_{i}(s,v^{\src})$ for each $v^{\src}\in \supp(A^{\src}_{i})$. In the third line, we exploit that, for each $v^{\src}\in \supp(A^{\src}_{i})$ s.t. $M''^{,\src}_{i}(v^{\src})>A^{\src}_{i}(v^{\src})$, $M^{\src}_{i}(v^{\src}) = A^{\src}_{i}(v^{\src})$ by the construction of $M_{i}$, which implies $\bar{U}_{i}(s,v^{\src})\leq O(M^{\src}_{i}(v^{\src}))$ combining $\bar{U}_{i}(s,v^{\src}) = x\cdot A^{\src}_{i}(v^{\src})$.

The observation below is useful when constructing the moving cut $C_{\CM}$.

\begin{observation}
$\supp(A^{\src}_{U})$ is disjoint from $V^{\src}_{\sat}$, and $\supp(A^{\sink}_{U})$ is disjoint from $V^{\sink}_{\sat}$.
\label{ob:NoSatVertices}
\end{observation}
\begin{proof}
We have $\supp(A^{\src}_{U})$ is disjoint from $V^{\src}_{\sat}$ because for each vertex $v^{\src}\in V^{\src}_{\sat}$, virtual nodes in $A^{\src}(v^{\src})\setminus A^{\src}_{i}(v^{\src})$ have been matched in the previous iterations, and virtual nodes in $A^{\src}_{i}(v^{\src})$ are matched by $M_{i}$ by \Cref{claim:4.19}. Also $\supp(A^{\sink}_{U})$ is disjoint from $V^{\sink}_{\sat}$ just by definition. 
\end{proof}

\

\noindent\textbf{Construction of $C_{\CM}$.} For each edge $e\in E(G)$ corresponding to two directed edges $e_{1},e_{2}\in H_{i}$, if $\max\{\bar{w}_{i}(e_{1}),\bar{w}_{i}(e_{2})\}\geq 1/m^{3}$ (this threshold is from property (3) in \Cref{lemma:AdditionalProperties}), we let $C_{\CM}(e) = \lfloor 3h_{\CM}\cdot\max\{\bar{w}_{i}(e_{1}),\bar{w}_{i}(e_{2})\}\rfloor$, otherwise $C_{\CM}(e) = 0$. 

Now we show the correctness of $C_{\CM}$. Consider an arbitrary pair of vertices $v^{\src}\in \supp(A^{\src}_{U})$ and $v^{\src}\in \supp(A^{\sink}_{U})$, and an $h_{\CM}$-length path $P_{G}$ connecting them on $G$. $P_{G}$ corresponds to a path $P_{H} = s\to v^{\src}\to\cdots P_{G}\cdots\to v^{\sink}\to t$ on $\bar{H}_{i}$ (by adding back the source and sink edges $(s,v^{\src})$ and $(v^{\sink},t)$) with length at most $h_{\CM}+2$, so $\bar{w}_{i}(P_{H}) \geq 1$ by the feasibility of $\bar{w}_{i}$. Because $v^{\src}\notin V^{\src}_{\sat}$ and $v^{\sink}\notin V^{\sink}_{\sat}$ by \Cref{ob:NoSatVertices}, we have $w'_{i}(s,v^{\src})<1/m$ and $w'_{i}(v^{\sink},t)<1/m$. Therefore,
\begin{align*}
C_{\CM}(P_{G}) &\geq \sum_{e\in P_{G}}\lfloor 3\cdot h_{\CM}\cdot(\max\{\bar{w}_{i}(e_{1}),\bar{w}_{i}(e_{2})\} - 1/m^{3}) \rfloor\\
&\geq \sum_{e\in P_{G}}(3\cdot h_{\CM}\cdot \max\{\bar{w}_{i}(e_{1}),\bar{w}_{i}(e_{2})\}-1-3h_{\CM}/m^{3})\\
&\geq (3\cdot h_{\CM}\cdot\sum_{e\in P_{G}} \max\{\bar{w}_{i}(e_{1}),\bar{w}_{i}(e_{2})\}) - (1+3h_{\CM}/m^{3})\cdot|E(P_{G})|\\
&\geq 3h_{\CM}\cdot(1-2/m)- h_{\CM}\cdot (1+3/m^{2})\\
&> h_{\CM}.
\end{align*}
Here the first inequality is by the definition of $C_{\CM}$. About the fourth inequality, note that
\[
\sum_{e\in P_{G}}\max\{\bar{w}_{i}(e_{1}),\bar{w}_{i}(e_{2})\}\geq \bar{w}_{i}(P_{H}) - 2/m \geq 1 - 2/m
\]
because the source and sink edges on $P_{H}$ have $\bar{w}_{i}$-weight at most $1/m$. Moreover, 
\[
(1+3h_{\CM}/m^{3})\cdot |E(P_{G})| \leq |E(P_{G})| + 3\cdot h_{\CM}\cdot |E(P_{G})|/m^{3}\leq h_{\CM}\cdot(1+3/m^{2})
\]
because $|E(P_{G})|\leq h_{\CM}$ and $|E(P_{G})|\leq m$.

Furthermore, $C_{\CM}$ has size 
\begin{align*}
|C_{\CM}|&\leq O(h_{\CM}\cdot |\bar{w}_{i}|)\leq O(h_{\CM}\cdot \val(\bar{f}_{i})\cdot\log n)\leq O(h_{\CM}\cdot \val(f_{i})\cdot\log^{2} n)\\
&\leq O(h_{\CM}\cdot \phi\cdot |A_{\src}|\cdot \log^{2} n) = O(h_{\CM}\cdot \phi_{\CM}\cdot |A_{\src}|).
\end{align*}
Here $|\bar{w}_{i}|\leq O(\val(\bar{f}_{i}))$ is by the approximation guarantee of the local flow algorithm on $\bar{H}_{i}$. Also, 
$\val(\bar{f}_{i})\leq O(\val(f_{i}))$ is by the following reasons. First, $\bar{f}_{i}$ is a feasible $(h_{\CM}+2)$-length flow in $\bar{H}_{i}$, so it is a feasible $(h_{\CM}+2)$-length flow in $x\cdot H_{i}$ (recall that $H_{i} = H(A^{\src}_{i}, A^{\sink}_{i})$ and $\bar{H}_{i} = H(x\cdot A^{\src}_{i},A^{\sink}_{i})$). This means $\val(\bar{f}_{i})$ is at most $x$ times the value of maximum $(h_{\CM}+2)$-length flow in $H_{i}$. Hence $\val(\bar{f}_{i})\leq O(x\cdot \val(f_{i})) = O(\val(f_{i})\log n)$ because $f_{i}$ is a $O(\log n)$-approximation $(h_{\CM}+2)$-length maxflow and $x$ is a constant.

\paragraph{Construction of Landmarks.} We still use $i$ to denote the last iteration of the local flow algorithm. %

The first step is to construct a set ${\cal P}$ of $h_{\CM}$-length paths on $G$ that \textit{covers} $\supp(C_{\CM})$ by taking all flow paths of $\bar{f}_{i}$ with source and sink edges removed, where ``cover'' means that, for each edge $e\in E(G)$ with $C_{\CM}(e)>0$, there exists a path $P\in {\cal P}$ s.t. $e\in P$. The reason is as follows. Each edge $e\in E(G)$ (corresponding to $e_{1},e_{2}\in \bar{H}_{i}$) with $C_{\CM}(e)>0$ has $\max\{\bar{w}_{i}(e_{1}),\bar{w}_{i}(e_{2})\}\geq 1/m^{3}$, so either $\bar{f}_{i}(e_{1})>0$ or $\bar{f}_{i}(e_{2})>0$ by property (3) in \Cref{lemma:AdditionalProperties}, which implies $e\in P$ for some $P\in{\cal P}$. 

Next, we decompose paths in ${\cal P}$ into \textit{heavy edges} and \textit{segments}. Let $\sigma' = h_{\CM}/(3\cdot \kappa_{\sigma})$. Heavy edges are edges $e\in E({\cal P})$ s.t. $\ell_{G-C_{\CM}}(e)\geq \sigma'$. Segments are constructed by processing each path $P\in{\cal P}$ as follows. First, remove heavy edges from $P$. Second, for each remaining subpath $\wtilde{P}$, further decompose it into segments (i.e. shorter subpaths) $P'$ s.t. $\sigma'\leq \ell_{G-C_{\CM}}(P')\leq 3\sigma'$. This can be done by picking a prefix $P'$ of $\wtilde{P}$ with $\sigma'\leq \ell_{G-C_{\CM}}(P')\leq 2\sigma'$ repeatedly until the last subpath of $\wtilde{P}$ has $(G-C_{\CM})$-length less than $\sigma'$, and then concatenating the last subpath to its previous segment.

\begin{claim}
For some $\eta = \Theta(\log m)$, the total number of heavy edges is at most $O(\val(\bar{f}_{i})\cdot h_{\CM}/(\eta\cdot\sigma') + |C_{\CM}|/\sigma')$, and the total number of segments is at most $O(\val(\bar{f}_{i})\cdot h_{\CM}/(\eta\cdot\sigma') + |C_{\CM}|/(\eta\cdot\sigma'))$. 
\end{claim}
\begin{proof}

First we have $|{\cal P}|\leq \val(\bar{f}_{i})/\eta$ because each path $P_{H}\in \supp(\bar{f}_{i})$ has $\bar{f}_{i}(P_{H})\geq \eta$ by property (1) in \Cref{lemma:AdditionalProperties}. Again by this property, each edge in $G$ will only appear in $O(1/\eta)$ paths in ${\cal P}$, which means $\sum_{P\in{\cal P}}C_{\CM}(P)\leq O(|C_{\CM}|/\eta)$.

To bound the number of heavy edges, note that each heavy edge $e$ has either $\ell_{G}(e)\geq \sigma'/2$ or $C_{\CM}(e)\geq \sigma'/2$ because $\ell_{G-C_{\CM}}(e)\geq \sigma'$. For the heavy edges with $\ell_{G}(e)\geq \sigma'/2$, there are at most $O(|{\cal P}|\cdot h_{\CM}/\sigma')$ of them  because each path in ${\cal P}$ has length at most $h_{\CM}$ on $G$. For the heavy edges with $C_{\CM}(e)\geq \sigma'/2$, the number of them is obviously $O(|C_{\CM}|/\sigma')$. Therefore, the total number of heavy edges is $O(|{\cal P}|\cdot h_{\CM}/\sigma' + |C_{\CM}|/\sigma') = O(\val(\bar{f}_{i})\cdot h_{\CM}/(\eta\cdot\sigma') + |C_{\CM}|/\sigma')$.

We can bound the number of segments similarly. Each segment $P'$ has either $\ell_{G}(P') \geq \sigma'/2$ or $C_{\CM}(P')\geq \sigma'/2$. There are at most $O(|{\cal P}|\cdot h_{\CM}/\sigma')$ segments with $\ell_{G}(P')\geq \sigma'/2$. The number of segments with $C_{\CM}(P')\geq \sigma'/2$ is at most $O(\sum_{P\in{\cal P}}C_{\CM}(P)/\sigma')$. Thus the total number of segments is
\[
O(|{\cal P}|\cdot h_{\CM}/\sigma' + \sum_{P\in{\cal P}}C_{\CM}(P)/\sigma') = O(\val(\bar{f}_{i})\cdot h_{\CM}/(\eta\cdot\sigma') + |C_{\CM}|/(\eta\cdot\sigma')).
\]
\end{proof}

Given the decomposition, we construct the landmark set $L$ as follows. For each heavy edge $e=(u,v)$, we add both $u$ and $v$ into $L$ (we say $u$ and $v$ are landmarks \textit{created} by $e$). For each segment $P'$, we add an arbitrary vertex $v\in P'$ into $L$ (similarly, $v$ is the landmark \textit{created} by $P'$). The size of $L$ is linear to the number of heavy edges and segments, so 
\begin{align*}
|L| &= O(\val(\bar{f}_{i})\cdot h_{\CM}/(\eta\cdot\sigma') + |C_{\CM}|/\sigma') + O(\val(\bar{f}_{i})\cdot h_{\CM}/(\eta\cdot\sigma') + |C_{\CM}|/(\eta\cdot\sigma'))\\
&\leq O(h_{\CM}\cdot\phi_{\CM}\cdot |A_{\src}|\cdot\log^{2}n/(\eta\cdot\sigma'))\\
&\leq O(\phi_{\CM}\cdot |A^{\src}|\cdot \log^{3} n\cdot\kappa_{\sigma}),
\end{align*}
Here the first inequality is because $\val(\bar{f}_{i})\leq O(x\cdot \val(f_{i}))\leq O(\val(f_{i})\log n)$ (we have shown this when bounding $|C_{\CM}|$). Moreover, recall that $\val(f_{i})\leq \phi\cdot|A^{\src}_{i}| = \phi_{\CM}\cdot|A^{\src}_{i}|/\log^{2}n$ by the condition of the second case, and $|C_{\CM}|\leq O(h_{\CM}\cdot\phi_{\CM}\cdot |A_{\src}|)$. The last inequality is because $\eta = \Theta(\log n)$ and $\sigma' = h_{\CM}/(3\cdot\kappa_{\sigma})$.

$L$ is a landmark set of $C_{\CM}$ on $G$ with distortion $\sigma_{\CM} = 3\sigma' = h_{\CM}/\kappa_{\sigma}$ by the following reasons. All edges $e\in E(G)$ with $\ell_{G-C_{\CM}}(e)>\sigma_{\CM}$ are heavy edges, so their endpoints are landmarks. For those $e\in E(G)$ with $C_{\CM}(e)>0$ and $\ell_{G-C_{\CM}}(e)\leq \sigma_{\CM}$, it must belong to a segment (or it is a heavy edge itself), so the landmark created by this segment (or this heavy edge) has $(G-C_{\CM})$-distance at most $3\sigma'=\sigma_{\CM}$ to endpoints of $e$.

\paragraph{Time Analysis.} The running time is dominated by the local flow subroutine of each iteration. For iteration $i$, the local flow subroutine will be run on graph $H_{i}$ s.t. for each $v\in V(G)$, $\deg_{H_{i}}^{+}(v) \leq \deg_{G}(v) + 2$ (because each vertex may incident to the source edge and sink edge added to $H_{i}$), $\Delta_{H_{i}}(v) = A^{\src}_{i}(v) \leq A^{\src}(v)$ and $\nabla_{H_{i}}(v) = A^{\sink}_{i}(v) \geq A^{\sink}(v) - M^{\sink}_{\CM}(v)$. 
By \Cref{thm:LocalLengthConstrainedFlow}, the running time of local flow on $H_{i}$ is (pick parameter $c = 3$) 
\begin{align*}
&~~~~\poly(h_{\CM},\log n)\cdot\sum_{v\in V(H_{i})\setminus\{s,t\}}(3\cdot\Delta_{H_{i}}(v) + \max\{\deg^{+}_{H_{i}}(v)-3\cdot\nabla_{H_{i}}(v),0\})\\
&\leq\poly(h_{\CM}, \log n)\cdot\sum_{v\in V(G)}(3\cdot A^{\src}(v) + \max\{\deg_{G}(v)+2-3\cdot(A^{\sink}(v)-M^{\sink}_{\CM}(v)),0\})\\
&\leq \poly(h_{\CM}, \log n)\cdot\\
&~~~~\sum_{v\in V(G)}(3\cdot A^{\src}(v) + \max\{\deg_{G}(v)+2-3\cdot(A^{\sink}(v)+A^{\src}(v)),0\} + 3\cdot A^{\src}(v) + 3\cdot M_{\CM}^{\sink}(v))\\
&\leq \poly(h_{\CM}, \log n)\cdot (9\cdot |A^{\src}| + \sum_{v\in V(G)}\max\{\deg_{G}(v) - (A^{\src}(v) + A^{\sink}(v)), 0\})\\
&\leq\poly(h_{\CM}, \log n)\cdot|A^{\src}|.
\end{align*}
The third inequality is because $\sum_{v\in V(G)}M^{\sink}_{\CM}(v) = |M_{\CM}| \leq |A^{\src}|$ and $A^{\src}(v) + A^{\sink}(v)\geq 1$ and the last inequality is by $A^{\src}(v) + A^{\sink}(v)\geq \deg_{G}(v)$ for all $v\in V(G)$. The last iteration will run one more local flow subroutine on graph $\bar{H}_{i}$. The running time bound is the same by the same analysis. Finally, the total running time is $\poly(h_{\CM},\log n)\cdot |A^{\src}|/\phi_{\CM}$ because the number of iterations is $O(\log n/\phi) = O(\log^{3}/\phi_{\CM})$ by \Cref{lemma:IterationsOfCutmatch}.

%% file: 13-path_reporting.tex
\section{Dynamic Routers with Path Reporting}
\label{sect:DynamicRouter}
In this section, we introduce a dynamic router algorithm, which is a key tool of the dynamic length-constrained expander decomposition in \Cref{sect:DynamicCertifiedED,sect:DynDenseED}. As we mentioned in \Cref{sect:GlobalParameters}, the number of batched updates received by a router (denoted by $t_{\local}$ in \Cref{thm:Router}) might not be the same with the global time parameter $t$. We now specify the relations between $t_{\local}$ and $t$. In fact, when we apply \Cref{thm:DynamicED} in \Cref{sect:DynamicCertifiedED}, we have $t_{\local} = O(t) = O(1/\epsilon)$. When we apply \Cref{thm:Router} in \Cref{sect:DynDenseED}, we have $t_{\local} = O(t\cdot\lambda_{\insLM,t}) = O(1/\epsilon^{5})$, where $\lambda_{\insLM,t} = O(1/\epsilon^{4})$ is from \Cref{lemma:LandmarkClosure}. To avoid clutter, we manually set an upper bound $t_{\local} = \poly(1/\epsilon)$ in the statement of \Cref{thm:Router}.

\begin{theorem}[Dynamic Routers]
Given a set $V$, there is an algorithm that initializes a $(h^{(0)}_{\rou},\gamma^{(0)}_{rt})$-router $R^{(0)}$ with $V(R^{(0)}) = V$ for node-weighting $\mathds{1}(V)$ with $h^{(0)}_{\rou}=\lambda_{\rt,h}(0)$, $\gamma^{(0)}_{\rt} = \kappa_{\rt,\gamma}(0)$ and maximum degree $\kappa_{\rou,\deg} = n^{O(\epsilon^{4})}$, where $\lambda_{\rt,h}(i) = 2^{O(1/\epsilon^{4})}+O(i)$ and $\kappa_{\rt,\gamma}(i) = 2^{O(1/\epsilon^{4}+i)}$ are global parameters (with respect to $i$). For some $t_{\local} = \poly(1/\epsilon)$, the algorithm can also maintain $R$ under a sequence $\pi^{(1)},\pi^{(2)},...,\pi^{(t_{\local})}$ of the following batched updates. 
\begin{itemize}
\item (Edge Deletion) If $\pi_{i}=E_{F}\subseteq E(R^{(i-1)})$ denotes edge deletions, the algorithm computes a set of pruned vertices $V_{\prune}\subseteq V(R^{(i-1)})$ and a $(h^{(i)}_{\rou},\gamma^{(i)}_{rt})$-router $R^{(i)}$ which is the subgraph of $R^{(i-1)}$ induced by $V(R^{(i-1)})\setminus V_{\prune}$ s.t. $R^{(i)}$ has no edge in $E_{F}$,
\[
h^{(i)}_{\rt}=\lambda_{\rt,h}(i),\ \gamma^{(i)}_{\rt} = \kappa_{\rt,\gamma}(i)
\text{ and }  |V_{\prune}|=\lambda_{\rt,\prune}(i)\cdot|\pi_{i}|
\]
where $\lambda_{\rt,\prune}(i) = t_{\local}^{O(1/\epsilon^{4})}\cdot 2^{O(i)} = (1/\epsilon)^{O(1/\epsilon^{4})}\cdot 2^{O(i)}$.

\item (Matching Insertion) If $\pi_{i}=M$ denotes a matching between $V(R^{(i-1)})$ and a set of new vertices $V_{\new}$, the new router $R^{(i)} = R^{(i-1)}\cup M$ (in particular, $V(R^{(i)}) = V(R^{(i-1)})\cup V_{\new}$) is a $(h^{(i)}_{\rt},\gamma^{(i)}_{\rt})$-router with
\[
h^{(i)}_{\rou}=\lambda_{\rt,h}(i),\text{ and }\gamma^{(i)}_{\rt} = \kappa_{\rt,\gamma}(i).
\]
\end{itemize}
The initialization time is $|V(R^{(0)})|\cdot n^{O(\epsilon)}$ and the update time for $\pi^{(i)}$ is $t_{\local}^{O(1/\epsilon^{4})}\cdot 2^{O(i)}\cdot |\pi^{(i)}|= (1/\epsilon)^{O(1/\epsilon^{4})}\cdot 2^{O(i)}\cdot |\pi^{(i)}|$.

Furthermore, at any moment $0\leq i\leq t$, given two vertices $u,v\in V(R^{(i)})$, there is an algorithm that outputs a path $P$ connecting $u$ and $v$ in $R^{(i)}$ with length at most $h^{(i)}_{\rt} = \lambda_{\rt,h}(i)$ in time $O(|P|/\epsilon^{4})$.
\label{thm:Router}
\label{thm:RouterPathReport}
\end{theorem}
\begin{proof}
At any time $0\leq i\leq t$, the router $R^{(i)}$ is always the union of a \textit{core graph} $J^{(i)}$ and an \textit{affiliated forest} ${\cal T}^{(i)}$. Vertices in $V(J^{(i)})$ are \textit{core vertices} and other vertices in $V(R^{(i)})\setminus V(J^{(i)})$ are \textit{affiliated vertices}. Each core vertex $v$ one-one corresponds to an \textit{affiliated tree} $T_{v} \in {\cal T}^{(i)}$, which means each tree $T_{v}\in{\cal T}^{(i)}$ only contains one core vertex $v$ and the other vertices in $T_{v}$ are affiliated vertices. Although the tree edges are undirected, we manually set the cored vertex $v$ to be the root of $T_{v}$ to establish the notions of ancestors and descendants on affiliated trees. Let $\tau^{(i)}$ denote the maximum tree size in ${\cal T}^{(i)}$.

\paragraph{Initialization.} We first introduce a recursive partition of $V$ called the \emph{group hierarchy}. Let $b = \Theta(n^{\epsilon^{4}})$ be an integral parameter only used in this proof. We decompose the vertex set $V$ into \emph{groups} (i.e. subsets of vertices) recursively. We use ${\cal B}^{(0)}$ to denote the initial set of all groups and use ${\cal B}^{(0)}_{j}$ to denote the initial set of groups at level $j$, and the first level only has a single group containing all vertices $V$, i.e. ${\cal B}^{(0)}_{1} = \{V\}$. For each level $j$ and each group $B\in{\cal B}^{(0)}_{j}$ with size $|B|\geq 10\cdot b$, we partition $B$ into $b$ subgroups, each of which has size between $\lfloor|B|/b\rfloor$ and $\lceil|B|/b\rceil$, and then we add these subgroups into level $j+1$. We call these subgroups the \emph{children} of $B$ and we call $B$ the \emph{parent} of them. Observe that the groups form a tree structure. Furthermore, the number of levels is at most $O(\log_{b}|V|) = O(1/\epsilon^{4})$.

Initially, the core graph $J^{(0)}$ has vertices $V(J^{(0)}) = V$, and we construct the edges $E(J^{(0)})$ based on the groups. Concretely, consider each group $B\in{\cal B}$. If $B$ has no children (namely, $B$ is a leaf in the hierarchy), we add a clique among vertices in $B$. If $B$ has children, for each pair of its children $B_{1}$ and $B_{2}$, we add a matching $M_{B_{1},B_{2}}$ between $B_{1}$ and $B_{2}$ with size $\min\{|B_{1}|,|B_{2}|\}$. Recall that $|B_{1}|$ and $|B_{2}|$ differ by at most one, so at most one vertex in $B_{1}$ or $B_{2}$ is not incident to this matching. The edge set $E(J^{(0)})$ is the union of the cliques and matchings mentioned above. 

We set the initial router $R^{(0)} = J^{(0)}$. Namely, there is no affiliated vertices and each affiliated tree contains just a single core vertex.

The maximum degree of $R^{(0)}$ is $n^{O(\epsilon^{4})}$ by the following reasons. For each vertex $v\in V(R^{(0)})$ and each group $B$ containing $v$, the clique or matchings added by $B$ will bring at most $b$ incident edges to $v$. Because the number of group $B$ containing $v$ is at most $O(1/\epsilon^{4})$ (i.e. the number of levels), we have $\deg_{R^{(0)}}(v)\leq O(1/\epsilon^{4})\cdot b\leq n^{O(\epsilon^{4})}$. 

The initialization time is obviously proportional to the graph size, i.e. $|R^{(0)}| = |V|\cdot n^{O(\epsilon^{4})}$.

\paragraph{The Edge Deletion Update.} When $\pi^{(i)} = E_{F}\subseteq E(R^{(i-1)})$ is a batch of edge deletions, we will compute the set $V_{\prune}$ of pruned vertices as follows. We let $V_{\aff} = V(E_{F})$ collect all endpoints of deleted edges, called \emph{affected vertices}. We then call $V_{\aff,J} = V_{\aff}\cap V(J^{(i-1)})$ \emph{affected core vertices} and naturally $V_{\aff,{\cal T}} = V_{\aff}\setminus V(J^{(j-1)})$ is the set of \emph{affected affiliated vertices}. 

We construct the set $V_{\prune,J}$ of \emph{pruned core vertices} by the following procedure. At the beginning, we let $V_{\prune,J} = V_{\aff,J}$. Then, we iterate $j$ from the maximum level to the first level of ${\cal B}^{(i-1)}$. For each group $B\in {\cal B}^{(i-1)}_{j}$, if $|B\cap V_{\prune,J}|\geq |B|/(100t_{\local})$, we add the whole group $B$ into $V_{\prune,J}$. Finally, the new group hierarchy ${\cal B}^{(i)}$ is obtained by dropping from ${\cal B}^{(i-1)}$ all groups added into $V_{\prune,J}$.

The whole set $V_{\prune}$ of pruned vertices contains all pruned core vertices $V_{\prune,J}$, all affected affiliated vertices $V_{\aff,{\cal T}}$, and the descendants of $V_{\prune,J}\cup V_{\aff,{\cal T}}$ in the affiliated forest.

\paragraph{The Size of $V_{\prune}$.} We first bound the number of pruned core vertices (i.e. $|V_{\prune,J}|$) as follows. Recall the procedure constructing $V_{\prune,J}$. At the beginning of this procedure, $|V_{\prune,J}| = |V_{\aff,J}|\leq |V_{\aff}| \leq 2|E_{F}|$. Then the procedure goes from the maximum level to the first level. At each level $j$ and each group $B\in {\cal B}_{j}^{(i-1)}$, we will add the whole group $B$ into $V_{\prune,J}$ if currently $|B\cap V_{\prune,J}|\geq |B|/(100t)$. Because the groups at level $j$ are disjoint, $|V_{\prune,J}|$ will blow up by a factor of at most $100t$ after processing level $j$. Finally, we have $|V_{\prune,J}|\leq (100t_{\local})^{O(1/\epsilon^{4})}|E_{F}| \leq t_{\local}^{O(1/\epsilon^{4})}|E_{F}|$ because the number of levels is at most $O(1/\epsilon^{4})$.

Next, we can easily bound the number of pruned vertices by $|V_{\prune}|\leq (|V_{\aff,{\cal T}}| + |V_{\prune,J}|)\cdot \tau^{(i)}\leq t_{\local}^{O(1/\epsilon^{4})}\cdot 2^{O(i)}\cdot|E_{F}|$. Here $\tau^{(i)}$ is the maximum size of affiliated trees in ${\cal T}^{(i)}$ (trivially $\tau^{(i)}\leq 2^{i}$). The first inequality is because each affected affiliated vertex and pruned core vertex has at most $\tau^{(i)}$ descendants.

\paragraph{The Matching Insertion Update.} When $\pi^{(i)} = M$ is a matching between new vertices $V_{\new}$ and a subset of old vertices, we let all new vertices be affiliated vertices, and keep the core graph $J^{(i)} = J^{(i-1)}$ unchanged (also ${\cal B}^{(i)} = {\cal B}^{(i-1)}$). Then we add these new affiliated vertices to the affiliated forest, i.e. ${\cal T}^{(i)} = {\cal T}^{(i-1)}\cup M$. In other words, we just update $R^{(i)}$ to be $R^{(i-1)}\cup M$. %

\paragraph{The Update Time.} The matching insertion update obviously takes $O(|\pi^{(i)}|)$ time. The edge deletion update time can be upper bounded by $O(|V_{\prune,J}|/\epsilon^{4})$ because at each level we just need to scan the current $V_{\prune,J}$, and the number of levels is at most $O(1/\epsilon^{4})$. We have shown $|V_{\prune,J}| \leq t_{\local}^{O(1/\epsilon^{4})}\cdot 2^{O(i)}\cdot|\pi^{(i)}|$ and the number of levels is $O(1/\epsilon^{4})$, so the update time is $V_{\prune,J}\cdot O(1/\epsilon^{4}) = t_{\local}^{O(1/\epsilon^{4})}\cdot 2^{O(i)}\cdot|\pi^{(i)}|$.

\paragraph{The Quality of the Router.} We now show that $R^{(i)}$ is an $(h^{(i)}_{\rt},\gamma^{(i)}_{\rt})$-router for node-weighting $\mathds{1}(V)$ with $h^{(i)}_{\rt} = 2^{O(1/\epsilon^{4})} + O(i)$ and $\gamma^{(i)}_{\rt} = 2^{O(1/\epsilon^{4}+i)}$. Before proving this, we first show that each group $B\in B^{(i)}$ only has a small fraction of pruned vertices (i.e. \Cref{ob:PathReportGroupSurvive}), which implies a matching between two surviving groups still has a large fraction of surviving matching edges (i.e. \Cref{ob:PathReportMatchingSurvive}).

\begin{observation}
For each group $B\in B^{(i)}$, $|B\setminus V(J^{(i)})|\leq 0.01|B|$.
\label{ob:PathReportGroupSurvive}
\end{observation}
\begin{proof}
Because $B\in B^{(i)}$, by our update algorithm, $B$ is not dropped in all edge deletion update so far, which means at most $|B|/(100t)$ many vertices are pruned in each previous edge deletion update. Therefore, the total number of pruned vertices in $B$ is at most $i\cdot|B|/(100t)\leq 0.01|B|$.
\end{proof}

\begin{observation}
For each pair of children $(B_{1},B_{2})$ of some group $B\in {\cal B}^{(i)}$, $|M_{B_{1},B_{2}}\cap J^{(i)}|\geq 0.8\max\{|B_{1}|,|B_{2}|\}$.
\label{ob:PathReportMatchingSurvive}
\end{observation}
\begin{proof}
By \Cref{ob:PathReportGroupSurvive}, we can show that $M_{B_{1},B_{2}}$ still has large fraction of surviving matching edges, i.e.
\begin{align*}
|M_{B_{1},B_{2}}\cap J^{(i)}|&\geq |M_{B_{1},B_{2}}| - \max\{|B_{1}\setminus V(J^{(i)})|,|B_{2}\setminus V(J^{(i)})|\}\\
&\geq \min\{|B_{1}|,|B_{2}|\} - 0.01\max\{|B_{1}|,|B_{2}|\}\\
&\geq 0.8\max\{|B_{1}|,|B_{2}|\},
\end{align*}
where the last inequality is because $|B_{1}|$ and $|B_{2}|$ differ by at most 1 and $|B_{1}|,|B_{2}|\geq 10$ (since $B$ has children which implies $|B|\geq 10 b$).
\end{proof}

\begin{lemma}
An arbitrary $\mathds{1}(V(J^{(i)}))$-respecting demand can be routed on the core graph $J^{(i)}$ with length $2^{O(1/\epsilon^{4})}$ and congestion $2^{O(1/\epsilon^{4})}$.
\label{lemma:CoreGraphRouting}
\end{lemma}
\begin{proof}
We will show this by induction. Let $\bar{j} = O(1/\epsilon^{4})$ be the number of levels. Let the induction hypothesis at level $j$ be that, for each group $B\in{\cal B}^{(i)}_{j}$, an arbitrary $\mathds{1}(V(J^{(i)})\cap B)$-respecting demand can be routed on the subgraph $J^{(i)}\cap B$ with length $4^{\bar{j}-j+1}$ and congestion $4^{\bar{j}-j+1}$. In the base case when $j$ is the maximum level, the hypothesis trivially holds because each group $B$ at the maximum level is a leaf in the hierarchy, which means the subgraph $J^{(i)}\cap B$ contains a clique on vertices $V(R^{(i)})\cap B$. 

Consider an arbitrary level $j$ and assume the hypothesis holds for level $j+1$. For each group $B\in {\cal B}^{(i)}_{j}$, if $B$ is a leaf, then the hypothesis holds for $B$ by the same reason. Otherwise, $B$ has several children in $B^{(i-1)}$, and all surviving vertices $V(J^{(i)})\cap B$ are contained in these children. Let $D$ be an arbitrary $\mathds{1}(V(J^{(i)})\cap B)$-respecting demand on $V(J^{(i)})\cap B$. We decompose $D$ into inner-children demands and inter-children demands, i.e.
\[
D = \sum_{\text{children }B_{1}\text{ of }B} D_{\mid B_{1}} + \sum_{\text{children pairs }(B_{1},B_{2})}D_{\mid (B_{1},B_{2})},
\]
where $D_{\mid B_{1}}$ collect all demand pairs of $D$ with both endpoints in $B_{1}$, and $D_{B_{1},B_{2}}$ collect demand pairs with one endpoint in $B_{1}$ and the other one in $B_{2}$. 

For each inner-children demand $D_{\mid B_{1}}$, the hypothesis tells that $D_{\mid B_{1}}$ can be routed on $J^{(i)}\cap B_{1}$ with length $4^{\bar{j}-j}$ and congestion $4^{\bar{j}-j}$. Because the children of $B$ are disjoint, routing all inner-children demands $\sum D_{\mid B_{1}}$ has length $4^{\bar{j}-j}$ and congestion $4^{\bar{j}-j}$ on $J^{(i)}\cap B$.

For each inter-children demand $D_{\mid(B_{1},B_{2})}$ and each demand pair $(u_{1},u_{2})\in \supp(D_{\mid(B_{1},B_{2})})$ (s.t. $u_{1}\in B_{1}$ and $u_{2}\in B_{2}$), we will route it through the surviving matching edges in $M_{B_{1},B_{2}}\cap J^{(i)}$ evenly. Concretely, for each $e=(w_{1},w_{2})\in M_{B_{1},B_{2}}\cap J^{(i)}$ (where $w_{1}\in B_{1}$ and $w_{2}\in B_{2}$), we will route $D_{\mid(B_{1},B_{2})}(u_{1},u_{2})/|M_{B_{1},B_{2}}\cap J^{(i)}|$ units of demand through $e$, and to complete the routing path, we will request the children $B_{1},B_{2}$ to route $D_{\mid(B_{1},B_{2})}(u_{1},u_{2})/|M_{B_{1},B_{2}}\cap J^{(i)}|$ units from $u_{1}$ to $w_{1}$ and $w_{2}$ to $u_{2}$ in $J^{(i)}\cap B_{1}$ and $J^{(i)}\cap B_{2}$ respectively. 

Now we show the length and congestion of the above routing of the total inter-children demand $\sum D_{\mid(B_{1},B_{2})}$ on $J^{(i)}\cap B$. For each matching $M_{B_{1},B_{2}}$, the congestion of each surviving matching edge can be bounded by
\begin{align*}
\sum_{(u_{1},u_{2})\in \supp(D_{\mid(B_{1},B_{2})})} D_{\mid(B_{1},B_{2})}(u_{1},u_{2})/|M_{B_{1},B_{2}}\cap J^{(i)}|\leq \min\{|B_{1}|,|B_{2}|\}/(0.8\max\{|B_{1}|,|B_{2}|\})\leq 2,
\end{align*}
where the first inequality uses \Cref{ob:PathReportMatchingSurvive}.
For each child $B_{1}$ and each surviving vertex $v_{1}\in J^{(i)}\cap B_{1}$, the newly requested inner-children demand involving $v_{1}$ can be bounded as follows. When $v$ acts as $u_{1}$, the contribution is at most
\[
\sum_{\text{children }B_{2}\text{ of }B}\ \sum_{v_{2}\in B_{2}\cap V(J^{(i)})} D_{\mid(B_{1},B_{2})}(v_{1},v_{2}) \leq \sum_{v_{2}\in (B\setminus B_{1})\cap V(J^{(i)})} D(v_{1},v_{2})\leq 1 
\]
because $D$ is a $\mathds{1}(B\cap V(J^{(i)}))$-respecting demand. When $v_{1}$ acts as $w_{1}$ (i.e. $v_{1}$ incident to some surviving matching edges), the contribution is at most
\begin{align*}
&~~~~\sum_{\text{children }B_{2}\text{ of }B}\ \sum_{(u_{1},u_{2})\in \supp(D_{\mid(B_{1},B_{2})})} D_{\mid(B_{1},B_{2})}(u_{1},u_{2})/|M_{B_{1},B_{2}}\cap J^{(i)}|\\
&\leq \sum_{u_{1}\in B_{1}\cap V(J^{(i)})}\ \sum_{u_{2}\in (B\setminus B_{1})\cap V(J^{(i)})} D(u_{1},u_{2})/(0.8|B_{1}|)\\
&\leq 2.
\end{align*}
In the first-line expression, the first summation is enumerating the possible surviving matching edges incident to $v$ (at most one for each other children $B_{2}$), and the second summation is bounding the demands routed through this edge. The first inequality use \Cref{ob:PathReportMatchingSurvive}. The above argument shows that the newly request inner-children demand for each child $B_{1}$ is $3\cdot \mathds{1}(B_{1}\cap V(J^{(i)}))$-respecting, so it can be routed on $J^{(i)}\cap B_{1}$ with length $4^{\bar{j}-j}$ and congestion $3\cdot 4^{\bar{j}-j}$. Finally, concatenate the routing on surviving matching edges and the routing of newly requested inner-children demands, the inter-children demand $\sum D_{\mid(B_{1},B_{2})}$ can be routed on $J^{(i)}\cap B$ with length $2\cdot 4^{\bar{j}-j}+1\leq 4^{\bar{j}-j+1}$ and congestion $3\cdot 4^{\bar{j}-j}$.

Combining two parts, $D$ can be routed on $B\cap J^{(i)}$ with length $4^{\bar{j}-j+1}$ and congestion $(1+3)\cdot 4^{\bar{j}-j} = 4^{\bar{j}-j+1}$ as desired.

\end{proof}

We now ready to show the quality of the router. Recall that $\tau^{(i)}$ denote the maximum size of trees in ${\cal T}^{(i)}$, and trivially we have $\tau^{(i)}\leq 2^{i}$. Consider an arbitrary $\mathds{1}(V(R^{(i)}))$-respecting demand $D$. We will route $D$ in two phases. In the first phase, we route all demands on affiliated vertices to core vertices along the affiliated trees, which will cause length $h_{\cal T} = i$ and congestion $\gamma_{\cal T} = \tau^{(i)} = 2^{i}$ because the affiliated trees in ${\cal T}^{(i)}$ have maximum depth $i$. The first phase will reduce the original demand $D$ to a new demand $D_{\core}$ on core vertices. Note that $D_{\core}$ is a $\tau^{(i)}\cdot \mathds{1}(V(J^{(i)}))$-respecting demand on $J^{(i)}$. In the second phase, we route the demand $D_{\core}$ on the core graph $J^{(i)}$ using \Cref{lemma:CoreGraphRouting}, which has length $h_{J} = 2^{O(1/\epsilon^{4})}$ and congestion $\gamma_{J} = 2^{O(1/\epsilon^{4}+i)}$. By concatenating the routings in two phases, we can obtain a routing of $D$ on $R^{(i)}$ with length $h_{\rt} = 2\cdot h_{\cal T} + h_{J} = 2^{O(1/\epsilon^{4})} + O(i)$ and congestion $\gamma_{\rt} = \gamma_{\cal T} + \gamma_{J} = 2^{O(1/\epsilon^{4} + i)}$.

\paragraph{The Path Reporting Query.} Finally, we describe the path reporting query algorithm. Let $u,v\in V(R^{(i)})$ be the input vertices. Let $u_{J},v_{J}$ be the roots of the affiliated trees containing $u$ and $v$. The path between $u$ and $u_{J}$ and the path between $v$ and $v_{J}$ can be easily outputted by working on the affiliated trees, so it remains to output a path connecting $u_{J}$ and $v_{J}$ on $J^{(i)}$.

The idea to output a path connecting $u_{J}$ and $v_{J}$ is similar (but simpler) to the proof of the quality of the router. The algorithm is recursive. In each recursive step, our input is a group $B$ and two vertices $u_{1},u_{2}\in B\cap V(J^{(i)})$ (the input of the initial recursive step is $u_{J},v_{J}$ and the unique group $B=V\in {\cal B}^{(i)}_{1}$, which obviously satisfies the requirement), and the task is to output a path connecting $u_{1}$ and $v_{2}$ in $B\cap J^{(i)}$. If $B$ is a leaf in the group hierarchy, there is a clique on $B\cap V(J^{(i)})$, so we can just return the edge connecting $u_{1}$ and $v_{2}$. Otherwise, $B$ has several children. If there is a child $B_{1}$ of $B$ containing both $u_{1},u_{2}$, it suffices to invoke a recursive step with input $(B_{1},u_{1},v_{2})$. From now suppose $u_{1}$ and $u_{2}$ are inside two different children $B_{1}$ and $B_{2}$. We pick an arbitrary surviving matching edge $e=(w_{1},w_{2})\in M_{B_{1},B_{2}}\cap J^{(i)}$ (s.t. $w_{1}\in B_{1}$ and $w_{2}\in B_{2}$). Then we invoke two recursive steps with inputs $(B_{1},u_{1},w_{1})$ and $(B_{2},w_{2},u_{2})$ respectively, which output two paths from $u_{1}$ to $w_{1}$ and from $w_{2}$ to $u_{2}$. Concatenating these two paths and the matching edge $e$, we obtain a path connecting $u_{1}$ and $u_{2}$ in $B\cap J^{(i)}$. 

The above algorithm compute a path connecting $u_{J}$ and $v_{J}$ on $J^{(i)}$ with length $2^{O(1/\epsilon^{4})}$, because the group hierarchy has at most $O(1/\epsilon^{4})$ levels. Hence finally we obtain a path connecting $u$ and $v$ with length at most $2^{O(1/\epsilon^{4})} + O(i)$. The path-reporting time is $O(|P|/\epsilon^{4})$ because finding each edge takes time proportional to the number of levels.

\end{proof}

%% file: 5-dynamic_ED.tex
\section{Dynamic Length-Constrained Expander Decomposition}
\label{sect:DynamicCertifiedED}

In this section, we establish an algorithm to maintain a (weak) length-constrained expander decomposition $C$ on an online-batch dynamic graph $G$. In fact, instead of just the integral moving cut $C$, we will maintain an object called \emph{certified expander decomposition} (see \Cref{def:CertifiedED}), which is slightly different with the notion of witnessed-ED in \Cref{def:WitnessedED}. The main theorem of this section is \Cref{thm:DynamicED}.

\begin{definition}[Certified Expander Decomposition]
Let $G$ be a graph with a node-weighting $A$. A \textit{certified expander decomposition} (\textit{certified-ED} for short) of $A$ on $G$ is denoted by $(C, L,{\cal N},{\cal R},\Pi_{{\cal R}\to G})$ where

\begin{itemize}
\item $C$ is an integral moving cut on $G$.
\item $L$ is a landmark set of $C$ on $G$ with distortion $\sigma$.%
\item ${\cal N}$ denotes a $b$-distributed $(h_{\cov}, h_{\sep}, \omega)$-neighborhood cover of $A$ on $G-C$ with $h_{\cov} = h_{\sep} = h$.
\item ${\cal R}=\{R^{S}\mid S\in{\cal N}\}$ is a collection of routers s.t. for each cluster $S\in{\cal N}$, $R^{S}$ is a router with $V(R^{S})\supseteq S$, and it is maintained by \Cref{thm:Router} under at most $f$ update batches. In particular, each vertex in $V(R^{S})$ is either \textit{redundant} or corresponding to a distinct virtual node in $S$. 
\item $\Pi_{{\cal R}\to G}$ denotes an embedding of all routers (i.e. $\bigcup_{R\in{\cal R}} R$) into $G$ with length $h_{\emb}$ and congestion $\gamma_{\emb}$. In this embedding, each router vertex $v_{R}$ is mapped to the vertex $v_{G}\in V(G)$ which owns the corresponding virtual node of $v_{R}$ (if $v_{R}$ is a redundant vertex, then it can be mapped to an arbitrary vertex $v_{G}\in V(G)$).
\end{itemize}
The parameters $\sigma,h_{\trace},\gamma_{\trace},b,h,\omega,f,h_{\emb},\gamma_{\emb}$ are called \textit{quality parameters}, which are part of the parameters of the certified expander decomposition, and we will store them explicitly. In particular, we call $h$ the \emph{length parameter} of this certified-ED. The theorem statements in this section will refer to these parameters. We omit them in the tuple $(C,L,{\cal N},{\cal R},\Pi_{{\cal R}\to G})$ just to avoid clutter.
\label{def:CertifiedED}
\end{definition}

By exploiting elementary data structures like binary search trees, all accesses to certified-EDs in our update algorithm can be done in polylogarithmic time, which is negligible because we allow $n^{O(\epsilon)}$ overhead in the update time. However, in \Cref{sect:FastAccessToCertifiedED} we will design access interfaces with constant access time, in order to achieve 
sublogarithmic query time.

The main differences between our certified-ED and the witnessed-ED in \Cref{def:WitnessedED} are as follows. First, the clusters in our neighborhood cover 
are sets of virtual nodes in $A$ rather than vertices in $V(G)$. Second, the routers in certified-EDs have router vertex corresponding to virtual nodes, and they are instances of our dynamic router algorithm \Cref{thm:Router}. Third, we store an additional landmark set of $C$, which is for reducing the recourse.

For better understanding, we point out that a certified-ED essentially implies some kind of flow characterization of the graph $G$ with moving cut $C$ by \Cref{lemma:CertifiedEDToFlow}. However, Strictly speaking, $C$ is \emph{not} a length-constrained expander decomposition of $G$, because the flow characterization from \Cref{lemma:CertifiedEDToFlow} cannot guarantee that $G-C$ is a length-constrained expander. Precisely, \Cref{lemma:CertifiedEDToFlow} only talks about routability on $G$, while a length-constrained expander decomposition $C$ requires routability on $G-C$. Nonetheless, this \emph{weak} form of length-constrained expander decomposition is sufficient in our scenario, because indeed we only need short paths on $G$ (not $G-C$) to certify pairwise distances. In fact, we only require the length guarantee from \Cref{lemma:CertifiedEDToDist} to obtain distance approximations, and the congestion guarantee of the embedding $\Pi_{{\cal R}\to G}$ is purely for fast update time.

\begin{lemma}
Let $G$ be a graph with a node-weighting $A$. If there exists a certified-ED $(C,L,{\cal N},{\cal R},\Pi_{{\cal R}\to G})$ of $A$ on $G$, then any $h$-length $A$-respecting demand on $G-C$ can be routed on $G$ with length $\lambda_{\rt,h}(f)\cdot h_{\emb}$ and congestion $\kappa_{\rt,\gamma}(f)\cdot \gamma_{\emb}$. 
\label{lemma:CertifiedEDToFlow}
\end{lemma}
The proof of \Cref{lemma:CertifiedEDToFlow} is quite similar to the proof of \Cref{lemma:WitnessToFlow}, so we omit it here.

\begin{lemma}
Let $G$ be a graph with a node-weighting $A$ and a certified-ED $(C,L,{\cal N},{\cal R},\Pi_{{\cal R}\to G})$ of $A$ on $G$. For any pair of vertices $u,v\in V(G)$ s.t. there exists a cluster $S\in{\cal N}$ with $u,v\in \supp(S)$, we have $\dist_{G}(u,v)\leq \lambda_{\rt,h}(f)\cdot h_{\emb}$.
\label{lemma:CertifiedEDToDist}
\end{lemma}
\begin{proof}
This is because each router has diameter at most $\lambda_{\rt,h}(f)$ and it can be embedded into $G$ with length $h_{\emb}$.
\end{proof}

\Cref{thm:DynamicED} concludes our online-batch dynamic algorithm for a dynamic graph $G$ with a dynamic node-weighting $A$. We note that the batched updates for $(G,A)$ are defined slightly different from those just for a graph. Precisely, in \Cref{thm:DynamicED}, each batched update $\pi^{(i)}$ contains five types of unit updates: edge insertions/deletions, isolated vertex insertions/deletions and node-weighting insertions. Without loss of generality, we assume the unit updates in $\pi^{(i)}$ have the same type.
\begin{itemize}
\item When $\pi^{(i)}$ represents batched edge insertions, we are given an edge set $E_{\new}$, which updates $E(G^{(i)}) = E(G^{(i-1)})\cup E_{\new}$, and we say $\pi^{(i)}$ has size $|\pi^{(i)}| = |E_{\new}|$.
\item When $\pi^{(i)}$ represents batched edge deletions, we are given an edge set $F\subseteq E(G^{(i-1)})$, which updates $E(G^{(i)}) = E(G^{(i-1)})\setminus F$. The size is $|\pi^{(i)}| = |F|$.
\item When $\pi^{(i)}$ represents batched isolated vertex insertions, the input is a vertex set $V_{\new}$ with a node-weighting $A_{\new}$ on it, which updates $V(G^{(i)}) = V(G^{(i-1)})\cup V_{\new}$ and $A^{(i)} = A^{(i-1)}\uplus A_{\new}$.\footnote{Recall the definition of operation $\uplus$ from \Cref{sect:PairwiseCovers}.} The size is $|\pi^{(i)}| = \sum_{v\in V_{\new}} A_{\new}(v)$.
\item When $\pi^{(i)}$ represents batched isolated vertex deletions, the input is a set $V_{\del}\subseteq V(G^{(i-1)})$ of isolated vertices on $G^{(i-1)}$, which updates $V(G^{(i)}) = V(G^{(i-1)})\setminus V_{\new}$ and $A^{(i)} = (A^{(i-1)})_{\mid V(G^{(i)})}$. The size is $|\pi^{(i)}| = \sum_{v\in V_{\del}} A^{(i-1)}(v)$.
\item When $\pi^{(i)}$ represents batched node-weighting insertions, the input is a node-weighting $A_{\new}$ on $V(G^{(i-1)})$, which updates $A^{(i)} = A^{(i-1)} + A_{\new}$. The size is $\pi^{(i)} = |A_{\new}|$.
\end{itemize}

Particularly, a batched isolated vertex insertions/deletions update of $(G,A)$ has size depending on the node weighting. However, when we apply \Cref{thm:DynamicED} in the following sessions, the node-weighting $A$ is actually always bounded by a constant, so the size of batched isolated vertex insertions/deletions will still be linear to the size of vertex set as usual.

\begin{theorem}
Let $G$ be a dynamic graph with an incremental node-weighting $A$ s.t. $A^{(0)}\geq \mathds{1}(V(G^{(0)}))$ under $t$ batches of updates $\pi^{(1)},\pi^{(2)},...,\pi^{(t)}$ of edge insertions/deletions, isolated vertex insertions/deletions and node-weighting insertions. 
Given parameters $h$ and $\phi$, there is an algorithm that maintains a certified-ED $(C,L,{\cal N},{\cal R},\Pi_{{\cal R}\to G})$ of some $\bar{A}$ on $G$ satisfying the following.
\begin{itemize}
\item At any time $0\leq i\leq t$, 
$\bar{A}^{(i)} = A^{(i)} + \deg_{G^{(i)}}\geq \deg_{G^{(i)}} + \mathds{1}(V(G^{(i)}))$.
\item The integral moving cut $C$ has initial size $|C^{(0)}| \leq O(|A^{(0)}|\cdot\phi\cdot\lambda_{\dynED,h}\cdot h)$, where \[
\lambda_{\dynED,h} = 2^{O(t)} = 2^{O(1/\epsilon)}.
\]
\item The landmark set $L$ has initial size $|L^{(0)}| \leq O(|A^{(0)}|\cdot\phi)$. Over all time, $L$ has distortion $\sigma = O(h\cdot\lambda_{\dynED,h}/\kappa_{\sigma})$, where $\kappa_{\sigma} = n^{O(\epsilon^{4})}$ is from \Cref{thm:cutmatch}.

\item ${\cal N}$ is a $b$-distributed $(h_{\cov},h_{\sep},\omega)$-neighborhood cover of $\bar{A}$ on $G-C$ with
\[
b = O(t),\ 
h_{\cov} = h_{\sep} = h,\ \omega = t\cdot\kappa_{\PC,\omega},
\]
where $\kappa_{\NC,\omega} = n^{O(\epsilon^{4})}$ is from \Cref{thm:NeighborhoodCovers}.
\item Each router in ${\cal R}$ is maintained by \Cref{thm:Router} under at most $f = O(t)$ batched updates.
\item $\Pi_{{\cal R}\to G}$ is a $(h_{\emb},\gamma_{\emb})$-embedding with
\begin{align*}
&h_{\emb} = O(\lambda_{\init,h}\cdot\lambda_{\dynED,h}\cdot h),\text{ and}\\
&\gamma_{\emb} = O(\kappa_{\dynED,\gamma}/\phi),
\end{align*}
where $\lambda_{\init,h} = 2^{\poly(1/\epsilon)}$ will be defined in \Cref{thm:InitCertifiedED} and
\[
\kappa_{\dynED,\gamma} = (\lambda_{\rt,\prune}(t) \cdot (t\cdot \kappa_{\PC,\omega}\cdot\kappa_{\CM,\gamma}+\kappa_{\init,\gamma})\cdot(\kappa_{\CM,L} + \kappa_{\init,L}))^{O(t)} = n^{O(\epsilon^{3})}.
\]
\end{itemize}
Furthermore, at each update step $\pi^{(i)}$,
\begin{itemize}
\item $C^{(i)} = (C^{(i-1)})_{\mid E(G^{(i)})} + C^{(i)}_{\new}$ with $|C^{(i)}_{\new}| \leq O(|\pi^{(i)}|\cdot\lambda_{\dynED,h}\cdot h)$.
\item $L^{(i)} = (L^{(i-1)}\cap V(G^{(i)})) \cup L^{(i)}_{\new}$ with $|L^{(i)}_{\new}|\leq O(|\pi^{(i)}|)$.
\item 
the recourse from ${\cal N}^{(i-1)}$ and ${\cal N}^{(i)}$ is
\[
\recourse({\cal N}^{(i-1)}\to {\cal N}^{(i)}) = O(\lambda_{\rt,\prune}(t)\cdot t\cdot\kappa_{\PC,\omega}\cdot\kappa_{\dynED,\gamma}\cdot|\pi^{(i)}|/\phi) = n^{O(\epsilon^{3})}\cdot|\pi^{(i)}|/\phi.
\]
\end{itemize}
The initialization time is $(|G^{(0)}| + |A^{(0)}|)\cdot \poly(h)\cdot n^{O(\epsilon)}$. The update time for a batch $\pi^{(i)}$ is
$|\pi^{(i)}|\cdot\poly(h)\cdot n^{O(\epsilon)}/\phi$.

\label{thm:DynamicED}
\end{theorem}

We establish the proof of \Cref{thm:DynamicED} as follows. In \Cref{sect:InitED}, we show the initialize algorithm of certified-EDs. In \Cref{sect:DynEDNodeWeightingInsertion,sect:DynEDEdgeDeletion,sect:DynEDEdgeInsertion}, we show the update algorithms for one batched update of pure node-weighting insertions, edge deletions and edge insertions respectively. Finally, in \Cref{sect:ProofOfDynED}, we complete the proof of \Cref{thm:DynamicED}.

\subsection{Initialization}
\label{sect:InitED}

Our initialization algorithm of certified-EDs is \Cref{thm:InitCertifiedED}, which basically first computes a witnessed-ED by \Cref{thm:WitnessedED}, embeds our dynamic routers on it using expander routing \Cref{thm:WitnessExpanderRouting}, and initialzes the landmark set using \Cref{thm:InitLandmark}.%

\begin{lemma}[Initialization of Landmarks]
Let $G$ be a graph with a node-weighting $A\geq \mathds{1}(V(G))$ and an $(h,h_{\cal R},\gamma_{\cal R},h_{\Pi},\gamma_{\Pi})$-witnessed-ED $(C,{\cal N}_{\wit},{\cal R}_{\wit},\Pi_{{\cal R}_{\wit}\to G-C})$. Let $\omega_{\wit}$ be the width of ${\cal N}_{\wit}$. Given a parameter $\sigma\leq h$, there is an algorithm that computes a new integral moving cut $C'=C_{\mid\supp(C')}$, i.e. $C'$ is the restriction of $C$ on some edge set $\supp(C')\subseteq \supp(C)$, and a landmark set $L$ satisfying the following.
\begin{itemize}
\item $L$ is a landmark set of $C'$ with distortion $\sigma$ and size $|L| = O(\lambda_{\PC,\diam}\cdot \kappa_{\PC,\omega}\cdot |C'|/\sigma)$.%
\item For any simple path $P$ on $G$ with $\ell_{G-C}(P)\geq \sigma$, we have $\ell_{G-C}(P)/2\leq \ell_{G-C'}(P)\leq \ell_{G-C}(P)$.
\end{itemize}
The running time is $(|G| + |A|)\cdot n^{O(\epsilon)}$.%
\label{thm:InitLandmark}
\end{lemma}

\begin{proof}
We first construct a neighborhood cover ${\cal N}$ of vertices $V(G)$ (i.e. of node-weighting $\mathds{1}(V(G))$) on $G-C$ with 
\[
h_{\cov} = \sigma/\lambda_{\PC,\diam},\  h_{\diam} = \sigma\text{ and }\omega = \kappa_{\PC,\omega}
\]
by applying \Cref{thm:NeighborhoodCovers}. Next, we construct $C'$ and $L$ from $C$ as follows. We will select some edges $e\in\supp(C)$ (and set $C'(e) = C(e)$) into the new cut $C'$ by the following two steps.
\begin{enumerate}
\item\label{line:InitLM1} For each cluster $S\in{\cal N}$, let $E_{C,S}$ denote the edges in $\supp(C)$ with \emph{both} endpoints in $S$. If their total cut values $C(E_{C,S})\geq h_{\cov}/10$, we select $E_{C,S}$ into $C'$. Then we pick arbitrary $d(S) = \lceil|V(E_{C,S})|/(h_{\cov}/10)\rceil$ many vertices in $S$ as landmarks into $L$ (we say these $d(S)$ landmarks are \textit{corresponding} to $S$). Note that we can choose such many vertices in $S$ because $|S|\geq |V(E_{C,S})|\geq \lceil|V(E_{C,S})|/(h_{\cov}/10)\rceil$.
\item\label{line:InitLM2} Also, for each edge $e=(u,v)$ with $\ell_{G-C}(e)\geq h_{\cov}/10$ and $C(e)\geq \ell_{G-C}(e)/10$, we select $e$ into $C'$ with $C'(e) = C(e)$ and add its endpoints $u$ and $v$ into $L$. 
\end{enumerate}
In other words, for each cut edge $e\in\supp(C)$ not satisfying $\ell_{G-C}(e)\geq h_{\cov}/10$ and $C(e)\geq \ell_{G-C}(e)/10$ (i.e. the conditions in step \ref{line:InitLM2}), $e$ will be excluded from $C'$ if all clusters $S$ covering $e$ have small total cut values $C(E_{C,S})< h_\cov/10$. 

\paragraph{Correctness.} The distortion and size of the $L$ constructed above is shown by \Cref{claim:InitLandmarkDistortion}. The second guarantee is shown by \Cref{claim:InitLMSimilarDistanceMetric}.

\begin{claim}
$L$ is a landmark set of $C'$ on $G$ with distortion $\sigma = h/\kappa_{\sigma}$. The size of $L$ is bounded by $|L| \leq O(\lambda_{\PC,\diam}\cdot\kappa_{\PC,\omega}\cdot\kappa_{\sigma}\cdot |C'|/h)$.
\label{claim:InitLandmarkDistortion}
\end{claim}
\begin{proof}
For each edge $e\in \supp(C')$ with $\ell_{G-C'}(e)> \sigma$, note that $e$ must be added in step \ref{line:InitLM2} because $\ell_{G-C}(e) = \ell_{G-C'}(e)> \sigma\geq h_{\diam}$. Hence its endpoints will be added to $L$.

For each edge $e\in \supp(C')$ with $\ell_{G-C'}(e)\leq \sigma$, if it is added to $C'$ in step \ref{line:InitLM2}, then trivially both its endpoints have themselves as landmarks. If $e$ is added in step \ref{line:InitLM1} via some cluster $S$, then the landmarks from this cluster will have $(G-C)$-distance at most $h_{\diam}\leq \sigma$ to both $e$'s endpoints. 

Next we bound the size of $L$. In step \ref{line:InitLM1}, each cluster $S\in {\cal N}$ will add $d(S)$ landmarks if $C(E_{C,S})\geq h_{\cov}/10$. Because $C(E_{C,S})\geq h_{\cov}/10$ and $|V(E_{C,S})|\leq 2|E_{C,S}|\leq 2\cdot C(E_{C,S})$, we have $d(S)\leq O(C(E_{C,S})/h_{\cov})$. Furthermore, because each edge $e\in \supp(C)$ will appear in the set $E_{C,S}$ of at most $\omega$ clusters $S$ (one from each clustering), we have 
\[
\sum_{\substack{S\in{\cal N}\text{ 
s.t. }\\C(E_{C,S})\geq h_{\cov}/10}}C(E_{C,S}) = \sum_{\substack{S\in{\cal N}\text{ 
s.t. }\\C(E_{C,S})\geq h_{\cov}/10}}C'(E_{C,S}) \leq \omega\cdot |C'|,
\]
where the first equality is because we add $E_{C,S}$ to $C'$ with the same cut values for all $S$ s.t. $C(E_{C,S})\geq h_{\cov/10}$.
Therefore, the total number of landmarks added in step \ref{line:InitLM1} can be bounded by
\[
\sum_{\substack{S\in{\cal N}\text{ 
s.t. }\\C(E_{C,S})\geq h_{\cov}/10}} d(S)
\leq O(\sum_{\substack{S\in{\cal N}\text{ 
s.t. }\\C(E_{C,S})\geq h_{\cov}/10}}C(E_{C,S})/h_{\cov}) \leq O(\omega\cdot|C'|/h_{\cov}).
\]
For step \ref{line:InitLM2}, each edge that puts its endpoints into $L$ have $C'(e) = C(e) \geq \ell_{G-C}(e)/10\geq h_{\cov}/100$. Hence the number of landmarks added in step \ref{line:InitLM2} is bounded by $O(|C'|/h_{\cov})$. In conclusion
\[
|L|\leq O(\omega\cdot|C'|/h_{\cov}) + O(|C'|/h_{\cov}) = O(\lambda_{\PC,\diam}\cdot\kappa_{\PC,\omega}\cdot |C'|/\sigma).
\]
\end{proof}

\begin{claim}
For an arbitrary simple path $P$ with $\ell_{G-C}(P)\geq h_{\cov}$, we have $\ell_{G-C}(P)/2\leq \ell_{G-C'}(P)\leq \ell_{G-C}(P)$. 
\label{claim:InitLMSimilarDistanceMetric}
\end{claim}
\begin{proof}
First, $\ell_{G-C'}(P)\leq \ell_{G-C}(P)$ holds trivially because $C'$ is a subset of $C$. To show $\ell_{G-C}(P)/2\leq \ell_{G-C'}(P)$, consider the moving cut we drop on $P$, denoted by $C_{P,\drop} = (C-C')_{\mid P}$, and we will show it has size at most a constant fraction of $\ell_{G-C}(P)$, i.e. $C_{P,\drop}(P)\leq 1/2\cdot \ell_{G-C}(P)$. 

For each $e\in \supp(C_{P,\drop})$, if $\ell_{G-C}(e)\geq h_{\cov}/10$, it must have $C(e)<\ell_{G-C}(e)/10$, Otherwise it has been included in $C'$ in step \Cref{line:InitLM2}. Therefore,
\[
\sum_{\substack{e\in\supp(C_{P,\drop}),\\\text{ s.t. }\ell_{G-C}(e)\geq h_{\cov/10}}} C(e)\leq \ell_{G-C}(P)/10.
\]
For those edges $e\in\supp(C_{P,\drop})$ s.t. $\ell_{G-C}(e)<h_{\cov}/10$, we use a charging argument. For each such edge $e$, we charge $g(e',e) = C(e)/h_{\cov}$ dollar to each edge $e'\in P$ s.t. \underline{one} endpoint of $e'$ has $P$-distance at most $h_{\cov}$ to \underline{both} endpoints of $e$, where the $P$-distance between two vertices on $P$ is the 
$(G-C)$-length of subpath between them.  Now, let $g(e')$ be the total dollars charging to $e'$ in the whole charging process. We have
\begin{align*}
\sum_{\substack{e\in\supp(C_{P,\drop}),\\\text{s.t. }\ell_{G-C}(e)<h_{\cov}/10}} C(e)&\leq \sum_{\substack{e\in\supp(C_{P,\drop}),\\\text{s.t. }\ell_{G-C}(e)<h_{\cov}/10}}\sum_{e'\in P\text{ charged by $e$}}(C(e)/h_{\cov}))\cdot \ell_{G-C}(e')\leq \sum_{e'\in P}\ell_{G-C}(e')\cdot g(e')
\end{align*}
where the first inequality holds because the edges $e'$ charged by $e$ have total $(G-C)$-length at least $h_{\cov}$, which is true because $\ell_{G-C}(P)\geq h_{\cov}$ and $\ell_{G-C}(e)<h_{\cov}/10 \le h_\cov$\footnote{For better understanding, here we need $\ell_{G-C}(e)\leq h_{\cov}$ because there would have been no edge $e'$ charged by $e$ if $\ell_{G-C}(e) > h_{\cov}$.}.

Then, all we need is an upper bound on $g(e')$ for each $e'=(u',v')\in P$. By our charging process, each edge $e=(u,v)$ that charges $e'$ must have both $u$ and $v$ with $(G-C)$-distance to $u'$ (or $v'$) at most $h_{\cov}$. Let $S_{u'}$ and $S_{v'}$ be the clusters in ${\cal N}$ covering $\Ball_{G-C}(u',h_{\cov})$ and $\Ball_{G-C}(v',h_{\cov})$ respectively. Then either $e\in E_{C,S_{u'}}$ or $e\in E_{C,S_{v'}}$. Therefore,
\begin{align*}
g(e')&\leq \sum_{e\in E_{C,S_{u'}}\cap \supp(C_{P,\drop})} g(e',e) + \sum_{e\in E_{C,S_{v'}}\cap \supp(C_{P,\drop})} g(e',e)\\
& =  (\sum_{e\in E_{C,S_{u'}\cap \supp(C_{P,\drop})}} C(e) + \sum_{e\in E_{C,S_{v'}}\cap \supp(C_{P,\drop})} C(e))/h_{\cov}\\
&\leq (2h_{\cov}/10)/h_{\cov} = 1/5.
\end{align*}
The last inequality is because each cluster $S$ in ${\cal N}$ will either keep all moving cuts in $C'$ or drop all of them if $C(E_{C,S})<h_{\cov}/10$.

Finally, we can conclude 
\[
\sum_{e\in\supp(C_{P,\drop})}C(e)\leq \ell_{G-C}(P)\cdot(1/10+1/5)\leq \ell_{G-C}(P)/2,
\]
so $\ell_{G-C'}(P) = \ell_{G-C}(P)-C_{P,\drop}(P)\geq \ell_{G-C}(P)/2$.

\end{proof}

\paragraph{The Running Time.} The construction of ${\cal N}$ takes $(|G| + |A|)\cdot n^{O(\epsilon)}$ time by \Cref{thm:NeighborhoodCovers}. The construction of $C'$ and $L$ just needs to scan ${\cal N}$ and $|E(G)|$, which takes $\wtilde{O}(\size({\cal N}) + |E(G)|) \leq \wtilde{O}(\kappa_{\PC,\omega}\cdot |G|) = |G|\cdot n^{O(\epsilon)}$ time. %

\end{proof}

\begin{lemma}[Initialization of Certified ED]
Given a graph $G$ with a node-weighting $A\geq \mathds{1}(V(G))$ and parameters $h$ and $\phi$, there is an algorithm that computes a certified-ED $(C, L,{\cal N},{\cal R},\Pi_{{\cal R}\to G})$ of $A$ on $G$ satisfying the following.
\begin{itemize}
\item $|C|\leq  O(\kappa_{\init,C}\cdot\phi\cdot h\cdot|A|)$ where $\kappa_{\init,C} = \kappa_{\ED,C}\cdot\lambda_{\PC,\diam} = n^{O(\epsilon^{4})}$.
\item $L$ has distortion $\sigma = h/\kappa_{\sigma}$ and size $|L| = O(\kappa_{\init,L}\cdot\phi\cdot|A|)$%
, where
\begin{align*}
\kappa_{\init,L} &= \lambda_{\PC,\diam}\cdot\kappa_{\PC,\omega}\cdot\kappa_{\init,C}\cdot\kappa_{\sigma} = n^{O(\epsilon^{4})}.
\end{align*}
\item ${\cal N}$ is a $b$-distributed $(h_{\cov},h_{\sep},h_{\diam},\omega)$-neighborhood cover of $A$ on $G-C$ with 
\[b = 1,\ h_{\cov}=h,\ h_{\sep}=h,\text{ and }\omega = \kappa_{\PC,\omega},
\]
where $\kappa_{\PC,\omega} = n^{O(\epsilon^{4})}$ is from \Cref{thm:NeighborhoodCovers}.
Furthermore, if a cluster $S$ includes a virtual node of some vertex $v$, then $S$ includes all virtual nodes of $v$.
\item Routers ${\cal R}$ are initialized by \Cref{thm:Router} (undergoing $f=0$ batched update).%
\item Embedding $\Pi_{{\cal R}\to G}$ has length $h_{\emb} = O(\lambda_{\init,h}\cdot h)$ and congestion $\gamma_{\emb} = O(\kappa_{\init,\gamma}/\phi)$ where
\begin{align*}
\kappa_{\init,\gamma} &= \kappa_{\PC,\omega}\cdot\kappa_{\rou,\deg}\cdot \kappa_{\emb,\gamma} = n^{O(\epsilon^{4})},\\
\lambda_{\init,h} &= \lambda_{\emb,h}\cdot \lambda^{2}_{\ED,h}\cdot \lambda_{\PC,\diam} = 2^{\poly(1/\epsilon)}.
\end{align*}
\end{itemize}
The running time is $(|G|+|A|)\cdot \poly(h)\cdot n^{O(\epsilon)}$.
\label{thm:InitCertifiedED}
\end{lemma}
\begin{proof}
By applying \Cref{thm:WitnessedED} on $G$ and $A$ with parameters $h_{\wit} = 2\cdot h_{\diam} = 2\cdot \lambda_{\PC,\diam}\cdot h$ and $\phi$, we can compute an $(h_{\wit},h_{\cal R},\gamma_{\cal R},h_{\Pi},\gamma_{\Pi})$-witnessed-ED $(C_{\wit},{\cal N}_{\wit},{\cal R}_{\wit}, \Pi_{{\cal R}_{\wit}\to G-C_{\wit}})$ satisfying the following.
\begin{itemize}
\item $|C_{\wit}|\leq \kappa_{\ED,C}\cdot h_{\wit}\cdot \phi \cdot |A|$.
\item ${\cal N}_{\wit}$ is a neighborhood cover of $V(G)$ on $G-C_{\wit}$ with cover radius $h_{\wit}$ and width $\omega_{\wit} = \kappa_{\PC,\omega}$.
\item $h_{\cal R} = \lambda_{\ED,h}$, $h_{\Pi} = \lambda_{\ED,h}\cdot h_{\wit}$, $\gamma_{\cal R} = 1$ and $\gamma_{\Pi} = 1/\phi$.
\end{itemize}

Next, we apply \Cref{thm:InitLandmark} on $G$ and the witnessed-ED $(C_{\wit},{\cal N}_{\wit},{\cal R}_{\wit}, \Pi_{{\cal R}_{\wit}\to G-C_{\wit}})$. The output is a new integral moving cut $C = (C_{\wit})_{\mid \supp(C)}$ and a landmark set $L$ of $C$ on $G$ with distortion $\sigma = h/\kappa_{\sigma}$ and size 
\[|L| = O(\lambda_{\PC,\diam}\cdot \kappa_{\PC,\omega}\cdot |C|/\sigma) \leq O(\lambda_{\PC,\diam}\cdot\kappa_{\PC,\omega}\cdot\kappa_{\sigma}\cdot|C_{\wit}|/h) = O(\kappa_{\init,L}\cdot\phi\cdot|A|)\]

The neighborhood cover ${\cal N}$ can be initialized by invoking \Cref{thm:NeighborhoodCovers} on graph $G-C$ with node-weighting $A$ and length parameter $h$, so it will have $h_{\cov} = h$, $h_{\sep} = h$, $\omega = \kappa_{\PC,\omega}$ and diameter $h_{\diam} = \lambda_{\PC,\diam}\cdot h$. Next, we use \Cref{thm:Router} to initialize an router $R^{S}$ with $V(R^{S}) = S$ for each cluster in $S\in {\cal N}$. In particular, initially there is no redundant vertex in the router $R^{S}$. Let ${\cal R} = \{R^{S}\mid S\in{\cal N}\}$ collect all the routers.

Lastly, we need to compute the embedding $\Pi_{{\cal R}\to G}$. To clarify, recall that each cluster $S\in {\cal N}$ is a subset of virtual nodes in $A$, so each router vertex $v_{\virtual}\in V(R^{S})$ will corresponding to a virtual node $v_{\virtual}\in A$, and when defining the embedding, we will map $v$ to the vertex $v_{\vertex}\in V(G)$ that owns $v$. Then we define a demand $D$ on vertices $V(G)$ by adding one unit of demand to $D(u_{\vertex},v_{\vertex})$ for each router edge $e=(u_{\virtual},v_{\virtual})\in E({\cal R})$. We apply \Cref{thm:WitnessExpanderRouting} on $G-C_{\wit}$ and its witness $({\cal N}_{\wit},{\cal R}_{\wit},\Pi_{{\cal R}_{\wit}\to G-C_{\wit}})$ with $D$ as the input demand. Note that the input satisfies the requirements of \Cref{thm:WitnessExpanderRouting}, because by definition $({\cal N}_{\wit},{\cal R}_{\wit},\Pi_{{\cal R}_{\wit}\to G-C_{\wit}})$ is also a $(h_{\wit},h_{\cal R},\kappa_{\rt,\deg}\cdot\kappa_{\NC,\omega}\cdot\gamma_{\cal R},\gamma_{\Pi})$-witness of $(\kappa_{\rt,\deg}\cdot \kappa_{\PC,\omega})\cdot A$ on $G-C_{\wit}$, and by \Cref{claim:DInitED}, $D$ is an $h_{\wit}$-length $((\kappa_{\rt,\deg}\cdot \kappa_{\PC,\omega})\cdot A)$-respecting demand on $G-C_{\wit}$.
The output is an integral routing of $D$ on $G-C_{\wit}$ with 
\begin{align*}
h_{\emb}&=O(\lambda_{\rr,h}\cdot h_{\cal R}\cdot h_{\Pi}) = O(\lambda_{\rr,h}\cdot\lambda^{2}_{\ED,h}\cdot\lambda_{\PC,\diam}\cdot h),\\
\gamma_{\emb}&=O(\kappa_{\rt,\deg}\cdot\kappa_{\PC,\omega}\cdot \kappa_{\rr,\gamma}\cdot\gamma_{\cal R}\cdot \gamma_{\Pi}) = O(\kappa_{\rt,\deg}\cdot\kappa_{\PC,\omega}\cdot \kappa_{\rr,\gamma}/\phi).
\end{align*}
We can simply assign flow paths of the integral routing of $D$ to each unit of demand, and then obtain the embedding $\Pi_{{\cal R}\to G}$ with the same bounds on length and congestion.

\begin{claim}
$D$ is a $((\kappa_{\rt,\deg}\cdot \kappa_{\PC,\omega})\cdot A)$-respecting demand on $G-C_{\wit}$ with length $h_{\wit} = 2\cdot h_{\diam}$.\label{claim:DInitED}
\end{claim}
\begin{proof}
For each virtual node $v_{\virtual}\in A$, it will appear in at most $\omega = \kappa_{\PC,\omega}$ many clusters of ${\cal N}$. For each cluster $S\ni v_{\virtual}$, the router vertex in $R^{S}$ corresponding to $v_{\virtual}$ has degree at most $\kappa_{\rt,\deg}$. Hence there are totally $\kappa_{\rt,\deg}\cdot \kappa_{\PC,\omega}$ many router edges connecting a virtual node $v$ summing over all routers. Then for each vertex $v_{\vertex}\in V(G)$, there are at most $\kappa_{\rt,\deg}\cdot\kappa_{\PC,\omega}\cdot A(v)$ many router edges with endpoint mapped to $v_{\vertex}$, so $D$ is $((\kappa_{\rt,\deg}\cdot\kappa_{\PC,\omega})\cdot A)$-respecting. 

Now we show that for each demand pair $(u_{\vertex},v_{\vertex})\in \supp(D)$, $\dist_{G-C_{\wit}}(u_{\vertex},v_{\vertex})\leq 2\cdot h_{\diam}$. Consider a router edge $e=(u_{\virtual},v_{\virtual})$ which adds one unit of demand to $D(u_{\vertex},v_{\vertex})$. Then $u_{\virtual}\in A(u_{\vertex})$ and $v_{\virtual}\in A(v_{\vertex})$ are virtual nodes owned by $u_{\vertex}$ and $v_{\vertex}$ respectively. Because $e$ is a router edge, there is a cluster $S\in{\cal N}$ s.t. $u_{\virtual},v_{\virtual}\in S$, which implies $\dist_{G-C}(u_{\vertex},v_{\vertex})\leq h_{\diam}$. Recall that \Cref{thm:InitLandmark} guarantees, for any simple path $P$ on $G$ with $\ell_{G-C_{\wit}}(P)\geq \sigma$, we have $\ell_{G-C_{\wit}}(P)\leq 2\cdot \ell_{G-C}(P)$. Consider the shortest $u_{\vertex}$-to-$v_{\vertex}$ path $P$ on $G-C$. If $\ell_{G-C_{\wit}}(P)\geq \sigma$, then $\dist_{G-C_{\wit}}(u_{\vertex},v_{\vertex})\leq \ell_{G-C_{\wit}}(P)\leq 2\cdot\ell_{G-C}(P)\leq 2\cdot h_{\diam}$, otherwise trivially $\dist_{G-C_{\wit}}(u_{\vertex},v_{\vertex})\leq \ell_{G-C_{\wit}}(P)<\sigma \leq 2\cdot h_{\diam}$.
\end{proof}

\paragraph{The Running Time.} 
Computing the witnessed-ED takes $|G|\cdot \poly(2\cdot h_{\diam})\cdot n^{O(\epsilon)}$ time by \Cref{thm:WitnessedED}.
Computing $C$ and its landmark set $L$ also takes $|G|\cdot n^{O(\epsilon)}$ 
time by \Cref{thm:InitLandmark}.
Constructing the neighborhood cover ${\cal N}$ takes $(|G|+|A|)\cdot n^{O(\epsilon)}$ time by \Cref{thm:NeighborhoodCovers}. Initializing the routers takes $(|A|\cdot \kappa_{\NC,\omega})\cdot n^{O(\epsilon)}$ time by \Cref{thm:Router}, because the total size of clusters is $\size({\cal N})\leq O(|A|\cdot \kappa_{\NC,\omega})$. Lastly, computing the embedding $\Pi_{{\cal R}\to G}$ takes $(|A|\cdot\kappa_{\NC,\omega} + |A|\cdot\kappa_{\rt,\deg}\cdot\kappa_{\NC,\omega}\cdot\omega_{\wit})\cdot \poly(h_{\diam})\cdot \lambda_{\ED,h}\cdot n^{O(\epsilon)}$ time by \Cref{thm:WitnessExpanderRouting}, because this step we works with node-weighting $(\kappa_{\rt,\deg}\cdot\kappa_{\NC,\omega})\cdot A$. In conclusion, the total time is $(|G|+|A|)\cdot \poly(h)\cdot n^{O(\epsilon)}$.

\end{proof}

\subsection{Batched Node-Weighting Insertion}
\label{sect:DynEDNodeWeightingInsertion}

To insert a node-weighting $A_{\new}$ into a certified-ED of an old node-weighting $A$, our strategy is to compute a cutmatch between $A_{\new}$ and $A$ using our local cutmatch algorithm by \Cref{thm:cutmatch}. $A_{\new}$ is partitioned into a matching part $A_{\new,M}$ and an unmatched part $A_{\new,U}$. Intuitively, we can just add the matching part $A_{\new,M}$ into the old certified-ED because the matching has low length and congestion. For the unmatched part, we can initialize a new certified-ED of $A_{\new,U}$ on a \emph{local} subgraph, because $A_{\new,U}$ is far from most of the old virtual nodes by adding a small new cut.

\begin{lemma}
Let $G$ be a graph with a node-weighting $A$ and a certified-ED $(C,{L},{\cal N}, {\cal R},\Pi_{{\cal R}\to G})$ of $A$ on $G$ in \Cref{def:CertifiedED}. 
Given a node-weighting $A_{\new}$ s.t. $\wtilde{A} = A + A_{\new}\geq \deg_{G} + \mathds{1}(V(G))$ and parameters $h$ and $\phi$, there is an algorithm that computes a certified-ED $(\wtilde{C},\wtilde{L},\wtilde{N},\wtilde{R},\Pi_{\wtilde{\cal R}\to G})$ of $\wtilde{A}$ on $G$ s.t.
\begin{itemize}
\item $\wtilde{C}=C+C_{\new}$ with $|C_{\new}| \leq O(\kappa_{\init,C}\cdot \phi\cdot h\cdot |A_{\new}|)$.
\item $\wtilde{L} = L \cup {L}_{\new}$ with distortion $\wtilde{\sigma} = \sigma + O(h/\kappa_{\sigma})$ and $|{L}_{\new}| = O((\kappa_{\CM,L} + \kappa_{\init,L})\cdot \phi\cdot|A_{\new}|)$.
\item $\wtilde{\cal N}$ is a $\wtilde{b}$-distributed $(\wtilde{h}_{\cov},\wtilde{h}_{\sep},\wtilde{\omega})$-neighborhood cover of $\wtilde{A}$ on $G-\wtilde{C}$ with 
\[\wtilde{b} = b+1,\ \wtilde{h}_{\cov} = h/3,
\ \wtilde{h}_{\sep} = h/3,
\text{ and }\wtilde{\omega} = \omega + \kappa_{\PC,\omega}.
\]
\item Routers in $\wtilde{\cal R}$ are maintained by \Cref{thm:Router} under $\wtilde{f} = f+1$ update batches. 
\item $\Pi_{\wtilde{\cal R}\to G}$ is a $(\wtilde{h}_{\emb},\wtilde{\gamma}_{\emb})$-embedding with
\begin{align*}
&\wtilde{h}_{\emb} = \max\{h_{\emb},O(\lambda_{\init,h}\cdot h)\},\\
&\wtilde{\gamma}_{\emb} = \gamma_{\emb} + O((\omega\cdot\kappa_{\CM,\gamma} + \kappa_{\init,\gamma})/\phi).
\end{align*}
\end{itemize}
The recourse from ${\cal N}$ to $\wtilde{\cal N}$ is $\recourse({\cal N}\to\wtilde{\cal N}) = O(|A_{\new}|\cdot(\omega + \kappa_{\PC,\omega}))$.
The running time is $|A_{\new}|\cdot \poly(h)\cdot n^{O(\epsilon)}\cdot (1/\phi +2^{O(f)}\cdot \omega)$.

\label{thm:NodeWeightingInsertion}

\end{lemma}

\begin{proof}

We first apply the cutmatch subroutine, \Cref{thm:cutmatch}, on $G-C$ with source node-weighting $A^{\src}=A_{\new}$, sink node-weighting $A_{\sink}=A$ and parameters $h_{\CM}=h/3$ and $\phi_{\CM} = \phi$. The output is 
\begin{itemize}
\item a partition $A^{\src}_{M},A^{\src}_{U}$ of $A^{\src}$,
\item a partition $A^{\sink}_{M},A^{\sink}_{U}$ of $A^{\sink}$ with $|A^{\sink}_{M}| = O(|A^{\src}_{M}|)$,
\item a matching $M_{\CM}$ between $A^{\src}_{M}$ and $A^{\sink}_{M}$ with size $|M_{\CM}| = |A^{\src}_{M}|$ and its $h_{\CM}$-length embedding $\Pi_{M_{\CM}\to G-C}$ which has congestion $\gamma_{\CM} = O(\kappa_{\CM,\gamma}/\phi_{\CM})$, and
\item a moving cut $C_{\CM}$ with size $|C_{\CM}|\leq \phi_{\CM}\cdot h_{\CM}\cdot |A^{\src}|$ s.t. $A^{\src}_{U}$ is $h_{\CM}$-separated from $A^{\sink}_{U}$ in $G-C-C_{\CM}$.
\item a landmark set $L_{\CM}$ of $C_{\CM}$ on $G-C$ with distortion $\sigma_{\CM} = O(h/\kappa_{\sigma})$ and size $|L_{\CM}|\leq O(\kappa_{\CM,L}\cdot\phi_{\CM}\cdot |A^{\src}|)$%
\end{itemize}

We define two node-weightings $A_{1}=A^{\sink}\cup A^{\src}_{M}=A\cup A^{\src}_{M}$ and 
\[
A_{2} =\wtilde{A}\setminus\{\text{virtual nodes of vertices in }\supp(A_{U}^{\sink})\}.
\]
For better understanding, equivalently $A_{2} = (A^{\src}\cup A_{M}^{\sink})\setminus\{\text{virtual nodes of vertices in }\supp(A_{U}^{\sink})\}$. 
Note that $A^{\src}_{U}\subseteq A_{2}$ because $\supp(A^{\src}_{U})$ is $h_{\CM}$-separated from $\supp(A^{\sink}_{U})$ (it means $\supp(A^{\src}_{U})$ and $\supp(A^{\sink}_{U})$ are disjoint). $A_{2}$ has size $|A_{2}|\leq|A_{\new}|+|A^{\sink}_{M}|\leq |A_{\new}| + O(|A^{\src}_{M}|)\leq O(|A_{\new}|)$.

We first compute $({\cal N}_{1}, {\cal R}_{1}, \Pi_{{\cal R}_{1}\to G})$ by updating $({\cal N}, {\cal R}, \Pi_{{\cal R}\to G})$. Concretely, for each original cluster $S\in{\cal N}$ and its router $R\in{\cal R}$, let $M_{S} = \{(u,v)\in M_{\CM}\mid \text{the sink virtual node $v\in S$}\}$
collect all matching edges in $M_{\CM}$ with one endpoint falling in $S$. Then we update $S$ and $R$ to $S_{1} = S\cup M^{\src}_{S}$ (recall that $M^{\src}_{S}$ denotes the source virtual nodes in $M_{S}$) and $R_{1} = R\cup M_{S}$ via a matching insertion update in \Cref{thm:Router}. Moreover, the new router $R_{1}$ has embedding $\Pi_{R_{1}\to G} = \Pi_{R\to G}\cup \Pi_{M_{S}\to G}$, where $\Pi_{M_{S}\to G}$ can be obtained from $\Pi_{M_{\CM}\to G-C}$ (note that the graphs $G$ and $G-C$ have the same structure). 

Next, we initialize a certified-ED on $(G-C-C_{\CM})[A_{2}]$ by applying \Cref{thm:InitCertifiedED} with node-weighting $A_{2}$ and parameters $\phi_{\ED} = \phi$ and $h_{\ED} = h/3$. The output $(C_{\ED}, L_{\ED}, {\cal N}_{2},{\cal R}_{2},\Pi_{{\cal R}_{2}\to G-C-C_{\CM}})$ satisfies the following.
\begin{itemize}
\item $|C_{\ED}|\leq O(\kappa_{\init,C}\cdot\phi_{\ED}\cdot h_{\ED}\cdot |A_{2}|)$.
\item $L_{\ED}$ is a landmark set of $C_{\ED}$ on $(G-C-C_{\CM})[A_{2}]$ with distortion $\sigma_{\ED} = O(h_{\ED}/\kappa_{\sigma})$ and size $|L_{\ED}| = O(\kappa_{\init,L}\cdot\phi_{\ED}\cdot |A_{2}|)$. %
Note that $L_{\ED}$ is automatically a landmark set of $C_{\ED}$ on $G-C-C_{\CM}$ with the same distortion and size.
\item ${\cal N}_{2}$ is an $(h_{\cov,2},h_{\sep,2},\omega_{2})$-neighborhood cover of $A_{2}$ on $G-C-C_{\CM}-C_{\ED}$ with $h_{\cov,2} = h_{\ED}, h_{\sep,2} = h_{\ED}, h_{\diam} = \lambda_{\PC,\diam}\cdot h_{\ED}$ and $\omega_{2} = \kappa_{\PC,\omega}$.
\item Routers in ${\cal R}_{2}$ are initialized by \Cref{thm:Router}.
\item The embedding $\Pi_{{\cal R}_{2}\to (G-C-C_{\CM})[A_{2}]}$ has congestion $\gamma_{\emb,2} = O(\kappa_{\init,\gamma}/\phi_{\ED})$ and dilation $h_{\emb,2} = O(\lambda_{\init,h}\cdot h_{\ED})$.
\end{itemize}

The new certified-ED $(\wtilde{C},\wtilde{L},\wtilde{\cal N},\wtilde{\cal R},\Pi_{\wtilde{\cal R}\to G})$ is defined by
\begin{gather*}
\wtilde{C} = C+C_{\new}\text{ where }C_{\new} = C_{\CM} + C_{\ED},\\ 
\wtilde{L} = L\cup {L}_{\new}\text{ where }{L}_{\new} = L_{\CM}\cup L_{\ED},\\
(\wtilde{\cal N}, \wtilde{\cal R}, \Pi_{\wtilde{\cal R}\to G}) = ({\cal N}_{1},{\cal R}_{1},\Pi_{{\cal R}_{1}\to G})\cup ({\cal N}_{2},{\cal R}_{2},\Pi_{{\cal R}_{2}\to G-C-C_{\CM}}).
\end{gather*}

\paragraph{The Quality.} The quality of $(\wtilde{C},\wtilde{L},\wtilde{\cal N},\wtilde{\cal R},\Pi_{\wtilde{\cal R}\to G})$ can be verified as follows. 
The new moving cut has size
\[|C_{\new}|\leq O(\phi_{\CM}\cdot h_{\CM}\cdot |A^{\src}|) + O(\kappa_{\init,C}\cdot\phi_{\ED}\cdot h_{\ED}\cdot |A_{2}|) \leq O(\kappa_{\init,C}\cdot\phi \cdot h\cdot |A_{\new}|)\]
because $h_{\ED} = h_{\CM} = h/3$, $|A^{\src}| = |A_{\new}|$ and $|A_{2}|\leq O(|A_{\new}|)$. By \Cref{lemma:LandmarkUnion}, $\wtilde{L} = L\cup L_{\CM}\cup L_{\ED}$ is a landmark set of $\wtilde{C} = C + C_{\CM} + C_{\ED}$ on $G$ with distortion $\wtilde{\sigma} = \sigma + \sigma_{\CM} + \sigma_{\ED} = \sigma + O(h/\kappa_{\sigma})$, size
\[
|L_{\new}| \leq O(\kappa_{\CM,L}\cdot \phi_{\CM}\cdot |A^{\src}|) + O(\kappa_{\init,L}\cdot\phi_{\ED}\cdot |A_{2}|)\leq O((\kappa_{\CM,L} +\kappa_{\init,L})\cdot\phi\cdot |A_{\new}|).
\]
The quality of $\wtilde{\cal N}$ is shown in \Cref{lemma:ProperiesOfTildeN}. In particular, the ball cover ${\cal S}_{v}$ of each $v\in \wtilde{A}$ can be maintained explicitly following the proof of \Cref{lemma:ProperiesOfTildeN}. The routers from ${\cal R}_{1}$ are updated from ${\cal R}$ via one batched update, and ${\cal R}_{2}$ are just initialized with no update. Lastly, the length and congestion of $\Pi_{\wtilde{\cal R}\to G}$ are 
\begin{align*}
\wtilde{h}_{\emb} &= \max\{h_{\emb}, h_{\CM}, h_{\emb,2}\} = \max\{h_{\emb},O(\lambda_{\init,h}\cdot h)\}\\
\wtilde{\gamma}_{\emb} &= \gamma_{\emb} + \omega\cdot\gamma_{\CM} + \gamma_{\emb,2} = \gamma_{\emb} + O((\omega\cdot\kappa_{\CM,\gamma}+\kappa_{\init,\gamma})/\phi).
\end{align*}
The term $\gamma_{\CM}$ is multiplied by $\omega$ because each matching edge will be added to routers at most $\omega$ times.

\begin{lemma}
$\wtilde{\cal N}$ is a $\wtilde{b}$-distributed $(\wtilde{h}_{\cov},\wtilde{h}_{\sep},\wtilde{\omega})$-pairwise cover of $\wtilde{A}$ on $G-\wtilde{C}$ with $\wtilde{b} = b+1$, $\wtilde{h}_{\cov} = h/3$, $\wtilde{h}_{\sep} = h/3$ and $\wtilde{\omega} = \omega + \kappa_{\PC,\omega}$.
\label{lemma:ProperiesOfTildeN}
\end{lemma}
\begin{proof}
It is relatively simple to see the bounds on $\wtilde{h}_{\sep}$ and $\wtilde{\omega}$. The width $\wtilde{\omega} = \omega + \omega_{2} = \omega + \kappa_{\PC,\omega}$. Regarding the separation, note that $\wtilde{h}_{\sep} = \max\{h_{\sep,1},h_{\sep,2}\}$, where $h_{\sep,1}$ denotes the separation of ${\cal N}_{1}$, and $h_{\sep,2}$ is already known to be $h_{\ED} = h/3$. To see $h_{\sep,1}= h/3$, we need to show that for two virtual nodes $u,v\in A_{1}$ belonging to two different clusters $S_{u},S_{v}\in {\cal N}_{1}$, we have $\dist_{G-\wtilde{C}}(u,v)> h/3$. Consider the worst case that $u,v\in A_{1}\cap A_{\new}$ are both new virtual nodes. Let $u'\in S_{u}$ and $v'\in S_{v}$ be the old virtual nodes matched to $u$ and $v$ in $M_{\CM}$ respectively. The separation of ${\cal N}$ gives $\dist_{G-C}(u',v')\geq h_{\sep} = h$. Therefore,
\begin{align*}
\dist_{G-\wtilde{C}}(u,v) &\geq \dist_{G-C}(u,v)\\
&\geq \dist_{G-C}(u',v')-\dist_{G-C}(u,u')-\dist_{G-C}(v,v')\\
&> h_{\sep} - 2h_{\CM}
= h/3.
\end{align*}
The case that at most one of $u,v$ belongs to $A_{1}\cap A_{\new}$ can be argued in a similar way.

We now show that $\wtilde{h}_{\cov} = h/3$ and $\wtilde{b} = b+1$. Consider an arbitrary virtual node $v\in \wtilde{A}$.

\underline{Case 1.} Suppose $v\in A\setminus A_{2}$. That is, $v$ is a virtual node of some vertex in $\supp(A^{\sink}_{U})$. By definition of ${\cal N}$, we know $\Ball_{G-C,A}(v,h_{\cov})$ is covered by a collection ${\cal S}_{v}=\{S_{1},...,S_{b}\} \subseteq {\cal N}$ of $b$ clusters. We will show that $\Ball_{G-\wtilde{C},\wtilde{A}}(v,\wtilde{h}_{\cov})$ is covered by $b$ many clusters $\wtilde{\cal S}_{v} = \{\wtilde{S}_{1},...,\wtilde{S}_{b}\}\subseteq {\cal N}_{1}$, where each new cluster $\wtilde{S}_{b'}$ originates from the old cluster $S_{b'}$. Let $u\in\wtilde{A}$ be an arbitrary virtual node s.t. $\dist_{G-\wtilde{C}}(u,v)\leq \wtilde{h}_{\cov}$.

\underline{Subcase 1(a).} If $u\in A$, then $\dist_{G-C}(u,v)\leq \dist_{G-\wtilde{C}}(u,v)\leq\wtilde{h}_{\cov}\leq h_{\cov}$, so $u$ is inside some old cluster $S\in{\cal S}_{v}$ and the new cluster $\wtilde{S}\in \wtilde{\cal S}_{v}$.

\underline{Subcase 1(b).} Suppose $u\in A^{\src}_{M}$. Let $u'\in A^{\sink}_{M}\subseteq A$ be the virtual node matched to $u$ by $M_{\CM}$. Then
\begin{align*}
\dist_{G-C}(u',v)&\leq \dist_{G-C}(u',u)+\dist_{G-C}(u,v)\\
&\leq \dist_{G-C}(u',u)+\dist_{G-\wtilde{C}}(u,v)\\
&\leq h_{\CM}+\wtilde{h}_{\cov}\leq h_{\cov},
\end{align*}
which means $u'$ belongs to some old cluster $S\in{\cal S}_{v}$. The new cluster $\wtilde{S}\in\wtilde{S}_{v}$ will include $u$ by our construction.

\underline{Subcase 1(c).} Suppose $u\in A^{\src}_{U}$. However, this is impossible because $\dist_{G-\wtilde{C}}(u,v)\leq \wtilde{h}_{\cov} = h_{\CM}$ contradicts that $A^{\src}_{U}$ is $h_{\CM}$-separated from $A^{\sink}_{U}$ in $G-\wtilde{C}$.

\underline{Case 2.} Suppose $v\in A\cap A_{2}$. Again, $\Ball_{G-C}(v,h_{\cov})$ is covered by a size-$b$ collection ${\cal S}_{v}\subseteq {\cal N}$. Furthermore, $\Ball_{(G-\wtilde{C})[A_{2}]}(v,h_{\ED})$ is covered by a cluster $\wtilde{S}_{v,2}\in{\cal N}_{2}$ because ${\cal N}_{2}$, is a (1-distributed) neighborhood cover. We now show that $\Ball_{G-\wtilde{C}}(v,\wtilde{h}_{\cov})$ is covered by $b+1$ many clusters in $\wtilde{\cal N}$, i.e. the $\wtilde{S}_{v,2}\in{\cal N}_{2}$ and the new clusters $\wtilde{\cal S}_{v,1}\subseteq {\cal N}_{1}$ original from ${\cal S}_{v}$. Consider an arbitrary $u\in \Ball_{G-\wtilde{C},\wtilde{A}}(v,\wtilde{h}_{\cov})$.

\underline{Subcases 2(a) and 2(b).} If $u\in A$ or $u\in A^{\src}_{M}$, then $u$ belongs to some cluster in $\wtilde{\cal S}_{v}$ by the same arguments as subcases 1(a) and 1(b) respectively.

\underline{Subcase 2(c).} Suppose $u\in A^{\src}_{U}\subseteq A_{2}$. Then $\dist_{(G-\wtilde{C})[A_{2}]}(u,v) = \dist_{G-\wtilde{C}}(u,v) \leq \wtilde{h}_{\cov} = h_{\ED}$ by \Cref{ob:Subcase2c}, so $u$ belongs to $\wtilde{S}_{v}$.
\begin{observation}
For $u\in A^{\src}_{U}$ and $v\in A_{2}$ s.t. $\dist_{G-\wtilde{C}}(u,v)\leq \wtilde{h}_{\cov}$, we have $\dist_{(G-\wtilde{C})[A_{2}]}(u,v)= \dist_{G-\wtilde{C}}(u,v)$.
\label{ob:Subcase2c}
\end{observation}
\begin{proof}
Consider the shortest $u$-$v$ path $P$ on $G-\wtilde{C}$. We know $\ell_{G-\wtilde{C}}(P) = \dist_{G-\wtilde{C}}(u,v) \leq \wtilde{h}_{\cov} = h_{\CM}$, so $P$ cannot go through any vertex $w\in\supp(\wtilde{A}\setminus A_{2}) = \supp(A^{\sink}_{U})$ because $w$ is $h_{\CM}$-separated from $u$ even on $G-C-C_{\CM}$ (let alone $G-\wtilde{C}$). Therefore, $P$ is totally inside $(G-\wtilde{C})[A_{2}]$ and the observation holds.
\end{proof}

\underline{Case 3.} Suppose $v\in A^{\src}_{M}\setminus A_{2}$. This infers $v$ is owned by some vertex in $\supp(A^{\sink}_{U})$. Let $v'\in A^{\sink}_{M}\subseteq A$ be the virtual node matched to $v$ by $M_{\CM}$. Let ${\cal S}_{v'}\subseteq {\cal N}$ be those $b$ clusters covering $\Ball_{G-C,A}(v',h_{\cov})$ and let $\wtilde{\cal S}_{v',1} \subseteq {\cal N}_{1}$ be the clusters original from ${\cal S}_{v'}$. Observe that, for every $S \in \wtilde{\cal S}_{v',1}$, we have $v \in S$ because $v$ is matched to $v'$ by $M_\CM$ and so our algorithm adds $v$ to every cluster $S \in \wtilde{\cal S}_{v',1}$. We will next show that $\Ball_{G-\wtilde{C},\wtilde{A}}(v,\wtilde{h}_{\cov})$ is covered by the new clusters $\wtilde{\cal S}_{v',1}$. Consider an arbitrary virtual node $u\in\Ball_{G-\wtilde{C},\wtilde{A}}(v,\wtilde{h}_{\cov})$.

\underline{Subcase 3(a).} If $u\in A$, then
\begin{align*}
\dist_{G-C}(u,v')&\leq \dist_{G-C}(u,v) + \dist_{G-C}(v,v')\\
&\leq \dist_{G-\wtilde{C}}(u,v) + \dist_{G-C}(v,v')\\
&\leq \wtilde{h}_{\cov} + h_{\CM}\leq h_{\cov},
\end{align*}
and ${\cal S}_{v'}$ covers $u$, so does $\wtilde{\cal S}_{v',1}$.

\underline{Subcase 3(b).} If $u\in A^{\src}_{M}$, let $u'\in A^{\sink}_{M}\subseteq A$ be the virtual node matched to $u$. Similarly,
\begin{align*}
\dist_{G-C}(u',v')&\leq \dist_{G-C}(u',u)+\dist_{G-C}(u,v)+\dist_{G-C}(v,v')\\
&\leq \dist_{G-C}(u',u)+\dist_{G-\wtilde{C}}(u,v)+\dist_{G-C}(v,v')\\
&\leq h_{\CM}+\wtilde{h}_{\cov}+h_{\CM}\leq h_{\cov},
\end{align*}
so $u'$ is covered by some cluster $S\in {\cal S}_{v'}$. By the construction, $u$ will be added to $S$, so $u$ is covered by $\wtilde{\cal S}_{v',1}$. 

\underline{Subcase 3(c).} It is impossible that $u\in A^{\src}_{U}$ because $v$ is owned by some vertex in $\supp(A^{\sink}_{U})$ and $\supp(A^{\src}_{U})$ is $(h_{\CM}=\wtilde{h}_{\cov})$-separated from $\supp(A^{\sink}_{U})$ in $G-\wtilde{C}$.

\underline{Case 4.} Suppose $v\in A^{\src}_{M}\cap A_{2}$ which matches $v'\in A^{\sink}_{M}\subseteq A$ in $M_{\CM}$. Again, let ${\cal S}_{v'}\subseteq {\cal N}$ be the $b$ clusters covering $\Ball_{G-C,A}(v',h_{\cov})$. Moreover, let $\wtilde{S}_{v,2}\in {\cal N}_{2}$ be the cluster covering $\Ball_{(G-\wtilde{C})[A_{2}],A_{2}}(v,h_{\ED})$. Then we will show that $\Ball_{G-\wtilde{C},\wtilde{A}}(v,\wtilde{h}_{\cov})$ is covered by $\wtilde{S}_{v}$ and the new clusters $\wtilde{\cal S}_{v'}\subseteq {\cal N}_{1}$ original from ${\cal S}_{v',1}$.

\underline{Subcases 4(a) and 4(b).} If $u\in A$ and $u\in A^{\src}_{M}$, $\wtilde{\cal S}_{v'}$ covers $u$ by the same argument as subcases 3(a) and 3(b) respectively.

\underline{Subcase 4(c).} If $u\in A^{\src}_{U}\subseteq A_{2}$, then $u\in \wtilde{S}_{v}$ by the same argument as subcase 2(c).

\underline{Case 5.} Suppose $v\in A^{\src}_{U}$. Consider an arbitrary $u\in \Ball_{G-\wtilde{C},\wtilde{A}}(v,\wtilde{h}_{\cov})$. We know $u$ is not owned by any vertex in $\supp(A^{\sink}_{U})$ by the $h_{\CM}$-separation, so $u\in A_{2}$. By \Cref{ob:Subcase2c}, the cluster $\wtilde{S}_{v}\in{\cal N}_{2}$ covering $\Ball_{(G-\wtilde{C})[A_{2}],A_{2}}(v,h_{\ED})$ will cover $\Ball_{G-\wtilde{C},\wtilde{A}}(v,\wtilde{h}_{\cov})$. 

\end{proof}

\paragraph{The Recourse from ${\cal N}$ to $\wtilde{\cal N}$.} First, ${\cal N}_{1}$ can be updated from ${\cal N}$ by at most $|A_{\new}|\cdot \omega$ virtual node insertions, because each virtual node in $A_{\new}$ can be inserted at most $\omega$ times (at most once for each clustering). Second, $\wtilde{\cal N}$ can be updated from ${\cal N}_{1}$ by adding ${\cal N}_{2}$, where ${\cal N}_{2}$ has at most $\kappa_{\PC,\omega}\cdot |A_{2}|$ virtual nodes summing over all clusters. Therefore, the recourse from ${\cal N}$ to $\wtilde{\cal N}$ is 
\[
\recourse({\cal N}\to \wtilde{\cal N}) \leq |A_{\new}|\cdot \omega + |A_{2}|\cdot \kappa_{\PC,\omega} = O(|A_{\new}|\cdot (\omega + \kappa_{\PC,\omega})).
\]

\paragraph{The Running Time.} The running time of the cutmatch subroutine is $\poly(h_{\CM})\cdot|A^{\src}|/\phi_{\CM}$ by \Cref{thm:cutmatch}. When we update $({\cal N},{\cal R},\Pi_{{\cal R}\to G})$ to $({\cal N}_{1}, {\cal R}_{1}, \Pi_{{\cal R}_{1}\to G})$, the running time is dominated by updating the routers. 
The total size of matchings added to routers is $\sum_{S\in{\cal N}}|M_{S}| \leq \omega\cdot |M_{\CM}|\leq \omega\cdot |A_{\new}|$ because each matching edge appears at most $\omega$ times. By \Cref{thm:Router}, the total time to update routers is $2^{O(f)}\cdot\omega\cdot|A_{\new}|\cdot 2^{\poly(1/\epsilon)}$. By \Cref{thm:InitCertifiedED}, the construction time of $(C_{\ED},L_{\ED},{\cal N}_{2},{\cal R}_{2},\Pi_{{\cal R}_{2}\to G-C-C_{\CM}})$ is $|A_{\new}|\cdot \poly(h_{\ED})\cdot n^{O(\epsilon)}$, because the size of subgraph $(G-C-C_{\CM})[A_{2}]$ is bounded by $\deg_{G}(\supp(A_{2})) \leq |A_{2}|\leq O(|A_{\new}|)$, where $\deg_{G}(\supp(A_{2})) \leq |A_{2}|$ follows that each vertex $v\in \supp(A_{2})$ has $A_{2}(v)\geq \deg_{G}(v)$. In conclusion, the total running time is bounded by $|A_{\new}|\cdot \poly(h)\cdot n^{O(\epsilon)}\cdot (1/\phi + 2^{O(f)}\cdot\omega)$.

\end{proof}

\subsection{Batched Edge Deletions}
\label{sect:DynEDEdgeDeletion}

A batched edge deletion update to the certified-ED can be reduced to a batched node-weighting insertion. The key point is that we have a low-congestion embedding of routers on $G$, which means deleting one edge will only destroy a small number of router edges. Removing these router edges from our dynamic routers will only generate a small set of pruned virtual nodes. Then we can reinsert the pruned virtual node by a batched node-weighint insertion.

\begin{lemma}
Let $G$ be a graph with a node-weighting $A\geq \deg_{G} + \mathds{1}(V(G))$ and a certified-ED $(C,L,{\cal N}, {\cal R},\Pi_{{\cal R}\to G})$ of $A$ on $G$ in \Cref{def:CertifiedED}.
Given parameters $\phi,h$, and a batch of edge deletions $F\subseteq E(G)$, there is an algorithm that computes a certified-ED $(\wtilde{C},\wtilde{L},\wtilde{N},\wtilde{R},\Pi_{\wtilde{\cal R}\to G})$ of $\wtilde{A} = A - \deg_{F}$ on $\wtilde{G} = G\setminus F$ s.t.
\begin{itemize}
\item $\wtilde{C}=C+C_{\new}$ with 
$|C_{\new}| \leq O(\kappa_{\init,C}\cdot\phi\cdot h\cdot \lambda_{\rt,\prune}(f)\cdot \gamma_{\emb}\cdot |F|)$.
\item 
$\wtilde{L} = L\cup L_{\new}$ with distortion $\wtilde{\sigma} = \sigma + O(h/\kappa_{\sigma})$ and
\[
|L_{\new}| = O((\kappa_{\CM,L} + \kappa_{\init,L})\cdot\phi\cdot \lambda_{\rt,\prune}(f)\cdot\gamma_{\emb}\cdot|F|).
\]
\item $\wtilde{\cal N}$ is a $\wtilde{b}$-distributed $(\wtilde{h}_{\cov},\wtilde{h}_{\sep},\wtilde{\omega})$-neighborhood cover of $A$ on $G-\wtilde{C}$ with 
\[
\wtilde{b} = b+1,\ 
\wtilde{h}_{\cov} = h/3,
\ \wtilde{h}_{\sep} = h/3,
\text{ and }\wtilde{\omega} = \omega + \kappa_{\PC,\omega}.
\]
\item Routers in $\wtilde{\cal R}$ is maintained by \Cref{thm:Router} undergoing at most $\wtilde{f} = f+2$ batched updates. %
\item $\Pi_{\wtilde{\cal R}\to \wtilde{G}}$ is a $(\wtilde{h}_{\emb},\wtilde{\gamma}_{\emb})$-embedding with
\begin{align*}
&\wtilde{h}_{\emb} = \max\{h_{\emb},O(\lambda_{\init,h}\cdot h)\},\\
&\wtilde{\gamma}_{\emb} = \gamma_{\emb} + O((\omega\cdot\kappa_{\CM,\gamma} + \kappa_{\init,\gamma})/\phi).
\end{align*}
\end{itemize}
The recourse from ${\cal N}$ to $\wtilde{\cal N}$ is \[\recourse({\cal N}\to \wtilde{\cal N}) = O(\kappa_{\rt,\prune(f)}\cdot\gamma_{\emb}\cdot(\omega + \kappa_{\PC,\omega})\cdot|F|).\]
The running time is
$\lambda_{\rt,\prune}(f)\cdot\gamma_{\emb}\cdot |F|\cdot\poly(h)\cdot n^{O(\epsilon)}\cdot(1/\phi + 2^{O(f)}\cdot \omega)$.

\label{thm:EdgeDeletion}
\end{lemma}

\begin{proof}

Let ${\cal R}_{\edge,F} = \Pi^{-1}_{{\cal R}\to G}(F)$ collect all edges in routers with embedding paths destroyed by $F$. 
Then for each cluster $S\in{\cal N}$ and its router $R\in{\cal R}$, let $R_{\edge,F} = {\cal R}_{\edge,F}\cap E(R)$ and we perform an edge deletion update $R_{\edge,F}$ on $R$ by applying \Cref{thm:Router}, which outputs a prune set $R_{\prune}\subseteq V(R)$ and a new router $R' = R[V(R)\setminus R_{\prune}]$. Let $S_{\prune} = R_{\prune}\cap S$ be the set of pruned virtual nodes of $S$ (namely, we obtain $S_{\prune}$ by excluding redundant vertices in $R_{\prune}$). Let $A_{\prune} = \bigcup_{S\in{\cal N}} S_{\prune}$ collect all pruned virtual nodes.

Now we partition $A$ into $A_{\prune}$ and $A' = A\setminus A_{\prune}$. Then we update $(C,L,{\cal N},R,\Pi_{{\cal R}\to G})$ to a certified-ED $(C',L',{\cal N}',{\cal R}',\Pi_{{\cal R}'\to\wtilde{G}})$ of $A'$ on $\wtilde{G}$ by removing $A_{\prune}$. Concretely, for each original cluster $S\in{\cal S}$ and its router $R\in{\cal R}$, we will update them to a new cluster $S'$ corresponding the new router $R' = R[V(R)\setminus R_{\prune}]$ by the following steps.
\begin{enumerate}
\item Remove $A_{\prune}\cap S$ from $S$ and denote this new cluster by $S'$. Note that the new router $R' = R[V(R)\setminus R_{\prune}]$ satisfies $S' = S\setminus A_{\prune}\subseteq V(R)\setminus R_{\prune}= V(R')$, so $R'$ is a valid router for $S'$. Technically, some vertices in $V(R')$ may still correspond to some virtual nodes in $A_{\prune}$ (precisely, these virtual nodes are $A_{\prune}\setminus S_{\prune}$), so we will mark all such vertices in $R'$ redundant.

\item For each edge removed from $R$ due to the previous update via \Cref{thm:Router}, remove its embedding path from $\Pi_{{\cal R}\to G}$.
\end{enumerate}
Lastly, we let $C' = C_{\mid E(\wtilde{G})}$ by removing cut values on deleted edges (this step is just to make $C'$ well-defined on $E(\wtilde{G})$). Also, we let $L' = L\cup L'_{\new}$ with $L'_{\new} = V(F)$.%

Observe that $(C',L',{\cal N}',{\cal R}',\Pi_{{\cal R}'\to \wtilde{G}})$ is a certified-ED of $A'$ on $\wtilde{G}$, and it has the same quality parameters with those of $(C,L,{\cal N},{\cal R},\Pi_{{\cal R}\to G})$, except that routers in ${\cal R}'$ undergo one more batched update. Concretely, the quality of $L'$ is guaranteed by \Cref{lemma:LandmarkEdgeDel}. The quality of ${\cal N'}$ follows \Cref{ob:CoversOnSubgraph} and the fact that ${\cal N}'$ is the restriction of ${\cal N}$ on $A'$. $\Pi_{{\cal R}'\to\wtilde{G}}$ is a valid embedding on $\wtilde{G}$ because we remove all router edges with embedding paths going through deleted edges by \Cref{thm:Router}, and their embedding paths are removed in step 2.

Next, we can obtain a certified-ED $(C'',L'',{\cal N}'',{\cal R}'',\Pi_{{\cal R}''\to \wtilde{G}})$ of $A$ on $\wtilde{G}$ by applying \Cref{thm:NodeWeightingInsertion} on graph $\wtilde{G}$ and $(C',L',{\cal N}',{\cal R}',\Pi_{{\cal R}'\to \wtilde{G}})$ with the new node-weighting $A_{\prune}$. By \Cref{claim:SizeOfAPrune} and \Cref{thm:NodeWeightingInsertion}, this certified-ED has the desired quality, except that it is for node-weighting $A$ instead of $\wtilde{A} = A - \deg_{F}$, but this can be easily fixed. We update 
${\cal N}''$ to $\wtilde{\cal N}$ as follows. For each vertex $v\in V(F)$ incident to some deleted edges, we choose arbitrary $\deg_{F}(v)$ virtual nodes in $A(v)$, remove these virtual nodes from $\wtilde{\cal N}$, and mark the router vertices corresponding to them redundant. The final certified-ED is $(\wtilde{C},\wtilde{L},\wtilde{\cal N},\wtilde{\cal R},\Pi_{\wtilde{\cal R}\to\wtilde{G}}) = (C'',L'',\wtilde{\cal N},{\cal R}'',\Pi_{{\cal R}''\to \wtilde{G}})$.

\begin{claim}
$|A_{\prune}|\leq \lambda_{\rt,\prune}(f)\cdot\gamma_{\emb}\cdot|F|$.
\label{claim:SizeOfAPrune}
\end{claim}
\begin{proof}
First we have $|{\cal R}_{\edge,F}| \leq |F|\cdot\gamma_{\emb}$ because each deleted edge can destroy at most $\gamma_{\emb}$ embedding paths. By definition, $\sum_{R\in{\cal R}} |R_{\edge,F}| = {\cal R}_{\edge,F}$. By \Cref{thm:Router}, each router $R\in {\cal R}$ has $|R_{\prune}|\leq \lambda_{\rt,\prune}(f)\cdot |R_{\edge,F}|$. Therefore, 
\begin{align*}
|A_{\prune}| &\leq \sum_{S\in{\cal S}} |S_{\prune}| \leq \sum_{R\in{\cal R}}|R_{\prune}|
\leq \sum_{R\in{\cal R}}\lambda_{\rt,\prune}(f)\cdot|R_{\edge,F}|
\leq \lambda_{\rt,\prune}(f)\cdot\gamma_{\emb}\cdot|F|.
\end{align*}
\end{proof}

\paragraph{The Recourse from ${\cal N}$ to $\wtilde{\cal N}$.} The recourse $\recourse({\cal N}\to \wtilde{\cal N})$ can be bounded as follows. From ${\cal N}$ to ${\cal N}'$, the number of virtual node deletions is at most $|A_{\prune}|\cdot \omega$ (each virtual node in $A_{\prune}$ may be removed once for each clustering). The recourse from ${\cal N'}$ to ${\cal N}''$ is at most $O((\omega + \kappa_{\PC,\omega})\cdot|A_{\prune}|)$ by \Cref{thm:NodeWeightingInsertion}. Lastly, the recourse from ${\cal N}''$ to $\wtilde{\cal N}$ is at most $O((\omega + \kappa_{\NC,\omega})\cdot|F|)$ because each virtual node in $\deg_{F}$ may appear in at most $\omega'' = \omega + \kappa_{\NC,\omega}$ clusters. Therefore, we have
\begin{align*}
\recourse({\cal N}\to \wtilde{\cal N}) &= |A_{\prune}|\cdot \omega + O(|A_{\prune}|\cdot(\omega + \kappa_{\PC,\omega})) + O((\omega + \kappa_{\NC,\omega})\cdot|F|)\\
&= O(\lambda_{\rt,\prune}(f)\cdot\gamma_{\emb}\cdot(\omega + \kappa_{\PC,\omega})\cdot |F|).
\end{align*}

\paragraph{The Running Time.} Now we analyse the running time. The update from $(C,L,{\cal N},{\cal R},\Pi_{{\cal R}\to G})$ to $(C',L',{\cal N}',{\cal R}',\Pi_{{\cal R}'\to\wtilde{G}})$ is dominated by the edge deletion updates to routers in ${\cal R}$ via \Cref{thm:Router}. This step takes $2^{O(f)}\cdot \gamma_{\emb}\cdot |F|\cdot 2^{\poly(1/\epsilon)}$ time because the total size of edge deletions updates is $|{\cal R}_{\edge,F}|\leq |F|\cdot\gamma_{\emb}$ (see the proof of \Cref{claim:SizeOfAPrune}). From $(C',L',{\cal N}',{\cal R}',\Pi_{{\cal R}'\to\wtilde{G}})$ to $(C'',L'',{\cal N}'',{\cal R}'',\Pi_{{\cal R}''\to \wtilde{G}})$, it takes 
$|A_{\prune}|\cdot\poly(h)\cdot n^{O(\epsilon)}\cdot(1/\phi + 2^{O(f+1)}\cdot\omega)$
time by \Cref{thm:NodeWeightingInsertion}. Finally, updating ${\cal N}''$ to $\wtilde{\cal N}$ takes $O((\omega + \kappa_{\NC,\omega})\cdot |F|)$. Therefore, the total running time is
$\lambda_{\rt,\prune}(f)\cdot\gamma_{\emb}\cdot |F|\cdot\poly(h)\cdot n^{O(\epsilon)}\cdot(1/\phi + 2^{O(f)}\cdot \omega)$.

\end{proof}

\subsection{Batched Edge Insertions}
\label{sect:DynEDEdgeInsertion}

The update algorithm for a batched edge insertion is just a simple corollary of the node-weighting insertion algorithm, because we can just assign cut value $h$ to all new edges to ``block'' them. A small point is that we need to add a new node-weighting $\deg_{E_{\new}}$ to preserve the invariant $A\geq \deg_{G} + \mathds{1}(V(G))$.

\begin{lemma}
Let $G$ be a graph with a node-weighting $A\geq \deg_{G} + \mathds{1}(V(G))$ and a certified-ED $(C,L,{\cal N},{\cal R},\Pi_{{\cal R}\to G})$ of $A$ on $G$ in \Cref{def:CertifiedED}.
Given a batch of edge insertions $E_{\new}$ and parameters $\phi$ and $h$, there is an algorithm that computes a certified-ED $(\wtilde{C},\wtilde{L}, \wtilde{\cal N},\wtilde{\cal R},\Pi_{\wtilde{\cal R}\to G})$ of $\wtilde{A} = A + \deg_{E_{\new}}\geq \deg_{\wtilde{G}}$ on $\wtilde{G} = G\cup E_{\new}$ s.t. 
\begin{itemize}
\item $\wtilde{C}=C+C_{\new}$ with $|C_{\new}| \leq O((\kappa_{\init,C}\cdot \phi + 1)\cdot h\cdot |E_{\new}|)$.
\item $\wtilde{L} = L \cup L_{\new}$ with distortion $\wtilde{\sigma} = \sigma + O(h/\kappa_{\sigma})$ and
\[
|L_{\new}| = O(((\kappa_{\CM,L} + \kappa_{\init,L})\cdot \phi + 1)\cdot|E_{\new}|).
\]
\item $\wtilde{\cal N}$ is a $\wtilde{b}$-distributed $(\wtilde{h}_{\cov},\wtilde{h}_{\sep},\wtilde{\omega})$-neighborhood cover of $\wtilde{A}$ on $G-\wtilde{C}$ with 
\[
\wtilde{b} = b+1,\ \wtilde{h}_{\cov} = h/3,
\ \wtilde{h}_{\sep} = h/3,
\text{ and }\wtilde{\omega} = \omega + \kappa_{\PC,\omega}.
\]
\item Routers in $\wtilde{\cal R}$ are maintained by \Cref{thm:Router} under $f+1$ updates. 
\item $\Pi_{\wtilde{\cal R}\to G}$ is a $(\wtilde{h}_{\emb},\wtilde{\gamma}_{\emb})$-embedding with
\begin{align*}
&\wtilde{h}_{\emb} = \max\{h_{\emb},O(\lambda_{\init,h}\cdot h)\},\\
&\wtilde{\gamma}_{\emb} = \gamma_{\emb} + O((\omega\cdot\kappa_{\CM,\gamma} + \kappa_{\init,\gamma})/\phi).
\end{align*}
\end{itemize}
The recourse from ${\cal N}$ to $\wtilde{\cal N}$ is $\recourse({\cal N}\to\wtilde{\cal N}) = O(|E_{\new}|\cdot(\omega + \kappa_{\PC,\omega}))$.
The running time is
$|E_{\new}|\cdot\poly(h)\cdot n^{O(\epsilon)}\cdot (1/\phi + 2^{O(f)}\cdot\omega)$.
\label{thm:EdgeInsertion}
\end{lemma}
\begin{proof}

We first construct a certified-ED $(C',L',{\cal N}',{\cal R}',\Pi_{{\cal R}'\to \wtilde{G}})$ of $A$ on $\wtilde{G}$ by putting large cut values at the new edges. Precisely, we let
\begin{itemize}
\item $C'=C+C_{E_{\new}}$, where $C_{E_{\new}}$ assigns value $C_{E_{\new}}(e) = h$ to each $e\in E_{\new}$,
\item $L' = L\cup L_{E_{\new}}$, where $L_{E_{\new}} = V(E_{\new})$ so that the $C_{E_{\new}}$-vertices have themselves as landmarks,
\item $({\cal N}',{\cal R}',\Pi_{{\cal R}'\to \wtilde{G}}) = ({\cal N},{\cal R},\Pi_{{\cal R}\to G})$.
\end{itemize}
It can be easily checked that $(C',L',{\cal N}',{\cal R}',\Pi_{{\cal R}'\to \wtilde{G}})$ has the same quality parameters with those of $(C,L,{\cal N},{\cal R},\Pi_{{\cal R}\to G})$. Intuitively, the reason is that the new cut $C_{E_{\new}}$ blocks all new edges (each $e\in E_{\new}$ has $\ell_{\wtilde{G}-C'}(e) = \ell_{\wtilde{G}}(e) + C'(e)\geq h+1$, since the original edge length is at least $1$), so ${\cal N}$ is still a valid distributed neighborhood cover in the graph $\wtilde{G}-C'$.

Then we apply \Cref{thm:NodeWeightingInsertion} on $(\wtilde{G},A)$ and $(C',L',{\cal N}',{\cal R}',\Pi_{{\cal R}'\to \wtilde{G}})$ with node-weighting increasing $A_{\new} = \deg_{E_{\new}}$. The output is exactly a certified-ED $(\wtilde{C},\wtilde{L},\wtilde{\cal N},\wtilde{\cal R},\Pi_{\wtilde{\cal R}\to \wtilde{G}})$ of $\wtilde{A} = A + \deg_{E_{\new}}$ with desired quality parameters. In particular, 
$\wtilde{C} = C' + C'_{\new}$ with $|C'_{\new}|\leq O(\kappa_{\init,C}\cdot\phi\cdot h\cdot |A_{\new}|)$ by \Cref{thm:NodeWeightingInsertion}, so we can rewrite $\wtilde{C} = C + C_{\new}$ where 
$C_{\new} = C_{E_{\new}} + C'_{\new}$ with size $|C_{\new}| = O((\kappa_{\init,C}\cdot\phi + 1)\cdot h\cdot |E_{\new}|)$.
Similarly, $\wtilde{L} = L' \cup L'_{\new}$ with $|L'_{\new}| \leq O((\kappa_{\CM,L} + \kappa_{\init,L})\cdot \phi\cdot |A_{\new}|)$ by \Cref{thm:NodeWeightingInsertion}, so $\wtilde{L} = L \cup L_{\new}$ where 
$L_{\new} = L_{E_{\new}} \cup L'_{\new}$ with size $|L_{\new}|\leq O(((\kappa_{\CM,L} + \kappa_{\init,L})\cdot \phi + 1)\cdot |A_{\new}|)$.

By \Cref{thm:NodeWeightingInsertion}, the recourse from ${\cal N}$ to $\wtilde{\cal N}$ is $O(|A_{\new}|\cdot(\omega + \kappa_{\PC,\omega})) = O(|E_{\new}|\cdot(\omega + \kappa_{\PC,\omega}))$. The running time is dominated by the subroutine \Cref{thm:NodeWeightingInsertion}, which takes time $|E_{\new}|\cdot\poly(h)\cdot n^{O(\epsilon)}\cdot(1/\phi + 2^{O(f)}\cdot\omega)$.

\end{proof}

\subsection{Proof of \Cref{thm:DynamicED}}
\label{sect:ProofOfDynED}
During the whole process, we will maintain a certified-ED $(C^{(i)}, L^{(i)}, {\cal N}^{(i)}, {\cal R}^{(i)}, \Pi_{{\cal R}^{(i)}\to G^{(i)}})$ of node-weighting $\bar{A}^{(i)}$ on $G^{(i)}$, where
\[
\bar{A}^{(i)} = A^{(i)} + \deg_{G^{(i)}}.
\]

\paragraph{Initialization.} We initialize $(C^{(0)},L^{(0)},{\cal N}^{(0)},{\cal R}^{(0)},\Pi_{{\cal R}^{(0)}\to G^{(0)}})$ of $\bar{A}^{(0)}$ on $G^{(0)}$ by applying \Cref{thm:InitCertifiedED} with parameters $\phi^{(0)} = \phi/\kappa_{\init,L}$ and $h^{(0)} = h\cdot \lambda_{\dynED,h}$. The output satisfies that 
\begin{itemize}
    \item $|C^{(0)}|\leq O(\kappa_{\init,C}\cdot\phi^{(0)}\cdot h^{(0)}\cdot |\bar{A}^{(0)}|)\leq O(|A^{(0)}|\cdot\phi\cdot\lambda_{\dynED,h}\cdot h)$ because $\kappa_{\init,C}\leq \kappa_{\init,L}$.
    \item $L^{(0)}$ has distortion $\sigma^{(0)} = O(h^{(0)}/\kappa_{\sigma}) = O(\lambda_{\dynED,h}\cdot h/\kappa_{\sigma})$ and \[|L^{(0)}|\leq O(\kappa_{\init,L}\cdot\phi^{(0)}\cdot|\bar{A}^{(0)}|)\leq O(|A^{(0)}|\cdot\phi).
    \]
    \item ${\cal N}^{(0)}$ is a $b^{(0)}$-distributed $(h^{(0)}_{\cov},h^{(0)}_{\sep},\omega)$-pairwise cover of $\bar{A}^{(0)}$ of $G^{(0)} - C^{(0)}$ with
    \[
    b^{(0)} = 1,\ 
    h^{(0)}_{\cov} = h^{(0)}_{\sep} = h^{(0)} = \lambda_{\dynED,h}\cdot h,\text{ and }\omega = \kappa_{\PC,\omega}.
    \]
    \item Routers in ${\cal R}^{(0)}$ are just initialized by \Cref{thm:Router}, so they are under $f^{(0)} = 0$ batched update.
    \item $\Pi_{{\cal R}^{(0)}\to G^{(0)}}$ is a $(h^{(0)}_{\emb},\gamma^{(0)}_{\emb})$-embedding with
    \begin{align*}
    h_{\emb} &= O(\lambda_{\init,h}\cdot h^{(0)}) = O(\lambda_{\init,h}\cdot\lambda_{\dynED,h}\cdot h),\\
    \gamma_{\emb} &= O(\kappa_{\init,\gamma}/\phi^{(0)}) = O(\kappa_{\init,\gamma}\cdot\kappa_{\init,L}/\phi).
    \end{align*}
\end{itemize}

\paragraph{Batched Edge Insertions/Deletions and Node-Weighting Insertions.} For now we assume $\pi^{(i)}$ represents edge insertions/deletions or node-weighting insertions (the case $\pi^{(i)}$ represents isolated vertex insertions/deletions will be discussed later). We say $\pi^{(i)}$ is \textit{incremental} if it represents node-weighting increasing or edge insertions, and $\pi^{(i)}$ is \textit{decremental} if it represents edge deletions. 

For each update $\pi^{(i)}$, we will apply \Cref{thm:NodeWeightingInsertion}, \Cref{thm:EdgeDeletion}, or \Cref{thm:EdgeInsertion}, depending on whether $\pi^{(i)}$ represents node-weighting insertion, edge deletions or edge insertions respectively, on $(C^{(i-1)},L^{(i-1)},{\cal N}^{(i-1)},{\cal R}^{(i)},\Pi_{{\cal R}^{(i-1)}\to G^{(i-1)}})$ with length parameter $h^{(i)} = h^{(i-1)}/3$ and congestion parameter
\[
\phi^{(i)} = \left\{
\begin{aligned}
&1/(\kappa_{\CM,L} + \kappa_{\init,L}),\text{ if $\pi^{(i)}$ is incremental}\\
&1/((\kappa_{\CM,L} + \kappa_{\init,L})\cdot\lambda_{\rt,\prune}(f^{(i-1)})\cdot\gamma_{\emb}^{(i-1)}),\text{ if $\pi^{(i)}$ is decremental}
\end{aligned}
\right.
\]
Directly, we have
\begin{itemize}
\item $C^{(i)} = C^{(i-1)}+ C_{\new}^{(i)}$ with
\begin{align*}
|C_{\new}^{(i)}| &\leq 
\left\{
\begin{aligned}
&O((\kappa_{\init,C}\cdot\phi^{(i)}+1)\cdot h^{(i-1)}\cdot|\pi^{(i)}|),\text{ if $\pi^{(i)}$ is incremental}\\
&O(\kappa_{\init,C}\cdot\phi^{(i)}\cdot h^{(i-1)}\cdot \lambda_{\rt,\prune}(f^{(i-1)})\cdot\gamma^{(i-1)}_{\emb}\cdot |\pi^{(i)}|),\text{ if $\pi^{(i)}$ is decremental}
\end{aligned}
\right.\\
& = O(h^{(i-1)}\cdot|\pi^{(i)}|)\\
&\leq O(\lambda_{\dynED,h}\cdot h\cdot |\pi^{(i)}|).
\end{align*}

\item $L^{(i)} = L^{(i-1)}\cup L_{\new}^{(i)}$ with distortion $\sigma^{(i)} = \sigma^{(i-1)} + O(h^{(i-1)}/\kappa_{\sigma})$ and
\begin{align*}
|L^{(i)}_{\new}| &\leq \left\{
\begin{aligned}
&O(((\kappa_{\CM,L} + \kappa_{\init,L})\cdot\phi^{(i)} + 1)\cdot |\pi^{(i)}|),\text{ if $\pi^{(i)}$ is incremental}\\
&O((\kappa_{\CM,L} + \kappa_{\init,L})\cdot\phi^{(i)}\cdot\lambda_{\rt,\prune}(f^{(i-1)})\cdot\gamma^{(i-1)}_{\emb}\cdot|\pi^{(i)}|),\text{ if $\pi^{(i)}$ is decremental}
\end{aligned}
\right.\\
&= O(|\pi^{(i)}|).
\end{align*}

\item ${\cal N}$ is a $b^{(i)}$-distributed $(h^{(i)}_{\cov},h^{(i)}_{\sep},\omega^{(i)})$-pairwise cover with
\[
b^{(i)} = b^{(i-1)} + 1,\ h^{(i)}_{\cov} = h^{(i)}_{\sep} = h^{(i)} = h^{(i-1)}/3,\ \omega^{(i)} = \omega^{(i-1)} + \kappa_{\PC,\omega}.
\]
\item Routers in ${\cal R}^{(i)}$ are maintained under $f^{(i)} = f^{(i-1)} + O(1)$ updates.
\item $\Pi_{{\cal R}^{(i)}\to G^{(i)}}$ has 
\begin{align*}
h^{(i)}_{\emb} &= \max\{h^{(i-1)}_{\emb}, \lambda_{\init,h}\cdot h^{(i-1)}\}\text{ and }\\
\gamma^{(i)}_{\emb} &= \gamma^{(i-1)}_{\emb} + O((\omega^{(i-1)}\cdot\kappa_{\CM,\gamma} + \kappa_{\init,\gamma})/\phi^{(i)}).
\end{align*}
\end{itemize}

\paragraph{The Quality of Certified-EDs.} Now we show that the certified-ED $(C^{(i)},L^{(i)},{\cal N}^{(i)},{\cal R}^{(i)},\Pi_{{\cal R}^{(i)}\to G^{(i)}})$ meet the requirements. Some of them are easy to see (or they have already shown above), so we will omit some proofs.
\begin{itemize}
\item The bounds of $|C_{\new}^{(i)}|$ and $|L_{\new}^{(i)}|$ have already shown above, and obviously we have $\omega^{(i)}\leq t\cdot\kappa_{\PC,\omega}$, $f^{(i)}\leq O(t)$ and $b^{(i)}\leq O(t)$.
\item The bounds of %
$h^{(i)}_{\cov}, h^{(i)}_{\sep}, h^{(i)}_{\emb}$ follows the fact that, for 
each $i$, $h^{(i)}\leq h^{(0)} = \lambda_{\dynED,h}\cdot h$ and $h^{(i)}\geq h^{(t)} = h^{(0)}/3^{t}\geq h$.
\item About the congestion of $\Pi_{{\cal R}^{(i)}\to G^{(i)}}$, we have
\begin{align*}
\gamma^{(i)}_{\emb} &= \gamma^{(i-1)}_{\emb} + O((\omega^{(i-1)}\cdot\kappa_{\CM,\gamma} + \kappa_{\init,\gamma})/\phi^{(i)})\\
&\leq \gamma^{(i-1)}_{\emb} + O((\omega^{(i-1)}\cdot\kappa_{\CM,\gamma} + \kappa_{\init,\gamma})\cdot(\kappa_{\CM,L} + \kappa_{\init,L})\cdot\lambda_{\rt,\prune}(f^{(i-1)})\cdot\gamma_{\emb}^{(i-1)})\\
&\leq O((t\cdot\kappa_{\PC,\omega}\cdot\kappa_{\CM,\gamma} + \kappa_{\init,\gamma})\cdot(\kappa_{\CM,L} + \kappa_{\init,L})\cdot \lambda_{\rt,\prune}(t))\cdot\gamma^{(i-1)}_{\emb}.
\end{align*}

Combining $\gamma^{(0)}_{\emb} = O(\kappa_{\init,\gamma}\cdot\kappa_{\init,L}/\phi)$, we get
\[
\gamma^{(i)}_{\emb} = (\lambda_{\rt,\prune}(t) \cdot (t\cdot \kappa_{\PC,\omega}\cdot\kappa_{\CM,\gamma}+\kappa_{\init,\gamma})\cdot(\kappa_{\CM,L} + \kappa_{\init,L}))^{O(t)}/\phi.
\]
\end{itemize}

\paragraph{Isolated Vertex Insertions/Deletions} It is easy to handle $\pi^{(i)}$ if it represents isolated vertex insertions/deletions. 

Suppose that $\pi^{(i)}$ represents isolated vertex insertions $V_{\new}$ with $A_{\new}$ as their node-weighting (i.e. after this update, $A^{(i)}(v) = A_{\new}(v)$ for each $v\in V_{\new}$). For each $v\in V_{\new}$, we create a cluster $S = A_{\new}(v)$ collecting all virtual nodes of $v$, and add $S$ into an arbitrary clustering ${\cal S}\in{\cal N}^{(i-1)}$. The router $R$ corresponding to $S$ can be initialized by \Cref{thm:Router}. $R$ can be embedded into $G$ trivially (i.e. the embedding path of each router edge degenerates to the singleton vertex $v$), because each virtual node in $V(R)$ are owned by the same vertex $v\in V(G)$. There is no change to $C$ and $L$. Obviously the certified-ED after update is of $\bar{A}^{(i)}$ on $G^{(i)}$ with quality unchanged (note that for each $v\in V_{\new}$, $\bar{A}^{(i)}(v) = A^{(i)}(v) + \deg_{G^{(i)}}(v) = A_{\new}(v)$).

Suppose that $\pi^{(i)}$ represents isolated vertex deletions $V_{\del}\subseteq V(G^{(i-1)})$. For each vertex $v\in V^{\del}$, we remove all clusters $S\in{\cal N}^{(i-1)}$ s.t. $S\subseteq \bar{A}^{(i-1)}(v)$. In fact, we have $\bigcup_{S\subseteq \bar{A}^{(i-1)}(v)} = \bar{A}^{(i-1)}(v)$ (since $v$ is isolated), so this will remove all virtual nodes $\bar{A}^{(i-1)}(v)$ from ${\cal N}^{(i-1)}$. Conceptually, for each cluster $S$ that is removed, we will further remove its corresponding router $R$ from ${\cal R}^{(i-1)}$ and its embedding, but in the actual implementation, it is safe to leave them unchanged, because our algorithm will not access $R$ and its embedding after the isolated vertex $v$ and cluster $S$ are removed. Note that because $R$ is connected and some virtual nodes in $V(R)$ are owned by the isolated vertex $v$, the embedding of $R$ into $G^{(i-1)}$ is trivial (i.e. the embedding path of each router edge degenerates to the singleton vertex $v$). Finally, we also need to remove $V_{\del}$ from the landmark set $L$.

\paragraph{The Recourse from ${\cal N}^{(i-1)}$ to ${\cal N}^{(i)}$.} 
When $\pi^{(i)}$ represents isolated vertex insertions, the recourse is trivially at most $\sum_{v\in V_{\new}} A_{\new}(v) = |\pi^{(i)}|$. For isolated vertex deletions, the recourse is at most $\sum_{v\in V_{\del}} A^{(i-1)}(v) = |\pi^{(i)}|$.
Combining \Cref{thm:NodeWeightingInsertion,thm:EdgeDeletion,thm:EdgeInsertion}, the recourse is
\[
O(\lambda_{\rt,\prune}(f^{(i-1)})\cdot\gamma^{(i-1)}_{\emb}\cdot(\omega^{(i-1)} + \kappa_{\PC,\omega})\cdot |\pi^{(i)}|)\leq \lambda_{\rt,\prune}(t)\cdot t\cdot\kappa_{\PC,\omega}\cdot\kappa_{\dynED,\gamma}\cdot|\pi^{(i)}|/\phi.
\]
because the bottleneck is when $\pi^{(i)}$ represents batched edge deletions. 

\paragraph{The Running Time.} By \Cref{thm:InitCertifiedED}, the initialization time is $(|G^{(0)}| + |\bar{A}^{(0)}|)\cdot\poly(h^{(0)})\cdot n^{O(\epsilon)} = (|G^{(0)}| + |A^{(0)}|)\cdot\poly(h)\cdot n^{O(\epsilon)}$ because $h^{(0)} = h\cdot\lambda_{\dynED,h}$. By \Cref{thm:NodeWeightingInsertion,thm:EdgeDeletion,thm:EdgeInsertion} and the update algorithm for batched isolated vertex insertions/deletions, the bottleneck is batched edge deletions and the update time for batch $\pi^{(i)}$ is 
\begin{align*}
\lambda_{\rt,\prune}(f^{(i-1)})\cdot\gamma^{(i-1)}_{\emb}\cdot |\pi^{(i)}|\cdot\poly(h^{(i)})\cdot n^{O(\epsilon)}\cdot (1/\phi^{(i)} + 2^{O(f^{(i-1)})}\cdot \omega^{(i-1)})
= |\pi^{(i)}|\cdot\poly(h)\cdot n^{O(\epsilon)}/\phi.
\end{align*}

%% file: 6-dense_dynamicED.tex
\section{Dynamic Length-Constrained Expander Decomposition with Density}
\label{sect:DynDenseED}

In this section, we will generalize the notion of certified-EDs in \Cref{def:CertifiedED} to $\rho$-dense certified-EDs defined in \Cref{def:DenseCertifiedED}, and extend the online-batch dynamic certified-ED algorithm to $\rho$-dense certified-EDs. Roughly speaking, the \emph{density} $\rho$ (defined in \Cref{def:Density}) is essentially weight functions on node-weightings (in other words, the ``weights of weights'' of vertices). The motivation of defining this second level of weights is to reduce the recourse on the pairwise cover ${\cal N}$ in the dynamic certified-ED algorithm in \Cref{sect:DynamicCertifiedED}. 

To be specific, the recourse on ${\cal N}$ is roughly $1/\phi$ by \Cref{thm:DynamicED}, and the bottleneck is the subroutine \Cref{thm:EdgeDeletion} for handling batched edge deletions. However, to obtain the low-recourse dynamic vertex sparsifier in \Cref{sect:DynSparsifier}, we will require dynamic certified-ED algorithms with recourse on ${\cal N}$ independent of $1/\phi$. To achieve this, the high level idea is to put $1/\phi$ \emph{items} (see \Cref{def:Density}) on each virtual node and define the routers on items instead of virtual nodes. Then, when handling batched edge deletions, roughly speaking, each unit of vertex pruning on routers now cause only amortized $\phi$ (instead of $1$) unit of virtual node pruning, which reduces the recourse of ${\cal N}$ by a factor $1/\phi$ if we keep the congestion of the embedding unchanged. On the other hand, routing more things with the same congestion requires more routability of the certified-ED, so the size of the cut and landmark set will now depend on the total density instead of the number of virtual nodes.

\begin{definition}[$\rho$-Dense Node-Weighting]
Given a graph $G$ and a node-weighting $A$ on $G$, the \emph{density} of $A$ is a function $\rho:A\to\mathbb{N}_{>0}$. The \emph{density} of each virtual node $v\in A$ is $\rho(v)$ and let $\rho(A) = \sum_{v\in A}\rho(v)$. Equivalently, we say a virtual node $v$ owns $\rho(v)$ many \emph{items}. Without ambiguity, we also use $\rho(v)$ to denote the set of \emph{items} owned by $v$, and use $\rho(A):=\bigcup_{v\in A}\rho(v)$ to denote the whole set of items. A node-weighting $A$ is unit-dense if each virtual node $v\in A$ has $\rho(v)=1$. 
\label{def:Density}
\end{definition}

\begin{definition}[$\rho$-Dense Pairwise Cover]

Let $G$ be a graph with $\rho$-dense node-weighting $A$. A pairwise cover ${\cal N}$ of $A$ on $G$ is further $\rho$-dense, if each cluster $S\in{\cal N}$ is assigned a density function (called the \textit{$S$-density}) $\rho_{S}:S\to \mathbb{N}_{>0}$ s.t. $\rho_{S}\geq \rho_{\mid S}$.
\label{def:PairwiseCoverDensity}
\end{definition}

\Cref{def:DenseCertifiedED} generalizes the notion of certified-EDs in \Cref{def:CertifiedED}, but with the following main differences. 
(1) We only need ${\cal N}$ to be a pairwise cover, which is weaker than the notion of $b$-distributed neighborhood cover in \Cref{def:CertifiedED}, but ${\cal N}$ should be $\rho$-dense. 
(2) Each router $R$ of $S\in{\cal N}$ has vertices corresponding to items in $\rho_{S}(S)$ instead of virtual nodes in $S$.

\begin{definition}[$\rho$-Dense Certified-ED]

Let $G$ be a graph with $\rho$-dense node-weighting $A$. A $\rho$-dense certified-ED $(C,L,{\cal N},{\cal R},\Pi_{{\cal R}\to G})$ of $A$ on $G$ satisfies that 
\begin{itemize}
\item $C$ is an integral moving cut on $G$,
\item $L$ is a landmark set of $C$ on $G$ with distortion $\sigma$,
\item ${\cal N}$ denotes a $\rho$-dense $(h_{\cov}, h_{\sep}, \omega)$-pairwise cover of $A$ on $G-C$ with $h_{\cov} = h_{\sep} = h$
\item ${\cal R}=\{R^{S}\mid S\in{\cal N}\}$ is a collection of routers s.t. for each cluster $S\in{\cal N}$, $R^{S}$ is a router maintained by \Cref{thm:Router} under at most $f$ update batches s.t. $V(R^{S})\supseteq \rho_{S}(S)$. Precisely, each vertex $v\in V(R^{S})\cap \rho_{S}(S)$ is corresponding to the item in $\rho_{S}(S)$, and each vertex $v\in V(R^{S})\setminus \rho_{S}(S)$ is a \emph{redundant} vertex corresponding to nothing. 

This is equivalent to the following 
\emph{density property}: for each virtual node $v\in S$, its $S$-density $\rho_{S}(v)$ is the number of vertices in $R^{S}$ corresponding to items of $v$.

\item $\Pi_{{\cal R}\to G}$ denotes an embedding of all routers into $G$ with length $h_{\emb}$ and congestion $\gamma_{\emb}$. In this embedding, each router vertex $v_{R}$ is mapped to the vertex $v_{G}\in V(G)$ which owns the corresponding item of $v_{R}$ (if $v_{R}$ is a redundant router vertex, then it can be mapped to an arbitrary vertex $v_{G}\in V(G)$).
\end{itemize}

\label{def:DenseCertifiedED}
\end{definition}

The organization of this section is similar to \Cref{sect:DynamicCertifiedED}. In \Cref{sect:InitDenseCertifiedED}, we show the initialization algorithm for $\rho$-dense certified-EDs. In \Cref{sect:DynDenseEDNWI,sect:DynDenseEDEdgeDeletion,sect:DynDenseEDEdgeInsertion}, we discuss how to maintain a $\rho$-dense certified-ED under batched update of node-weighting insertion, edge deletion and edge insertion respectively. 

Although we can obtain an online-batch dynamic $\rho$-dense certified-ED algorithm (analogous to \Cref{thm:DynamicED} for certified-ED) combining subroutines in \Cref{sect:InitDenseCertifiedED,sect:DynDenseEDNWI,sect:DynDenseEDEdgeDeletion,sect:DynDenseEDEdgeInsertion}, we did not formalize it here, because we need a specialized version of dynamic $\rho$-dense certified-EDs algorithm for our application in \Cref{sect:DynSparsifier}. Roughly speaking, in \Cref{sect:DynSparsifier}, we need to maintain a collection of dynamic $\rho$-dense certified-EDs of a same node-weighting $A$ with different length scales simultaneously, and further enforce that the node-weighting $A$ should include all landmark sets $L$ of this collection. Therefore, in \Cref{sect:DynDenseEDLandmarkInsertion}, we design an additional subroutine which can insert the newly generated landmarks to the node-weighting. This specialized version of dynamic $\rho$-dense certified-EDs algorithm is integrated in the proof of \Cref{thm:NonHopReducingEmulator}.

\subsection{Initialization}
\label{sect:InitDenseCertifiedED}

\begin{lemma}
Given a graph $G$ with a $\rho$-dense positive node-weighting $A$ and parameters $h$ and $\phi$, there is an algorithm that computes a $\rho$-dense certified-ED $(C, L,{\cal N},{\cal R},\Pi_{{\cal R}\to G})$ of $A$ on $G$ satisfying the following.
\begin{itemize}
\item $|C|\leq  O(\kappa_{\init,C}\cdot\phi\cdot h\cdot \rho(A))$.
\item $L$ has distortion $\sigma = h/\kappa_{\sigma}$ and size $|L| = O(\kappa_{\init,L}\cdot\phi\cdot \rho(A))$.

\item ${\cal N}$ is a $\rho$-dense $(h_{\cov},h_{\sep},\omega)$-pairwise cover of $A$ on $G-C$ with 
\[b = 1,\ h_{\cov}=h,\ h_{\sep}=h\text{ and }\omega = \kappa_{\PC,\omega}.
\]
\item Routers ${\cal R}$ are initialized by \Cref{thm:Router} undergoing $f=0$ batched update.%
\item Embedding $\Pi_{{\cal R}\to G}$ has length $h_{\emb} = O(\lambda_{\init,h}\cdot h)$ and congestion $\gamma_{\emb} = O(\kappa_{\init,\gamma}/\phi)$.
\end{itemize}
The running time is $(|G|+\rho(A))\cdot \poly(h)\cdot n^{O(\epsilon)}$.
\label{thm:InitDenseCertifiedED}
\end{lemma}
\begin{proof}
In fact, this is a simple corollary of \Cref{thm:InitCertifiedED}. We define a node-weighting $A'$ s.t. for each vertex $v\in V(G)$, $A'(v) = \rho(A(v))$ be the number of items of virtual nodes of $v$. Note that items in $\rho(A)$ one-one corresponding to virtual nodes in $A'$. Then we invoke \Cref{thm:InitCertifiedED} on graph $G$ and node-weighting $A'$ with the same parameters. The output certified-ED $(C,L,{\cal N}',{\cal R},\Pi_{{\cal R}\to G})$ is almost what we desire. 

We just need to modify ${\cal N}'$ to ${\cal N}$ to make it fit with our definition of dense certified-ED. Precisely, ${\cal N}'$ can be viewed as a pairwise cover of items in $\rho(A)$. For each cluster $S'\in{\cal N}'$, we define a new cluster $S=\{v\in A\mid \text{some item of $v$ is included by $S'$}\}$, and set the $S$-density function be $\rho_{S} = \rho_{\mid S}$. The density property of $(C,L,{\cal N},{\cal R},\Pi_{{\cal R}\to G})$ holds because $\{v\in A\mid\text{some item of $v$ is included by $S'$}\} = \{v\in A\mid\text{all items of $v$ are included by $S'$}\}$ for each $S'\in{\cal N}'$ (this is from the guarantee of \Cref{thm:InitCertifiedED}: if a cluster $S'$ includes a virtual node in $A'$ of some vertex $v\in V(G)$, then $S'$ includes all virtual nodes in $A'$ of $v$). 
\end{proof}

\subsection{Batched Node-Weighting Insertion}
\label{sect:DynDenseEDNWI}

\begin{theorem}
Let $G$ be a graph with a $\rho$-dense node-weighting $A$ and a 
$\rho$-dense certified-ED $(C, L, {\cal N}, {\cal R}, \Pi_{{\cal R}\to G})$ of $A$ on $G$ in \Cref{def:DenseCertifiedED}. Given a $\rho_{\new}$-dense node-weighting $A_{\new}$ s.t. $\wtilde{A} = A+A_{\new}\geq \deg_{G} + \mathds{1}(V(G))$, and parameter $\phi$, there is an algorithm that computes a $\wtilde{\rho}$-dense certified-ED $(\wtilde{C}, \wtilde{L}, \wtilde{\cal N}, \wtilde{\cal R},\Pi_{\wtilde{\cal R}\to G})$ of $\wtilde{A}$ with density $\wtilde{\rho}=\rho \uplus \rho_{\new}$\footnote{Recall the definition of operation $\uplus$ from \Cref{sect:PairwiseCovers}.} on $G$ s.t.
\begin{itemize}
\item $\wtilde{C}=C+C_{\new}$ with 
$|C_{\new}|\leq O(\kappa_{\init,C}\cdot\omega\cdot\phi\cdot h\cdot\rho_{\new}(A_{\new}))$
\item $\wtilde{L} = L\cup L_{\new}$ with distortion $\wtilde{\sigma} = \sigma + O(h/\kappa_{\sigma})$ and
\[
|L_{\new}|\leq O((\kappa_{\CM,L} + \kappa_{\init,L})\cdot\omega\cdot\phi\cdot\rho_{\new}(A_{\new})),
\]
\item $\wtilde{\cal N}$ is a $\wtilde{\rho}$-dense $(\wtilde{h}_{\cov},\wtilde{h}_{\sep},\wtilde{\omega})$-pairwise cover with
\[
\wtilde{h}_{\cov}=h/3,\ \wtilde{h}_{\sep} = h/3\text{ and }\wtilde{\omega} = \omega + \kappa_{\PC,\omega}
\]
\item Routers in $\wtilde{\cal R}$ are maintained by \Cref{thm:Router} under $\wtilde{f} = f+1$ updates.
\item $\Pi_{\wtilde{\cal R}\to G}$ is a $(\wtilde{h}_{\emb},\wtilde{\gamma}_{\emb})$-embedding with 
\begin{align*}
\wtilde{h}_{\emb} &= \max\{h_{\emb}, O(\lambda_{\init,h}\cdot h)\}\\
\wtilde{\gamma}_{\emb} &= \gamma_{\emb} + O((\omega\cdot\kappa_{\CM,\gamma} + \kappa_{\init,\gamma})/\phi)
\end{align*}
\end{itemize}
Given a node-weighting $\wtilde{A}^{\star}\subseteq \wtilde{A}$, the recourse from ${\cal N}$ to $\wtilde{\cal N}$ restricted on $\wtilde{A}^{\star}$ is
\[
\recourse_{\mid \wtilde{A}^{\star}}({\cal N}\to\wtilde{\cal N}) =O(\rho_{\new}(A_{\new})\cdot\omega\cdot\kappa_{\PC,\omega}/\min_{v\in \wtilde{A}^{\star}}\wtilde{\rho}(v)).
\]
The running time is $\omega\cdot\rho_{\new}(A_{\new})\cdot\poly(h)\cdot n^{O(\epsilon)}\cdot(1/\phi + 2^{O(f)}\cdot\omega)$.

\label{thm:NWInsertWithDensity}
\end{theorem}
\begin{proof}
We first set up notations in this proof to avoid confusion. We will use subscripts $v_\vertex$ to refer to a (real) vertex, $v$ without subscript to refer to a virtual node, and $v_\ite$ to refer to an item. Also, for a set $\rho'$ of items, we use $\supp_{\virtual}(\rho')$ to denote the set of virtual nodes owning some items in $\rho'$, and naturally, $\supp_{\vertex}(\rho') = \supp(\supp_{\virtual}(\rho'))$ denote the set of vertices owning some items in $\rho'$. 

The first step is to apply the local cutmatch subroutine, i.e. \Cref{thm:cutmatch}. We define source density function $\rho^{\src}=\omega\cdot\rho_{\new}(A_{\new})$ and sink density-function $\rho^{\sink} = \rho(A)$. To elaborate, for each virtual node $v\in A_{\new}$, $\rho^{\src}(v) = \omega\cdot\rho_{\new}(v)$ is the total number of items owned by $v$ scaled up by $\omega$, but for each virtual node 
$v\in A$, $\rho^{\sink}(v) = \rho(v)$ is not scaled up. Furthermore, we can extend the domain of density functions $\rho^{\src}$ and $\rho^{\sink}$ to \emph{vertices} (originally, the domain of a density function is a virtual node set) in a natural way, so they are valid inputs to a general form (see \Cref{remark:GeneralCutmatch}) of \Cref{thm:cutmatch}. That is, for each vertex $v_{\vertex}\in V(G)$, $\rho^{\src}(v_{\vertex}) = \sum_{\text{virtual node $v\in A_{\new}(v_{\vertex})$}} \rho^{\src}(v)$ is the total number of items of all $A_{\new}$-virtual nodes at $v_{\vertex}$, scaled up by $\omega$ at the end. Similarly, $\rho^{\src}(v_{\vertex}) = \sum_{v\in A(v_{\vertex})}\rho(v)$ is the total number of items of all $A$-virtual nodes at $v_{\vertex}$.

We apply \Cref{thm:cutmatch} on $G-C$ with source density function $\rho^{\src}$, sink density function $\rho^{\sink}$, parameters $h_{\CM}=h/3$ and $\phi_{\CM} = \phi$. The output is
\begin{itemize}
\item a partition $\rho^{\src}_{M},\rho^{\src}_{U}$ of $\rho^{\src}$,
\item a partition $\rho^{\sink}_{M},\rho^{\sink}_{U}$ of $\rho^{\sink}$ with $|\rho^{\sink}_{M}| = O(|\rho^{\src}_{M}|)$,
\item a matching $M_{\CM}$ between $\rho^{\src}_{M}$ and $\rho^{\sink}_{M}$ with size $|M_{\CM}| = |\rho^{\src}_{M}|$, and an embedding $\Pi_{M_{\CM}\to G-C}$ with length $h_{\CM}$ and congestion $\gamma_{\CM}=O(\kappa_{\CM,\gamma}/\phi_{\CM})$
\item a moving cut $C_{\CM}$ with size $|C_{\CM}|\leq \phi_{\CM}\cdot h_{\CM}\cdot |\rho^{\src}|$ s.t. $\rho^{\src}_{U}$ is $h_{\CM}$-separated from $\rho^{\sink}_{M}$ in $G-C-C_{\CM}$.
\item a landmark set $L_{\CM}$ of $C_{\CM}$ on $G-C$ with distortion $\sigma_{\CM} = O(h/\kappa_{\sigma})$ and size $|L_{\CM}| = O(\kappa_{\CM,L}\cdot \phi_{\CM}\cdot |\rho^{\src}|)$.
\end{itemize}
We define node-weightings $A_{1}$ and $A_{2}$ as follows. Let $\wtilde{A} = A + A_{\new}$ be the total node-weighting with density $\wtilde{\rho} = \rho\uplus \rho_{\new}$. Then 
\[
A_{1}=A+\{v\in A_{\new}\mid \rho^{\src}(v)\subseteq \rho^{\src}_{M}\}\text{ with density }\rho_{1}=\wtilde{\rho}_{\mid A_{1}},
\]
and
\[
A_{2} = \wtilde{A}\setminus \wtilde{A}(\supp_{\vertex}(\rho^{\sink}_{U}))\text{ with density }\rho_{2} = \wtilde{\rho}_{\mid A_{2}}.
\]
That is, $A_{2}$ excludes virtual nodes of vertices in $\supp_{\vertex}(\rho^{\sink}_{U})$. 
\begin{claim}
$\rho_{2}(A_{2})\leq O(\omega\cdot\rho_{\new}(A_{\new}))$.
\end{claim}
\begin{proof}
By the definition of $A_{2}$, each virtual node $v\in A_{2}$ has $\rho^{\sink}_{U}(v) = 0$. If $v\in A_{\new}$ is a new virtual node, $\rho_{2}(v) = \wtilde{\rho}(v) = \rho_{\new}(v) \leq \rho^{\src}(v)$. If $v\in A$ is an old virtual node, $\rho_{2}(v) = \wtilde{\rho}(v) = \rho(v) = \rho^{\sink}(v) = \rho^{\sink}_{M}(v) + \rho^{\sink}_{U}(v) = \rho^{\sink}_{M}(v)$. Therefore, we have $\rho_{2}(A_{2})\leq |\rho^{\src}| + |\rho^{\sink}_{M}|$. Combining $|\rho^{\sink}_{M}| = O(|\rho^{\src}_{M}|)\leq O(|\rho^{\src}|)$, we can conclude that $\rho_{2}(A_{2})\leq O(|\rho^{\src}|) = O(\omega\cdot \rho_{\new}(A_{\new}))$.
\end{proof}

\paragraph{Obtain a Certified-ED of $A_{1}$ by Update.} We will update $({\cal N},{\cal R},\Pi_{{\cal R}\to G})$ to $({\cal N}_{1},{\cal R}_{1},\Pi_{{\cal R}_{1}\to G})$ in what follows. For better understanding, with the $({\cal N}_{1},{\cal R}_{1},\Pi_{{\cal R}_{1}\to G})$ constructed below, $(C + C_{\CM},L\cup L_{\CM},{\cal N}_{1},{\cal R}_{1},\Pi_{{\cal R}_{1}\to G})$ is actually a $\rho_{1}$-dense certified-ED of $A_{1}$ on $G$ with slightly worse quality parameters (but we will not prove this explicitly). 

First, for each cluster $S\in{\cal N}$ and its router $R\in{\cal R}$, we will update them to be $S_{1}$ and $R_{1}$. Concretely, for each virtual node $v\in S$, we first establish an injection from $v$'s items in $\rho^{\sink}$ to $v$'s items in $R$, so that we can formally talk about adding matching edges in $M_{\CM}$ to $R$ later. Precisely, $v$ has $\rho^{\sink}(v)$ many items in $\rho^{\sink}$, while it has $\rho_{S}(v)\geq \rho(v) = \rho^{\sink}(v)$ many items in $R$, so we just select arbitrary $\rho^{\sink}(v)$ many $v$'s items in $R$ and one-one correspond them with items in $\rho^{\sink}(v)$.

Next, we let 
\[
M'_{S} = \{(u_{\ite},v_{\ite})\in M_{\CM}\mid \text{sink item $v_{\ite}$ belongs to some virtual node in $S$}\}
\]
be the sub-matching restricted on $S$. However, simply adding the whole $M'_{S}$ into $R$ may violate the density property from \Cref{def:PairwiseCoverDensity} (it states that $\rho_S(u) \ge \wtilde{\rho}(u)$ for all virtual nodes $u \in S$). We will explain soon why the violation might happen. %
But to fix this issue, we take 
\[
M_{S} = \{(u_{\ite},v_{\ite})\in M'_{S}\mid \text{the virtual node $u$ owning $u_{\ite}$ has at least $\rho_{\new}(u)$ source items in $M'_{S}$}\}
\]

We update $S$ to $S_{1}=S\cup\supp_{\virtual}(M_{S}^{\src})$\footnote{Recall that $M^{\src}_{S}$ denote the set of source items of $M^{S}$, and for each virtual node $u\in\supp_{\virtual}(M_{S}^{\src})$, $M^{\src}_{S}(u)$ denote the number of $u$'s items in $M^{\src}_{S}$.} and update $R$ to $R_{1} = R\cup M_{S}$ via \Cref{thm:Router}. For each new virtual node $u \in \supp_{\virtual}(M_{S}^{\src})$, set its $S$-density $\rho_{S}(u)=M^{\src}_{S}(u)$. The new router $R_{1}$ has embedding $\Pi_{R_{1}\to G} = \Pi_{R\to G}\cup \Pi_{M_{S}\to G}$, where $M_{S\to G}$ can be easily obtained from $\Pi_{M_{\CM}\to G-C}$.

With this update to $S$, we can verify that, for each virtual node $u \in S$, $\rho_S(u) \ge \wtilde{\rho}(u)$. If $u \in A$ is an old virtual node, then this condition clearly holds, i.e., $\rho_S(u) \ge \rho(u)$ even before update and $\rho_S(u)$ does not change. 
Otherwise, we have $u \in A_\new$. Then the condition $\rho_S(u) \ge \rho_\new(u) = \wtilde{\rho}_{u}$ holds by the construction of $M_S$. Note that if we instead added $M'_S$ to $R$, this condition may not hold because the number of source items owned by $u$ that is matched by $M'_S$ may be less than $\rho_{\new}(u)$.

\paragraph{Compute a Certified-ED of $A_{2}$ from Scratch.} Next, we initialize a $\rho_{2}$-dense certified-ED of $A_{2}$ on $(G-C-C_{\CM})[A_{2}]$ using \Cref{thm:InitDenseCertifiedED}. The output is $(C_{\ED}, L_{\ED}, {\cal N}_{2}, {\cal R}_{2}, \Pi_{{\cal R}_{2}\to (G-C-C_{\CM})[A_{2}]})$ satisfying that
\begin{itemize}
\item $|C_{\ED}|\leq O(\kappa_{\init,C}\cdot\phi_{\ED}\cdot h_{\ED}\cdot \rho_{2}(A_{2}))$.
\item $L_{\ED}$ is a landmark set of $C_{\ED}$ on $(G-C-C_{\CM})[A_{2}]$ with distortion $\sigma_{\ED} = h_{\ED}/\kappa_{\sigma}$ and size $|L_{\ED}| = O(\kappa_{\init,L}\cdot\phi_{\ED}\cdot \rho_{2}(A_{2}))$.
\item ${\cal N}_{2}$ is a $\rho_{2}$-dense $(h_{\cov,2},h_{\sep,2},\omega_{2})$-pairwise cover of $A_{2}$ on $G-C-C_{\CM}-C_{\ED}$ with $h_{\cov,2} = h_{\ED}, h_{\sep,2} = h_{\ED}$ and $\omega_{2} = \kappa_{\PC,\omega}$.
\item Routers ${\cal R}_{2}$ are initialized by \Cref{thm:Router}.
\item The embedding $\Pi_{{\cal R}_{2}\to (G-C-C_{\CM})[A_{2}]}$ has congestion $\gamma_{\emb,2} = O(\kappa_{\init,\gamma}/\phi_{\ED})$ and length $h_{\emb,2} = O(\lambda_{\init,h}\cdot h_{\ED})$.
\end{itemize}

\paragraph{The New Certified-ED and Its Quality.} The new certified-ED $(\wtilde{C},\wtilde{L},\wtilde{\cal N},\wtilde{\cal R},\Pi_{\wtilde{\cal R}\to G})$ is given by
\begin{gather*}
\wtilde{C} = C+C_{\new}\text{ 
where }C_{\new} = C_{\CM} + C_{\ED},\\ 
\wtilde{L} = L\cup L_{\new}\text{ 
where }L_{\new} = L_{\CM}\cup L_{\ED},\\
(\wtilde{\cal N}, \wtilde{\cal R}, \Pi_{\wtilde{\cal R}\to G}) = ({\cal N}_{1},{\cal R}_{1},\Pi_{{\cal R}_{1}\to G})\cup ({\cal N}_{2},{\cal R}_{2},\Pi_{{\cal R}_{2}\to G-C-C_{\CM}}).
\end{gather*}

Towards its quality, the size of the new cut $C_{\new}$ is given by
\begin{align*}
|C_{\new}|\leq O(\phi_{\CM}\cdot h_{\CM}\cdot |\rho^{\src}|) +O(\kappa_{\init,C}\cdot\phi_{\ED}\cdot h_{\ED}\cdot\rho_{2}(A_{2}))\leq O(\omega\cdot\kappa_{\init,C}\cdot\phi\cdot h\cdot \rho_{\new}(A_{\new})).
\end{align*}
because $\rho_{2}(A_{2})\leq O(\omega\cdot\rho_{\new}(A_{\new}))$.
The distortion of $\wtilde{L}$ is $\wtilde{\sigma} = \sigma + \sigma_{\CM} + \sigma_{\ED} = \sigma + O(h/\kappa_{\sigma})$ by \Cref{lemma:LandmarkUnion}, and the new landmark set $L_{\new}$ has size
\begin{align*}
|L_{\new}|\leq O(\kappa_{\CM,L}\cdot\phi_{\CM}\cdot|\rho^{\src}|) + O(\kappa_{\init,L}\cdot\phi_{\ED}\cdot\rho_{2}(A_{2}))\leq O((\kappa_{\CM,L} + \kappa_{\init,L})\cdot\omega\cdot\phi\cdot\rho_{\new}(A_{\new})).
\end{align*}
The quality of $\wtilde{\cal N}$ is given by \Cref{lemma:DenseTildeN}. Each router $R_{1}\in {\cal R}_{1}$ is updated from some $R\in{\cal R}$ via one matching insertion, so $\wtilde{f} = f+1$. For embedding $\Pi_{\wtilde{\cal R}\to G}$, we have
\begin{align*}
\wtilde{h}_{\emb} &= \max\{h_{\emb}, h_{\CM}, h_{\emb, 2}\} = \max\{h_{\emb}, O(\lambda_{\init,h}\cdot h)\},\\
\wtilde{\gamma}_{\emb} &= \gamma_{\emb} + \omega\cdot\gamma_{\CM} + \gamma_{\emb,2} = \gamma_{\emb} = O((\omega\cdot\kappa_{\CM,\gamma} +\kappa_{\init,\gamma})\cdot\phi),
\end{align*}
where the $\omega$ factor of $\gamma_{\CM}$ is because each matching edge will be added into routers at most $\omega$ times (at most one in each clustering of ${\cal N}$).

\begin{lemma}
$\wtilde{\cal N}$ is a $\wtilde{\rho}$-dense pairwise cover of $\wtilde{A}$ on $G-\wtilde{C}$ with $\wtilde{h}_{\cov} = \wtilde{h}_{\sep} = h/3$ and $\wtilde{\omega} = \omega + \kappa_{\PC,\omega}$.\label{lemma:DenseTildeN}
\end{lemma}
\begin{proof}

The proof of the separation and the width is exactly the same with that in the proof of \Cref{lemma:ProperiesOfTildeN}, so we omit it here. 

${\wtilde{\cal N}}$ is $\wtilde{\rho}$-dense by the following reasons. First, when computing ${\cal N}_{1}$, each time we add a new virtual node $v$ to some original cluster $S$, we will assign $\rho_{S}(v) \geq \rho_{\new}(v)$, so each cluster $S_{1}\in{\cal S}_{1}$ has $\rho_{S_{1}}\geq \wtilde{\rho}(S_{1})$. Also, clusters in ${\cal N}_{2}$ respect the density function $\wtilde{\rho}$ by construction, so we can conclude that $\wtilde{\cal N}$ is $\wtilde{\rho}$-dense.

In the remaining proof, we will show that $\wtilde{\cal N}$ has cover radius $\wtilde{h}_{\cov}$. It is equivalent to show for each pair of virtual nodes $u,v\in \wtilde{A}$ s.t. $\dist_{G-\wtilde{C}}(u,v)\leq \wtilde{h}_{\cov}$, there is a cluster $S\in\wtilde{\cal N}$ containing both $u$ and $v$. We will study by cases. One can check that the five cases below will cover all pairs of virtual nodes in $\wtilde{A}$. 

\underline{Case 1.} If $u,v\in A$, then because $\dist_{G-C}(u,v)\leq \dist_{G-\wtilde{C}}(u,v)\leq \wtilde{h}_{\cov}\leq h_{\cov}$, there is a cluster $S\in {\cal N}$ containing both $u,v$. By the construction of ${\cal N}_{1}$, the cluster $S_{1}$ originated from $S$ will have $u,v\in S_{1}$.

\underline{Case 2.} Suppose $u\in A_{\new}\cap A_{1}$ and $v\in A$. For each item $u_{\ite}\in \omega\cdot\rho_{\new}(u)$ of $u$, it is matched to an item of some virtual node $u'$ via $M_{\CM}$. Observe that
\begin{align*}
\dist_{G-C}(u',v)&\leq \dist_{G-C}(u',u) + \dist_{G-C}(u,v)\\
&\leq \dist_{G-C}(u',u) + \dist_{G-\wtilde{C}}(u,v)\\
&\leq h_{\CM}+\wtilde{h}_{\cov}\leq h_{\cov},
\end{align*}
which means each $u$'s item $u_{\ite}$ is matched to some cluster containing $v$. Note that there are at most $\omega$ clusters containing $v$ in ${\cal N}$, and there are totally $\rho^{\src}(u) = \omega\cdot\rho_{\new}(u)$ items of $u$ in the matching $M_{\CM}$ by the definition of $A_{1}$. By averaging, there must be a cluster $S\in{\cal N}$ with $v\in S$ that receives at least $\rho_{\new}(u)$ items of $u$. Therefore, by the construction algorithm, $u$ will be added to the cluster $S_{1}\in{\cal N}_{1}$ originated from $S$.

\underline{Case 3.} Suppose both virtual nodes $u,v\in A_{\new}\cap A_{1}$. For each pair of units $u_{\ite}\in \omega\cdot\rho_{\new}(u)$ of $u$ and $v_{\ite}\in\omega\cdot\rho_{\new}(v)$ of $v$, $M_{\CM}$ will match $u_{\ite}$ and $v_{\ite}$ to items of some virtual nodes $u'$ and $v'$ respectively. Similarly,
\begin{align*}
\dist_{G-C}(u',v')&\leq \dist_{G-C}(u',u) + \dist_{G-C}(u,v) + \dist_{G-C}(v',v)\\
&\leq \dist_{G-C}(u',u) + \dist_{G-\wtilde{C}}(u,v) + \dist_{G-C}(v',v)\\
&\leq h_{\CM}+\wtilde{h}_{\cov}+h_{\CM}\\
&\leq h_{\cov},
\end{align*}
which means for each item pair $u_{\ite}$ and $v_{\ite}$ of $u$ and $v$, there is at least a cluster that both $u_{\ite}$ and $v_{\ite}$ are matched to. Note that in each clustering of ${\cal N}$, there is at most one cluster that can receive items of $u$ via $M_{\CM}$, because ${\cal N}$ has separation $h_{\sep}>2 h_{\CM}$. Hence the number of clusters that can receive both items of $u$ and $v$ via $M_{\CM}$ is at most $\omega$ (at most one from each clustering). The total number of such pairs of $x$ and $y$ is $\omega^{2}\rho_{\new}(u)\rho_{\new}(v)$. By averaging, there exists a cluster that receives at least $\omega\rho_{\new}(u)\rho_{\new}(v)$ pairs, which means it receives at least $\rho_{\new}(u)$ items of $u$ and at least $\rho_{\new}(v)$ items of $v$, because a cluster will receive at most $A^{\src}(u) = \omega\cdot\rho_{\new}(u)$ and $A^{\src}(v) = \omega\cdot\rho_{\new}(v)$ units of $u$ and $v$ respectively. Therefore, by the construction, both $u$ and $v$ will be added to this cluster. 

\underline{Case 4.} Suppose $u\in A_{\new}\setminus A_{1}$ and $v\in A_{2}$. By definition, the virtual node $u$ has at least one unit in $A^{\src}_{U}$, so the vertex $u$ is $h_{\CM}$-separated from vertices $\supp(A^{\sink}_{U})$ on $G-\wtilde{C}$. This means $\dist_{(G-\wtilde{C})[A_{2}]}(u,v) = \dist_{G-\wtilde{C}}(u,v)$ because the shortest $u$-$v$ path on $G-\wtilde{C}$ cannot go through any vertex in $\supp(A^{\sink}_{U}) = \supp(\wtilde{A}\setminus A_{2})$. Therefore, a cluster from ${\cal N}_{2}$ will contain both $u$ and $v$.

\underline{Case 5.} Suppose $u\in A_{\new}\setminus A_{1}$ and $v\in A\setminus A_{2}$. This case is impossible because we have argued that $u$ is $h_{\CM}$-separated from the vertex $v\in \supp(A\setminus A_{2}) = \supp(A^{\src}_{U})$ on $G-\wtilde{C}$, contradicting that $\dist_{G-\wtilde{C}}(u,v)\leq h/3 = h_{\CM}$.

\end{proof}

\paragraph{The Recourse from ${\cal N}$ to $\wtilde{\cal N}$ Restricted on $\wtilde{A}^{\star}$.} From ${\cal N}$ to ${\cal N}_{1}$, each virtual node in $A_{\new}$ will be added into ${\cal N}$ at most $\omega$ times (at most one for each clustering), so 
\[
\recourse_{\mid \wtilde{A}^{\star}}({\cal N}\to {\cal N}_{1})\leq \omega\cdot|A_{\new}\cap \wtilde{A}^{\star}|\leq \omega \cdot \rho_{\new}(A_{\new})/\min_{v\in\wtilde{A}^{\star}}\wtilde{\rho}(v).
\]
Because $\wtilde{\cal N}$ is the union of ${\cal N}_{1}$ and ${\cal N}_{2}$, trivially $\recourse_{\mid\wtilde{A}^{\star}}({\cal N}_{1}\to\wtilde{\cal N})$ is at most the size of ${\cal N}_{2}$ restricted on $\wtilde{A}^{\star}$, i.e. $\size({\cal N}_{2,\mid \wtilde{A}^{\star}})$. Recall that ${\cal N}_{2}$ is a pairwise cover of $A_{2}$ with width $\kappa_{\PC,\omega}$, so $\size({\cal N}_{2,\mid \wtilde{A}^{\star}})\leq \kappa_{\PC,\omega}\cdot|A_{2}\cap \wtilde{A}^{\star}|$, so.
\begin{align*}
\recourse_{\mid\wtilde{A}^{\star}}({\cal N}_{1}\to\wtilde{\cal N}) 
&\leq \kappa_{\PC,\omega}\cdot|A_{2}\cap \wtilde{A}^{\star}|
\leq \kappa_{\PC,\omega}\cdot \wtilde{\rho}(A_{2})/\min_{v\in \wtilde{A}^{\star}}\wtilde{\rho}(v).
\end{align*}
Recall that we have shown $\wtilde{\rho}(A_{2}) = \rho_{2}(A_{2})\leq O(\omega\cdot \rho_{\new}(A_{\new}))$, so
\begin{align*}
\recourse_{\mid\wtilde{A}^{\star}}({\cal N}\to\wtilde{\cal N})
&\leq \recourse_{\mid \wtilde{A}^{\star}}({\cal N}\to {\cal N}_{1}) + \recourse_{\mid\wtilde{A}^{\star}}({\cal N}_{1}\to\wtilde{\cal N}) \leq O(\kappa_{\PC,\omega}\cdot\omega\cdot\rho_{\new}(A_{\new})/\min_{v\in\wtilde{A}^{\star}}\wtilde{\rho}(v)).
\end{align*}

\paragraph{The Running Time.} The running time of the cutmatch subroutines is $|\rho^{\src}|\cdot\poly(h_{\CM},\log n)/\phi_{\CM}$.
When we update $({\cal N}, {\cal R}, \Pi_{{\cal R}\to G})$ to $({\cal N}_{1},{\cal R}_{1}, \Pi_{{\cal R}_{1}\to G})$, the running time is dominated by the update to routers. Because $\sum_{S\in{\cal S}}|M_{S}|\leq \omega \cdot|M_{\CM}|\leq O(\omega^{2}\cdot \rho_{\new}(A_{\new}))$, the total running time for this step is $2^{O(f)}\cdot\omega^{2}\cdot \rho_{\new}(A_{\new})\cdot 2^{\poly(1/\epsilon)}$
by \Cref{thm:Router}. The construction time of $(C_{\ED}, L_{\ED}, {\cal N}_{2}, {\cal R}_{2}, \Pi_{{\cal R}_{2}\to (G-C-C_{\CM})[A_{2}]})$ is
$\omega\cdot\rho_{\new}(A_{\new})\cdot\poly(h_{\ED})\cdot n^{O(\epsilon)}$,
because the size of subgraph $(G-C-C_{\CM})[A_{2}]$ is bounded by $\deg_{G}(\supp(A_{2}))\leq \rho_{2}(A_{2})\leq O(\omega\cdot\rho_{\new}(A_{\new}))$, where the first inequality is because, for each vertex $v_{\vertex}\in\supp(A_{2})$, we have $\wtilde{A}(v_{\vertex})\subseteq A_{2}$ and $\deg_{G}(v_{\vertex})\leq \wtilde{A}(v_{\vertex})\leq \sum_{v\in \wtilde{A}(v_{\vertex})}1\leq \sum_{v\in \wtilde{A}(v_{\vertex})}\rho_{2}(v)$.
In conclusion, the total running time is
\[
\omega\cdot\rho_{\new}(A_{\new})\cdot \poly(h)\cdot n^{O(\epsilon)}\cdot (1/\phi + 2^{O(f)}\cdot\omega).
\]

\end{proof}

\subsection{Batched Edge Deletions}
\label{sect:DynDenseEDEdgeDeletion}

When we apply \Cref{thm:EdgeDelWithDensity} in the proof of 
\Cref{thm:NonHopReducingEmulator} (i.e. \Cref{sect:ProofOfDynSparsifier}), we will set the parameter $\mu = 2t$, where $t = \Theta(1/\epsilon)$ is the global time parameter.

\begin{theorem}
Let $G$ be a graph with a $\rho$-dense node-weighting $A\geq \deg_{G} + \mathds{1}(V(G))$ and a $\rho$-dense certified-ED $(C,{\cal N},{\cal R},\Pi_{{\cal R}\to G})$ of $A$ on $G$ in \Cref{def:DenseCertifiedED}. Given a set of edge deletions $F\subseteq E(G)$ with parameters $\phi$ and $\mu$, there is an algorithm that computes a $\tilde{\rho}$-dense certified ED $(C, L, {\cal N}, {\cal R}, \Pi_{{\cal R}\to G})$ of $A$ on $\wtilde{G} = G\setminus F$ with $\wtilde{\rho}=\lceil(1-1/\mu)\rho\rceil$ s.t.
\begin{itemize}
\item $\wtilde{C}=C_{\mid E(\wtilde{G})}+C_{\new}$ with 
$|C_{\new}|\leq O(\kappa_{\init,C}\cdot\omega\cdot\phi\cdot \mu \cdot \lambda_{\rt,\prune}(f)\cdot\gamma_{\emb}\cdot|F|\cdot h)$.

\item $\wtilde{L} = L\cup L_{\new}$ with distortion distortion $\wtilde{\sigma} = \sigma + O(h/\kappa_{\sigma})$ and
\[
|L_{\new}| \leq O((\kappa_{\CM,L} + \kappa_{\init,L})\cdot\omega\cdot\phi\cdot\mu\cdot \lambda_{\rt,\prune}(f)\cdot\gamma_{\emb}\cdot|F|)
\]

\item ${\cal N}$ is a $\wtilde{\rho}$-dense $(\wtilde{h}_{\cov},\wtilde{h}_{\sep},\wtilde{\omega})$-pairwise cover with
\[
\wtilde{h}_{\cov} = \wtilde{h}_{\sep} =  h/3,\text{ and }\wtilde{\omega} = \omega + \kappa_{\PC,\omega}.
\]
\item Each router $\wtilde{R}\in\wtilde{\cal R}$ is maintained by \Cref{thm:Router} under $\wtilde{f} = f+O(1)$ updates.
\item $\Pi_{\wtilde{\cal R}\to G}$ is a $(\wtilde{h}_{\emb},\wtilde{\gamma}_{\emb})$-embedding with
\begin{align*}
\wtilde{h}_{\emb} &= \max\{h_{\emb},\lambda_{\init,h}\cdot h\}\\
\wtilde{\gamma}_{\emb} &= \gamma_{\emb} + O((\omega\cdot\kappa_{\CM,\gamma} + \kappa_{\init,\gamma})/\phi)
\end{align*}
\end{itemize}
Given a node-weighting $\wtilde{A}^{\star}\subseteq \wtilde{A}$, the recourse from ${\cal N}$ to $\wtilde{\cal N}$ restricted on $\wtilde{A}^{\star}$ is
\[
\recourse_{\mid \wtilde{A}^{\star}}({\cal N}\to\wtilde{\cal N}) =  O(\omega\cdot\kappa_{\PC,\omega}\cdot\mu\cdot\lambda_{\rt,\prune}(f)\cdot\gamma_{\emb}\cdot |F|/\min_{v\in \wtilde{A}^{\star}}\rho(v)).
\]
The running time is $\omega\cdot\mu\cdot \lambda_{\rt,\prune}(f)\cdot\gamma_{\emb}\cdot |F|\cdot\poly(h)\cdot n^{O(\epsilon)}\cdot(1/\phi + 2^{O(f)}\cdot \omega)$.
\label{thm:EdgeDelWithDensity}
\end{theorem}
\begin{proof}
We first compute the set ${\cal R}_{\edge,F} = \Pi^{-1}_{{\cal R}\to G} (F)$ which collects all edges in routers with underlying embedding paths going through some deleted edges in $F$. For each router $R\in{\cal R}$ corresponding to a cluster $S\in{\cal N}$, let $R_{\edge,F}={\cal R}_{\edge,F}\cap E(R)$. We apply an edge-deletions update to $R$ with $R_{\edge,F}$ as inputs by \Cref{thm:Router}, which returns a set $R_{\prune}\subseteq V(R)$ of pruned vertices. Recall that each router vertex in $R_{\prune}$ is either a redundant vertex or corresponding to an item in $\rho_{S}(S)$. For each virtual node $v\in S$, let $R_{\prune}(v)$ denote the number of $v$'s item in $R_{\prune}$. We define the pruned virtual nodes of $S$ to be 
\[
S_{\prune}=\{v\in S\mid R_{\prune}(v)\geq \lceil \rho(v)/\mu\rceil\}.
\]
Finally, let $A_{\prune}=\bigcup_{S\in{\cal N}}S_{\prune}$ collect all pruned virtual nodes. 

Now we partition $A$ into $A_{\prune}$ and $A' = A\setminus A_{\prune}$, and let the density of $A_{\prune}$ be $\rho_{\prune} = \rho_{\mid A_{\prune}}$ and the density of $A'$ be $\rho' = \lceil(1-1/\mu)\rho_{\mid A'}\rceil$. Note that we can obtain a $\rho'$-dense certified-ED $(C',L',{\cal N}',{\cal R}',\Pi_{{\cal R}'\to \wtilde{G}})$ of $A'$ on $\wtilde{G}$ by simply removing $A_{\prune}$ from $(C,L,{\cal N},{\cal R},\Pi_{{\cal R}\to G})$. Concretely, for each cluster $S\in{\cal N}$ and its router $R$, we do the following steps. 
\begin{enumerate}
\item Remove virtual nodes $A_{\prune}\cap S$ from $S$. Note that the new $R$ after the pruning from \Cref{thm:Router} may still have vertices corresponding to items of some virtual node in $A_{\prune}\cap S$ and we will mark all such router vertices redundant. 
\item For each remaining virtual node $v$ in $S$, subtract its $S$-density $\rho_{S}(v)$ by $R_{\prune}(v)$, because these many items of $v$ were pruned from $R$ via \Cref{thm:Router} before.
\item For each edge removed from $R$ due to the previous update via \Cref{thm:Router}, remove its embedding path from $\Pi_{{\cal R}\to G}$. 
\end{enumerate}
Lastly, we take $C' = C_{\mid E(\wtilde{G})}$ by removing cut values on deleted edges. Also, let $L' = L\cup V(F)$.%

We observe that the $\rho'$-dense certified-ED $(C',L',{\cal N}',{\cal R}',\Pi_{{\cal R}'\to \wtilde{G}})$ of $A'$ on $\wtilde{G}$ has the same quality parameters as those of the original $\rho$-dense certified-ED $(C,L,{\cal N},{\cal R},\Pi_{{\cal R}\to G})$ of $A$ on $G$, except that each router in ${\cal R'}$ suffers one more batched update. In particular, $\Pi_{{\cal R}'\to\wtilde{G}}$ is a valid embedding on $\wtilde{G}$ because we remove all router edges with embedding paths going through deleted edges by \Cref{thm:Router}, and these embedding paths are removed in step (3). Also, the new cluster ${\cal N}'$ is $\rho'$-dense because for each cluster $S'\in{\cal N'}$ (updated from $S\in{\cal N}$) and virtual node $v\in S'$, we have 
\[\rho_{S'}(v) = \rho_{S}(v)-R_{\prune}(v)\geq \rho(v) - R_{\prune}(v)\geq \lceil (1-1/\mu)\cdot\rho(v)\rceil=\rho'(v)
\]
because $R_{\prune}(v)\leq\lfloor\rho(v)/\mu\rfloor$ (otherwise $v$ will be selected into $S_{\prune}\subseteq A_{\prune}$ and it will not belong to $S'$). 

We can obtain a $\wtilde{\rho}$-dense certified-ED $(\wtilde{C},\wtilde{L},\wtilde{\cal N},\wtilde{\cal R},\Pi_{\wtilde{\cal R}\to G})$ of $A$ on $\wtilde{G}$ by applying \Cref{thm:NWInsertWithDensity} on $(C',L',{\cal N}',{\cal R}',\Pi_{{\cal R}'\to \wtilde{G}})$ with the new $\rho_{\prune}$-dense node weighting $A_{\prune}$. The output has the desired quality by plugging the bound of $\rho_{\prune}(A_{\prune})$ (see \Cref{claim:BoundOfAprune}) into the quality bounds of \Cref{thm:NWInsertWithDensity}.

\begin{claim}
$\rho_{\prune}(A_{\prune})=\rho(A_{\prune})\leq \mu \cdot\lambda_{\rt,\prune}(f)\cdot\gamma_{\emb}\cdot|F|$
\label{claim:BoundOfAprune}
\end{claim}
\begin{proof}
First, the number of router edges with embedding paths going through deleted edges, i.e. $|{\cal R}_{\edge,F}|$, is at most $|F|\cdot\gamma_{\emb}$. By definition, $\sum_{R\in{\cal R}}|R_{\edge,F}| = |{\cal R}_{\edge,F}|$. By \Cref{thm:Router}, each router $R\in{\cal R}$ has $|R_{\prune}|\leq \lambda_{\rt,\prune}(f)\cdot |R_{\edge,F}|$. Therefore, we have
\begin{align*}
\rho_{\prune}(A_{\prune}) &= \rho(A_{\prune}) = \sum_{v\in A_{\prune}}\rho(v) 
\leq \sum_{S\in{\cal N}}\sum_{v\in S_{\prune}}\rho(v) 
\leq \sum_{S\in{\cal N}}\sum_{v\in S_{\prune}}\mu\cdot R_{\prune}(v)\\
&\leq \mu\cdot\sum_{S\in{\cal N}} |R_{\prune}|\leq \mu\cdot\lambda_{\rt,\prune}(f)\cdot\sum_{S\in{\cal N}} |R_{\edge,F}|= \mu\cdot\lambda_{\rt,\prune}(f)\cdot|{\cal R}_{\edge,F}|\\
&\leq \mu\cdot\lambda_{\rt,\prune}(f)\cdot\gamma_{\emb}\cdot |F|.
\end{align*}
\end{proof}
\paragraph{The Recourse from ${\cal N}$ to $\wtilde{\cal N}$ Restricted on $\wtilde{A}^{\star}$.} Because ${\cal N}'$ is obtained from ${\cal N}$ by removing virtual nodes in $A_{\prune}$ and the width is $\omega$, we have
\[
\recourse_{\mid \wtilde{A}^{\star}}({\cal N}\to{\cal N}')\leq \omega\cdot|A_{\prune}\cap \wtilde{A}^{\star}|\leq \rho(A_{\prune})/\min_{v\in \wtilde{A}^{\star}}\rho(v). 
\]
From ${\cal N'}$ to $\wtilde{\cal N}$, by \Cref{thm:NWInsertWithDensity}, we have
\[
\recourse_{\mid \wtilde{A}^{\star}}({\cal N}'\to\wtilde{\cal N}) \leq O(\rho_{\prune}(A_{\prune})\cdot\omega\cdot\kappa_{\PC,\omega}/\min_{v\in \wtilde{A}^{\star}}\rho(v)).
\]
Therefore, by \Cref{claim:BoundOfAprune}, the total recourse is 
\begin{align*}
\recourse_{\mid \wtilde{A}^{\star}}({\cal N}\to\wtilde{\cal N})
&\leq O(\rho(A_{\prune})\cdot\omega\cdot\kappa_{\PC,\omega}/\min_{v\in \wtilde{A}^{\star}}\rho(v))\\
&\leq O(\omega\cdot\kappa_{\PC,\omega}\cdot\mu\cdot\lambda_{\rt,\prune}(f)\cdot\gamma_{\emb}\cdot |F|/\min_{v\in \wtilde{A}^{\star}}\rho(v)).
\end{align*}

\paragraph{The Running Time.} The running time to update $(C,L,{\cal N},{\cal R},\Pi_{{\cal R}\to G})$ to $(C',L',{\cal N}',{\cal R}',\Pi_{{\cal R}'\to\wtilde{G}})$ is dominated the edge deletion updates to the routers, which take totally 
$2^{O(f)}\cdot |{\cal R}_{\edge,F}|\cdot 2^{\poly(1/\epsilon)}\leq 2^{O(f)}\cdot |F|\cdot\gamma_{\emb}\cdot 2^{\poly(1/\epsilon)}$
time. From $(C',L',{\cal N}',{\cal R}',\Pi_{{\cal R}'\to\wtilde{G}})$ to $(\wtilde{C},\wtilde{L},\wtilde{\cal N},\wtilde{\cal R},\Pi_{\wtilde{\cal R}\to\wtilde{G}})$, it takes 
$\omega\cdot\rho_{\prune}(A_{\prune})\cdot\poly(h)\cdot n^{O(\epsilon)}\cdot(1/\phi + 2^{O(f+1)}\cdot\omega)$
time. Therefore, the total running time is
\[
\omega\cdot\mu\cdot \lambda_{\rt,\prune}(f)\cdot\gamma_{\emb}\cdot |F|\cdot\poly(h)\cdot n^{O(\epsilon)}\cdot(1/\phi + 2^{O(f)}\cdot \omega).
\]

\end{proof}

\subsection{Batched Edge Insertions}
\label{sect:DynDenseEDEdgeInsertion}

\begin{theorem}
Let $G$ be a graph with a $\rho$-dense node-weighting $A\geq \deg_{G} + \mathds{1}(V(G))$ and a $\rho$-dense certified-ED $(C,L,{\cal N}, {\cal R}, \Pi_{{\cal R}\to G})$ of $A$ on $G$ in \Cref{def:DenseCertifiedED}. Given a batch of edge insertions $E_{\new}$ with parameter $\phi$, there is an algorithm that computes a certified-ED $(\wtilde{C}, \wtilde{L}, \wtilde{{\cal N}}, \wtilde{{\cal R}}, \Pi_{\wtilde{\cal R}\to \wtilde{G}})$ of a $\wtilde{\rho}$-dense $\wtilde{A}$ on $\wtilde{G} = G\cup E_{\new}$ satisfying that
\begin{itemize}
\item $\wtilde{A} = A + A_{\new}$ and $\wtilde{\rho} = \rho\uplus \rho_{\new}$, where $A_{\new} = \deg_{E_{\new}}$ with unit density $\rho_{\new}$.

\item $\wtilde{C}=C+C_{\new}$ with $|C_{\new}|\leq O((\omega \cdot \kappa_{\init,C}\cdot\phi + 1)\cdot h\cdot |E_{\new}|)$.

\item $\wtilde{L} = L \cup L_{\new}$ with distortion $\wtilde{\sigma} = \sigma + O(h/\kappa_{\sigma})$ and
\[
|L_{\new}|\leq O((\omega\cdot(\kappa_{\CM,L} + \kappa_{\init,L})\cdot\phi + 1)\cdot|E_{\new}|);
\]
\item $\wtilde{N}$ is a $\wtilde{\rho}$-dense $(\wtilde{h}_{\cov},\wtilde{h}_{\sep},\wtilde{\omega})$-pairwise cover with
\[
\wtilde{h}_{\cov}=h/3,\ \wtilde{h}_{\sep} = h/3\text{ and }\wtilde{\omega} = \omega + \kappa_{\PC,\omega}
\]
\item Routers in $\wtilde{\cal R}$ are maintained by \Cref{thm:Router} under $f+1$ updates.

\item $\Pi_{\wtilde{\cal R}\to G}$ is a $(\wtilde{h}_{\emb},\wtilde{\gamma}_{\emb})$-embedding with 
\begin{align*}
\wtilde{h}_{\emb} &= \max\{h_{\emb}, O(\lambda_{\init,h}\cdot h)\}\\
\wtilde{\gamma}_{\emb} &= \gamma_{\emb} + O((\omega\cdot\kappa_{\CM,\gamma} + \kappa_{\init,\gamma})/\phi)
\end{align*}

\end{itemize}
The recourse from ${\cal N}$ to $\wtilde{\cal N}$ is $\recourse({\cal N}\to\wtilde{\cal N}) = O(|E_{\new}|\cdot\omega\cdot\kappa_{\PC,\omega})$.
The running time is $|E_{\new}|\cdot\poly(h)\cdot n^{O(\epsilon)}\cdot (1/\phi + 2^{O(f)}\cdot\omega)$.

\label{thm:EdgeInsertWithDensity}
\end{theorem}

\begin{proof}
We first update $(C,L,{\cal N},{\cal R},\Pi_{{\cal R}\to G})$ to $(C',L',{\cal N},{\cal R},\Pi_{{\cal R}\to \wtilde{G}})$ as follows.
\begin{enumerate}
\item Let $C' = C + C_{E_{\new}}$, where $C_{E_{\new}}$ assigns $C_{E_{\new}}(e) = h$ to each $e\in E_{\new}$.
\item Let $L' = L\cup L_{E_{\new}}$, where $L_{E_{\new}} = V(E_{\new})$ includes all endpoints of new edges, so the $C_{E_{\new}}$-vertices have themselves as landmarks.
\item Keep $\Pi_{{\cal R}\to \wtilde{G}} = \Pi_{{\cal R}\to G}$ unchanged.
\end{enumerate}
Obviously, $(C',L',{\cal N},{\cal R},\Pi_{{\cal R}\to \wtilde{G}})$ is a $\rho$-dense certified-ED of $A$ on $\wtilde{G}$ with quality parameters the same as those of $(C,L,{\cal N},{\cal R},\Pi_{{\cal R}\to G})$.

Next, let $A_{\new} = \deg_{E_{\new}}$ with density $\rho_{\new}(v) = 1$ for each virtual node $v\in A_{\new}$. Then apply \Cref{thm:NWInsertWithDensity} on the $\rho$-dense certified-ED $(C',L',{\cal N},{\cal R},\Pi_{{\cal R}\to \wtilde{G}})$ of $A$ on $\wtilde{G}$ with the new $\rho_{\new}$-dense node-weighing $A_{\new}$. The output $(\wtilde{C},\wtilde{L},\wtilde{\cal N},\wtilde{\cal R},\Pi_{\wtilde{\cal R}\to \wtilde{G}})$ is exactly a $\wtilde{\rho}$-dense certified-ED of $\wtilde{A} = A + A_{\new}$ on $\wtilde{G}$ with density $\wtilde{\rho} = \rho\uplus \rho_{\new}$ and the desired quality parameters. 

By \Cref{thm:NWInsertWithDensity}, the recourse from ${\cal N}$ to $\wtilde{\cal N}$ is at most
$O(\rho_{\new}(A_{\new})\cdot\omega \cdot \kappa_{\PC,\omega})\leq O(|E_{\new}|\cdot\omega\cdot\kappa_{\PC,\omega})$, and the running time is $|E_{\new}|\cdot\poly(h)\cdot n^{O(\epsilon)}\cdot(1/\phi + 2^{O(f)}\cdot\omega)$.

\end{proof}

\subsection{Batched Landmark Insertions into Node-Weighting}
\label{sect:DynDenseEDLandmarkInsertion}

\begin{lemma}
Let $G$ be a graph with node-weighting $A\geq \deg_{G}$ and parameter $\phi$. For each $1\leq j\leq \bar{j}$, suppose we are given a $\rho$-dense certified-ED $(C_{j},L_{j},{\cal N}_{j},{\cal R}_{j},\Pi_{{\cal R}_{j}\to G})$ of $A$ on $G$ s.t.
\begin{itemize}
\item $L_{j}$ has distortion $\sigma_{j}$ and a subset $L^{\star}_{j}\subseteq L_{j}$ specified, and let $\tau = \sum_{1\leq j\leq \bar{j}}|L^{\star}_{j}|$;
\item ${\cal N}_{j}$ is a $(h_{\cov,j},h_{\sep,j},\omega_{j})$-pairwise cover of $A$ on $G$ with $h_{\cov,j} = h_{\sep,j} = h_{j}$;
\item ${\cal R}_{j}$ are maintained by \Cref{thm:Router} under $f_{j}$ updates;
\item $\Pi_{{\cal R}_{j}\to G}$ is a $(h_{\emb,j},\gamma_{\emb,j})$-embedding.
\end{itemize}
We define global parameters $\kappa_{\insLM,\gamma} = n^{O(\epsilon^{4})}$ and $\lambda_{\insLM,t} = \log_{\kappa_{\insLM,\gamma}}n =  O(1/\epsilon^{4})$. There is an algorithm that computes, for each $1\leq j\leq \bar{j}$, a $\wtilde{\rho}$-dense certified-ED $(\wtilde{C}_{j},\wtilde{L}_{j}, \wtilde{\cal N}_{j},\wtilde{\cal R}_{j},\Pi_{\wtilde{\cal R}_{j}\to G})$ of $\wtilde{A}$ on $G$ s.t.
\begin{itemize}
\item $\wtilde{C}_{j} = C_{j} + C_{\new,j}$ with $|C_{\new,j}| = O(h_{j}\cdot\tau)$
\item $\wtilde{L}_{j} = L_{j}\cup L_{\new,j}$ with distortion $\wtilde{\sigma}_{j} = \sigma_{j} + O(h_{j}/\kappa_{\sigma})$ and $|L_{\new,j}| = O(\tau)$.
\item $\wtilde{A} = A + A_{\new}$ with density $\wtilde{\rho} = \rho \uplus \rho_{\new}$, where \[
A_{\new} = \mathds{1}(\bigcup_{j}L^{\star}_{j}) + \mathds{1}(\bigcup_{j}L_{\new,j}\setminus \bigcup_{j}L_{j}),
\]
and $A_{\new}$ has $1/\phi$-uniform density $\rho_{\new}$.
\item $\wtilde{N}_{j}$ is a $(\wtilde{h}_{\cov,j},\wtilde{h}_{\sep,j},\wtilde{\omega}_{j})$-pairwise cover of $\wtilde{A}$ on $G$ with
\[
\wtilde{h}_{\cov,j} = \wtilde{h}_{\sep,j} = h_{j}/2^{O(\lambda_{\insLM,t})}\text{ and }\wtilde{\omega}_{j} = \omega_{j} + \lambda_{\insLM,t}\cdot\kappa_{\PC,\omega}.
\]
\item Routers in $\wtilde{\cal R}_{j}$ are maintained by \Cref{thm:Router} under at most $\wtilde{f}_{j} = f_{j} + O(\lambda_{\insLM,t})$ update batches.

\item $\Pi_{\wtilde{\cal R}_{j}\to G}$ is a $(\wtilde{h}_{\emb,j},\wtilde{\gamma}_{\emb,j})$-embedding with 
\begin{align*}
\wtilde{h}_{\emb,j} &= \max\{h_{\emb,j}, O(\lambda_{\init,h}\cdot h_{j})\},\\
\wtilde{\gamma}_{\emb,j} &= \gamma_{\emb,j} +  O( \wtilde{\omega}_{j}\cdot(\kappa_{\CM,L} + \kappa_{\init,L})\cdot(\wtilde{\omega}_{j}\cdot\kappa_{\CM,\gamma} + \kappa_{\init,\gamma})\cdot\kappa_{\insLM,\gamma}\cdot\lambda_{\insLM,t}/\phi)
\end{align*}

\end{itemize}
Given an arbitrary set $\wtilde{A}^{\star}$ of virtual nodes, the recourse from ${\cal N}_{j}$ to $\wtilde{\cal N}_{j}$ restricted on $\wtilde{A}^{\star}$ is
\[
\recourse_{\mid\wtilde{A}^{\star}}({\cal N}_{j}\to \wtilde{\cal N}_{j}) = O(\tau\cdot\wtilde{\omega}_{j}\cdot\kappa_{\PC,\omega}/(\phi\cdot\min_{v\in \wtilde{A}^{\star}\cap \wtilde{A}}\wtilde{\rho}(v))).
\]
The running time is $(\tau/\phi)\cdot n^{O(\epsilon)}\cdot \sum_{j}(\wtilde{\omega}_{j}^{2}\cdot\poly(h_{j})\cdot(1/\phi + 2^{O(\wtilde{f}_{j})}))$.

\label{lemma:LandmarkClosure}
\end{lemma}

\begin{proof}

Roughly speaking, this subroutine wants to add the landmarks $\bigcup_{j}L^{\star}_{j}$ into the node-weighting. This can be done by calling the dynamic ED subroutine, at a cost of creating more landmarks. Hence our algorithm is iterative. Let $L^{\star,(i)}_{j}$ be the new landmarks created in iteration $i$ of the level-$j$ certified-ED (initially $L^{\star,(0)}_{j} = L^{\star}_{j}$). Then at each iteration $i$, we add $\bigcup_{j}L^{\star,(i-1)}_{j}$ (excluding repetitive landmarks) into the node-weighting of all certified-EDs, which will create new landmarks $L^{\star,(i)}_{j}$ for the next iteration. We will set proper parameters to make sure the number of new landmarks drops by a factor roughly $\kappa_{\insLM,L}$ in each iteration, so the number of iteration is bounded by $\lambda_{\insLM,t}$.

Formally, let $(C_{j}^{(0)},L_{j}^{(0)},{\cal N}_{j}^{(0)},{\cal R}_{j}^{(0)},\Pi_{{\cal R}_{j}^{(0)}\to G})$ be the initial certified-ED for each $j$. In each iteration $i\geq 1$, we let
\[
T_{\new}^{(i)} = \bigcup_{j} L^{\star,(i-1)}_{j}\setminus \bigcup_{j} L^{(i-2)}_{j}\text{ (initially $T^{(1)}_{\new} = \bigcup_{j} L^{\star,(0)}_{j}$)} %
\]
If $T^{(i)}_{\new}$ is empty, we terminate the algorithm.

Otherwise, let $A_{\new}^{(i)} = \mathds{1}(T^{(i)}_{\new})$ be the unit node-weighting over $T^{(i)}_{\new}$ with $1/\phi$-uniform density $\rho^{(i)}_{\new}$. Then we call \Cref{thm:NWInsertWithDensity} on $(C_{j}^{(i-1)}, L_{j}^{(i-1)}, {\cal N}_{j}^{(i-1)}, {\cal R}_{j}^{(i-1)}, \Pi_{{\cal R}_{j}^{(i-1)}\to G})$ with the $\rho^{(i)}_{\new}$-dense node-weighting $A^{(i)}_{\new}$ and parameter 
\[
\phi_{j}^{(i)} = \phi/(\kappa_{\insLM,\gamma}\cdot \omega^{(i-1)}_{j}\cdot(\kappa_{\CM,L} + \kappa_{\init,L})).
\]
Let $h^{(i-1)}_{j}$ refer to $h^{(i-1)}_{\cov,j}$. The output is a $\rho^{(i)}$-dense certified-ED $(C_{j}^{(i)}, L_{j}^{(i)}, {\cal N}_{j}^{(i)}, {\cal R}_{j}^{(i)}, \Pi_{{\cal R}_{j}^{(i)}\to G})$ of $A^{(i)} = A^{(i-1)} + A^{(i)}_{\new}$ on $G$ with density $\rho^{(i)} = \rho^{(i-1)}\uplus \rho^{(i)}_{\new}$. It further satisfies that
\begin{itemize}
\item $C_{j}^{(i)} = C_{j}^{(i-1)} + C^{(i)}_{\new,j}$ with
\begin{align*}
|C^{(i)}_{\new,j}| \leq O( \kappa_{\init,C}\cdot \omega_{j}^{(i-1)}\cdot\phi_{j}^{(i)}\cdot h^{(i-1)}_{j}\cdot \rho^{(i)}_{\new}(A^{(i)}_{\new}))\leq O(h^{(i-1)}_{j}\cdot |A^{(i)}_{\new}|/\kappa_{\insLM,\gamma}),
\end{align*}
where the second inequality is by $\rho^{(i)}_{\new}(A^{(i)}_{\new}) = |A^{(i)}_{\new}|/\phi$ and $\kappa_{\CM,L} + \kappa_{\init,L}\geq \kappa_{\init,C}$.

\item $L_{j}^{(i)} = L_{j}^{(i-1)}\cup L^{\star,(i)}_{j}$ has distortion $\sigma^{(i)}_{j} = \sigma^{(i-1)}_{j} + O(h^{(i-1)}_{j}/\kappa_{\sigma})$ and
\[
|L^{\star,(i)}_{j}|\leq O((\kappa_{\CM,L} + \kappa_{\init,L})\cdot\omega^{(i-1)}_{j}\cdot\phi_{j}^{(i)}\cdot \rho^{(i)}_{\new}(A^{(i)}_{\new}))\leq O(|A^{(i)}_{\new}|/\kappa_{\insLM,\gamma}).
\]
\item ${\cal N}^{(i)}_{j}$ has quality parameters
\[
h^{(i)}_{\cov,j} = h^{(i)}_{\sep,j} = h^{(i-1)}_{\cov,j}/3,\text{ and }\omega^{(i)}_{j} = \omega^{(i-1)}_{j} + \kappa_{\PC,\omega}.
\]
\item Routers in ${\cal R}^{(i)}_{j}$ are maintained by \Cref{thm:Router} under $f^{(i)}_{j} = f^{(i-1)}_{j} + O(1)$ batched updates.
\item The embedding $\Pi_{{\cal R}^{(i)}_{j}\to G}$ has
\begin{align*}
h^{(i)}_{\emb,j} &= \max\{h^{(i-1)}_{\emb,j}, O(\lambda_{\init,h}\cdot h^{(i)}_{j})\}\\
\gamma^{(i)}_{\emb,j} &= \gamma^{(i-1)}_{\emb,j} + O((\omega^{(i-1)}_{j}\cdot \kappa_{\CM,\gamma} + \kappa_{\init,\gamma})/\phi_{j}^{(i)})\\
&= \gamma^{(i-1)}_{\emb,j} + O(\omega^{(i-1)}_{j}\cdot(\kappa_{\CM,L}+\kappa_{\init,L})\cdot(\omega^{(i-1)}_{j}\cdot\kappa_{\CM,\gamma} + \kappa_{\init,\gamma})\cdot\kappa_{\insLM,\gamma}/\phi)
\end{align*}
\end{itemize}

\begin{claim}
The number of iteration is at most $\lambda_{\insLM,t}$, and $\sum_{i} |A^{(i)}_{\new}|\leq O(\tau)$. 
\label{claim:NumberOfIterationsInsLM}
\end{claim}
\begin{proof}
For each $i\geq 2$, $|A^{(i)}_{\new}|\leq \sum_{j}|L^{\star,(i-1)}_{j}|\leq O(\bar{j}\cdot |A^{(i-1)}_{\new}|/\kappa_{\insLM,\gamma})$. Because $|A^{(1)}_{\new}|\leq n$, $\kappa_{\insLM,\gamma} = n^{O(\epsilon)}$ and $\bar{j} = O(\log h) = O(\log n)$, after at most $O(1/\epsilon)\leq \lambda_{\insLM,t}$ iterations, $|L^{(i)}_{\new}| = |A^{(i)}_{\new}| < 1$ and the algorithm terminates. Because $|A^{(1)}_{\new}|\leq \sum_{j}|L^{\star,(0)}_{j}| = \tau$ and $|A^{(i)}_{\new}|$ decreases exponentially over time, we have $\sum_{i}|A^{(i)}_{\new}| = O(\tau)$.
\end{proof}

\paragraph{Quality of the Final Certified-EDs.} For each $j$, the final (after the last iteration) certified-EDs $(\wtilde{C}_{j},\wtilde{L}_{j},\wtilde{\cal N}_{j},\wtilde{\cal R}_{j},\Pi_{\wtilde{\cal R}_{j}\to G})$ has qualtiy parameters as follows.
\begin{itemize}
\item $\wtilde{C}_{j} = C_{j} + C_{\new,j}$, where $C_{\new,j} = \sum_{i} C_{\new,j}^{(i)}$ with size
\begin{align*}
|C_{\new,j}| &= \sum_{i} |C_{\new,j}^{(i)}|\leq \sum_{i}O(h_{j}^{(i-1)}\cdot |A^{(i)}_{\new}|/\kappa_{\insLM,\gamma})\leq O(h_{j}\cdot\tau/\kappa_{\insLM,\gamma}).
\end{align*}

\item $\wtilde{L}_{j} = L_{j}\cup L_{\new,j}$, where $L_{\new,j} = \bigcup_{i}L_{j}^{\star,(i)}$ with size
\[
|L_{\new,j}| = \sum_{i}|L_{j}^{\star,(i)}|\leq\sum_{i}O(|A^{(i)}_{\new}|/\kappa_{\insLM,\gamma})\leq O(\tau/\kappa_{\insLM,\gamma}).
\]

\item $\wtilde{A} = A + A_{\new}$, where
\begin{align*}
A_{\new} &= \sum_{i} A^{(i)}_{\new} = \sum_{i} \mathds{1}(T^{(i)}_{\new})\\
& = \mathds{1}(\bigcup_{j}L^{\star,(0)}_{j}) + \sum_{i\geq 2} \mathds{1}(\bigcup_{j} L^{\star,(i-1)}_{j}\setminus \bigcup_{j}L^{(i-2)}_{j})\\
&= \mathds{1}(\bigcup_{j}L^{\star,(0)}_{j}) + \mathds{1}(\bigcup_{i\geq 2,j}L^{\star,(i-1)}_{j}\setminus \bigcup_{j}L^{(0)}_{j})\\
&= \mathds{1}(\bigcup_{j}L^{\star}_{j}) + \mathds{1}(\bigcup_{j}L_{\new,j}\setminus\bigcup_{j}L_{j})
\end{align*}
and $A_{\new}$ has $1/\phi$-uniform density $\rho_{\new}$.

\item The bounds on $\wtilde{h}_{\cov,j},\wtilde{h}_{\sep,j},\wtilde{\omega}_{j}$ and $\wtilde{f}_{j}$ are trivial by \Cref{claim:NumberOfIterationsInsLM}.

\item The embedding $\Pi_{\wtilde{\cal R}_{j}\to G}$ Trivially has length $\wtilde{h}_{\emb,j} = \max\{h_{\emb,j},O(\lambda_{\init,h}\cdot h_{j})\}$. 
Regarding the congestion, because in each iteration $i$,
\[
\gamma^{(i)}_{\emb,j} = \gamma^{(i-1)}_{\emb,j} + O(\omega^{(i-1)}_{j}\cdot(\kappa_{\CM,L}+\kappa_{\init,L})\cdot(\omega^{(i-1)}_{j}\cdot\kappa_{\CM,\gamma} + \kappa_{\init,\gamma})\cdot\kappa_{\insLM,\gamma}/\phi),
\]
and the number of iteration at most $\lambda_{\insLM,t}$, the final congestion is
\begin{align*}
\wtilde{\gamma}_{\emb,j}= \gamma_{\emb,j} +  O( \wtilde{\omega}_{j}\cdot(\kappa_{\CM,L} + \kappa_{\init,L})\cdot(\wtilde{\omega}_{j}\cdot\kappa_{\CM,\gamma} + \kappa_{\init,\gamma})\cdot\kappa_{\insLM,\gamma}\cdot\lambda_{\insLM,t}/\phi).
\end{align*}

\end{itemize}

\paragraph{The Recourse from ${\cal N}_{j}$ to $\wtilde{\cal N}_{j}$ Restricted on $\wtilde{A}^{\star}$.} We have
\begin{align*}
&\recourse_{\mid \wtilde{A}^{\star}}({\cal N}_{j}\to\wtilde{\cal N}_{j}) \leq \sum_{i\geq 1} \recourse_{\mid (\wtilde{A}^{\star}\cap A^{(i)})}({\cal N}_{j}^{(i-1)}\to {\cal N}_{j}^{(i)})\\
&\leq \sum_{i\geq 1}O(\rho_{\new}^{(i)}(A_{\new}^{(i)})\cdot\omega^{(i-1)}_{j}\cdot\kappa_{\PC,\omega}/\min_{v\in \wtilde{A}^{\star}\cap A^{(i)}}\rho^{(i)}(v))\\
&\leq \sum_{i\geq 1}O(|A_{\new}^{(i)}|\cdot\omega^{(i-1)}_{j}\cdot\kappa_{\PC,\omega}/(\phi\cdot\min_{v\in \wtilde{A}^{\star}\cap A^{(i)}}\rho^{(i)}(v)))\\
&\leq O(\tau\cdot\wtilde{\omega}_{j}\cdot\kappa_{\PC,\omega}/(\phi\cdot\min_{v\in \wtilde{A}^{\star}\cap \wtilde{A}}\wtilde{\rho}(v))),
\end{align*}
where the second inequality is by \Cref{thm:NWInsertWithDensity}.

\paragraph{The Running Time.}
At each iteration $i$ and for each $j$, the time to get $(C_{j}^{(i)}, L_{j}^{(i)}, {\cal N}_{j}^{(i)}, {\cal R}_{j}^{(i)}, \Pi_{{\cal R}_{j}^{(i)}\to G})$ is 
$\omega_{j}^{(i-1)}\cdot \rho_{\new}^{(i)}(A_{\new}^{(i)})\cdot\poly(h^{(i-1)}_{j})\cdot (1/\phi_{j}^{(i)} + 2^{O(f_{j}^{(i-1)})}\cdot\omega_{j}^{(i-1)})$. 
Summing over all $i$ and $j$, the total running time is bounded by
\begin{align*}
&~~~~\lambda_{\insLM,t}\cdot\sum_{j}\wtilde{\omega}_{j}\cdot(\tau/\phi)\cdot \poly(h_{j})\cdot n^{O(\epsilon)}\cdot (\kappa_{\insLM,\gamma}\cdot\wtilde{\omega}_{j}\cdot(\kappa_{\CM,L}+\kappa_{\init,L})/\phi + 2^{O(\wtilde{f}_{j})}\cdot\wtilde{\omega}_{j})\\
&\leq (\tau/\phi)\cdot n^{O(\epsilon)}\cdot\sum_{j}(\wtilde{\omega}_{j}^{2}\cdot\poly(h_{j})\cdot(1/\phi + 2^{O(\wtilde{f}_{j})}))
\end{align*}

\end{proof}

%% file: 7-vertex_sparsifier.tex
\section{Dynamic Vertex Sparsifiers for Bounded Distances}
\label{sect:DynSparsifier}

In this section, we introduce our online-batch dynamic vertex sparsifier algorithm, where the sparsifier will preserve pairwise distance between a given terminal set. The main theorem of this section is \Cref{thm:NonHopReducingEmulator}.

\begin{definition}[Vertex Sparsifiers]
Let $G$ be a graph with terminals $T\subseteq V(G)$, an $(\alpha_{\low},\alpha_{\up},h)$-sparsifier $H$ of $T$ on $G$ is a weighted graph satisfying the following.
\begin{enumerate}
\item\label{property:Sparsifier1} $T\subseteq V(H)$.
\item\label{property:Sparsifier2} For each $u,v\in T$, there is $\dist_{H}(u,v)\cdot\alpha_{\low}\geq \dist_{G}(u,v)$.
\item\label{property:Sparsifier3} For each $u,v\in T$ s.t. $\dist_{G}(u,v)\leq h$, there is $\dist_{H}(u,v)\leq \alpha_{\up}\cdot\dist_{G}(u,v)$.
\end{enumerate}
The \emph{stretch/approximation} of the sparsifier $H$ is $\alpha_{\low}\cdot\alpha_{\up}$. In particular, we say $H$ is an $(\alpha_{\low},\alpha_{\up})$-sparsifer of $T$ on $G$ if property 3 holds for every pair of $u,v\in T$ without restriction.
\label{def:Sparsifier}
\end{definition}

Before stating \Cref{thm:NonHopReducingEmulator}, we first explain the term of \emph{fresh} vertices, which appears in \Cref{Item:SparsifierVertexIdentifier} of \Cref{thm:NonHopReducingEmulator}. In a dynamic algorithm, we say a vertex is fresh, if at the moment it is created, we assign it a \emph{global} vertex identifier different from the identifiers of all vertices appear before. Equivalently, we can image that there is a large \emph{global} pool of vertex identifiers, and we request an unpicked identifier when creating this vertex. Whether a newly created vertex is fresh or not is its inherent property, and we are not allowed to change the global identifier of a vertex once it has been created. We note that in \Cref{thm:NonHopReducingEmulator}, when we say a vertex $v$ in the sparsifer $H$ of $G$ is fresh, besides $v\notin V(G)$, we also know that $v$ has identifier different from any vertex created before in an algorithm that uses \Cref{thm:NonHopReducingEmulator} as a subroutine.

It is important to specify the global identifiers of vertices, because in some following sections (e.g. \Cref{sect:DynHopEmu}), we will work on several graphs simultaneously and define some new graph using vertices from different graphs. The new graph is well-defined only when global identifiers exist, because we need to know whether two vertices from two different graphs correspond to the same vertex in the new graph or not. This is also the reason why we cannot allow change the identifier of a vertex, because changing the identifier may cause vertex split operation to the new graph, which is expensive. 

\begin{theorem}[Dynamic vertex sparsifiers]
\label{thm:NonHopReducingEmulator}

Let $G$ be a dynamic graph with an incremental terminal set $T\subseteq V(G)$ and parameters $h$ and $\phi\leq 1/t$ under $t$ batches of updates $\pi^{(1)},\pi^{(2)},...\pi^{(t)}$ of edge insertions/deletions, isolated vertex insertions/deletions and terminals insertions s.t. at each time $1\leq i\leq t$,
\[
|\pi^{(i)}| \leq \phi\cdot |G^{(i)}|.
\]
There is an algorithm that maintains an $(\lambda_{H,\low},\lambda_{H,\up},h)$-sparsifier $H$ of $T$ on $G$, where
\[
\lambda_{H,\low} = \lambda_{\rt,h}(t\cdot\lambda_{\insLM,t})\cdot\lambda_{\init,h}\cdot 2^{O(t\cdot\lambda_{\insLM,t})} = 2^{\poly(1/\epsilon)}\text{ and }\lambda_{H,\up} = O(1).
\]
The algorithm further guarantee the following.
\begin{enumerate}
\item\label{Item:SparsifierVertexIdentifier} $H$ is further a $(\lambda_{H,\low},\lambda_{H,\up},h)$-sparsifier of extended terminals $\bar{T}$ s.t. $\bar{T} = V(H)\cap V(G)$ and $T\subseteq \bar{T}$. Vertices in $V(H)\setminus \bar{T}$ are fresh vertices.

\item\label{Item:SparisiferSize} The size of $H$ is bounded by $|H|\leq \kappa_{H,\size}\cdot(|T| + \phi\cdot|G|)$
where
\[
\kappa_{H,\size} = \kappa_{j}^{3}\cdot 2^{O(t\cdot\lambda_{\insLM,t})}\cdot \lambda_{\rt,\prune}(t\cdot\lambda_{\insLM,t}) \cdot \kappa_{\PC,\omega}^{3}\cdot (\kappa_{\CM,L} +\kappa_{\init,L})\cdot(\kappa_{\CM,\gamma} + \kappa_{\init,\gamma}) = n^{O(\epsilon^{4})},
\]
and $\kappa_{j} = O(\log n)$ is a global upper bound of $\lceil 100\cdot\log h\rceil$.

\item The number of vertices in $H$ is bounded by $|V(H)|\leq n^{O(\epsilon^{4})}\cdot |V(G)|$.

\item The recourse from $H^{(i-1)}$ to $H^{(i)}$ is $\recourse(H^{(i-1)}\to H^{(i)})\leq \kappa_{H,\rcs}\cdot|\pi^{(i)}|$, where
\[
\kappa_{H,\rcs} = \kappa_{j}^{3}\cdot 2^{O(t\cdot\lambda_{\insLM,t})}\cdot \lambda_{\rt,\prune}(t\cdot\lambda_{\insLM,t}) \cdot \kappa_{\PC,\omega}^{4}\cdot (\kappa_{\CM,L} +\kappa_{\init,L})\cdot(\kappa_{\CM,\gamma} + \kappa_{\init,\gamma})\cdot \kappa_{\insLM,\gamma} = n^{O(\epsilon^{4})}.
\]

\item The initialization time is $(|G^{(0)}| + |T^{(0)}|/\phi)\cdot\poly(h)\cdot n^{O(\epsilon)}/\phi$ and the time to handle the batch $\pi^{(i)}$ is $|\pi^{(i)}|\cdot\poly(h)\cdot n^{O(\epsilon)}/\phi^{2}$.

\item At any time, given a path $P_{H}$ on $H$ connecting some vertices $u,v\in \bar{T}$, the algorithm can compute a $u$-$v$ path $P_{G}$ on $G$ with $\ell_{G}(P_{G})\leq \lambda_{H,\low}\cdot \ell_{H}(P_{H})$ and $|P_{G}|\geq |P_{H}|/2$. The path-unfolding time is $|P_{G}|\cdot 2^{O(1/\epsilon^{4})}$.
\end{enumerate}
\end{theorem}

We emphasize that an important advantage of our dynamic sparsifier algorithm is that it has \emph{low recourse}, independent of $1/\phi$. This is crucial for our dynamic length-constrained expander hierarchy algorithm in \Cref{sect:ExpanderHierarchy}, where we recursively apply \Cref{thm:NonHopReducingEmulator} to build some kind of hierarchy of sparsifiers with roughly $\log n/\log(1/\phi)$ levels but still require small final recourse.

However, one may notice that \Cref{thm:NonHopReducingEmulator} has some drawbacks, and we now discuss how to fix them. Nonetheless, \Cref{thm:NonHopReducingEmulator} itself is sufficient for its application in \Cref{sect:ExpanderHierarchy}.

\begin{itemize}
\item First, each batched update $\pi^{(i)}$ is required to have size no more than $\phi\cdot |E(G^{(i)})|$. This can be easily fixed by a standard rebuilding technique. Namely, we rebuild the sparsifier if the batched update is too large. Indeed, we use this technique when we apply \Cref{thm:NonHopReducingEmulator} in \Cref{sect:ExpanderHierarchy}. 
\item Second, the size of the sparsifier has an additive term $\phi\cdot |E(G)|$, which means when the number of terminals is small, the sparsifier only sparsifies $G$ by a factor $\phi$. We can fix this by recursively applying \Cref{thm:NonHopReducingEmulator} until the sparsifier has size roughly $|T|$. See \Cref{sect:ExtendedSparsifier} for a formal proof.
\item Third, the update time polynomially depends on $h$, which means this algorithm will only be efficient when the sparsifier preserves bounded distance. We will see how to fix this in \Cref{sect:ExtendedSparsifier} by exploiting the stacking technique in \Cref{sect:Stacking}.
\end{itemize}

In \Cref{sect:EDsToSparsifier}, we first show a static construction of the sparsifier from a collection of (specialized) certified-EDs with different length parameters, and then make the sparsifier dynamic by assuming these certified-EDs are dynamic. In \Cref{sect:ProofOfDynSparsifier}, we complete the proof of \Cref{thm:NonHopReducingEmulator} by showing a dynamic algorithm for these certified-EDs. In order to achieve low recourse, we will exploit $\rho$-dense certified-EDs in \Cref{sect:DynDenseED}.

\subsection{Construction Based on Landmarks}
\label{sect:EDsToSparsifier}

\Cref{thm:CoversToEmulators} shows a static construction of the sparsifier, given a collection of \emph{specialized} certified-EDs at different levels $j$ with length parameters $2,4,..,2^{j},...,h$. Precisely, the static construction will only access the moving cuts, landmark sets and pairwise covers (the routers and embedding are only for controlling the diameter of pairwise covers on $G$). The certified-EDs are specialized in the sense that they should be with respect to an \emph{extended} terminal set $\bar{T}$ including the original terminal set $T$ and all the landmarks. 

In fact, the sparsifier we construct is also with respect to this extended terminal set $\bar{T}$. The construction is given in the first paragraph in the proof of \Cref{thm:CoversToEmulators}. The intuition that the sparsifier will have low stretch is as follows. Consider a pair of extended terminals $u$ and $v$. We look at the level $j$ with length parameter $2^{j}$ slightly larger than $\dist_{G}(u,v)$. If $\dist_{G-C_{j}}(u,v)$ only increases a little bit compared to $\dist_{G}(u,v)$, then the pairwise cover at level $j$ (with cover radius $2^{j}$ on $G-C_{j}$) will have a cluster covering $u$ and $v$ and the corresponding star graph will preserve the $u$-$v$ distance. Otherwise, $\dist_{G-C_{j}}(u,v)$ increases quite a lot, which means there is a $C_{j}$-vertex $w'$ on the middle of the $u$-$v$ shortest path on $G$. Then we look at the landmark $w\in L_{j}\subseteq \bar{T}$ of $w'$ (which is quite close to $w'$), and consider pairs $(u,w)$ and $(w,v)$ recursively.

\begin{lemma}
Let $G$ be a graph with terminal set $T$ and a parameter $h$. For each integer $1\leq j\leq \bar{j}=\lceil \log(100\cdot h) \rceil$, let $h_{j} = 2^{j}$ and suppose we are given an integral moving cut $C_{j}$, a landmark set $L_{j}$ and a pairwise cover ${\cal N}_{j}$ satisfying the following.
\begin{itemize}
\item $L_{j}$ is a landmark set of $C_{j}$ on $G$ with distortion $\sigma_{j} = h_{j}/\log^{2} n$.
\item ${\cal N}_{j}$ is a pairwise cover of $\bar{T} = T\cup(\bigcup_{j'}L_{j'})$ on $G-C_{j}$ with cover radius $h_{\cov,j} = h_{j}$ and width $\omega_{j}$. Furthermore, the diameter of ${\cal N}_{j}$ \underline{on $G$} is $h_{\Gdiam,j}\leq h_{j}\cdot\alpha_{\low}$ for some fixed parameter $\alpha_{\low}\geq 1$.
\end{itemize}
There is an algorithm that computes an $(\alpha_{\low},\alpha_{\up},h)$-sparsifier $H$ of $\bar{T}$ on $G$ with $\alpha_{\up} = O(1)$ and $|H| = O(\sum_{j}(|C_{j}|/h_{j} + \omega_{j}\cdot|\bar{T}|))$. The running time is $O(\sum_{j}(|C_{j}| + \omega_{j}\cdot|\bar{T}|))$

\label{thm:CoversToEmulators}
\end{lemma}
\begin{proof}
First, we define a set $E_{\heavy}$ of \textit{heavy edges}, which collect all edges $e\in E(G)$ s.t. $C_{j}(e)\geq h_{j}/10$ for some $j$. Note that each edge $e=(u,v)\in E_{\heavy}$ must have endpoints $u,v\in \bar{T}$, because $C_{j}(e)\geq h_{j}/10\geq \sigma_{j}$ implies $u,v\in L_{j}\subseteq \bar{T}$. Second, for each pairwise cover ${\cal N}_{j}$ and each cluster $S\in{\cal N}_{j}$, let $H^{\star}_{S}$ be a star graph with $V(H^{\star}_{S}) = S\cup\{v_{S}\}$, where $v_{S}$ is the \emph{artificial center}, and $E(H^{\star}_{S}) = \{(v_{S},v)\mid v\in \supp(S)\}$ in which each edge has length $h_{j}$. Exceptionally, when $|S| = 1$, its star graph $H^{\star}_{S}$ degenerate to an isolated vertex, which is the unique vertex $v\in S$. We call $H^{\star} = \bigcup_{1\leq j\leq \bar{j},S\in{\cal N}_{j}} H^{\star}_{S}$ the \textit{star union} of the collection $\{{\cal N}_{j}\mid 1\leq j\leq\bar{j}\}$. At last, the sparsifier $H$ is simply the union of the star union and heavy edges, i.e. $H = H^{\star}\cup E_{\heavy}$.

The quality of $H$ is shown by \Cref{lemma:SparsifierCorrectness}. The number of vertices in $H$ is at most $O(\sum_{j}\size({\cal N}_{j}))$. The number of edges in $H$ is at most $|E(H)|\leq \sum_{j}(O(|C_{j}|/h_{j}) + \size({\cal N}_{j}))\leq O(\sum_{j}(|C_{j}|/h_{j} + \omega_{j}\cdot|\bar{T}|))$. The bound on running time is straightforward.

\begin{lemma}
$H$ is a $(\alpha_{\low},\alpha_{\up},h)$-sparsifer of $\bar{T}$ on $G$ with $\alpha_{\up} = O(1)$.
\label{lemma:SparsifierCorrectness}
\end{lemma}
\begin{proof}
We show the three properties of sparsifers in \Cref{def:Sparsifier} one by one.

\paragraph{Property \ref{property:Sparsifier1}.} From the definition of pairwise cover (i.e. \Cref{def:PairwiseCover}), for each ${\cal N}_{j}$, we have $\bigcup_{S\in{\cal N}_{j}}S = \bar{T}$. By the construction of $H$, we have $\bar{T}\subseteq V(H^{\star})= V(H)$.

\paragraph{Property \ref{property:Sparsifier2}.} Let $u,v$ be an arbitrary pair of vertices in $\bar{T}$, and let $P_{H}$ be the shortest $u$-$v$ path in $H$. Consider an arbitrary pair of consecutive $\bar{T}$-vertices $x,y$ on $P_{H}$. That is, the subpath of $P_{H}$ from $x$ to $y$, denoted by $P_{H}(x,y)$, is internally disjoint with $\bar{T}$. By the construction of $H$, there are two cases.

In the first case, $P_{H}(x,y)$ is made up with a single edge $e=(x,y)\in E_{\heavy}$. This edge has the same length in both $G$ and $H$, so $\dist_{G}(x,y)\leq \ell_{H}(P_{H}(x,y))$. In the second case, $P_{H}(x,y)$ is made up with two edges and a internal vertex $w$, where $w$ is the artificial center of some cluster $S\in {\cal N}_{j}$ for some $j$. We have $x,y\in S$ and $\dist_{G}(x,y)\leq \alpha_{\low}\cdot h_{j}$. Moreover, $\ell_{H}(P_{H}(x,y)) = 2h_{j}$ because these two edges have length $h_{j}$. Therefore, $\dist_{G}(x,y)\leq \alpha_{\low}\cdot \ell_{H}(P_{H}(x,y))$.

Finally, by taking the concatenation of these subpaths $P_{H}(x,y)$, we have
\[
\dist_{G}(u,v)\leq \sum_{(x,y)}\dist_{G}(x,y)\leq \sum_{(x,y)}\alpha_{\low}\cdot \ell_{H}(P_{H}(x,y)) = \alpha_{\low}\cdot \dist_{H}(u,v).
\]

\paragraph{Property \ref{property:Sparsifier3}.} Consider a pair of vertices $u,v\in\bar{T}$ with $\dist_{G}(u,v)\leq h$. We now show that $\dist_{H}(u,v)\leq O(1)\cdot \dist_{G}(u,v)$. Let $P^{\star}$ be the shortest $u$-$v$ path on $G$.

Our strategy is to construct a set ${\cal P}$ of paths $P_{1},P_{2},...,P_{\bar{x}}$ called \textit{roads} s.t. each road $P_{x}\in{\cal P}$ has endpoints $u_{x},v_{x}\in\bar{T}$. These roads should further have the following properties. (1) First, $u_{1} = u$, $v_{\bar{x}} = v$ and $v_{x} = u_{x+1}$ for each $1\leq x\leq \bar{x}-1$. That is, the roads in ${\cal P}$ are laid end to end, where the first road $P_{1}$ starts from $u$ and the last road $P_{\bar{x}}$ ends at $v$. (2) Second, $\sum_{P_{x}\in {\cal P}} \ell_{G}(P_{x})\leq O(1)\cdot \dist_{G}(u,v)$. Namely, the total $G$-length of roads in ${\cal P}$ is at most a constant multiple of the real $\dist_{G}(u,v)$. (3) Each $(u_{x},v_{x})$ has $\dist_{H}(u_{x},v_{x})\leq O(1)\cdot\ell_{G}(P_{x})$. 

These three properties are enough to certify that $\dist_{H}(u,v)\leq O(1)\cdot \dist_{G}(u,v)$, because they imply
\begin{align*}
\dist_{H}(u,v) \leq \sum_{1\leq x\leq \bar{x}}\dist_{H}(u_{x},v_{x})
\leq \sum_{P_{x}\in{\cal P}}O(1)\cdot\ell_{G}(P_{x})
\leq O(1)\cdot \dist_{G}(u,v).
\end{align*}

\textbf{Construction of ${\cal P}$.} We construct ${\cal P}$ in an algorithmic way. We start with an initial set ${\cal P}_{1}$ which only contains a single road $P_{1,1} = P^{\star}$. The $y$-th phase will process a road set ${\cal P}_{y}$. Roughly speaking, each road $P_{y,x}\in {\cal P}_{y}$ may be split\footnote{The split operation \textit{may not} simply choose a vertex  on $P_{y,x}$ as the breakpoint and break $P_{y,x}$ into two subpaths.} into two new roads $P_{y+1,x'},P_{y+1,x''}\in {\cal P}_{y+1}$, or added into ${\cal P}_{y+1}$ with no change (i.e. it becomes $P_{y+1,x'} = P_{y,x}\in {\cal P}_{y+1}$). Intuitively, all the roads generated by this procedure will form a tree structure ${\cal T}$. Naturally, we can define the \textit{ancestor-descendant} relationships between roads. In particular, we call the road $P_{y,x}$ the \textit{parent} of $P_{y+1,x'}$ and $P_{y+1,x''}$, and the road $P_{1,1} = P^{\star}$ is the \textit{root road} of ${\cal T}$.

Now we describe the algorithm in details. In the $y$-th iteration, we do the following for each road $P_{y,x}\in {\cal P}_{y}$. Let $P_{y-1,\hat{x}}$ be its parent if $y\geq 2$. Let $u_{y,x}, v_{y,x}\in \bar{T}$ be the endpoints of $P_{y,x}$. Let $j^{\star}(P_{y,x})$ be the unique level s.t. $10\cdot\ell_{G}(P_{y,x})< h_{j^{\star}(P_{y,x})} \leq 20\cdot\ell_{G}(P_{y,x})$, and let $j(P_{y,x}) = \min\{j^{\star}(P_{y,x}), j(P_{y-1,\hat{x}})\}$ (if $P_{y,x}$ has no parent, then $j(P_{y,x}) = j^{\star}(P_{y,x})$). We will study by cases. If every road $P_{y,x}\in{\cal P}_{y}$ falls in case 1, we stop the algorithm with ${\cal P} = {\cal P}_{y}$ and let $\bar{y} = y$ denote the index of the last iteration. For simplicity, let $P = P_{y,x}$, $j = j(P_{y,x})$, $u = u_{y,x}$ and $v = v_{y,x}$. 

\underline{The Invariant (\Cref{lemma:Invariant}).} Before going into the case study, we first assume this important invariant
\[
5\cdot\ell_{G}(P)\leq h_{j}\leq 20\cdot\ell_{G}(P).
\] 
It is proven in \Cref{lemma:Invariant} after finishing the case study because the proof depends on some notations and claims below. However, we emphasize that there is no circular argument (see the discussion in \Cref{remark:NoCircularArgument}). 

Intuitively, this invariant says, although we enforce that the $j$-value of $P$ cannot be larger than the $j$-value of its parent $\hat{P}$, $j(P)$ will only slightly smaller than the best $j^{\star}(P)$ and it is still fine to work with $j(P)$ in the case study.

\underline{Case 1.} The first case is that $P$ has only one edge or $\ell_{G-C_{j}}(P)\leq h_{j}$. We add $P'=P$ into ${\cal P}_{y+1}$ with $P$ as its parent.

\begin{claim}
In case 1, we have $\dist_{H}(u,v)\leq 40\cdot \ell_{G}(P)$.
\label{claim:Case1}
\end{claim}
\begin{proof}
Suppose $P$ has only one edge $e$. If $C_{j}(e)<h_{j}/10$, then $\dist_{G-C_{j}}(u,v)\leq\ell_{G-C_{j}}(e) \leq \ell_{G}(P)+h_{j}/10\leq h_{j}/5 + h_{j}/10 \leq h_{j}$, which means there is a cluster $S\in{\cal N}_{j}$ containing both $u$ and $v$. Hence $\dist_{H}(u,v)\leq \dist_{H^{\star}_{S}}(u,v)\leq 2h_{j}\leq 40\cdot\ell_{G}(P)$. If $C_{j}(e)\geq h_{j}/10$, then $e\in E_{\heavy}$ and $\dist_{H}(u,v) \leq \ell_{G}(e) = \ell_{G}(P)$.

Suppose $\ell_{G-C_{j}}(P)\leq h_{j}$. Then similarly, $(u,v)$ is covered by a cluster $S\in{\cal N}_{j}$, and $\dist_{H}(u,v)\leq 40\cdot\dist_{G}(u,v)$.
\end{proof}

\underline{Case 2.} If $P$ has at least two edges and there is an edge $e\in P$ s.t. $\ell_{G-C_{j}}(e)> \sigma_{j}$, we let $w$ be an endpoint of $e$ that is not the same with $u$ or $v$ ($w$ must exist since $P$ has at least two edges). Let $P'$ and $P''$ be subpaths of $P$ from $u$ to $w$ and from $w$ to $v$ respectively, and then add $P'$ and $P''$ into ${\cal P}_{y+1}$ with $P$ as their parent.

\begin{claim}
In case 2, we have (\rmnum{1}) $\ell_{G-C_{j}}(P'),\ell_{G-C_{j}}(P'')\leq \ell_{G-C_{j}}(P)$, (\rmnum{2}) $\ell_{G}(P'),\ell_{G}(P'')\leq \ell_{G}(P)$, and (\rmnum{3}) $\ell_{G}(P') + \ell_{G}(P'') = \ell_{G}(P)$. Furthermore, the endpoints of $P'$ and $P''$ are inside $\bar{T}$.
\label{claim:Case2}
\end{claim}
\begin{proof}
The claims on lengths are trivial. To see that endpoints of $P'$ and $P''$ are inside $\bar{T}$, we just need to show $w\in\bar{T}$ because $u$ and $v$ are already known to be in $\bar{T}$. Recall that $w$ is an endpoint of some $e\in P$ with $\ell_{G-C_{j}}(e)>\sigma_{j}$. Because $L_{j}$ is a landmark set of $C_{j}$ on $G$ with distortion $\sigma_{j}$, we have $w\in L_{j}\subseteq \bar{T}$.
\end{proof}

\underline{Case 3.} In this case, $P$ has at least two edges and all edges $e\in P$ have $\ell_{G-C_{j}}(e)\leq \sigma_{j} = h_{j}/\log^{2} n$. Because $\ell_{G}(P)\leq h_{j}/5$ by the invariant, we have $C_{j}(P) \geq 4h_{j}/5$. Therefore, we can find an edge $e\in P\cap\supp(C_{j})$ with an endpoint $w'$ s.t. $\ell_{G-C_{j}}(P(u,w')),\ell_{G-C_{j}}(P(w',v))\leq 2\ell_{G-C_{j}}(P)/3$, where $P(u,w')$ and $P(w',v)$ are subpaths of $P$ from $u$ to $w'$ and from $w'$ to $v$ respectively. Let $w\in L_{j}\subseteq \bar{T}$ be the landmark of $w'$ s.t. $\dist_{G-C_{j}}(w,w')\leq \sigma_{j} = h_{j}/\log^{2}n$. 

Now we construct two roads $P'$ and $P''$ from $u$ to $w$ and from $w$ to $v$ as follows. Let $P_{w}$ be the shortest $w'$-$w$ path on $G-C_{j}$ with $\ell_{G-C_{j}}(P_{w}) = \dist_{G-C_{j}}(w',w)\leq \sigma_{j}$. Then $P'$ is the concatenation of $P(u,w')$ and $P_{w}$, and $P''$ is the concatenation of $P_{w}$ and $P(w',v)$ (and we enforce $P'$ and $P''$ to be simple paths by removing redundant parts). 

\begin{claim}
In case 3, we have (\rmnum{1}) $\ell_{G-C_{j}}(P'),\ell_{G-C_{j}}(P'')\leq (3/4)\cdot \ell_{G-C_{j}}(P)$, (\rmnum{2}) $\ell_{G}(P'),\ell_{G}(P'')\leq (1+O(1/\log^{2} n))\ell_{G}(P)$, and (\rmnum{3}) $\ell_{G}(P') + \ell_{G}(P'') \leq (1+O(1/\log^{2}n))\ell_{G}(P)$.
Furthermore, endpoints of $P'$ and $P''$ are inside $\bar{T}$.
\label{claim:Case3}
\end{claim}
\begin{proof}
For (\rmnum{1}), we have 
\begin{align*}
\ell_{G-C_{j}}(P')&\leq \ell_{G-C_{j}}(P(u,w')) + \ell_{G-C_{j}}(P_{w})\\
&\leq (2/3)\cdot\ell_{G-C_{j}}(P) + 20\cdot \ell_{G}(P)/\log^{2} n\\
&\leq (3/4)\cdot\ell_{G-C_{j}}(P),
\end{align*}
where the second inequality is because $\ell_{G-C_{j}}(P_{w})\leq \sigma_{j}\leq h_{j}/\log^{2}n\leq 20\cdot\ell_{G}(P)/\log^{2}n$ (recall the invariant $h_{j}\leq 20\cdot\ell_{G}(P)$). Similarly, we have $\ell_{G-C_{j}}(P'')\leq (3/4)\cdot\ell_{G-C_{j}}(P)$.

For (\rmnum{2}) and (\rmnum{3}), we have $\ell_{G}(P') + \ell_{G}(P'') \leq \ell_{G}(P) + 2\ell_{G}(P_{w})
\leq \ell_{G}(P) + 2\ell_{G-C_{j}}(P_{w})\leq \ell_{G}(P) + 2\sigma_{j}\leq \ell_{G}(P) + h_{j}/\log^{2} n\leq (1+40/\log^{2} n)\ell_{G}(P)$.
\end{proof}

\begin{lemma}
Let $P_{y}$ denote $P$. For an arbitrary ancestor $P_{\hat{y}}\in {\cal P}_{\hat{y}}$ of $P_{y}$ s.t. $j(P_{\hat{y}}) = j$, the tree-path from $P_{\hat{y}}$ to $P_{y}$ on ${\cal T}$ has at most $O(\log n)$ case-3 roads.
\label{lemma:NumberOfCase3Roads}
\end{lemma}
\begin{proof}
Let ${\cal T}(P_{\hat{y}},P_{y})$ denote the tree-path from $P_{\hat{y}}$ to $P_{y}$. For each 
$\hat{y}\leq y'\leq y$, let $P_{y'}\in{\cal P}_{y'}$ denote the ancestor of $P_{y}$ on level $y'$. We have $j(P_{y'}) = j$ for all $\hat{y}\leq y'\leq y$ because $j(P_{\hat{y}})=j(P_{y})$. Note that the case-1 roads on ${\cal T}(P_{\hat{y}},P_{y})$ are located consecutively as a suffix. Hence we only need to consider the prefix of ${\cal T}(P_{\hat{y}},P_{y})$ including only case-2 and case-3 roads. 

For each road $P_{y'}$ (with $y'\leq y-1$) on this prefix. If $P_{y'}$ is a case-2 road, we have $\ell_{G-C_{j}}(P_{y'+1})\leq \ell_{G-C_{j}}(P_{y'})$ by \Cref{claim:Case2}, otherwise $P_{y'}$ is a case-3 road and $\ell_{G-C_{j}}(P_{y'+1})\leq (3/4)\cdot \ell_{G-C_{j}}(P_{y'})$. Suppose there are $i$ many case-3 roads. Then, the $i$-th case-3 road on ${\cal T}(P_{\hat{y}},P_{y})$, denoted by $\wtilde{P}_{i}$ has $\ell_{G-C_{j}}(\wtilde{P}_{i})\leq (3/4)^{(i-1)}\cdot\ell_{G-C_{j}}(P_{\hat{y}})$. On the other hand, $\ell_{G-C_{j}}(\wtilde{P}_{i})\geq h_{j}\geq 1$ because it is a case-3 road. Therefore, there must be $i\leq O(\log n)$ because $\ell_{G-C_{j}}(P_{y_{1}})$ is polynomially bounded.

\end{proof}

\begin{lemma}
$5\cdot \ell_{G}(P)\leq h_{j}\leq 40\cdot \ell_{G}(P)$.
\label{lemma:Invariant}
\end{lemma}
\begin{proof}
Recall that $j^{\star}(P)$ is defined to be such that $10\cdot \ell_{G}(P)< h_{j^{\star}(P)}\leq 20\cdot\ell_{G}(P)$ and $j\leq j^{\star}(P)$, so trivially $h_{j}\leq h_{j^{\star}(P)}\leq 20\cdot\ell_{G}(P)$.

It remains to show $5\cdot \ell_{G}(P)\leq h_{j}$. Let $\wtilde{P}$ be the farthest ancestor of $P$ s.t. $j(\wtilde{P}) = j$. Then the parent of $\wtilde{P}$, denoted by $\widehat{\wtilde{P}}$, has $j(\widehat{\wtilde{P}}) > j(\wtilde{P})$, so combining $j(\wtilde{P}) = \min\{j(\widehat{\wtilde{P}}), j^{\star}(\wtilde{P})\}$, we have $j^{\star}(\wtilde{P}) = j(\wtilde{P}) = j$ . From the definition of $j^{\star}(\wtilde{P})$, we have $10\cdot\ell_{G}(\wtilde{P})<h_{j}$. We will see in a moment that $\ell_{G}(P)\leq 2\cdot\ell_{G}(\wtilde{P})$, which immediately implies $5\cdot\ell_{G}(P)<h_{j}$.

To show $\ell_{G}(P)\leq 2\cdot\ell_{G}(\wtilde{P})$, consider the tree-path ${\cal T}(\wtilde{P},P)$ from $\wtilde{P}$ to $P$. By \Cref{lemma:NumberOfCase3Roads} this tree-path has at most $O(\log n)$ case-3 roads. We walk from $\wtilde{P}$ to $P$ on this tree-path. By \Cref{claim:Case2}, passing a case-1 or case-2 road will not increase the road length on $G$. By \Cref{claim:Case3}, passing a case-3 road will increase the road length on $G$ by a factor $(1+O(\log^{2} n))$ multiplicatively. Therefore, We can conclude that $\ell_{G}(P)\leq (1+O(1/\log^{2} n))^{O(\log n)}\cdot\ell_{G}(\wtilde{P})\leq 2\cdot\ell_{G}(\wtilde{P})$.
\end{proof}

\begin{remark}
We emphasize that we \textit{did not} make a circular argument by clarifying the dependency between (a) \Cref{lemma:NumberOfCase3Roads,lemma:Invariant}, and (b) \Cref{claim:Case2} and \Cref{claim:Case3}. First, \Cref{claim:Case2} and \Cref{claim:Case3} for the current road $P$ depend on \Cref{lemma:Invariant} for $P$, because the algorithmic step processing $P$ needs \Cref{lemma:Invariant} for $P$. Second, \Cref{lemma:NumberOfCase3Roads} and \Cref{lemma:Invariant} for the current $P$ depend on \Cref{claim:Case2} and \Cref{claim:Case3} for all \textit{ancestor} $\wtilde{P}$ of $P$. Therefore, because ${\cal T}$ is a tree and \Cref{lemma:Invariant} holds for the root road trivially, all these claims and lemmas can be established inductively. 
\label{remark:NoCircularArgument}
\end{remark}

\textbf{Properties of ${\cal P}$.} For property (1), it is straightforward to see that the roads in ${\cal P}$ are laid end to end from $u$ to $v$. For property (3), because all roads in ${\cal P}$ are case-1 (by the terminate condition), each $P_{x}\in{\cal P}$ has $\dist_{H}(u_{x},v_{x})\leq 40\cdot \ell_{G}(P_{x})$ by \Cref{claim:Case1}. 

We now prove property (2), $\sum_{P_{x}\in{\cal P}} \ell_{G}(P_{x})\leq O(1)\cdot \dist_{G}(u,v)$, by induction. We call each road in ${\cal P} = {\cal P}_{\bar{y}}$ a \textit{leaf road}. For each road $P\in {\cal T}$, let $d(P)$ denote the maximum number of case-3 roads on the tree-path from $P$ to an arbitrary descendant leaf road of $P$, and let $g(P) = \sum_{\text{$P$'s descendant leaf road $P'$}}\ell_{G}(P')$ be the sum of $G$-length over all descendant leaf road $P'$ of $P$. Let the induction hypothesis be 
\[
g(P)\leq (1+O(1/\log^{2}n))^{d(P)}\cdot \ell_{G}(P),
\]
which trivially holds for all leaf road $P$. Now consider a non-leaf road $P$ and assume the induction hypothesis holds for $P$'s children. If $P$ is a case-1 road, the hypothesis also holds for $P$ trivially. If $P$ is a case-2 road with children $P'$ and $P''$, we have $d(P) = \max\{d(P'),d(P'')\}$ and $g(P) = g(P') + g(P'')$. Then
\begin{align*}
g(P) &= g(P') + g(P'')\\
&\leq (1+O(1/\log^{2}n))^{d(P')}\cdot \ell_{G}(P') + (1+O(1/\log^{2}n))^{d(P'')}\cdot \ell_{G}(P'')\\
&\leq (1+O(1/\log^{2}n))^{d(P)}\cdot (\ell_{G}(P') + \ell_{G}(P''))\\
&= (1+O(1/\log^{2}n))^{d(P)}\cdot \ell_{G}(P),
\end{align*}
where the last equation is by \Cref{claim:Case2}. Similarly, if $P$ is a case-3 road with children $P'$ and $P''$, we have $d(P) = \max\{d(P'),d(P'')\}+1$, and then
\begin{align*}
g(P) &= g(P') + g(P'')\\
&\leq (1+O(1/\log^{2}n))^{\max\{d(P'),d(P'')\}}\cdot (\ell_{G}(P') + \ell_{G}(P''))\\
&\leq (1+O(1/\log^{2}n))^{\max\{d(P'),d(P'')\}}\cdot (1+O(1/\log^{2}n))\cdot\ell_{G}(P)\\
& = (1+O(1/\log^{2}n))^{d(P)}\cdot \ell_{G}(P),
\end{align*}
where the second inequality is by \Cref{claim:Case3}.

Now consider the root road $P^{\star}$, the induction shows 
\[
g(P^{\star}) = (1+O(1/\log^{2}n))^{d(P^{\star})}\cdot \ell_{G}(P^{\star}).
\]
We will see in a moment that $d(P^{\star}) = O(\log^{2}n)$, which implies $\sum_{P_{x}\in{\cal P}}\ell_{G}(P_{x})=g(P^{\star}) \leq O(1)\cdot \ell_{G}(P^{\star}) = O(1)\cdot \dist_{G}(u,v)$.

To see $d(P^{\star}) = O(\log^{2}n)$, consider an arbitrary path on ${\cal T}$ from the root road to a leaf road $P_{\rm leaf}$, denoted by ${\cal T}(P^{\star},P_{\rm leaf})$. Because the $j$-value is monotonically decreasing, the roads on ${\cal T}(P^{\star},P_{\rm leaf})$ with the same $j$-value will located consecutively. Therefore, for each $1\leq j\leq \bar{j}$, the number of case-3 roads $P\in {\cal T}(P^{\star}, P_{\rm leaf})$ with $j(P) = P$ is at most $O(\log n)$ by \Cref{lemma:NumberOfCase3Roads}. This means the total number of case-3 roads on ${\cal T}(P^{\star},P_{\rm leaf})$ is at most $O(\bar{j}\cdot \log n) = O(\log^{2} n)$ because $h$ is polynomially bounded.

\end{proof}

\end{proof}

\Cref{lemma:DynamicCoverToSparsifier} is a dynamic version of \Cref{thm:NonHopReducingEmulator}, which is straightforward from the static construction of sparsifiers. 

\begin{lemma}
Let $G$ be a dynamic graph with an incremental terminal set $T$ under $t$ batches of updates $\pi^{(1)},...,\pi^{(t)}$ of edge insertions/deletions, isolated vertex insertions/deletions and terminals insertions. Given a parameter $h$, for each level $1\leq j\leq \bar{j}=\lceil\log(100\cdot h)\rceil$, let $h_{j} = 2^{j}$ and suppose we are given a sequence of $(C_{j}^{(i)}, L_{j}^{(i)}, {\cal N}_{j}^{(i)})$ over time period $0\leq i\leq t$ satisfying the following.
\label{lemma:DynamicCoverToSparsifier}
\begin{itemize}
\item The integral moving cut $C_{j}^{(i)} = (C_{j}^{(i-1)})_{\mid E(G^{(i)})} + C_{\new,j}^{(i)}$.
\item $L^{(i)}_{j}$ is a landmark set of $C^{(i)}_{j}$ on $G^{(i)}$ with distortion $\sigma_{j} = h_{j}/\log^{2} n$;
\item ${\cal N}_{j}^{(i)}$ is a pairwise cover of $\bar{T}^{(i)} = T^{(i)}\cup(\bigcup_{j'}L_{j'}^{(i)})$ on $G^{(i)}-C^{(i)}_{j}$ with cover radius $h_{\cov,j} \geq h_{j}$, width $\omega_{j}^{(i)}$. Furthermore, 
${\cal N}_{j}^{(i)}$ has diameter $h_{\Gdiam,j}\leq h_{j}\cdot\alpha_{\low}$ on $G^{(i)}$ for some fixed parameter $\alpha_{\low}$.
\end{itemize}
There is an algorithm that initializes and maintains an $(\alpha_{\low},\alpha_{\up},h)$-sparsifier $H$ of $\bar{T}$ on $G$ with $\alpha_{\up} = O(1)$, $|V(H^{(i)})| = O(\sum_{j}\omega^{(i)}_{j}\cdot|\bar{T}^{(i)}|)$ and size
\begin{align*}
|H^{(i)}|
&= O(\sum_{j}\omega_{j}^{(i)}\cdot |\bar{T}^{(i)}| + t\cdot\sum_{j}|C_{j}^{(i)}|/h_{j}).
\end{align*}
Furthermore, vertices in $V(H)\setminus \bar{T}$ are fresh.

The recourse from $H^{(i-1)}$ to $H^{(i)}$ is 
\[
O(|\pi^{(i)}| + \sum_{j}t\cdot|C^{(i)}_{\new,j}|/h_{j} + \recourse({\cal N}^{(i-1)}_{j}\to {\cal N}^{(i)}_{j})).
\]
The initialization time is $\wtilde{O}(\sum_{j}\omega^{(0)}_{j}\cdot|\bar{T}^{(0)}|+ |C_{j}^{(0)}|)$ and for each batched update $\pi^{(i)}$, the update time is $\wtilde{O}(\sum_{j}|C_{\new,j}^{(i)}| + \recourse({\cal N}_{j}^{(i-1)}\to{\cal N}_{j}^{(i)}))$.

\end{lemma}

\begin{proof}

Providing the static construction of sparsifiers in the proof of \Cref{thm:CoversToEmulators}, we just need it dynamic. Concretely, we will maintain the invariant $H^{(i)} = H^{\star,(i)}\cup E^{(i)}_{\heavy}$ over all time $0\leq i\leq t$. Here $H^{\star,(i)}$ is exactly the same as the star union in the proof of \Cref{thm:CoversToEmulators} given $\{(C^{(i)}_{j},L^{(i)}_{j},{\cal N}^{(i)}_{j})\}$, but $E^{(i)}_{\heavy}$ is a super set of the heavy edge set from \Cref{thm:CoversToEmulators}. We will generate the batched update $\pi^{(i)}_{H}$ to $H^{(i)}$ at time $i$ as follows.

\paragraph{Generating $\pi_{H}^{(i)}$.} For the heavy edge set, initially we let $E^{(0)}_{\heavy}$ collect all edges $e\in E(G^{(0)})$ s.t. $C_{j}^{(0)}(e)\geq h_{j}/(20\cdot t)$ for some $j$. Afterwards, at each time $i\geq 1$, let $E_{G,\delete}^{(i)}$ be the edges removed from $G^{(i-1)}$ to $G^{(i)}$, and let $E_{\heavy,\new}^{(i)}$ collect all edges $e\in G^{(i)}$ s.t. $C_{\new,j}^{(i)}(e)\geq h_{j}/(20\cdot t)$. Then we update $E^{(i-1)}_{\heavy}$ to $E^{(i)}_{\heavy} = E^{(i-1)}_{\heavy}\cup E^{(i)}_{\heavy,\new}\setminus E_{G,\delete}^{(i)}$, so the batched edge insertions $E^{(i)}_{\heavy,\new}$ and batched edge deletions $E^{(i)}_{G,\delete}$ will also apply to $H^{(i-1)}$, and we add them into $\pi_{H}^{(i)}$.

For the star union, we keep $H^{\star,(i)} = \bigcup_{1\leq j\leq \bar{j},S\in {\cal N}_{j}^{(i)}} H^{\star}_{S}$ to be the star union of $\{{\cal N}_{j}^{(i)}\mid 1\leq j\leq \bar{j}\}$ at all times. Note that $H^{\star,(i)}$ can be updated from $H^{\star,(i-1)}$ as follows. Consider the batched update sequence from ${\cal N}_{j}^{(i-1)}$ to ${\cal N}_{j}^{(i)}$ for all $j$. We add the following unit updates to $\pi_{H}^{(i)}$.
\begin{itemize}
\item Consider a vertex insertion which adds $v$ to a cluster $S$. We first apply an isolated vertex insertion $v$ if $v\in V(G^{(i)})\setminus V(G^{(i-1)})$ is a newly isolated vertex. If $|S|\geq 2$ before insertion, we add an edge $(v,v_{S})$ with length $h_{j}$ to $H^{\star,(i-1)}$. If $|S|=1$ (say $S=\{v'\}$) before, we add an isolated vertex $v_{S}$ and then add two edges $(v,v_{S})$ and $(v',v_{S})$ with lengths $h_{j}$. 
\item Similarly, for each virtual node deletions which removes $v$ from a cluster $S$, if $|S|\geq 2$ after the deletion, we delete the edge $(v,v_{S})$. If $|S|=1$ (say $S=\{v'\}$) after the deletion, we delete two edges $(v',v_{S}),(v,v_{S})$ and then the isolated vertex $v_{S}$. Finally, if $v\in V(G^{(i-1)})\setminus V(G^{(i)})$ is a removed isolated vertex, we apply an isolated vertex deletion $v$.
\end{itemize}

\paragraph{Correctness.} Note that at each moment $0\leq i\leq t$, the $H^{(i)}$ constructed here is a super graph of the sparsifier $H$ constructed in the proof of \Cref{thm:CoversToEmulators} given $\{(C_{j}^{(i)},L_{j}^{(i)},{\cal N}_{j}^{(i)})\mid 1\leq j\leq \bar{j}\}$, because the star union is exactly the same, and the heavy edge set from \Cref{thm:CoversToEmulators} is a subset of $E^{(i)}_{\heavy}$. The latter holds by the following reason. For each edge $e\in E(G^{(i)})$ with $C_{j}^{(i)}(e)\geq h_{j}/10$, let $\hat{i}$ denote the first time it is added into $G$. Then $C_{j}^{(i)}(e) = \sum_{\hat{i}\leq i'\leq i}C_{\new,j}^{(i)}(e)$ (we let $C^{(0)}_{\new,j}(e) = C^{(0)}_{j}(e)$ for simplicity), and by averaging, $C_{\new,j}^{(i')}(e)\geq h_{j}/(t\cdot 20)$ for some moment $\hat{i}\leq i'\leq i$, so $e$ was added to $E_{\heavy}^{(i')}$ at that time and remains in $E_{\heavy}^{(i)}$ because $e$ did not get deleted.

Therefore, we can conclude that $H^{(i)}$ is an $(\alpha_{\low},\alpha_{\up},h)$-sparsifier of $\bar{T}^{(i)}$ on $G^{(i)}$ with $\alpha_{\up} = O(1)$. 
The number of vertices in $H^{(i)}$ is at most $|V(H^{(i)})| \leq \sum_{j}\size({\cal N}^{(i)}_{j})\leq O(\sum_{j}\omega^{(i)}_{j}\cdot|\bar{T}^{(i)}|)$. The number of edges in $H^{(i)}$ is bounded by 
\begin{align*}
|E(H^{(i)})| &= |E(H^{\star,(i)})| + |E_{\heavy}^{(i)}|\\
&\leq \sum_{j}\size({\cal N}_{j}^{(i)}) + \sum_{j}|C^{(i)}_{j}|/(h_{j}/(20\cdot t))\\
&\leq 
O(\sum_{j}\omega_{j}^{(i)}\cdot |\bar{T}^{(i)}| + t\cdot\sum_{j}|C_{j}^{(i)}|/h_{j}).
\end{align*}
Note that $|E^{(i)}_{\heavy}|\leq \sum_{j}|C^{(i)}_{j}|/(h_{j}/(20\cdot t))$ because an edge $e$ with $C^{(i)}_{j}<h_{j}/(20\cdot t)$ will not be select in $E^{(i)}_{\heavy}$. 

Vertices in $V(H^{(i)})\setminus \bar{T}^{(i)}$ are fresh because they are the artificial centers of stars, and we can definitely assign totally new global identifers to them.

The recourse from $H^{(i-1)}$ to $H^{(i)}$ is bounded by the number of unit updates in $\pi_{H}^{(i)}$. By our construction of $\pi_{H}^{(i)}$, we have
\begin{align*}
&~~~~|E^{(i)}_{\heavy,\new}| + |E^{(i)}_{G,\delete}| + \sum_{j}O(\recourse({\cal N}_{j}^{(i-1)}\to {\cal N}_{j}^{(i)}))\\
&= O(|\pi^{(i)}| + \sum_{j}(t\cdot|C^{(i)}_{\new,j}|/h_{j} + \recourse({\cal N}^{(i-1)}_{j}\to {\cal N}^{(i)}_{j}))).
\end{align*}
The running time is straightforward from the algorithm description.

\end{proof}

\subsection{Proof of \Cref{thm:NonHopReducingEmulator}}
\label{sect:ProofOfDynSparsifier}

We will prove \Cref{thm:NonHopReducingEmulator} combining \Cref{lemma:DynamicCoverToSparsifier} and an online-batch dynamic algorithm for the collection of specialized certified-EDs shown below. The latter is quite similar to the algorithm in \Cref{thm:DynamicED}, except the following. First, we work on $\rho$-dense certified-EDs instead to reduce the recourse of pairwise covers (so the recourse of the sparsifier). Second, recall that the specialized certified-EDs should be with respect to the union of original terminals and all landmarks. Hence, after a usual update step (or the usual initialization step), we will invoke \Cref{lemma:LandmarkClosure} to add the newly generated landmarks to the node-weighting.

Again, in the following proof, we assume without loss of generality each batched update contains only one type of unit updates.

\begin{proof}

We have argued in \Cref{lemma:DynamicCoverToSparsifier} that to maintain a sparsifier, it suffices to maintain $\bar{j} = \lceil 100\cdot h\rceil$ levels of pairwise covers. To do this, we will maintain certifed-EDs instead. At the end of each update $i$ (and the initialization with $i=0$) and for each level $1\leq j\leq \bar{j}$, the up-to-date dense certified-ED is a $\bar{\rho}^{(i)}$-dense certified-ED $(\bar{C}_{j}^{(i)}, \bar{L}_{j}^{(i)}, \bar{{\cal N}}_{j}^{(i)}, \bar{{\cal R}}_{j}^{(i)}, \Pi_{\bar{{\cal R}}_{j}^{(i)}\to G^{(i)}})$ of $\bar{A}^{(i)}$ on $G^{(i)}$ where 
\begin{equation*}
\bar{A}^{(i)} = \deg_{G^{(i)}} + \mathds{1}(V(G^{(i)})) + \mathds{1}(T^{(i)}) + \mathds{1}(\bigcup_{j'} \bar{L}^{(i)}_{j'})
\end{equation*}
with density function
\begin{equation*}
\bar{\rho}^{(i)} \geq \left\{
\begin{aligned}
&1,\ v\in \deg_{G^{(i)}} + \mathds{1}(V(G^{(i)}))\\
&\frac{(1-1/\mu)^{i}}{\phi},\ v\in \mathds{1}(T^{(i)}) + \mathds{1}(\bigcup_{j'} \bar{L}^{(i)}_{j'})
\end{aligned}
\right.,
\end{equation*}
where we fix $\mu = 2t$ and it is the parameter fed to \Cref{thm:EdgeDelWithDensity}. We let $\bar{h}^{(i)}_{j} = \bar{h}^{(i)}_{\cov,j}$ be the cover radius of $\bar{\cal N}^{(i)}_{j}$ (namely, the length parameter of this dense certified-ED).

The framework of maintaining certified-EDs is as follows. The whole initialization step will invoke \Cref{thm:InitDenseCertifiedED} (paragraph \textbf{Initialization}) and then \Cref{lemma:LandmarkClosure} (paragraph \textbf{Insert New Landmarks to the Node-Weighting}). For each update step for $\pi^{(i)}$, if it represents isolated vertex insertions or deletions, we go to paragraph \textbf{Isolated Vertex Insertions/Deletions}. Otherwise, we will invoke one of \Cref{thm:EdgeDelWithDensity}, \Cref{thm:NWInsertWithDensity} and \Cref{thm:EdgeInsertWithDensity} (paragraphs \textbf{Edge Deletion}, \textbf{Edge Insertion} and \textbf{Terminals Insertion} respectively) depends on the type of $\pi^{(i)}$, and then again invoke \Cref{lemma:LandmarkClosure} (paragraph \textbf{Insert New Landmarks to the Node-Weighting}).

\paragraph{Initialization.}

Let $A^{(0)}$ be a node-weighting defined by
\begin{equation*}
A^{(0)}=\deg_{G^{(0)}}+\mathds{1}(T^{(0)})\text{ with density function }\rho^{(0)} = \left\{
\begin{aligned}
&1,\ v\in \deg_{G^{(0)}}\\
&1/\phi,\ v\in \mathds{1}(T^{(0)})
\end{aligned}
\right..
\end{equation*}

For each level $1\leq j\leq \bar{j}$, we let \[
h_{j} = 2^{j},\ 
h_{j}^{(0)}= h_{j}\cdot 2^{O(t\cdot\lambda_{\insLM,t})}.
\]
Then we apply \Cref{thm:InitDenseCertifiedED} on the initial graph $G^{(0)}$ with the $\rho^{(0)}$-dense node-weighting $A^{(0)}$, length parameter $h^{(0)}_{j}$ and congestion parameter $\phi^{(0)}_{j} = \phi/\kappa_{\init,L}$.

The output is a $\rho^{(0)}$-dense certified-ED $(C^{(0)}_{j},L^{(0)}_{j}, {\cal N}^{(0)}_{j},{\cal R}^{(0)}_{j},\Pi_{{\cal R}^{(0)}_{j}\to G^{(0)}})$ of $A^{(0)}$ on $G^{(0)}$ satisfying that 
\begin{itemize}
\item The initial moving cut $C^{(0)}_{j}$ has size
\begin{align*}
|C^{(0)}_{j}|&\leq O(\kappa_{\init,C}\cdot \phi^{(0)}_{j}\cdot h^{(0)}_{j}\cdot \rho^{(0)}(A^{(0)}))\leq O(\phi\cdot h^{(0)}_{j}\cdot \rho^{(0)}(A^{(0)}))\leq O(h^{(0)}_{j}\cdot (\phi\cdot|E(G^{(0)})| + |T^{(0)}|)),
\end{align*}
where the second inequality is by $\kappa_{\init,L}\geq \kappa_{\init,C}$.
\item $L^{(0)}_{j}$ has distortion $\sigma^{(0)}_{j} = O(h_{j}^{(0)}/\kappa_{\sigma})$ and size 
\begin{align*}
|L^{(0)}_{j}| &\leq O(\kappa_{\init,L}\cdot \phi^{(0)}_{j}\cdot\rho^{(0)}(A^{(0)}))\leq O(\phi\cdot\rho^{(0)}(A^{(0)}))\leq O(\phi\cdot |E(G^{(0)})| + |T^{(0)}|).
\end{align*}

\item ${\cal N}^{(0)}_{j}$ is a $\rho^{(0)}$-dense $(h^{(0)}_{\cov,j},h^{(0)}_{\sep,j},\omega^{(0)}_{j})$-pairwise cover of $A^{(0)}$ on $G^{(0)}$ with $h^{(0)}_{\cov,j}=h^{(0)}_{\sep,j} = h_{j}^{(0)}$ and $\omega_{j}^{(0)} = \kappa_{\PC,\omega}$.
\item Routers in ${\cal R}_{j}$ are initialized by \Cref{thm:Router} and suffer $f^{(0)} = 0$ batched update.
\item $\Pi_{{\cal R}^{(0)}_{j}\to G^{(0)}}$ has length and congestion
\begin{align*}
h^{(0)}_{\emb,j} &= O(\lambda_{\init,h}\cdot h_{j}^{(0)})\\
\gamma^{(0)}_{\emb,j} &= O(\kappa_{\init,\gamma}/\phi^{(0)}_{j}) = O(\kappa_{\init,L}\cdot \kappa_{\init,\gamma}/\phi)
\end{align*}
\end{itemize}

\paragraph{Edge Deletion.} 

Suppose the update $i$ is a batched edge deletion $\pi^{(i)} = F^{(i)}\subseteq E(G^{(i-1)})$ and the graph becomes $G^{(i)} = G^{(i-1)}\setminus F^{(i)}$.

For each level $j$ and its certified-ED $(\bar{C}_{j}^{(i-1)}, \bar{L}_{j}^{(i-1)}, \bar{{\cal N}}_{j}^{(i-1)}, \bar{{\cal R}}_{j}^{(i-1)}, \Pi_{\bar{{\cal R}}_{j}^{(i-1)}\to G^{(i-1)}})$ of $\bar{A}^{(i-1)}$ on $G^{(i-1)}$, we apply \Cref{thm:EdgeDelWithDensity} on it with edge deletions $F_{i}$ and parameters $\mu = 2t$ and
\[
\phi^{(i)}_{j} = t\cdot (\bar{\omega}^{(i-1)}_{j}\cdot\kappa_{\CM,\gamma} + \kappa_{\init,\gamma})/\bar{\gamma}_{\emb,j}^{(i-1)}
\]
The output is a $\rho^{(i)}$-dense certified-ED $(C_{j}^{(i)},L_{j}^{(i)},{\cal N}_{j}^{(i)}, {\cal R}_{j}^{(i)},\Pi_{{\cal R}_{j}^{(i)}\to G^{(i)}})$ of $A^{(i)}$ satisfying the following. To avoid clutter, we define a parameter $\chi_{j}^{(i)}$ for each $0\leq i\leq t,1\leq j\leq \bar{j}$, which will only be used in this proof.
\[
\chi_{j}^{(i)} = \bar{\omega}^{(i-1)}_{j}\cdot(\kappa_{\CM,L} + \kappa_{\init,L})\cdot\mu\cdot \lambda_{\rt,\prune}(\bar{f}^{(i-1)})\cdot t\cdot(\bar{\omega}^{(i-1)}\cdot\kappa_{\CM,\gamma} + \kappa_{\init,\gamma}).
\]

\begin{itemize}

\item $A^{(i)} = \bar{A}^{(i-1)} - \deg_{F^{(i)}}$ with density $\rho^{(i)} = \lceil(1-1/\lambda)\cdot\bar{\rho}^{(i-1)}\rceil$, where the virtual nodes $\deg_{F^{(i)}}$ we remove are from $\deg_{G^{(i-1)}}\subseteq \bar{A}^{(i-1)}$ (each of them has density $1$). Strictly speaking, the output of \Cref{thm:EdgeDelWithDensity} is a dense certified-ED of $\bar{A}^{(i-1)}$ instead of $A^{(i)}$, but this can be easily fixed by removing virtual nodes $\deg_{F^{(i)}}$ from ${\cal N}^{(i)}_{j}$, which will not change the quality.

\item $C_{j}^{(i)} = (\bar{C}_{j}^{(i-1)})_{\mid E(G^{(i)})}+ C_{\new,j}^{(i)}$ with 
\begin{align*}
|C^{(i)}_{\new,j}|&\leq O(\bar{\omega}_{j}^{(i-1)} \cdot\kappa_{\init,C} \cdot \phi^{(i)}_{j}\cdot \mu \cdot \lambda_{\rt,\prune}(\bar{f}^{(i-1)}_{j})\cdot\bar{\gamma}^{(i-1)}_{\emb,j}\cdot|\pi^{(i)}|\cdot \bar{h}^{(i-1)}_{j})\leq O(\chi_{j}^{(i-1)}\cdot |\pi^{(i)}|\cdot\bar{h}_{j}^{(i-1)}),
\end{align*}
by the definition of $\phi^{(i)}_{j}$ and $\kappa_{\init,C}\leq \kappa_{\CM,L} + \kappa_{\init,L}$.

\item $L_{j}^{(i)} = \bar{L}_{j}^{(i-1)}\cup L_{\new,j}^{(i)}$ with distortion $\sigma_{j}^{(i)} = \bar{\sigma}_{j}^{(i-1)} + O(\bar{h}_{j}^{(i-1)}/\kappa_{\sigma})$ and 
\begin{align*}
|L^{(i)}_{\new,j}|&\leq O(\bar{\omega}^{(i-1)}_{j}\cdot(\kappa_{\CM,L} + \kappa_{\init,L})\cdot\phi^{(i)}_{j}\cdot\mu\cdot \lambda_{\rt,\prune}(\bar{f}^{(i-1)}_{j})\cdot\bar{\gamma}_{\emb,j}^{(i-1)}\cdot|\pi^{(i)}|)\leq O(\chi_{j}^{(i-1)}\cdot|\pi^{(i)}|).
\end{align*}
\item The pairwise cover ${\cal N}_{j}^{(i)}$ of $A^{(i)}$ on $G^{(i)}$ has
\[
h^{(i)}_{\cov,j} = h^{(i)}_{\sep,j} = \bar{h}_{j}^{(i-1)}/3,\text{ and }\omega^{(i)}_{j} = \bar{\omega}^{(i-1)}_{j} + \kappa_{\PC,\omega}.
\]
\item Routers in ${\cal R}_{j}^{(i)}$ are maintained by \Cref{thm:Router} under at most $f^{(i)}_{j} = \bar{f}^{(i-1)}_{j} + O(1)$ updates.
\item $\Pi_{{\cal R}^{(i)}_{j}\to G^{(i)}}$ has
\begin{align*}
h^{(i)}_{\emb,j} &= \max\{\bar{h}^{(i-1)}_{\emb,j},O(\lambda_{\init,h}\cdot \bar{h}^{(i-1)}_{j})\}\\
\gamma^{(i)}_{\emb,j} &= \bar{\gamma}^{(i-1)}_{\emb,j} + O((\bar{\omega}_{j}^{(i)}\cdot\kappa_{\CM,\gamma} + \kappa_{\init,\gamma})/\phi^{(i)}_{j}) = (1+O(1/t))\cdot\bar{\gamma}^{(i-1)}_{\emb,j}.
\end{align*}
\end{itemize}

\paragraph{Terminals Insertion.}

Suppose the update $\pi^{(i)}$ is a batched-terminal insertion $T^{(i)}_{\new}$. Let $A^{(i)}_{\new} = \mathds{1}(T^{(i)}_{\new})$ with $1/\phi$-uniform density $\rho^{(i)}_{\new}$.

For each level $j$, we apply \Cref{thm:NWInsertWithDensity} on the certified-ED $(\bar{C}^{(i-1)}_{j}, \bar{L}^{(i-1)}_{j}, \bar{\cal N}^{(i-1)}_{j}, \bar{\cal R}^{(i-1)}_{j}, \Pi_{\bar{\cal R}^{(i-1)}_{j}\to G^{(i-1)}})$ of $\bar{A}^{(i-1)}_{j}$ on $G^{(i-1)}$ with the $\rho^{(i)}_{\new}$-dense node weighting $A^{(i)}_{\new}$ and parameter
\[
\phi^{(i)}_{j} = \phi/(\bar{\omega}^{(i-1)}_{j}\cdot( \kappa_{\CM,L} + \kappa_{\init,L})).
\]
The output is a $\rho^{(i)}$-dense certified-ED $(C^{(i)}_{j}, L^{(i)}_{j}, {\cal N}^{(i)}_{j}, {\cal R}^{(i)}_{j}, \Pi_{{\cal R}^{(i)}_{j}\to G^{(i)}})$ of $A^{(i)}$ on $G^{(i)}$ satisfying the following.

\begin{itemize}

\item $A^{(i)} = \bar{A}^{(i-1)} + A^{(i)}_{\new}$ with density $\rho^{(i)} = \bar{\rho}^{(i-1)}\uplus \rho^{(i)}_{\new}$.

\item $C^{(i)}_{j} = \bar{C}^{(i-1)}_{j} + C^{(i)}_{\new, j}$ with
\begin{align*}
|C^{(i)}_{\new, j}| &\leq O(\kappa_{\init,C}\cdot\bar{\omega}^{(i-1)}_{j}\cdot \phi^{(i)}_{j}\cdot \bar{h}^{(i-1)}_{j}\cdot\rho_{\new}(A^{(i)}_{\new}))\leq O(|\pi^{(i)}|\cdot \bar{h}^{(i-1)}_{j}),
\end{align*}
where the second inequality is because $\rho_{\new}(A^{(i)}_{\new}) = |A^{(i)}_{\new}|/\phi$ and $\kappa_{\init,C}\leq \kappa_{\CM,L} + \kappa_{\init,L}$.

\item $L^{(i)}_{j} = \bar{L}^{(i-1)}_{j} \cup L^{(i)}_{\new, j}$ with distortion $\sigma^{(i)}_{j} = \bar{\sigma}^{(i-1)}_{j} + O(\bar{h}^{(i-1)}_{j}/\kappa_{\sigma})$.
\begin{align*}
|L^{(i)}_{\new, j}|&\leq O(\bar{\omega}^{(i-1)}_{j}\cdot(\kappa_{\CM,L} + \kappa_{\init,L})\cdot \phi_{j}^{(i)}\cdot\rho_{\new}(A_{\new}))\leq O(|\pi^{(i)}|).
\end{align*}

\item The pairwise cover ${\cal N}_{j}^{(i)}$ of $A^{(i)}$ on $G^{(i)}$ has
\[
h^{(i)}_{\cov,j} = h^{(i)}_{\sep,j} = \bar{h}_{j}^{(i-1)}/3,\text{ and }\omega^{(i)}_{j} = \bar{\omega}^{(i-1)}_{j} + \kappa_{\PC,\omega}.
\]
\item Routers in ${\cal R}_{j}^{(i)}$ are maintained by \Cref{thm:Router} under at most $f^{(i)}= \bar{f}^{(i-1)} + O(1)$ updates.

\item $\Pi_{{\cal R}^{(i)}_{j}\to G^{(i)}}$ has length and congestion
\begin{align*}
h^{(i)}_{\emb,j} &= \max\{\bar{h}^{(i-1)}_{\emb,j},O(\lambda_{\init,h}\cdot \bar{h}^{(i-1)}_{j})\}\\
\gamma^{(i)}_{\emb,j} &= \bar{\gamma}^{(i-1)}_{\emb,j} + O((\bar{\omega}^{(i-1)}_{j}\cdot\kappa_{\CM,\gamma} + \kappa_{\init,\gamma})/\phi^{(i)}_{j})\\
&\leq \bar{\gamma}^{(i-1)}_{\emb,j} + O(\bar{\omega}^{(i-1)}_{j}\cdot\kappa_{\CM,\gamma} + \kappa_{\init,\gamma})\cdot \bar{\omega}^{(i-1)}_{j}\cdot(\kappa_{\CM,L} + \kappa_{\init,L})/\phi.
\end{align*}
\end{itemize}

\paragraph{Edge Insertion.}

If the update $\pi^{(i)}$ is a batched edge insertion $E^{(i)}_{\new}$, let $G^{(i)} = G^{(i-1)}\cup E^{(i)}_{\new}$. Similarly, we apply \Cref{thm:EdgeInsertWithDensity} on $(\bar{C}^{(i-1)}_{j}, \bar{L}^{(i-1)}_{j}, \bar{\cal N}^{(i-1)}_{j}, \bar{\cal R}^{(i-1)}_{j}, \Pi_{\bar{\cal R}^{(i-1)}_{j}\to G^{(i-1)}})$ of $\bar{A}^{(i-1)}_{j}$ on $G^{(i-1)}$ with new edges $E^{(i)}_{\new}$ and parameter
\[
\phi^{(i)}_{j} = \phi/(\bar{\omega}^{(i-1)}_{j}\cdot\kappa_{\CM,L} + \kappa_{\init,L}).
\]
The output is a $\rho^{(i)}$-dense certified-ED $(C^{(i)}_{j}, L^{(i)}_{j}, {\cal N}^{(i)}_{j}, {\cal R}^{(i)}_{j}, \Pi_{{\cal R}^{(i)}_{j}\to G^{(i)}})$ of $A^{(i)}$ on $G^{(i)}$ satisfying that $A^{(i)} = \bar{A}^{(i-1)} + A_{\new}^{(i)}$ and $\rho^{(i)} = \bar{\rho}^{(i-1)}\uplus \rho^{(i)}_{\new}$ where $A_{\new}^{(i)} = \deg_{E^{(i)}_{\new}}$ and $\rho^{(i)}_{\new} = \mathds{1}(A_{\new}^{(i)})$. All quality parameters have the same asymtotic bounds as those in the terminal insertion case.

\paragraph{Insert New Landmarks to the Node-Weighting.}

In each iteration $i$ (including the initialization with $i=0$), after computing $(C^{(i)}_{j}, L^{(i)}_{j}, {\cal N}^{(i)}_{j}, {\cal R}^{(i)}_{j}, \Pi_{{\cal R}^{(i)}_{j}\to G^{(i)}})$, we still need \Cref{lemma:LandmarkClosure} to add the new landmarks into the node-weighting. Precisely, let the new landmarks for each $j$ be $L^{\star,(i)}_{j} = L_{\new,j}^{(i)}\setminus \bigcup_{j'}\bar{L}_{j'}^{(i-1)}$ (if $i=0$, $L^{\star,(0)}_{j} = L^{(0)}_{j}$). Then we apply \Cref{lemma:LandmarkClosure} on the collection of certified-EDs $\{(C^{(i)}_{j}, L^{(i)}_{j}, {\cal N}^{(i)}_{j}, {\cal R}^{(i)}_{j}, \Pi_{{\cal R}^{(i)}_{j}\to G^{(i)}})\mid 1\leq j\leq \bar{j}\}$ of $A^{(i)}$ on $G^{(i)}$ with parameter $\phi$ and new landmarks $L^{\star,(i)}_{j}\subseteq L^{(i)}_{j}$ specified for each $j$.

The output is a collection of $\bar{\rho}^{(i)}$-dense certified-EDs $\{(\bar{C}_{j}^{(i)}, \bar{L}_{j}^{(i)}, \bar{\cal N}_{j}^{(i)}, \bar{\cal R}_{j}^{(i)},\Pi_{\bar{\cal R}_{j}^{(i)}\to G^{(i)}})\mid 1\leq j\leq \bar{j}\}$ of $\bar{A}^{(i)}$ on $G$ satisfying the following. Let $h^{(i)}_{j} = h^{(i)}_{\cov,j}$ be the cover radius of ${\cal N}_{j}^{(i)}$ for each $j$. 
Let $\tau^{(i)} = \sum_{j}|L^{\star,(i)}_{j}|$.

\begin{itemize}
\item $\bar{C}_{j}^{(i)} = C_{j}^{(i)} + \bar{C}_{\new,j'}^{(i)}$ with $|\bar{C}_{\new,j}^{(i)}| = O(h_{j}^{(i)} \cdot \tau^{(i)})$ 
\item $\bar{L}_{j}^{(i)} = L_{j}^{(i)}\cup \bar{L}_{\new,j}^{(i)}$ with distortion $\bar{\sigma}_{j}^{(i)} = \sigma_{j}^{(i)} + O(h_{j}^{(i)}/\kappa_{\sigma})$ and $|\bar{L}_{\new,j}^{(i)}| = O(\tau^{(i)})$.
\item $\bar{A}^{(i)} = A^{(i)} + \bar{A}_{\new}^{(i)}$ and $\bar{\rho}^{(i)} = \rho^{(i)}\uplus \bar{\rho}^{(i)}_{\new}$, where $\bar{A}_{\new}^{(i)} = \mathds{1}(\bigcup_{j}L^{\star,(i)}_{j}) + \mathds{1}(\bigcup_{j}\bar{L}^{(i)}_{\new,j}\setminus \bigcup_{j} L^{(i)}_{j})$ and $\bar{\rho}^{(i)}_{\new} = (1/\phi)\cdot \mathds{1}(\bar{A}^{(i)}_{\new})$.

\item $\bar{\cal N}^{(i)}_{j}$ is a $\bar{\rho}^{(i)}$-dense $(\bar{h}_{\cov,j}^{(i)},\bar{h}_{\sep,j}^{(i)},\bar{\omega}_{j}^{(i)})$-pairwise cover of $\bar{A}^{(i)}$ on $G^{(i)}$ with
\[
\bar{h}_{\cov,j}^{(i)} = \bar{h}_{\sep,j}^{(i)} = h_{j}^{(i)}/3^{\lambda_{\insLM,t}},\text{ and }\bar{\omega}_{j}^{(i)} = \omega_{j}^{(i)} + O(\lambda_{\insLM,t}\cdot\kappa_{\PC,\omega}).
\]
\item Routers in $\bar{\cal R}_{j}^{(i)}$ is maintained by \Cref{thm:Router} under at most $\bar{f}^{(i)} = f^{(i)} + O(\lambda_{\insLM,t})$ update batches.
\item $\Pi_{\bar{\cal R}_{j}^{(i)}\to G^{(i)}}$ has length $\bar{h}_{\emb,j}^{(i)} = \max\{h_{\emb,j}^{(i)}, O(\lambda_{\init,h}\cdot h_{j}^{(i)})\}$ and congestion $\bar{\gamma}_{\emb,j}^{(i)} = \gamma_{\emb,j}^{(i)} + \Gamma_{j}^{(i)}/\phi$
where we define
\begin{align*}
\Gamma^{(i)}_{j} = O( \bar{\omega}^{(i)}_{j}\cdot(\kappa_{\CM,L} + \kappa_{\init,L})\cdot(\bar{\omega}^{(i)}_{j}\cdot\kappa_{\CM,\gamma} + \kappa_{\init,\gamma})\cdot\kappa_{\insLM,\gamma}\cdot\lambda_{\insLM,t})
\end{align*}
as a parameter only used in this proof.
\end{itemize}

\paragraph{Isolated Vertex Insertions/Deletions.} Similar to the proof of \Cref{thm:DynamicED}, it is easy to handle isolated vertex insertions and deletions.

Regarding isolated vertex insertions $V_{\new}$, for each $v\in V_{\new}$, we assign it node-weighting $\bar{A}^{(i)}(v)=1$ and the density of the unique virtual node $v_{\virtual}\in \bar{A}^{(i)}(v)$ is $\bar{\rho}^{(i)}(v_{\virtual}) = 1$. For all level $j$, we create a singleton cluster $S = \{v_{\virtual}\}$ and add $S$ into an arbitrary clustering of $\bar{\cal N}^{(i-1)}_{j}$ (also initialize the corresponding router $R$ of $S$, and $R$ can be embedded into $G^{(i)}$ trivially).

Regarding isolated vertex deletions $V_{\del}\subseteq V(G^{(i-1)})$, we know each $v\in V_{\del}$ has $\bar{A}^{(i-1)}(v)\leq 3$, because $\deg_{G^{(i-1)}}(v) = 0$. Therefore, we just need to remove up to $3$ virtual nodes $\bar{A}^{(i-1)}(v)$ from $\bar{\cal N}^{(i-1)}_{j}$.

Obviously these two types of batched updates will not affect the quality of the certified-EDs. Trivially, the total recourse is $\sum_{j}\recourse(\bar{\cal N}^{(i-1)}_{j}\to \bar{\cal N}^{(i)}_{j})\leq O(\bar{\omega}^{(i)}_{j}\cdot\bar{j}\cdot |\pi^{(i)}|)$, and the update time is $O(\bar{\omega}^{(i)}_{j}\cdot\bar{j}\cdot |\pi^{(i)}|) = |\pi^{(i)}|\cdot n^{O(\epsilon)}$. As we will see, the bounds on the recourse and update time are asymptotically smaller than those of other types of batched updates. Therefore, to avoid clutter, we will ignore these two types in the following proof. 

\paragraph{Quality of Certified-EDs.} At the end of the initialization step $i=0$ and each update step $1\leq i\leq t$, for each level $1\leq j\leq \bar{j}$, the certified-ED $(\bar{C}_{j}^{(i)}, \bar{L}_{j}^{(i)}, \bar{\cal N}_{j}^{(i)}, \bar{\cal R}^{(i)}_{j},\Pi_{\bar{\cal R}_{j}^{(i)}\to G^{(i)}})$ of $\bar{A}^{(i)}$ on $G^{(i)}$ satisfies the following.

\underline{The moving cut $\bar{C}^{(i)}_{j}$.} Initially, $\bar{C}_{j}^{(0)}$ has size 
\begin{align*}
|\bar{C}^{(0)}_{j}| &= |C^{(0)}_{j}| + |\bar{C}^{(0)}_{\new,j}|\leq |C^{(0)}_{j}| + O(h_{j}^{(0)}\cdot\sum_{1\leq j'\leq \bar{j}}|L^{(0)}_{j'}|)\leq \kappa_{j}\cdot 2^{O(t\cdot\lambda_{\insLM,t})}\cdot h_{j}\cdot (\phi\cdot |E(G^{(0)})| + |T^{(0)}|).
\end{align*}

For each time $i\geq 1$, $\bar{C}^{(i)}_{j} = \bar{C}^{(i-1)}_{j} + C^{(i)}_{\new,j} + \bar{C}^{(i)}_{\new, j}$. Before deriving the size of the new moving cut, we first give a \emph{global} upper bound on $\chi_{j'}^{(i)}$ over all $0\leq i\leq t$ and $1\leq j'\leq \bar{j}$ (recall that $\chi_{j'}^{(i)}$ is defined in the paragraph \textbf{Edge Deletion}), which is
\[
\chi = O(\lambda_{\rt,\prune}(t\cdot\lambda_{\insLM,t})\cdot t^{4}\cdot\lambda_{\insLM,t}^{2}\cdot\kappa_{\PC,\omega}^{2}\cdot(\kappa_{\CM,L} +\kappa_{\init,L})\cdot(\kappa_{\CM,\gamma} + \kappa_{\init,\gamma})),
\]
and then upper bound $\tau^{(i)}$ for all $1\leq i\leq t$. Concretely, 
\[
\tau^{(i)} = \sum_{1\leq j'\leq\bar{j}}|L^{\star,(i)}_{j'}| \leq \sum_{1\leq j'\leq \bar{j}}|L^{(i)}_{\new,j'}|\leq \sum_{1\leq j'\leq\bar{j}}O(\chi_{j}^{(i-1)}\cdot|\pi^{(i)}|)\leq O(\kappa_{j}\cdot\chi\cdot |\pi^{(i)}|),
\]
because the bottleneck is the batched edge deletion update. Providing the bound on $\tau^{(i)}$, the size of new moving cut is
\begin{align*}
|C^{(i)}_{\new,j}| + |\bar{C}^{(i)}_{\new, j}|&\leq |C^{(i)}_{\new,j}| + O(h^{(i)}_{j}\cdot \tau^{(i)})\leq O(\chi_{j}^{(i-1)}\cdot |\pi^{(i)}|\cdot\bar{h}_{j}^{(i-1)}) + O(\kappa_{j}\cdot\chi\cdot|\pi^{(i)}|\cdot h^{(i)}_{j})\\
&\leq \kappa_{j}\cdot\chi\cdot|\pi^{(i)}|\cdot h_{j}\cdot 2^{O(t\cdot\lambda_{\insLM,t})},
\end{align*}
because $\bar{h}_{j}^{(i)}$ and $h_{j}^{(i)}$ have an upper bound $h_{j}^{(0)} = h_{j}\cdot 2^{O(t\cdot\lambda_{\insLM,t})}$ over all $i$.

\

\underline{The Landmark Set $\bar{L}^{(i)}_{j}$.} Similar to the calculation of the moving cut, the inital landmark set has size $|\bar{L}_{j}^{(0)}| = O(\kappa_{j}\cdot(\phi\cdot|E(G^{(0)})| + |T^{(0)}|))$. For each update step $i\geq 1$, $\bar{L}^{(i)}_{j} = \bar{L}^{(i-1)}_{j} \cup L^{(i)}_{\new, j} \cup \bar{L}^{(i)}_{\new, j}$ and the number of new landmarks is at most
\begin{align*}
&~~~~|L^{(i)}_{\new, j}| + |\bar{L}^{(i)}_{\new, j}| \leq |L^{(i)}_{\new,j}| + O(\tau^{(i)})\leq O(\kappa_{j}\cdot\chi\cdot |\pi^{(i)}|).
\end{align*}
The distortion $\bar{\sigma}^{(i)}_{j} = \sum_{0\leq i'\leq i}O(h^{(i)}_{j} + \bar{h}^{(i)}_{j})/\kappa_{\sigma}$ is always bounded by $2^{O(t\cdot\lambda_{\insLM,t})}\cdot h_{j}/\kappa_{\sigma}\leq h_{j}/\log^{2} n$ because $t = O(1/\epsilon), \lambda_{\insLM,t} = O(1/\epsilon)$ and $\kappa_{\sigma} = n^{\epsilon^{4}}$. 

\

\underline{The Node-Weighting $\bar{A}_{i}$ and Density $\bar{\rho}_{i}$.} Initially 
\begin{align*}
\bar{A}^{(0)} &= A^{(0)} + \bar{A}_{\new}^{(0)}\\
&= \deg_{G^{(0)}} + \mathds{1}(T^{(0)}) + \mathds{1}(\bigcup_{j}L_{j}^{(0)}) + \mathds{1}(\bigcup_{j}\bar{L}^{(0)}_{\new,j}\setminus \bigcup_{j}L^{(0)}_{j}) \\
&= \deg_{G^{(0)}} + \mathds{1}(T^{(0)}) + \mathds{1}(\bigcup_{j}\bar{L}_{j}^{(0)})
\end{align*}
For each time $i\geq 1$, by a simple induction, we have
\begin{align*}
\bar{A}^{(i)} &= \bar{A}^{(i-1)} + A^{(i)}_{\new} + \bar{A}^{(i)}_{\new} \\
&= \bar{A}^{(i-1)} + \mathds{1}(T^{(i)}_{\new}) - \deg_{F^{(i)}} + \deg_{E_{\new}^{(i)}} + \mathds{1}(\bigcup_{j}L^{(i)}_{\new,j}\setminus \bigcup_{j}\bar{L}^{(i-1)}_{j}) + \mathds{1}(\bigcup_{j}\bar{L}_{\new,j}^{(i)}\setminus \bigcup_{j}L^{(i)}_{j})\\
&= \deg_{G^{(i)}} + \mathds{1}(T^{(i)}) + \mathds{1}(\bigcup_{j}\bar{L}^{(i)}_{j})
\end{align*}

The density function is
\begin{equation*}
\bar{\rho}^{(i)} \geq \left\{
\begin{aligned}
&1,\ v\in \deg_{G^{(i)}}\\
&\frac{(1-1/\mu)^{i}}{\phi}\geq \Omega(1/\phi),\ v\in \mathds{1}(T^{(i)}) + \mathds{1}(\bigcup_{j}\bar{L}_{j}^{(i)})
\end{aligned}
\right.,
\end{equation*}
because $(1-1/\mu)^{i}\geq \Omega(1)$ (recall that we set $\mu = 2t$).

\

\underline{The Pairwise Cover $\bar{\cal N}^{(i)}_{j}$.} $\bar{\cal N}_{j}^{(i)}$ has quality parameters
\[
\bar{h}_{\cov,j}^{(i)} = \bar{h}_{\sep,j}^{(i)} =  h_{j}^{(0)}/2^{O(i\cdot\lambda_{\insLM,t})}\geq 2^{j}=h_{j},\ \bar{\omega}_{j}^{(i)} = O(t\cdot\lambda_{\insLM,t}\cdot\kappa_{\PC,\omega}),
\]
because we choose $h^{(0)}_{j} = 2^{O(t\cdot\lambda_{\insLM,t})}$ with sufficiently large constant hided in the exponent.

\

\underline{The Routers.} Routers in $\bar{\cal R}_{j}^{(i)}$ are maintained by \Cref{thm:Router} undergoing at most $\bar{f}^{(i)} = O(t\cdot\lambda_{\insLM,t})$ batched updates.

\

\underline{The Embedding.} $\Pi_{\bar{\cal R}_{j}^{(i)}\to G^{(i)}}$ has length, by induction,
\begin{align*}
\bar{h}^{(i)}_{\emb,j} &= \max\{\bar{h}^{(i-1)}_{\emb,j},O(\lambda_{\init,h}\cdot \bar{h}^{(i-1)}_{j}),O(\lambda_{\init,h}\cdot h^{(i)}_{j})\} \leq  \lambda_{\init,h}\cdot 2^{O(t\cdot\lambda_{\insLM,t})}\cdot h_{j}.
\end{align*}

The congestion $\bar{\gamma}_{\emb,j}^{(i)}$ can be bounded as follows. For the initialization and each $\pi^{(i)}$ representing terminal insertions or edge insertions, from $\bar{\gamma}^{(i-1)}_{\emb,j}$ to $\bar{\gamma}^{(i)}_{\emb,j}$, the congestion  increases additively and the increasing is dominated by the landmark insertion step. Formally, we have $\bar{\gamma}_{\emb,j}^{(0)} = O(\Gamma^{(0)}_{j}/\phi)$ and $\bar{\gamma}_{\emb,j}^{(i)} = \bar{\gamma}^{(i-1)}_{\emb,j} + O(\Gamma_{j}^{(i)}/\phi)$ if $\pi^{(i)}$ represents terminal insertion or edge insertion. We further define 
\[
\Gamma = O(t^{2}\cdot\lambda_{\insLM,t}^{3}\cdot\kappa_{\PC,\omega}^{2}\cdot (\kappa_{\CM,L} + \kappa_{\init,L})\cdot (\kappa_{\CM,\gamma} + \kappa_{\init,\gamma})\cdot\kappa_{\insLM,\gamma})
\]
to be a global upper bound of $\Gamma^{(i)}_{j}$ over all $0\leq i\leq t$ and $1\leq j\leq \bar{j}$. Then we can rewrite $\bar{\gamma}_{\emb,j}^{(0)} = O(\Gamma/\phi)$ and $\bar{\gamma}_{\emb,j}^{(i)} = \bar{\gamma}^{(i-1)}_{\emb,j} + O(\Gamma/\phi)$.

When $\pi^{(i)}$ represents a batched edge deletion, the congestion increases multiplicatively by a factor $(1+O(1/t))$ and then increases additively in the following landmark insertion step. Namely, $\bar{\gamma}^{(i)}_{\emb,j} = (1+O(1/t))\cdot \bar{\gamma}^{(i-1)}_{\emb,j} + O(\Gamma/\phi)$. Therefore,
\begin{align*}
\bar{\gamma}^{(i)}_{\emb,j} &= O(t\cdot\Gamma\cdot (1+O(1/t))^{t}/\phi) = O(t\cdot \Gamma/\phi).
\end{align*}

\paragraph{Construction of the sparsifier $H^{(i)}$.} For each $i\geq 0$, let $\bar{T}^{(i)} = T^{(i)}\cup \bigcup_{j}\bar{L}^{(i)}_{j}$. For each $1\leq j\leq \bar{j}$, the restriction of $\bar{\cal N}_{j}^{(i)}$ on $\mathds{1}(\bar{T}^{(i)})\subseteq \bar{A}^{(i)}$, denoted by $\bar{\cal N}^{(i)}_{\tmn,j}$, is exactly a pairwise cover of terminal set $\bar{T}^{(i)}$ on $G^{(i)}-\bar{C}^{(i)}_{j}$ with the same quality parameters
\[
\bar{h}^{(i)}_{\cov,j} = h_{j} = 2^{j},\ \bar{\omega}^{(i)}_{j} = O(t\cdot\lambda_{\insLM,t}\cdot\kappa_{\PC,\omega}),
\]
and by \Cref{lemma:CertifiedEDToDist}, its diameter on $G^{(i)}$ is 
\[
\bar{h}^{(i)}_{\Gdiam,j} = \lambda_{\rt,h}(\bar{f}^{(i)})\cdot\bar{h}^{(i)}_{\emb,j} \leq \lambda_{\rt,h}(t\cdot\lambda_{\insLM,t})\cdot \lambda_{\init,h} \cdot 2^{O(t\cdot\lambda_{\insLM,t})}\cdot h_{j}.
\]

Note that for the dynamic graph $G$ and the dynamic terminal set $T$, $\{\bar{C}_{j}^{(i)},\bar{L}_{j}^{(i)},\bar{\cal N}_{\tmn,j}^{(i)}\mid 1\leq j\leq \bar{j}\}$ is a valid input to \Cref{lemma:DynamicCoverToSparsifier} with $\alpha_{\low} = \lambda_{\rt,h}(t\cdot\lambda_{\insLM,t})\cdot\lambda_{\init,h}\cdot 2^{O(t\cdot\lambda_{\insLM,t})}
$. Then \Cref{lemma:DynamicCoverToSparsifier} will initialize and maintain a $(\alpha_{\low},\alpha_{\up},h)$-sparsifier $H$ of some $\bar{T}\supseteq T$ with $\alpha_{\up} = O(1)$ and the $\alpha_{\low}$ above. Furthermore, \Cref{lemma:DynamicCoverToSparsifier} guarantees that vertices in $V(H)\setminus \bar{T}$ are fresh.

\paragraph{The Size of $H^{(i)}$.} To bound the size of $H^{(i)}$, we will exploit the condition that $|\pi^{(i')}|\leq \phi\cdot |G^{(i')}|$ for each $1\leq i'\leq t$. We first give upper bounds of $|\bar{C}^{(i)}_{j}|$ and $|\bar{L}^{(i)}_{j}|$ for each $1\leq j\leq \bar{j}$. For the moving cut,
\begin{align*}
|\bar{C}^{(i)}_{j}| &\leq |\bar{C}^{(0)}_{j}| + \sum_{1\leq i'\leq i}(|C^{(i')}_{\new,j}| + |\bar{C}^{(i')}_{\new,j}|)\\
&\leq \kappa_{j}\cdot h_{j}\cdot 2^{O(t\cdot\lambda_{\insLM,t})}\cdot(\phi\cdot E(G^{(0)}) + |T^{(0)}|) + \kappa_{j}\cdot\chi\cdot h_{j}\cdot 2^{O(t\cdot\lambda_{\insLM,t})}\cdot\sum_{1\leq i'\leq i}|\pi^{(i')}|\\
&\leq \kappa_{j}\cdot\chi\cdot h_{j}\cdot 2^{O(t\cdot\lambda_{\insLM,t})}\cdot(|T^{(0)}| + \phi\cdot\sum_{0\leq i'\leq i}|E(G^{(i')})|)\\
&\leq \kappa_{j}\cdot\chi\cdot t\cdot h_{j}\cdot 2^{O(t\cdot\lambda_{\insLM,t})}\cdot(|T^{(i)}| + \phi\cdot E(G^{(i)})).
\end{align*}
Here, the third inequality is by $|\pi^{(i')}|\leq \phi\cdot |G^{(i')}|$ for all $1\leq i'\leq t$. The last inequality is because $|G^{(i')}|$ is almost unchanged during the time period $0\leq i'\leq i$. Precisely, for each $0\leq i'\leq i-1$, we have $|G^{(i')}| \leq |G^{(i'+1)}| + |\pi^{(i'+1)}| \leq (1+\phi)\cdot |G^{(i'+1)}|$, which implies $|G^{(i')}|\leq (1+\phi)^{i-i'}\cdot|G^{(i)}|\leq O(|G^{(i)}|)$ because we require that $\phi \leq 1/t$ in the statement of \Cref{thm:NonHopReducingEmulator}.

By the same argument, the number of landmarks in $\bar{L}^{(i)}_{j}$ is given by
\begin{align*}
|\bar{L}^{(i)}_{j}| &\leq |\bar{L}^{(0)}_{j}| + \sum_{1\leq i'\leq i} |L^{(i)}_{\new,j}| + |\bar{L}^{(i)}_{\new,j}|\\
&\leq O(\kappa_{j}\cdot(\phi\cdot |E(G^{(0)})| +|T^{(0)}|)) + O(\kappa_{j}\cdot\chi\cdot \sum_{1\leq i'\leq i}|\pi^{(i')}|)\\
&\leq O(\kappa_{j}\cdot\chi\cdot t\cdot (|T^{(i)}| + \phi\cdot |E(G^{(i)})|)).
\end{align*}

Now we are ready to bound $|V(H^{(i)})|$ and $|H^{(i)}|$. By \Cref{lemma:DynamicCoverToSparsifier}, we have
\[
|V(H^{(i)})| = O(\sum_{j}\bar{\omega}^{(i)}_{j}\cdot|\bar{T}^{(i)}|) \leq O(\kappa_{j}\cdot t\cdot\lambda_{\insLM,t}\cdot\kappa_{\NC,\omega}\cdot |V(G^{(i)})|)\leq n^{O(\epsilon^{4})}\cdot|V(G^{(i)})|
\]
because $|\bar{T}^{(i)}|$ has a trivial upper bound $|V(G^{(i)})|$ (since $\bar{T}^{(i)}\subseteq V(G^{(i)})$).
Regarding the size of $H^{(i)}$, we have
\begin{align*}
|H^{(i)}| &= O(\sum_{j}\bar{\omega}^{(i)}_{j}\cdot |\bar{T}^{(i)}| + t\cdot \sum_{j}|\bar{C}^{(i)}_{j}|/h_{j})\\
&\leq O(\kappa_{j}\cdot t\cdot\lambda_{\insLM,t}\cdot \kappa_{\PC,\omega}\cdot(|T^{(i)}| + \sum_{j}|\bar{L}^{(i)}_{j}|) + t\cdot \sum_{j}|\bar{C}^{(i)}_{j}|/h_{j})\\
&\leq \kappa_{j}^{3}\cdot\chi\cdot \kappa_{\PC,\omega}\cdot 2^{O(t\cdot\lambda_{\insLM,t})}\cdot(|T^{(i)}| + \phi\cdot |E(G^{(i)})|).
\end{align*}
Recall that $\chi =  O(\lambda_{\rt,\prune}(t\cdot\lambda_{\insLM,t})\cdot t^{4}\cdot\lambda_{\insLM,t}^{2}\cdot\kappa_{\PC,\omega}^{2}\cdot(\kappa_{\CM,L} +\kappa_{\init,L})\cdot(\kappa_{\CM,\gamma} + \kappa_{\init,\gamma}))$.

\paragraph{The Recourse from $H^{(i-1)}$ to $H^{(i)}$}. We first bound the recourse from $\bar{\cal N}_{j}^{(i-1)}$ to $\bar{\cal N}_{j}^{(i)}$ restricted on $\bar{T}^{(i)}$. Trivially, we have
\[
\recourse_{\mid \mathds{1}(\bar{T}^{(i)})}(\bar{\cal N}_{j}^{(i-1)}\to\bar{\cal N}_{j}^{(i)}) \leq \recourse_{\mid \mathds{1}(\bar{T}^{(i)})}(\bar{\cal N}^{(i-1)}_{j}\to{\cal N}^{(i)}_{j}) + \recourse_{\mid\mathds{1}(\bar{T}^{(i)})}({\cal N}^{(i)}_{j}\to\bar{\cal N}^{(i)}_{j}),
\]
and then we bound the two terms on the right hand side respectively. Recall that $\bar{T}^{(i)} = T^{(i)} \cup \bigcup_{j}\bar{L}^{(i)}_{j}$. The key point is that virtual nodes corresponding to terminals and landmarks have density $\Omega(1/\phi)$ at all time, so roughly speaking, the recourse of pairwise covers will drop by a factor $\Omega(1/\phi)$ when restricting on $\mathds{1}(\bar{T}^{(i)})$. Precisely, we have
\[
\min_{v\in \mathds{1}(\bar{T}^{(i)})\cap A^{(i)}} \rho^{(i)}(v)\geq \Omega(1/\phi)\text{ and }\min_{v\in \mathds{1}(\bar{T}^{(i)})\cap \bar{A}^{(i)}}\bar{\rho}^{(i)}(v)\geq \Omega(1/\phi).
\]
Therefore, by \Cref{thm:EdgeDelWithDensity} (since the batched edge deletion is the bottleneck),
\begin{align*}
\recourse_{\mid \mathds{1}(\bar{T}^{(i)})}(\bar{\cal N}^{(i-1)}_{j}\to{\cal N}^{(i)}_{j}) &\leq O(\bar{\omega}^{(i-1)}_{j}\cdot\kappa_{\PC,\omega}\cdot \mu\cdot \lambda_{\rt,\prune}(\bar{f}^{(i-1)}_{j})\cdot \bar{\gamma}_{\emb,j}^{(i-1)}\cdot |\pi^{(i)}|/\Omega(1/\phi))\\
&\leq O(t^{3}\cdot \lambda_{\insLM,t}\cdot \kappa^{2}_{\PC,\omega}\cdot \lambda_{\rt,\prune}(t\cdot\lambda_{\insLM,t})\cdot \Gamma\cdot |\pi^{(i)}|).
\end{align*}
where we use $\bar{\omega}^{(i-1)}_{j} = O(t\cdot\lambda_{\insLM,t}\cdot\kappa_{\NC,\omega})$, $\bar{f}^{(i-1)}_{j} = O(t\cdot \lambda_{\insLM,t})$ and $\bar{\gamma}^{(i)}_{\emb,j}\leq O(t\cdot \Gamma/\phi)$.
Also, by \Cref{lemma:LandmarkClosure}, 
\begin{align*}
\recourse_{\mid\mathds{1}(\bar{T}^{(i)})}({\cal N}^{(i)}_{j}\to\bar{\cal N}^{(i)}_{j})&\leq O(\tau^{(i)}\cdot\bar{\omega}^{(i)}_{j}\cdot\kappa_{\PC,\omega}/(\phi\cdot \Omega(1/\phi)))\\
&\leq O(\kappa_{j}\cdot \chi\cdot t\cdot \lambda_{\insLM,t}\cdot\kappa_{\PC,\omega}^{2}\cdot |\pi^{(i)}|).
\end{align*}
where the second inequality uses $\tau^{(i)}\leq O(\kappa_{j}\cdot\chi\cdot|\pi^{(i)}|)$.

Now we are ready to show the recourse from $H^{(i-1)}$ to $H^{(i)}$. By \Cref{lemma:DynamicCoverToSparsifier}, 
\begin{align*}
&~~~~\recourse(H^{(i-1)}\to H^{(i)})\\
&\leq O(|\pi^{(i)}| + \sum_{j}(t\cdot (|C^{(i)}_{\new,j}| + |\bar{C}^{(i)}_{\new,j}|)/h_{j} + \recourse(\bar{\cal N}_{\tmn,j}^{(i-1)}\to\bar{\cal N}_{\tmn,j}^{(i)})))\\
&\leq O(|\pi^{(i)}| + \sum_{j}(t\cdot (|C^{(i)}_{\new,j}| + |\bar{C}^{(i)}_{\new,j}|)/h_{j} + \bar{\omega}^{(i-1)}_{j}\cdot\recourse(\bar{T}^{(i-1)}\to\bar{T}^{(i)}) +\recourse_{\mid \mathds{1}(\bar{T}^{(i)})}(\bar{\cal N}_{j}^{(i-1)}\to\bar{\cal N}_{j}^{(i)})))\\
&\leq O(|\pi^{(i)}| + \sum_{j}(\kappa_{j}\cdot\chi\cdot 2^{O(t\cdot\lambda_{\insLM,t})}\cdot|\pi^{(i)}|
+t\cdot\lambda_{\insLM,t}\cdot\kappa_{\PC,\omega}\cdot \kappa^{2}_{j}\cdot\chi\cdot|\pi^{(i)}|\\
&~~~~~~~~~~~~~~~~~~~ 
+t^{3}\cdot \lambda_{\insLM,t}\cdot \kappa^{2}_{\PC,\omega}\cdot \lambda_{\rt,\prune}(t\cdot\lambda_{\insLM,t})\cdot \Gamma\cdot |\pi^{(i)}|\\
&~~~~~~~~~~~~~~~~~~~ 
+ \kappa_{j}\cdot \chi\cdot t\cdot \lambda_{\insLM,t}\cdot\kappa_{\PC,\omega}^{2}\cdot |\pi^{(i)}|))\\
&\leq O(\kappa_{j}^{3}\cdot 2^{O(t\cdot\lambda_{\insLM,t})}\cdot \lambda_{\rt,\prune}(t\cdot\lambda_{\insLM,t}) \cdot \kappa_{\PC,\omega}^{4}\cdot (\kappa_{\CM,L} +\kappa_{\init,L})\cdot(\kappa_{\CM,\gamma} + \kappa_{\init,\gamma})\cdot\kappa_{\insLM,\gamma}\cdot|\pi^{(i)}|)
\end{align*}
Here the second inequality is because, to update $(\bar{\cal N}^{(i-1)}_{j})_{\mid \bar{T}^{(i-1)}}$ to $(\bar{\cal N}^{(i)}_{j})_{\mid \bar{T}^{(i)}}$, we will first update $(\bar{\cal N}^{(i-1)}_{j})_{\mid \bar{T}^{(i-1)}}$ to $(\bar{\cal N}^{(i-1)}_{j})_{\mid \bar{T}^{(i)}}$ by adding vertices in $\bar{T}^{(i)}\setminus \bar{T}^{(i-1)}$, each of which takes $\bar{\omega}_{j}^{(i-1)}$ virtual node insertions to the pairwise cover. The third inequality is because $\recourse(\bar{T}^{(i-1)}\to \bar{T}^{(i)})\leq |\pi^{(i)}| + \sum_{j}(|L^{(i)}_{\new,j}| + \bar{L}^{(i)}_{\new,j})\leq O(\kappa^{2}_{j}\cdot\chi\cdot|\pi^{(i)}|)$ and $\bar{\omega}_{j}^{(i-1)} = O(t\cdot\lambda_{\insLM,t}\cdot\kappa_{\NC,\omega})$. We obtain the last inequality by plugging in the bounds of $\chi$ and $\Gamma$.

\paragraph{The Running Time.} 

\underline{The Initialization Step.} At the initialization step with $i=0$, summing over the initialization of dense certified-EDs (\Cref{thm:InitDenseCertifiedED}), the landmark insertion (\Cref{lemma:LandmarkClosure}) and the initialization step of \Cref{lemma:DynamicCoverToSparsifier}, the initialization time is 
\begin{align*}
&~~~~\sum_{j}(|G^{(0)}| + \rho^{(0)}(A^{(0)}))\cdot \poly(h_{j}^{(0)})\cdot n^{O(\epsilon)}\\
&+\sum_{j}(\tau^{(0)}/\phi)\cdot (\bar{\omega}^{(0)}_{j})^{2}\cdot \poly(h_{j}^{(0)})\cdot (1/\phi + 2^{O(\bar{f}^{(0)}_{j})})\\
&+\tilde{O}(\sum_{j}\bar{\omega}^{(0)}_{j}\cdot |\bar{T}^{(0)}| + |\bar{C}_{j}^{(0)}|)\\
&= (|G^{(0)}| + |T^{(0)}|/\phi)\cdot\poly(h)\cdot n^{O(\epsilon)}/\phi
\end{align*}

\underline{Update Steps.} For a update step $i\geq 1$, the time to obtain all $(C^{(i)}_{j}, L^{(i)}_{j}, {\cal N}^{(i)}_{j}, {\cal R}^{(i)}_{j}, \Pi_{{\cal R}^{(i)}_{j}\to G^{(i)}})$ is at most (the bottleneck is when $\pi^{(i)}$ represents batched edge deletions)
\begin{align*}
&~~~~\sum_{j}(\bar{\omega}^{(i-1)}_{j}\cdot \lambda\cdot\lambda_{\rt,\prune}(\bar{f}^{(i-1)}_{j})\cdot\bar{\gamma}_{\emb,j}^{(i-1)}\cdot |\pi^{(i)}|)\cdot \poly(\bar{h}^{(i-1)}_{j})\cdot n^{O(\epsilon)}\cdot (1/\phi^{(i)}_{j} + 2^{O(\bar{f}^{(i-1)}_{j})}\cdot\bar{\omega}^{(i-1)}_{j}).
\end{align*}
The time to compute $(\bar{C}^{(i)}_{j}, \bar{L}^{(i)}_{j}, {\cal N}^{(i)}_{j}, {\cal R}^{(i)}_{j}, \Pi_{{\cal R}^{(i)}_{j}\to G^{(i)}})$ is, by \Cref{lemma:LandmarkClosure},
\begin{align*}
&~~~~\sum_{j} (\tau^{(i)}/\phi)\cdot(\bar{\omega}_{j}^{(i-1)})^{2}\cdot \poly(h_{j})\cdot n^{O(\epsilon)}\cdot (1/\phi + 2^{O(\bar{f}_{j}^{(i-1)})}).
\end{align*}
The time to update $H^{(i)}$ to $H^{(i+1)}$ is, by \Cref{lemma:DynamicCoverToSparsifier},
\begin{align*}
&~~~~\tilde{O}(\sum_{j}(|C^{(i)}_{\new,j}| + |\bar{C}^{(i)}_{\new,j}|) + \recourse(\bar{\cal N}_{\tmn,j}^{(i-1)}\to\bar{\cal N}_{\tmn,j}^{(i)}))
\end{align*}
Summing over these three steps, the total update time is $|\pi^{(i)}|\cdot\poly(h)\cdot n^{O(\epsilon)}/\phi^{2}$.

\

\noindent\textbf{Path Unfolding.} At any time $i$, given a path $P_{H}$ in $H$ connecting some $u,v\in \bar{T}$, we construct a $u$-$v$ path $P_{G}$ in $G$ as follows. 

From the proof of \Cref{lemma:DynamicCoverToSparsifier}, we know $H = H^{\star,(i)}\cup E_{\heavy}$. Here $H^{\star}$ is the union of stars $H^{\star}_{S_{\tmn}}$ of those clusters $S_{\tmn}\in \bigcup_{j} \bar{\cal N}_{\tmn,j}$, and $E_{\heavy}\subseteq E(G)$ are some original edges. Therefore, we can decompose $P_{H}$ into original edges and \emph{star subpaths}, where a star subpath has exactly two edges from the same star. 

Consider a star subpath $P'_{H}$. Let $H^{\star}_{S_{\tmn}}$ (where $S_{\tmn}\in \bar{\cal N}_{\tmn,j}$ for some $j$) be the star containing $P'_{H}$, and let $u',v'\in S_{\tmn}$ denote the endpoints of $P'_{H}$. For now we assume $u',v'$ are two different vertices, and we will discuss the case $u'=v'$ later. We have $\ell_{H}(P'_{H}) = 2h_{j}$ because the star of a level-$j$ cluster has edge length $h_{j}$ (see the definition of stars in the proof of \Cref{thm:CoversToEmulators}). We now construct a $u'$-$v'$ path $P'_{G}$ on $G$. Note that there is a cluster $S\in \bar{\cal N}_{j}$ s.t. $S_{\tmn}\subseteq S$ (recall that $\bar{\cal N}_{\tmn,j}$ is the restriction of $\bar{\cal N}_{j}$ on $\mathds{1}(\bar{T})$). Let $R\in \bar{\cal R}_{j}$ be the router corresponding to $S$, and let $u'_{\virtual}$ and $v'_{\virtual}\in S\subseteq V(R)$ be arbitrary virtual nodes of $u'$ and $v'$ respectively. By \Cref{thm:RouterPathReport}, we can obtain a $u'_{\virtual}$-$v'_{\virtual}$ path $P'_{R}$ on $R$ with $|P'_{R}| \leq \lambda_{\rt,h}(\bar{f})$ in $O(|P'_{R}|)$ time. In the embedding $\Pi_{\bar{R}_{j}\to G}$, each edge $e=(x_{\virtual},y_{\virtual})\in P'_{R}$ (where $x_{\virtual},y_{\virtual}\in S$ are endpoints of $e$) is mapped into a path $P'_{G,e}$ on $G$ connecting vertices $x_{\vertex}$ and $y_{\vertex}$ (where $x_{\vertex},y_{\vertex}$ are the original vertex corresponding to virtual nodes $x_{\virtual},y_{\virtual}$) with length $\ell_{G}(P'_{G,e})\leq \bar{h}_{\emb,j}$. We construct $P'_{G}$ by concatenating of these $P'_{G,e}$ according to the order of edges $e$ on $P'_{R}$. Obviously, $P'_{G}$ is a $u'$-$v'$ path on $G$, and we have 
\[
\ell_{G}(P'_{G})\leq |P'_{R}|\cdot \bar{h}_{\emb,j}\leq \lambda_{\rt,h}(t\cdot \lambda_{\insLM,t})\cdot\lambda_{\init,h}\cdot 2^{O(t\cdot \lambda_{\insLM},t)}\cdot h_{j}.
\]
Because $u'$ and $v'$ are two different vertices, trivially we know $P'_{G}$ has at least one edge, i.e. $|P'_{G}|\geq 1$.

In the case that the endpoints $u',v'\in S_{\tmn}$ of $P'_{H}$ are the same (this may happen because $P'_{H}$ may not be a simple path). The reason why this case is special is that the above construction of $P'_{G}$ may give a $P'_{G}$ with no edge. This can be simply fixed as follows. Because the star $H^{\star}_{S_{\tmn}}$ has at least one edge, we have $|S_{\tmn}|\geq 2$ (otherwise $H^{\star}_{S_{\tmn}}$ will degenerate). we let $w'$ be an arbitrary vertex in $S_{\tmn}$ other than $u'$. We construct $P''_{G}$ for the vertex pair $(u',w')$ as above, and then let the real $P'_{G}$ be the concatenation of the forward $P''_{G}$ (from $u'$ to $w'$) and the backward $P'_{G}$ (from $w'$ to $u'$). Now $\ell_{G}(P'_{G})\leq |P'_{R}|\cdot \bar{h}_{\emb,j}\leq 2\cdot\lambda_{\rt,h}(t\cdot \lambda_{\insLM,t})\cdot\lambda_{\init,h}\cdot 2^{O(t\cdot \lambda_{\insLM},t)}\cdot h_{j}$ and $|P'_{G}|\geq 2$. 

Finally, we obtain $P_{G}$ from $P_{H}$ by substituting each star subpath $P'_{H}$ on $P_{H}$ with $P'_{G}$. We have $\ell_{G}(P_{G})\leq \lambda_{\rt,h}(t\cdot \lambda_{\insLM,t})\cdot\lambda_{\init,h}\cdot 2^{O(t\cdot \lambda_{\insLM},t)}\cdot \ell_{H}(P_{H})$ because we did not change the original edges on $P_{H}$, and each star subpath $P'_{H}$ with length $\ell_{H}(P'_{H}) = 2h_{j}$ is replaced by $P'_{G}$ with length $\ell_{G}(P'_{G})\leq 2\cdot\lambda_{\rt,h}(t\cdot \lambda_{\insLM,t})\cdot\lambda_{\init,h}\cdot 2^{O(t\cdot \lambda_{\insLM},t)}\cdot h_{j}$. Analogously, we can show $|P_{G}|\geq |P_{H}|/2$.

The path-unfolding time is $O(|P_{H}|\cdot\lambda_{\rt,h}(t\cdot\lambda_{\insLM,t})/\epsilon^{4}) + O(|P_{G}|) = 2^{O(1/\epsilon^{4})}\cdot |P_{G}|$. There is a term $O(|P_{H}|\cdot\lambda_{\rt,h}(t\cdot\lambda_{\insLM,t})/\epsilon^{4})$ because we will invoke the path-reporting algorithm in \Cref{thm:RouterPathReport} $O(|P_{H}|)$ times, each time it return some router path $P'$ with $|P'|\leq \lambda_{\rt,h}$ in time $O(|P'|/\epsilon^{4})$.

\end{proof}

%% file: 8-expander_hierarchy.tex
\section{Dynamic Length-Constrained Expander Hierarchy}
\label{sect:ExpanderHierarchy}

The notion of \emph{expander hierarchy} was first introduced by \cite{GRST21} with applications on maintaining tree-flow sparsifiers and congestion-competitive oblivious routing scheme. Later, in \cite{HRG22}, by developing the concept of lengh-constrained expander, they further introduced \emph{length-constrained expander hierarchy} and applied it on distributed oblivious routing with routing paths under certain length constraint.

This section shows a dynamic algorithm for the length-constrained expander hierarchy. Roughly speaking, the algorithm will invoke dynamic certified-EDs and dynamic vertex sparsifiers on top of each other alternatively. This hierarchy will be used in both \Cref{sect:DynHopEmu} and \Cref{sect:OracleShort} but with different input length parameters. 

We note that in the statement of \Cref{thm:ExpanderHierarchy}, we use the global vertex identifiers (see the discussion at the beginning of \Cref{sect:DynSparsifier}) to formally define the set intersection/union operations involving vertices from different vertex sparsifiers. Furthermore, we manually set an upper bound $\phi<o(1/(t\cdot\kappa_{H,\size}))$ just for simplifying the calculation.

\begin{theorem}[Expander Hierarchy]
Let $G$ be a dynamic graph under $t$ batched updates $\pi^{(i)},...,\pi^{(t)}$ of edge insertions/deletions and isolated vertex insertions/deletions. Given parameters $\phi\leq o(1/(t\cdot\kappa_{H,\size}))$ and $\bar{k}\leq \lambda_{k}$, where $\lambda_{k} = O(1/\epsilon^{2})$ is a sufficiently large global parameter, there is an algorithm that maintains an expander hierarchy $({\cal H}, {\cal D})$ with $\bar{k}$ levels, where ${\cal H}$ is a collection of vertex sparsifiers $\{H_{k}\mid 1\leq k\leq \bar{k}\}$, and ${\cal D}$ is a collection of certified-EDs $(C_{k},L_{k},{\cal N}_{k},{\cal R}_{k},\Pi_{{\cal R}_{k}\to H_{k}})$ of some node-weighting $A_{k}$ on $H_{k}$. For each level $1\leq k\leq \bar{k}$, let $h_{k}$ and $\bar{h}_{k}$ be additional given length parameters. The vertex sparsifiers satisfy that $H_{1} = G$ and for each $1\leq k\leq \bar{k}-1$, 
\begin{itemize}
\item $H_{k+1}$ is a $(\lambda_{H,\low},\lambda_{H,\up},\bar{h}_{k})$-sparsifier of $V(H_{k+1})\cap V(H_{k})$ on $H_{k}$ s.t. $L_{k}\subseteq V(H_{k+1})\cap V(H_{k})$. 
\item The size of $H_{k+1}$ is $|H_{k+1}|\leq O(t\cdot\kappa_{H,\size}\cdot \phi\cdot |H_{k}|)$, and the number of vertices in $H_{k+1}$ is $|V(H_{k+1})|\leq n^{O(\epsilon^{4})}\cdot |V(H_{k})|$. 
\item $V(H_{k})\cap V(H_{k+1}) = V(H_{k})\cap \bigcup_{k+1\leq k'\leq\bar{k}} V(H_{k'})$.
\item Given a path $P_{H_{k+1}}$ on $H_{k+1}$ connecting some vertices $u,v\in V(H_{k+1})\cap V(H_{k})$, the algorithm can compute a $u$-$v$ path $P_{H_{k}}$ on $H_{k}$ with $\ell_{H_{k}}(P_{H_{k}})\leq \lambda_{H,\low}\cdot \ell_{H_{k+1}}(P_{H_{k+1}})$ and $|P_{H_{k}}|\geq |P_{H_{k+1}}|/2$. The path-reporting time is $|P_{H_{k}}|\cdot 2^{O(1/\epsilon^{4})}$.
\end{itemize}
For each level $1\leq k\leq \bar{k}$, the certified-ED satisfies the following.
\begin{itemize}
    \item $A_{k}\geq \deg_{H_{k}} + \mathds{1}(V(H_{k}))$.
    \item $|C_{k}|\leq O(t\cdot \lambda_{\dynED,h}\cdot h_{k}\cdot \phi\cdot |H_{k}|)$.
    \item $L_{k}$ has size $|L_{k}|\leq O(t\cdot \phi\cdot |H_{k}|)$ and distrotion $\sigma_{k} = O(\lambda_{\dynED,h}\cdot h_{k}/\kappa_{\sigma})$.
    \item ${\cal N}_{k}$ is a $b_{k}$-distributed $(h_{\cov,k}, h_{\sep,k},\omega_{k})$-pairwise cover of $A_{k}$ on $H_{k}-C_{k}$ with
    \[
    b_{k} = O(t),\ h_{\cov,k} = h_{\sep,k} = h_{k},\ \omega_{k} = O(t\cdot\kappa_{\PC,\omega})
    \]
    \item Each router in ${\cal R}_{k}$ is maintained by \Cref{thm:Router} under at most $f_{k} = O(t)$ batched updates.
    \item $\Pi_{{\cal R}_{k}\to H_{k}}$ has $h_{\emb,k} = O(\lambda_{\init,h}\cdot \lambda_{\dynED,h}\cdot h_{k})$ and $\gamma_{\emb,k} = O(\kappa_{\dynED,\gamma}/\phi)$.
    \item At each time $1\leq i\leq t$, we have
    \begin{align*}
    \sum_{1\leq k\leq \bar{k}}\recourse({\cal N}^{(i-1)}_{k}\to{\cal N}^{(i)}_{k}) &\leq \kappa_{\EH,\PC,\rcs}\cdot |\pi^{(i)}|/\phi,
    \end{align*}
    where $\kappa_{\EH,\PC,\rcs} = \lambda_{\rt,\prune}(t)\cdot t\cdot \kappa_{\PC,\omega}\cdot \kappa_{\dynED,\gamma}\cdot 2^{O(\lambda_{k})}\cdot \kappa^{\lambda_{k}}_{H,\rcs} = n^{O(\epsilon^{2})}$.
\end{itemize}

The initialization time is $|G^{(0)}|\cdot\poly(\max_{k}\{h_{k},\bar{h}_{k}\})\cdot n^{O(\epsilon)}/\phi$,
and the update time for batch $\pi^{(i)}$ is $|\pi^{(i)}|\cdot \poly(\max_{k}\{h_{k},\bar{h}_{k}\})\cdot n^{O(\epsilon)}/\phi^{2}$.

\label{thm:ExpanderHierarchy}
\end{theorem}

\begin{proof} \

\paragraph{The Algorithm for Certifie-EDs.} The certified-EDs are maintained by \Cref{thm:DynamicED}. Precisely, at each level $k$, we will work on graph $H_{k}$ under some batched updates $\pi^{(1)}_{k},...,\pi^{(t)}_{k}$. Initially, $H_{1} = G$ and $\pi^{(1)}_{1} = \pi^{(1)},...,\pi^{(t)}_{1} = \pi^{(t)}$. The working graph and batched updates for $k\geq 2$ will be specified later. 

By applying \Cref{thm:DynamicED} on $H_{k}$ with length parameter $h_{k}$ and congestion parameter $\phi$, we can initialize and maintain a certified-ED $(C_{k},L_{k},{\cal N}_{k},{\cal R}_{k},\Pi_{{\cal R}_{k}\to H_{k}})$ of $A_{k}$ on $H_{k}$ satisfying the following.
\begin{itemize}
\item \underline{The Node-Weighting $A^{(i)}_{k}$.} At any time $i\geq 0$,
\[A_{k}^{(i)} = \deg_{H^{(i)}_{k}} + \mathds{1}(V(H^{(i)}_{k})).
\]

\item \underline{The Moving Cut $C^{(i)}_{k}$.} Initially, 
\[
|C^{(0)}_{k}| = O(\lambda_{\dynED,h}\cdot \phi\cdot |A^{(0)}_{k}|\cdot h_{k}) = O(\lambda_{\dynED,h}\cdot h_{k}\cdot \phi\cdot |H^{(0)}_{k}|).
\]
For each time $i\geq 1$, $C^{(i)}_{k} = (C^{(i-1)}_{k})_{\mid E(H^{(i)}_{k})} + C^{(i)}_{\new,k}$ where
\[
|C^{(i)}_{\new,k}| = O(\lambda_{\dynED,h}\cdot h_{k}\cdot |\pi^{(i)}_{k}|).
\]

\item \underline{The Landmark Set $L^{(i)}_{k}$.} Initially,
\[
|L^{(0)}_{k}| = O(\phi\cdot |A^{(0)}_{k}|) = O(\phi\cdot |H^{(0)}_{k}|).
\]
For each time $i\geq 1$, $L^{(i)}_{k} = (L^{(i-1)}_{k}\cap V(H^{(i)}_{k}))\cup L^{(i)}_{\new,k}$ where
\[
|L^{(i)}_{\new,k}| = O(|\pi^{(i)}_{k}|).
\]

\item The distortion of $L^{(i)}_{k}$, the distributed neighborhood cover ${\cal N}^{(i)}_{k}$, the routers ${\cal R}^{(i)}_{k}$ and the embedding $\Pi_{{\cal R}^{(i)}_{k}\to H^{(i)}_{k}}$ have quality parameters $\sigma_{k}$, $h_{\trace,k}$, $\gamma_{\trace,k}$, $b_{k}$, $h_{\cov,k}$, $h_{\sep,k}$, $\omega_{k}$, $f_{k}$, $h_{\emb,k}$ and $\gamma_{\emb,k}$ as shown in the theorem statement and we will not repeat them here.
\end{itemize}

\noindent\underline{The Reinitialization Condition.} At each time $i\geq 1$, we will reinitialize \Cref{thm:DynamicED} at this level $k$ if either \begin{itemize}
\item the dynamic certified-ED algorithm (\Cref{thm:DynamicED}) at level $k-1$ is reinitialized at time $i$, or
\item $|\pi^{(i)}_{k}|\geq \phi\cdot |H^{(i)}_{k}|$.
\end{itemize}

With the reinitialization steps, we can claim additional properties of $(C_{k}^{(i)}, L^{(i)}_{k}, {\cal N}^{(i)}_{k}, {\cal R}^{(i)}_{k}, \Pi_{{\cal R}^{(i)}_{k}\to H^{(i)}_{k}})$ at each time $i\geq 1$. If \Cref{thm:DynamicED} is not reinitialized at level $k$, we know $|\pi^{(i)}_{k}|\leq \phi\cdot |H^{(i)}_{k}|$, so we have 
\[
|C^{(i)}_{\new,k}| = O(\lambda_{\dynED,h}\cdot h_{k}\cdot |\pi^{(i)}_{k}|)\leq O(\lambda_{\dynED,h}\cdot h_{k}\cdot \phi\cdot |H^{(i)}_{k}|)
\]
and 
\[
|L^{(i)}_{\new,k}| = O(|\pi^{(i)}_{k}|)\leq O(\phi\cdot |H^{(i)}_{k}|). 
\]
Therefore, at any time $i\geq 0$, let $i^{\star}_{k}$ be the most recent reinitialization step of \Cref{thm:DynamicED} at level $k$, we can claim the following upper bounds on $|C^{(i)}_{k}|$ and $|L^{(i)}_{k}|$:
\[
|C^{(i)}_{k}| \leq O(\lambda_{\dynED,h}\cdot h_{k}\cdot\phi\cdot\sum_{i^{\star}_{k}\leq i'\leq i}|H^{(i')}_{k}|)\leq O(t\cdot\lambda_{\dynED,h}\cdot h_{k}\cdot\phi\cdot |H^{(i)}_{k}|),
\]
and 
\[
|L^{(i)}_{k}|\leq O(\phi\cdot \sum_{i^{\star}_{k}\leq i'\leq i}|H^{(i')}_{k}|)\leq O(t\cdot\phi\cdot |H^{(i)}_{k}|).
\]
Note that for each $i^{\star}_{k}\leq i'\leq i-1$, $|H^{(i')}_{k}|\leq |H^{(i'+1)}_{k}| + |\pi^{(i'+1)}_{k}|\leq (1+\phi)\cdot |H^{(i'+1)}_{k}|$, so $|H^{(i')}_{k}|\leq (1+\phi)^{i-i'}|H^{(i)}_{k}|\leq O(|H^{(i)}_{k}|)$ because the parameter $\phi$ is required to be at most $1/t$.

\paragraph{The Algorithm for Sparsifiers.} Given the certified-ED $(C_{k},L_{k},{\cal N}_{k},{\cal R}_{k},\Pi_{{\cal R}_{k}\to H_{k}})$ over time, we prepare the next-level graph $H_{k+1}^{(0)},...,H_{k+1}^{(t)}$ and the update batches $\pi^{(1)}_{k+1},...,\pi^{(t)}_{k+1}$ over time as follows. 

We maintain $H_{k+1}$ as a sparsifier of terminals $L_{k}$ on $H_{k}$ over time by applying \Cref{thm:NonHopReducingEmulator} with length parameter $\bar{h}_{k}$ and congestion parameter $\phi' = \phi/O(1)$, where $O(1)$ represents some large constant. Moreover, at each time $i\geq 1$, the batched update $\bar{\pi}^{(i)}_{k}$ fed to \Cref{thm:NonHopReducingEmulator} includes
\begin{itemize}
\item the batched edge insertions/deletions and isolated vertex insertions/deletions from $H^{(i-1)}_{k}$ to $H^{(i)}_{k}$, i.e. $\pi^{(i)}_{k}$, and
\item the terminal insertion from $L^{(i-1)}_{k}$ and $L^{(i)}_{k}$, i.e. $L^{(i)}_{\new,k}$,
\end{itemize}
if the dynamic certified-ED algorithm is \textit{not} reinitialized at level $k$ (we will discuss in a moment how to handle the opposite). Then $|\bar{\pi}^{(i)}_{k}|$ has size
\[
|\bar{\pi}^{(i)}_{k}| = |\pi^{(i)}_{k}| + |L^{(i)}_{\new,k}| \leq O(|\pi^{(i)}_{k}|)=O(\phi\cdot|H^{(i)}_{k}|)\leq \phi'\cdot |H^{(i)}_{k}|,
\]
so it meets the requirements of \Cref{thm:NonHopReducingEmulator}.

Also, $H^{(i)}_{k+1}$ is guaranteed to be a $(\lambda_{H,\low},\lambda_{H,\up},\bar{h}_{k})$-sparsifier of $V(H^{(i)}_{k})\cap V(H^{(i)}_{k+1})$ on $H^{(i)}_{k}$ s.t. $L^{(i)}_{k}\subseteq V(H^{(i)}_{k})\cap V(H^{(i)}_{k+1})$, and all vertices in $V(H^{(i)}_{k+1})\setminus V(H^{(i)}_{k})$ are fresh. The number of vertices in 
$H^{(i)}_{k+1}$ is $|V(H^{(i)}_{k+1})|\leq n^{O(\epsilon^{4})}\cdot |V(H^{(i)}_{k})|$. The size of $H^{(i)}_{k+1}$ is
\[
|H^{(i)}_{k+1}|\leq O(\kappa_{H,\size}\cdot(|L^{(i)}_{k}| + \phi\cdot |H^{(i)}_{k}|))\leq O(t\cdot\kappa_{H,\size}\cdot\phi\cdot|H^{(i)}_{k}|),
\]
and $H^{(i)}_{k+1}$ can be updated from $H^{(i-1)}_{k+1}$ by a batched update $\pi^{(i)}_{k+1}$ with size $|\pi^{(i)}_{k+1}| \leq \kappa_{H,\rcs}\cdot |\bar{\pi}^{(i)}_{k}|$. 

\begin{claim}
For each $1\leq k\leq \bar{k}-1$, we have $V(H^{(i)}_{k})\cap V(H^{(i)}_{k+1}) = V(H^{(i)}_{k})\cap \bigcup_{k+1\leq k'\leq \bar{k}} V(H^{(i)}_{k'})$.
\end{claim}
\begin{proof}
We first show that, for $0\leq i\leq t$ and $1\leq k<k'\leq\bar{k}$, the fresh vertices in $V(H^{(i)}_{k'})$ have global identifiers different from those of $V(H^{(i)}_{k})$, i.e. $V(H^{(i)}_{k'})\setminus V(H^{(i)}_{k'-1})$ and $V(H^{(i)}_{k})$ are disjoint. We will argue this inductively over time. Obviously, the statement holds at time $i=0$ because we initialize the sparsifiers bottom-up (i.e. in the order $H^{(0)}_{1},...,H^{(0)}_{\bar{k}}$). Fix a time $i\geq 1$. We assume the statement holds at time $i-1$. Then each fresh vertices $v$ in $V(H^{(i)}_{k'})$ is not in $V(H^{(i)}_{k})$ by the following reason. If $v$ is created at time $i$, then $v\notin V(H^{(i)}_{k})$ by the definition of fresh vertices (all vertices in $V(H^{(i)}_{k})$ are created before $v$ by our algorithm). Otherwise, $v$ is created before time $i$. By the assumption, $v\notin V(H^{(i-1)}_{k})$. It remains to show $v$ is not in $V(H^{(i)}_{k})\setminus V(H^{(i-1)}_{k})$. Because we update the sparsifiers bottom-up, each vertex in $V(H^{(i)}_{k})\setminus V(H^{(i-1)}_{k})$ is created without any knowledge about the global identifiers of $v$. In other words, each vertex in $V(H^{(i)}_{k})\setminus V(H^{(i-1)}_{k})$ either is a fresh vertex or has the same identifier with someone in $\bigcup_{k''\leq k}V(H^{(i-1)}_{k''})$, but $v$ is not in $V(H^{(i-1)}_{k''})$ for all $k''\leq k$ by the assumption. This implies $v\notin V(H^{(i)}_{k})\setminus V(H^{(i-1)}_{k})$.

Consider a fixed $k$. We now show $V(H^{(i)}_{k'})\cap V(H^{(i)}_{k})\subseteq V(H^{(i)}_{k+1})$ for each $k'\geq k+1$ by induction on $k'$. When $k'=k+1$, it trivially holds. When $k'\geq k+2$, assume this statement holds for $k'-1$. We have shown that the fresh vertices in $V(H^{(i)}_{k'})$ are not in $V(H^{(i)}_{k})$, so $V(H^{(i)}_{k'})\cap V(H^{(i)}_{k})$ is a subset of non-fresh vertices $V(H^{(i)}_{k'})\cap V(H^{(i)}_{k'-1})$, which implies $V(H^{(i)}_{k'})\cap V(H^{(i)}_{k})\subseteq V(H^{(i)}_{k'-1})\cap V(H^{(i)}_{k})\subseteq V(H^{(i)}_{k+1})$. 

The original claim is a simple corrollary of the above statement.

\end{proof}

\noindent\underline{The Reinitialization Condition.} At time $i\geq 1$, we will reinitialize \Cref{thm:NonHopReducingEmulator} at this level $k$ if the \Cref{thm:DynamicED} maintaining $(C_{k},L_{k},{\cal N}_{k},{\cal R}_{k},\Pi_{{\cal R}_{k}\to H_{k}})$ triggers a reinitialization step at time $i$.

\paragraph{Size of $\pi^{(i)}_{k}$ and $H^{(i)}_{k}$.} At time $i$, let $k^{\star}$ be the minimum level at which the dynamic certified-ED algorithm is reintialized. Then we have for each $1\leq k\leq k^{\star}-1$, $|\bar{\pi}^{(i)}_{k}|\leq O(|\pi^{(i)}_{k}|)$ and $|\pi^{(i)}_{k+1}|\leq \kappa_{H,\rcs}\cdot|\bar{\pi}^{(i)}_{k}|$. Therefore, for each $1\leq k\leq k^{\star}$, $|\pi^{(i)}_{k}|\leq O(\kappa_{H,\rcs})^{k-1}\cdot |\pi^{(i)}|$. For each $1\leq k\leq k^{\star}-1$, $|\bar{\pi}^{(i)}_{k}|\leq O(\kappa_{H,\rcs})^{k-1}\cdot |\pi^{(i)}|$

For each $1\leq k\leq \bar{k}$, we have
\[
|H^{(i)}_{k}|\leq O(t\cdot\kappa_{H,\size}\cdot\phi)^{k-1}\cdot |G^{(i)}|,
\]
because $|H^{(i)}_{k+1}|\leq O(t\cdot\kappa_{H,\size}\cdot\phi\cdot|H^{(i)}_{k}|)$ holds for each $1\leq k\leq \bar{k}-1$.

\paragraph{The Recourse of Pairwise Covers at Time $i$.} Let $k^{\star}$ be the minimum level at which the dynamic certified-ED algorithm is reinitialized. 

Then at each level $1\leq k\leq k^{\star}-1$, the dynamic certified-ED algorithm performs a normal update step at level $k$ and time $i$, 
by \Cref{thm:DynamicED}, we have
\begin{align*}
\recourse({\cal N}^{(i-1)}_{k}\to{\cal N}^{(i)}_{k})&\leq \lambda_{\rt,\prune}(t)\cdot t\cdot \kappa_{\PC,\omega}\cdot\kappa_{\dynED,\gamma}\cdot|\pi^{(i)}_{k}|/\phi\\
&\leq \lambda_{\rt,\prune}(t)\cdot t\cdot\kappa_{\PC,\omega}\cdot\kappa_{\dynED,\gamma}\cdot 2^{O(\bar{k})}\cdot\kappa_{H,\rcs}^{\bar{k}}\cdot|\pi^{(i)}|/\phi.
\end{align*}

Now consider level $k^{\star}$. The dynamic certified-ED algorithm triggers a reinitialization step at level $k^{\star}$ and time $i$. Note that by the reinitialization condition, we have $|\pi^{(i)}_{k^{\star}}|\geq \phi\cdot |H^{(i)}_{k^{\star}}|$, 
which implies
\[
|H^{(i)}_{k^{\star}}|\leq O(\kappa_{H,\rcs})^{k^{\star}-1}\cdot |\pi^{(i)}|/\phi.
\]
Furthermore, by the definition of $k^{\star}$, the dynamic sparsifier algorithm at level $k^{\star}-1$ that maintains $H_{k^{\star}}$ performs a normal update step at time $i$, so we have
\begin{align*}
|H^{(i-1)}_{k^{\star}}|&\leq |H^{(i)}_{k^{\star}}| + |\pi^{(i)}_{k^{\star}}| \leq O(\kappa_{H,\rcs})^{k^{\star}-1}\cdot |\pi^{(i)}|/\phi.
\end{align*}

Therefore, for each level $k\geq k^{\star}$, $\recourse({\cal N}^{(i-1)}_{k}\to{\cal N}^{(i)}_{k})$ can be bounded by
\begin{align*}
\recourse({\cal N}^{(i-1)}_{k}\to{\cal N}^{(i)}_{k})&\leq \size({\cal N}^{(i-1)}_{k}) + \size({\cal N}^{(i)}_{k})\\
&\leq \omega^{(i-1)}_{k}\cdot |H^{(i-1)}_{k}| + \omega^{(i)}_{k}\cdot |H^{(i)}_{k}|\\
&\leq \omega^{(i-1)}_{k}\cdot |H^{(i-1)}_{k^{\star}}| + \omega^{(i)}_{k}\cdot |H^{(i)}_{k^{\star}}|\\
&\leq O(t\cdot\kappa_{\PC,\omega})\cdot O(\kappa_{H,\rcs})^{k^{\star}-1}\cdot|\pi^{(i)}|/\phi\\
&\leq 2^{O(\bar{k})}\cdot t\cdot \kappa_{\PC,\omega}\cdot\kappa_{H,\rcs}^{\bar{k}}\cdot |\pi^{(i)}|/\phi.
\end{align*}
Here the first inequality is because we can simply delete the whole ${\cal N}^{(i-1)}_{k}$ and then add the whole ${\cal N}^{(i)}_{k}$. 
The second inequality is by $\size({\cal N}^{(i-1)}_{k})\leq \omega^{(i-1)}_{k}\cdot |A^{(i-1)}_{k}|\leq \omega^{(i-1)}_{k}\cdot |H^{(i-1)}_{k}|$.
The third inequality is because we have $|H^{(i)}_{k+1}|\leq O(t\cdot\kappa_{H,\size}\cdot\phi\cdot |H^{(i)}_{k}|)\leq o(|H^{(i)}_{k}|)$ for all $i$ and $k$ (here we use the condition $\phi\leq o(1/t\cdot\kappa_{H,\size})$).

In conclusion, we have
\begin{align*}
\sum_{1\leq k\leq \bar{k}}\recourse({\cal N}^{(i-1)}_{k}\to{\cal N}^{(i)}_{k}) &\leq \lambda_{\rt,\prune}(t)\cdot t\cdot \kappa_{\PC,\omega}\cdot \kappa_{\dynED,\gamma}\cdot 2^{O(\lambda_{k})}\cdot \kappa^{\lambda_{k}}_{H,\rcs}\cdot |\pi^{(i)}|/\phi.
\end{align*}

\paragraph{The Initialization and Update Time.} For the initialization, at level $k$, the dynamic certified-EDs algorithm takes time, by \Cref{thm:DynamicED},
\[
|A^{(0)}_{k}|\cdot \poly(h_{k})\cdot n^{O(\epsilon)} = |H^{(0)}_{k}|\cdot\poly(h_{k})\cdot n^{O(\epsilon)},
\]
because $|A^{(0)}_{k}| = O(|H^{(0)}_{k}|)$. The dynamic sparsifier algorithm takes time, by \Cref{thm:NonHopReducingEmulator},
\begin{align*}
(|H^{(0)}_{k}| + |L^{(0)}_{k}|/\phi)\cdot \poly(\bar{h}_{k})\cdot n^{O(\epsilon)}/\phi
\leq |H^{(0)}_{k}|\cdot \poly(\bar{h}_{k})\cdot n^{O(\epsilon)}/\phi,
\end{align*}
because $|L^{(0)}_{k}|\leq O(\phi\cdot |H^{(0)}_{k}|)$. The total initialization time summing over all levels is
\[
|G^{(0)}|\cdot\poly(\max_{k}\{h_{k},\bar{h}_{k}\})\cdot n^{O(\epsilon)}/\phi.
\]

About the update time for handling batch $\pi^{(i)}$, let $k^{\star}$ be the minimum level at which the dynamic certified-ED algorithm is reinitialized. Then for each level $1\leq k\leq k^{\star}-1$, the dynamic certified-ED algorithm needs update time
$|\pi^{(i)}_{k}|\cdot n^{O(\epsilon)}\cdot \poly(h_{k})/\phi$, and 
the dynamic sparsifier algorithm needs update time
$|\bar{\pi}^{(i)}_{k}|\cdot \poly(\bar{h}_{k})\cdot n^{O(\epsilon)}/\phi^{2}$.
For each level $k\geq k^{\star}$, the dynamic certified-ED and sparsifier algorithms are both reinitialized with running time 
\[
|H^{(i)}_{k}|\cdot\poly(\max\{h_{k},\bar{h}_{k}\})\cdot n^{O(\epsilon)}/\phi.
\]
Finally, combining $|\pi^{(i)}_{k}|,|\bar{\pi}^{(i)}_{k}|\leq O(\kappa_{H,\rcs})^{\bar{k}}\cdot |\pi^{(i)}|$ for each $k\leq k^{\star}-1$ and $|H^{(i)}_{k}|\leq o(|H^{(i)}_{k^{\star}}|)\leq |\pi^{(i)}_{k^{\star}}|/\phi$, the total update time is
\[
|\pi^{(i)}|\cdot \poly(\max_{k}\{h_{k},\bar{h}_{k}\})\cdot n^{O(\epsilon)}/\phi^{2}.
\]

\end{proof}

%% file: 10-low_distance_oracle.tex
\section{Dynamic Distance Oracles for the Low Distance Regime}
\label{sect:OracleShort}

Providing the dynamic length-constrained expander hierarchy, we are ready to design dynamic distance oracles. However, the oracles in this section only work for low distance because it has initialization time and update time depending on $\poly(h)$.

\begin{theorem}
Let $G$ be a dynamic graph under $t$ batches of updates $\pi^{(1)},\pi^{(2)},...,\pi^{(t)}$ of edge insertions/deletions and isolated vertex insertions/deletions. Given a parameter $h$, there is an algorithm that maintains an oracle ${\cal O}$ s.t. for each $0\leq i\leq t$, the oracle ${\cal O}^{(i)}$ supports the following query. Given two vertices $u,v\in V(G)$, ${\cal O}^{(i)}$ will declare either $\dist_{G^{(i)}}(u,v)> h$ or $\dist_{G^{(i)}}(u,v)\leq \lambda_{\query,\alpha}\cdot h$, where
\[
\lambda_{\query,\alpha} = 2^{O(\lambda_{k})}\cdot(\lambda_{H,\low}\cdot\lambda_{H,\up})^{\lambda_{k}}\cdot (\lambda_{\rt,h}(t)\cdot\lambda_{\init,h}\cdot\lambda_{\dynED,h})^{\lambda_{k}}.
\]
The initialization time is 
$|G^{(0)}|\cdot\poly(h)\cdot n^{O(\epsilon)}$, the update time to handle $\pi^{(i)}$ is
$|\pi^{(i)}|\cdot \poly(h)\cdot n^{O(\epsilon)}$, and the query time is $O(1/\epsilon^{3})$.

Furthermore, if ${\cal O}^{(i)}$ declares $\dist_{G^{(i)}}(u,v)\leq \lambda_{\query,\alpha}\cdot h$, it can output a path $P$ connecting $u$ and $v$ in $G^{(i)}$ with length $\ell_{G^{(i)}}(P)\leq \lambda_{\query,\alpha}\cdot h$ in additional $2^{O(1/\epsilon^{4})}\cdot |P|$ time.

\label{thm:LowDistanceOracle}
\end{theorem}

\subsection{Fast Certified-ED Access Interfaces}
\label{sect:FastAccessToCertifiedED}
Before showing the query algorithm, we first implement some constant-time access interfaces of certified-EDs, which allow us to achieve sublogarithmic query time at last. Let $(C,L,{\cal N},{\cal R},\Pi_{{\cal R}\to G})$ be a certified-ED for some node-weighting $A\geq \mathds{1}(V(G))$ on $G$. In fact, our query algorithm will work on the restriction of ${\cal N}$ on $\mathds{1}(V(G))\subseteq A$, denoted by ${\cal N}_{\vertex} = {\cal N}_{\mid \mathds{1}(V(G))}$. Therefore, the interfaces below are also for ${\cal N}_{\vertex}$. As mentioned in \Cref{def:PairwiseCover}, equivalently ${\cal N}$ is a pairwise cover of vertex set $V(G)$.

\paragraph{Membership Queries.} Given a cluster $S\in{\cal N}_{\vertex}$ and a vertex $v\in A$, the interface $\Member(S,v)$ answers whether $v\in S$ or not. 

We assign a distinct cluster ID to each cluster $S\in{\cal N}_{\vertex}$, and we also assign a distinct (integral) clustering ID in range $[1,\omega]$ to each clustering ${\cal S}\in{\cal N}_{\vertex}$. For each virtual node $v\in A$, we store an array with length $\omega$, where the $k$-th element is the cluster ID of $S$ s.t. $v\in S$ and the clustering ${\cal S}$ containing $S$ has clustering ID $k$ (if there is no such $S$, then put a symbol $\perp$). For each vertex insertion or deletion update to ${\cal N}_{\vertex}$, it takes constant time to update this structure.

To answer $\Member(S,v)$, we look at the clustering ID of ${\cal S}\ni S$, say $k$. Then $v\in S$ if and only if the cluster ID in the $k$-th position of the array of $v$ is the same as the cluster ID of $S$. The query time is $O(1)$.

\paragraph{$L$-Vertex Queries.} Given a cluster $S\in{\cal N}_{\vertex}$, the interface $\textsc{L-vertex}(S)$ returns an arbitrary vertex $v\in S\cap L$, or declares $S\cap L$ is empty.

For each cluster $S\in{\cal N}_{\vertex}$, we store all $L$-vertices in $S$ in a binary search tree. The query time is trivially $O(1)$ since we can return the $L$-vertex at the root of the BST.

Updating these BSTs takes $n^{O(\epsilon)}$ time for each unit update to the certified-EDs, which is affordable to our update algorithm. First, for each vertex insertion or deletion to ${\cal N}_{\vertex}$, updating the BST takes $O(\log n)$ time. Second, for each $L$-vertex insertion, it takes $\tilde{O}(\omega) = n^{O(\epsilon)}$ time to update the BST, because all certified-EDs we maintain have ${\cal N}$ with width at most $n^{O(\epsilon)}$. There is no $L$-vertex deletion, because $L$ is updated incrementally during our maintenance of certified-EDs.

\subsection{Proof of \Cref{thm:LowDistanceOracle}}
\begin{proof}

This data structure also exploits dynamic expander hierarchy in \Cref{thm:ExpanderHierarchy}. Precisely, we initialize and maintain an expander hierarchy $({\cal H},{\cal D})$ with $\bar{k} = \lambda_{k}$ levels by applying \Cref{thm:ExpanderHierarchy} on the dynamic graph $G$ with parameters $\kappa_{k}=\lambda_{k}$ and $\phi = \kappa_{\phi}=n^{\epsilon^{2}}$.
The additional length parameters at each level $1\leq k\leq \bar{k}$ are $\hat{h}_{k} = 4\cdot h_{k}$ (here we change the notation by passing $\hat{h}_{k}$ to the parameter $h_{k}$ in \Cref{thm:ExpanderHierarchy}) and $\bar{h}_{k} = (16\cdot\lambda_{\Gdiam} + 3)\cdot h_{k}$, where $h_{k} = ((16\cdot\lambda_{\Gdiam} + 3)\cdot\lambda_{H,\up})^{k-1}\cdot h$ for some parameter $\lambda_{\Gdiam} = O(\lambda_{\rt,h}(t)\cdot\lambda_{\init,h}\cdot\lambda_{\dynED,h})$ with sufficiently large hidden constant. 

The hierarchy $({\cal H},{\cal D})$ has properties listed in \Cref{thm:ExpanderHierarchy}. For clarity, we now restate the properties we will use in the proof. For each level $1\leq k\leq \bar{k}$, let $H_{k}$ be the vertex spasifier and let $(C_{k},L_{k},{\cal N}_{k},{\cal R}_{k},\Pi_{{\cal R}_{k}\to H_{k}})$ be the certified-ED of some node-weighting $A_{k}$ on $H_{k}$. The vertex sparsifiers satsify that $H_{1} = G$ and for each $1\leq k\leq \bar{k}-1$, 
\begin{itemize}
\item $H_{k+1}$ is a $(\lambda_{H,\low},\lambda_{H,\up},\bar{h}_{k})$-sparsifier of $L_{k}$ on $H_{k}$,
\item $H_{k+1}$ has size $|H_{k+1}|\leq O(t\cdot\kappa_{H,\size}\cdot \phi\cdot|H_{k}|)$.
\item Given a path $P_{H_{k+1}}$ on $H_{k+1}$ connecting some vertices $u,v\in \bar{T}_{k}$, there is an algorithm that computes a $u$-$v$ path $P_{H_{k}}$ on $H_{k}$ with $\ell_{H_{k}}(P_{H_{k}})\leq \lambda_{H,\low}\cdot \ell_{H_{k+1}}(P_{H_{k+1}})$ and $|P_{H_{k}}|\geq |P_{H_{k+1}}|/2$. The path-reporting time is $|P_{H_{k+1}}|\cdot 2^{O(1/\epsilon^{4})}$.
\end{itemize}
The certified-EDs satisfy that, for each $1\leq k\leq \bar{k}$,
\begin{itemize}
\item $L_{k}$ is a landmark set of $C_{k}$ on $H_{k}$ with distortion $\sigma_{k} = O(\lambda_{\dynED,h}\cdot \hat{h}_{k}/\kappa_{\sigma})$.

\item ${\cal N}_{k}$ is a $b_{k}$-distributed $(h_{\cov,k},\omega_{k},h_{\diam,k})$-neighborhood cover of some node-weighting $A_{k}\geq \deg_{H_{k}} + \mathds{1}(H_{k})$ on $H_{k}-C_{k}$ with
\[
b_{k} = O(t),\ h_{\cov,k} = \hat{h}_{k} = 4\cdot h_{k},\text{ 
and }\omega_{k} = O(t\cdot\kappa_{\PC,\omega}).
\]
In fact, our query algorithm will work on ${\cal N}_{k,\vertex} = ({\cal N}_{k})_{\mid \mathds{1}(V(H_{k}))}$, i.e. the restriction of ${\cal N}_{k}$ on vertices $V(H_{k})$. Obviously, ${\cal N}_{k,\vertex}$ the same quality parameters $b_{k}, h_{\cov,k}$ and $\omega_{k}$.

\item The routers in ${\cal R}_{k}$ are maintained by \Cref{thm:Router} under at most $f_{k} = O(t)$ batched updates, and the embedding $\Pi_{{\cal R}_{k}\to H_{k}}$ has length $h_{\emb,k} = O(\lambda_{\init,h}\cdot \lambda_{\dynED,h}\cdot \hat{h}_{k})$. We note that this property is only for showing that, by \Cref{lemma:CertifiedEDToDist}, each cluster $S\in{\cal N}_{k,\vertex}$ has diameter on $H_{k}$ at most 
\[
h_{\Gdiam,k} = \lambda_{\rt,h}(f_{k})\cdot h_{\emb,k} \leq \lambda_{\Gdiam}\cdot h_{k}
\]
where the last inequality holds by our choice of $\lambda_{\Gdiam}$.
\end{itemize}
\begin{claim}
At each time $0\leq i\leq t$, there exists a level $k_{0}\leq \bar{k}$ (depending on $i$) s.t. $C_{k_{0}}$ is empty.
\label{claim:EmptyLevel}
\end{claim}
\begin{proof}
Because we have $|H^{(i)}_{k+1}|\leq O(t\cdot \kappa_{H,\size}\cdot \phi\cdot |H^{(i)}_{k}|)\leq |H^{(i)}_{k}|/n^{\Theta(\epsilon^{2})}$ for each time $i$ and level $k$, because $t\cdot\kappa_{H,\size} = n^{O(\epsilon^{4})}$ and $\phi = 1/n^{\epsilon^{2}}$. Since the number of levels is $\lambda_{k} = O(1/\epsilon^{2})$ (with sufficiently large hidden constant factor), there exists $k_{0}$ s.t. the sparsifier $H^{(i)}_{k_{0}+1}$ is empty, which means $L^{(i)}_{k_{0}}$ and $C^{(i)}_{k_{0}}$ are empty.
\end{proof}
The above claim says that the hierarchy has at most $\lambda_{k} = O(1/\epsilon^{2})$ nontrivial levels over all time. In what follows, we let $k^{(0)}$ (depends on the current time $i$) denote the minimum level s.t. $C_{k_{0}}$ is empty.

\paragraph{The Query Algorithm}

We now describe a query algorithm that given two vertices $u,v\in V(G)$, outputs either \textsc{Close} or \textsc{Far}, where \textsc{Close} means $\dist_{G}(u,v)\leq \lambda_{\query,\alpha}\cdot h$, while \textsc{Far} means $\dist_{G}(u,v)>h$. 

The algorithm is bottom-up. When we are at level $k$, the input is two vertices $u_{k},v_{k}\in V(H_{k})$. Initially, we start from level $1$ and let the input be $u_{1} = u$ and $v_{1} = v$ (trivially $u_{1},v_{1}\in V(H_{1})$ because $H_{1} = G$). We do the following at level $k$.
\begin{enumerate}
\item We consider the ball cover ${\cal S}_{v_{k}}$ of $v_{k}$ in ${\cal N}_{k,\vertex}$ (see \Cref{def:DistributedNC} to recall the definition of ball cover). If there is a cluster $S\in {\cal S}_{v_{k}}$ containing $u_{k}$ (by querying $\textsc{Membership}(S,u_{k})$), we terminate the algorithm with output \textsc{Close}.

\item Otherwise, we have $u_{k}\notin \bigcup_{S\in {\cal S}_{v_{k}}}S$. If we have reached the level $k_{0}$ with empty $C_{k_{0}}$, we terminate the algorithm with output \textsc{Far}.

\item Otherwise, now we are at some level $k<k_{0}$, and we will pick the input $u_{k+1}$ and $v_{k+1}$ of the next level as follows. We select $v_{k+1}$ to be an arbitrary $L_{k}$-vertex in some cluster $S\in {\cal S}_{v_{k}}$ (querying $\textsc{L-vertex}(S)$ for each $S\in{\cal S}_{v_{k}}$). If there is no such $v'_{k}$, we terminate the algorithm with answer \textsc{Far}. The vertex $u_{k+1}$ is picked by an analogous algorithm. Note that $u_{k+1},v_{k+1}\in L_{k}\subseteq V(H_{k+1})$, and we proceed to the next level $k+1$. 
\end{enumerate}

\paragraph{The Query Time and Correctness.} The correctness of the query algorithm is shown by \Cref{lemma:QueryCorrectness1,lemma:QueryCorrectness2}. Consider the running time at each level $k$. Step 1 takes $O(b_{k})$ time, because the ball cover ${\cal S}_{v_{k}}$ has at most $b_{k}$ clusters and the membership query of $u$ and each cluster $S\in{\cal S}_{v_{k}}$ takes $O(1)$ time. Step 2 takes constant time trivially. Step 3 takes $O(b_{k})$ time because for each cluster $S$, selecting a $L_{k}$-vertex in $S$ (or declare there is no $L_{k}$-vertex) takes constant time. In conclusion, the total query time is
\[
O(\sum_{1\leq k\leq \bar{k}}b_{k}) = O(\bar{k}\cdot t) = O(\lambda_{k}\cdot t) = O(1/\epsilon^{3}).
\]

\begin{lemma}
If the input vertices $u$ and $v$ have $\dist_{G}(u,v)\leq h$, the query algorithm will output \textsc{Close}.
\label{lemma:QueryCorrectness1}
\end{lemma}
\begin{proof}
First we will prove, by induction, that each level $k$ reached by the algorithm has $\dist_{H_{k}}(u_{k},v_{k})\leq h_{k}$. This statement trivially holds for the first level. For each level $k\geq 2$ reached by the algorithm, suppose the statement holds for the level $k-1$. Then we have
\begin{align*}
\dist_{H_{k}}(u_{k+1},v_{k+1}) &\leq \dist_{H_{k}}(u_{k+1},u_{k}) + \dist_{H_{k}}(u_{k},v_{k}) + \dist_{H_{k}}(v_{k},v_{k+1})\\
&\leq 2\cdot h_{\Gdiam,k} + h_{k}\leq \bar{h}_{k}.
\end{align*}

Because $u_{k+1},v_{k+1}\in L_{k}$ and $H_{k+1}$ is a $(\lambda_{H,\low},\lambda_{H,\up},\bar{h}_{k})$-sparsifier of $L_{k}$ on $H_{k}$, we have 
\[
\dist_{H_{k+1}}(u_{k+1},v_{k+1})\leq \lambda_{H,\up}\cdot\dist_{H_{k}}(u_{k+1},v_{k+1})\leq h_{k+1}.
\]

Now we show that the algorithm will not terminate with output \textsc{Far}. If the algorithm reach the level $k=k_{0}$ with empty $C_{k_{0}}$, $\dist_{H_{k}-C_{k}}(u_{k},v_{k}) = \dist_{H_{k}}(u_{k},v_{k})\leq h_{k}\leq h_{\cov,k}$ and $u_{k}$ must be inside the ball cover ${\cal S}_{v_{k}}$, and the algorithm will terminate with output \textsc{Close}. For other level $k<k_{0}$, if the algorithm reaches level $k$ but does not terminate with output \textsc{Close} at step 1, then $\dist_{H_{k}-C_{k}}(u_{k},v_{k})>h_{\cov,k}\geq 4h_{k}$ and $u_{k+1}$ and $v_{k+1}$ can be found successfully by the following reason. Consider the $u_{k}$-to-$v_{k}$ shortest path $P$ on $G$. We know $\ell_{H_{k}}(P)\leq h_{k}$ but $\ell_{H_{k}-C_{k}}(P)\geq 4\cdot h_{k}$. Then there must be a $C_{k}$-vertex $v'_{k}\in P$ with $\dist_{H_{k}-C_{k}}(v_{k},v'_{k})\leq h_{k}$ by taking the $C_{k}$-vertex on $P$ closest to $v_{k}$. By the definition of landmarks, there exists a $L_{k}$-vertex $v_{k+1}$ s.t. $\dist_{H_{k}-C_{k}}(v'_{k},v_{k+1})\leq \sigma_{k}$. Therefore, $\dist_{H_{k}-C_{k}}(v_{k},v_{k+1})\leq h_{k} + \sigma_{k}\leq 4\cdot h_{k} = h_{\cov,k}$, because $\sigma_{k}\leq h_{k}$ from $\lambda_{\dynED,h} = 2^{O(t)} = 2^{O(1/\epsilon)}$ and $\kappa_{\sigma} = n^{\epsilon^{4}}$. The same argument holds for $u_{k+1}$.

\end{proof}

\begin{lemma}
If the query algorithm outputs \textsc{Close}, then $\dist_{G}(u,v)\leq \lambda_{\query,\alpha}\cdot h$.
\label{lemma:QueryCorrectness2}
\end{lemma}
\begin{proof}
Suppose the algorithm terminates with output \textsc{Close} at level $k^{\star}$. Then by the end condition,
\[
\dist_{H_{k^{\star}}}(u_{k^{\star}},v_{k^{\star}})\leq h_{\Gdiam,k^{\star}} = \lambda_{\Gdiam}\cdot h_{k^{\star}}.
\]
For each level $k<k^{\star}$, because $H_{k+1}$ is a $(\lambda_{H,\low},\lambda_{H,\up},\bar{h}_{k})$-sparsifier of $L_{k}$ on $H_{k}$ and $u_{k+1},v_{k+1}\in L_{k}$, we have
\[
\dist_{H_{k}}(u_{k+1},v_{k+1})\leq \lambda_{H,\low}\cdot \dist_{H_{k+1}}(u_{k+1},v_{k+1}),
\]
which implies 
\begin{align*}
\dist_{H_{k}}(u_{k},v_{k}) &\leq \dist_{H_{k}}(u_{k},u_{k+1}) + \dist_{H_{k}}(u_{k+1},v_{k+1}) + \dist_{H_{k}}(v_{k+1},v_{k})\\
&\leq 2\cdot h_{\Gdiam,k} + \dist_{H_{k}}(u_{k+1},v_{k+1})\\
&\leq 2\cdot \lambda_{\Gdiam}\cdot h_{k} + \lambda_{H,\low}\cdot \dist_{H_{k+1}}(u_{k+1},v_{k+1}).
\end{align*}

Combining all together, we have 
\begin{align*}
\dist_{G}(u,v)\leq O(\lambda_{H,\low}^{k^{\star}-1}\cdot \lambda_{\Gdiam}\cdot h_{k^{\star}})\leq 2^{O(\lambda_{k})}\cdot(\lambda_{H,\low}\cdot\lambda_{H,\up})^{\lambda_{k}}\cdot \lambda_{\Gdiam}^{\lambda_{k}}\cdot h
\end{align*}
\end{proof}

\paragraph{The Initialization and Update Time.} Note that the oracle ${\cal O}$ is exactly the length-constrained hierarchy $({\cal H},{\cal D})$, so by \Cref{thm:ExpanderHierarchy} and the fact that $\phi = \kappa_{\phi} = n^{\epsilon^{2}}$ and
\[
\max_{k}\{\hat{h}_{k},\bar{h}_{k}\} = O(\lambda_{\Gdiam}\cdot\lambda_{H,\up})^{\lambda_{k}} = n^{O(\epsilon)}\cdot h,
\]
the initialization time is $|G^{(0)}|\cdot\poly(h)\cdot n^{O(\epsilon)}$,
and the update time to handle $\pi^{(i)}$ is
$|\pi^{(i)}|\cdot \poly(h)\cdot n^{O(\epsilon)}$.

\paragraph{Path Reporting.} Given two vertices $u,v\in V(G)$, if the query algorithm returns \textsc{Close}, we can compute a path $P$ connecting $u$ and $v$ on $G$ as follows. More precisely, we consider the levels from top to bottom (i.e. from level $k^{\star}$ to $1$), and at each level $k$, we will generate a path $P_{H_{k}}$ connecting $u_{k}$ and $v_{k}$ on $H_{k}$. Then the path $P_{H_{1}}$ is what we want.

At the top level $k^{\star}$, because the output is \textsc{Close}, there is a cluster $S_{\vertex}\in{\cal N}_{k^{\star},\vertex}$ containing $u_{k^{\star}}$ and $v_{k^{\star}}$. Let $S\in{\cal N}_{k^{\star}}$ be the original cluster corresponding to $S_{\vertex}$, and let $u_{k^{\star},\virtual},v_{k^{\star},\virtual}\in S$ be arbitrary virtual nodes of $u_{k^{\star}}$ and $v_{k^{\star}}$ respectively. By \Cref{thm:RouterPathReport} on the router $R$ corresponding to $S$, we can obtain a $u_{k^{\star}}$-$v_{k^{\star}}$ path $P_{R}$ on $R$. Substituting each edge on $P_{R}$ with its embedding path in $\Pi_{{\cal R}_{k^{\star}}\to H_{k^{\star}}}$ lead to the desired $P_{H_{k^{\star}}}$. We have $\ell_{H_{k^{\star}}}(P_{H_{k^{\star}}})\leq \lambda_{\rt}(f_{k^{\star}})\cdot h_{\emb,k^{\star}}\leq \lambda_{\Gdiam}\cdot h_{k^{\star}}$, because $|P_{R}|\leq \lambda_{\rt}(f_{k^{\star}})$ and $\Pi_{{\cal R}_{k^{\star}}\to H_{k^{\star}}}$ has length $h_{\emb,k^{\star}}$.

At each level $1\leq k\leq \bar{k}-1$, note that we have generated a path $P_{H_{k+1}}$ connecting $u_{k+1}$ and $v_{k+1}$ on $H_{k+1}$ with length $\lambda_{\Gdiam}\cdot h_{k+1}$. By the guarantees of vertex sparsifiers, we can unfold $P_{H_{k+1}}$ to obtain a $u_{k+1}$-$v_{k+1}$ path $P'_{H_{k}}$ on $H_{k}$ with $\ell_{H_{k}}(P'_{H_{k}}) \leq \lambda_{H,\low}\cdot \ell_{H_{k+1}}(P_{H_{k+1}})$ and $|P'_{H_{k}}|\geq |P_{H_{k+1}}|/2$. Next, because there is some cluster in ${\cal N}_{k,\vertex}$ containing both $u_{k}$ and $u_{k+1}$, similar to the algorithm at level $k^{\star}$, we can obtain a $u_{k}$-$u_{k+1}$ path $P_{H_{k},u}$ on $H_{k}$ with $\ell_{H_{k}}(P_{H_{k},u})\leq \lambda_{\Gdiam}\cdot h_{k}$. Similarly, we can obtain a $v_{k}$-$v_{k+1}$ path $P_{H_{k},v}$ on $H_{k}$ with $\ell_{H_{k}}(P_{H_{k},v})\leq \lambda_{\Gdiam}\cdot h_{k}$. Finally, $P_{H_{k}}$ is the concatenation of $P_{H_{k},u}$, $P'_{H_{k}}$ and $P_{H_{k},v}$ with guarantee $\ell_{H_{k}}(P_{H_{k}})\leq 2\cdot\lambda_{\Gdiam}\cdot h_{k} + \lambda_{H,\low}\cdot\ell_{H_{k+1}}(P_{H_{k+1}})$ and $|P_{H_{k}}|\geq |P_{H_{k+1}}|/2$.

Following the analysis in \Cref{lemma:QueryCorrectness2}, the final path $P_{G} = P_{H_{1}}$ on $G = H_{1}$ is a $u$-$v$ path with length $\ell_{G}(P_{G})\leq \lambda_{\query,\alpha}\cdot h$. The total path-reporting time can be bounded by $\sum_{1\leq k\leq \bar{k}} |P_{H_{k}}|\cdot 2^{O(1/\epsilon^{4})}\leq 2^{O(1/\epsilon^{4})}\cdot|P_{G}|$, because $|P_{H_{k}}|\leq 2^{O(\lambda_{k})}\cdot|P_{G}|$ and $\lambda_{k} = O(1/\epsilon^{2})$.

\end{proof}

%% file: 9-hop_reducing_emulator.tex
\section{Dynamic Length-Reducing Emulators}
\label{sect:DynHopEmu}

In this section, we will show a dynamic length-reducing emulator algorithm. It will be used in \Cref{sect:Oracle} to remove the $\poly(h)$-dependency of the running time of \Cref{thm:LowDistanceOracle}.

\begin{definition}[Length-Reducing Emulators]
Let $G$ be a graph, an $(\alpha_{\low},\alpha_{\up},h)$-length-reducing emulator $Q$ is an unweighted graph with the following properties.
\begin{itemize}
\item[(1)] $V(G)\subseteq V(Q)$.
\item[(2)] Each pair of vertices $u,v\in V(G)$ s.t. $\dist_{G}(u,v)\leq h$ has $\dist_{Q}(u,v)\leq \alpha_{\up}$.
\item[(3)] Each pair of vertices $u,v\in V(G)$ has $\dist_{G}(u,v)\leq \alpha_{\low}\cdot h\cdot\dist_{Q}(u,v)$.
\end{itemize}
\end{definition}

\begin{theorem}
Let $G$ be a dynamic graph with parameters $h$ and $\phi$ under $t$ batches of updates $\pi^{(1)},\pi^{(2)},...,\pi^{(t)}$ of edge insertions/deletions and isolated vertex insertions/deletions. There is an algorithm that initializes and maintains a $(\lambda_{Q,\low},\lambda_{Q,\up},h)$-length-reducing-emulator $Q$ of $G$ with
\begin{align*}
\lambda_{Q,\low} &= \lambda_{\rt,h}(t)\cdot\lambda_{\init,h}\cdot\lambda_{\dynED,h}\cdot 2^{O(\lambda_{k})}\cdot (\lambda_{H,\up}\cdot\lambda_{H,\low})^{\lambda_{k}} = 2^{\poly(1/\epsilon)},\\
\lambda_{Q,\up} &= 4\cdot\lambda_{k} = O(1/\epsilon^{2}). 
\end{align*}
The algorithm further guarantee the following.
\begin{itemize}
\item The size of $Q$ is bounded by $|Q|\leq \kappa_{Q,\size}\cdot |G|$, where $\kappa_{Q,\size} = O(\lambda_{k}\cdot t\cdot\kappa_{\PC,\omega}) = n^{O(\epsilon^{4})}$. In additional, it can also be bounded by $|Q|\leq n^{O(\epsilon^{2})}\cdot |V(G)|$.
\item The recourse from $Q^{(i-1)}$ to $Q^{(i)}$ is $O(\kappa_{Q,\rcs}\cdot|\pi^{(i)}|)$, where $\kappa_{Q,\rcs} = \kappa_{\EH,\PC,\rcs}\cdot\kappa_{\phi}$
and $\kappa_{\phi} = n^{\epsilon^{2}}$. 
\item The initialization time is 
$|G^{(0)}|\cdot\poly(h)\cdot n^{O(\epsilon)}$
and the update time is
$|\pi^{(i)}|\cdot \poly(h)\cdot n^{O(\epsilon)}$.
\item Given a path $P_{Q}$ on $Q$ connecting some vertices $u,v\in V(G)$, the algorithm can compute a $u$-$v$ path $P_{G}$ on $G$ with $\ell_{G}(P_{G})\leq \lambda_{Q,\low}\cdot h\cdot \ell_{Q}(P_{Q})$ and $|P_{G}|\geq |P_{Q}|/2^{O(\lambda_{k})}$ in time $2^{O(1/\epsilon^{4})}\cdot|P_{G}|$.
\end{itemize}
\label{thm:HopReducingEmulator}
\end{theorem}

\begin{proof}

Our construction of the emulator $Q$ is based on the length-constrained expander hierarchy. Precisely, we initialize and maintain a length-constrained expander hierarchy $({\cal H}, {\cal D})$ with $\bar{k} = \lambda_{k}$ levels. Recall $\lambda_{k} = O(1/\epsilon^{2})$ is a sufficiently large global parameter and we define a global parameter $\kappa_{\phi} = n^{\epsilon^{2}}$. We apply \Cref{thm:ExpanderHierarchy} on $G$ with parameters $\phi = 1/\kappa_{\phi} = 1/n^{\epsilon^{2}}$ and $\bar{k} = \lambda_{k}$. The additional length parameter for each level $1\leq k\leq \bar{k}$ is $\hat{h}_{k} = 4\cdot h_{k}$ (here we change the notation by passing $\hat{h}_{k}$ to the parameter $h_{k}$ in \Cref{thm:ExpanderHierarchy}) and $\bar{h}_{k} = 2\cdot h_{k}$, where $h_{k} = h\cdot(2\cdot\lambda_{H,\up})^{k-1}$. 

For each level $1\leq k\leq \bar{k}$, let $H_{k}$ be the vertex spasifier and let $(C_{k},L_{k},{\cal N}_{k},{\cal R}_{k},\Pi_{{\cal R}_{k}\to H_{k}})$ be the certified-ED of some node-weighting $A_{k}$ on $H_{k}$. We now list the properties we will exploit in the remaining proof as follows. Regarding the sparsifiers, $H_{1} = G$, and for each $1\leq k\leq \bar{k}-1$,
\begin{itemize}
\item $H_{k+1}$ is a $(\lambda_{H,\low},\lambda_{H,\up},\bar{h}_{k})$-sparsifier of $V(H_{k+1})\cap V(H_{k})$ on $H_{k}$ s.t. $L_{k}\subseteq V(H_{k+1})\cap V(H_{k})$ (in particular, $H_{k+1}$ is a sparsifier of $L_{k}$ on $H_{k}$),
\item $H_{k+1}$ has size $|H_{k+1}|\leq O(t\cdot\kappa_{H,\size}\cdot \phi\cdot|H_{k}|)$ and number of vertices $|V(H_{k+1})|\leq n^{O(\epsilon^{4})}\cdot |V(H_{k})|$,
\item $V(H_{k})\cap V(H_{k+1}) = V(H_{k})\cap \bigcup_{k+1\leq k'\leq\bar{k}}V(H_{k'})$.
\item Given a path $P_{H_{k+1}}$ on $H_{k+1}$ connecting some vertices $u,v\in V(H_{k+1})\cap V(H_{k})$, there is an algorithm that computes a $u$-$v$ path $P_{H_{k}}$ on $H_{k}$ with $\ell_{H_{k}}(P_{H_{k}})\leq \lambda_{H,\low}\cdot \ell_{H_{k+1}}(P_{H_{k+1}})$ and $|P_{H_{k}}|\geq |P_{H_{k+1}}|/2$. The path-reporting time is $|P_{H_{k+1}}|\cdot 2^{O(1/\epsilon^{4})}$.
\end{itemize}
The certified-EDs satisfy that, for each $1\leq k\leq \bar{k}$,
\begin{itemize}
\item $L_{k}$ is a landmark set of $C_{k}$ on $H_{k}$ with distortion $\sigma_{k}= O(\lambda_{\dynED,h}\cdot \hat{h}_{k}/\kappa_{\sigma})$. At the top level, $C_{\bar{k}}$ and $L_{\bar{k}}$ are empty.
\item ${\cal N}_{k}$ is a $b_{k}$-distributed $(h_{\cov,k}, h_{\sep,k},\omega_{k})$-pairwise cover of $A_{k}\geq \deg_{H_{k}} + \mathds{1}(V(H_{k}))$ on $H_{k}-C_{k}$ with
\[
b_{k} = O(t),\ h_{\cov,k} = \hat{h}_{k} = 4\cdot h_{k},\ \omega_{k} = O(t\cdot\kappa_{\PC,\omega}).
\]
\item The routers in ${\cal R}_{k}$ are maintained by \Cref{thm:Router} under at most $f_{k} = O(t)$ batched updates.
\item The embedding $\Pi_{{\cal R}_{k}\to H_{k}}$ has length $h_{\emb,k} = O(\lambda_{\init,h}\cdot\lambda_{\dynED,h}\cdot \hat{h}_{k})$.
\end{itemize}

\paragraph{Construction of $Q$.} For each level $1\leq k\leq \bar{k}$, let ${\cal N}_{\vertex,k} = ({\cal N}_{k})_{\mid \mathds{1}(V(H_{k}))}$ by restricting ${\cal N}_{k}$ on vertices $V(H_{k})$. Note that ${\cal N}_{\vertex,k}$ is a pairwise cover of $T_{k}$ with the same quality parameters as those of ${\cal N}_{k}$.

For each cluster $S\in{\cal N}_{\vertex,k}$, let $Q^{\star}_{S}$ be the star graph with $V(Q^{\star}_{S}) = S \cup\{v_{S}\}$, where $v_{S}$ is the \emph{artificial center} of $S$ (which is a newly created fresh vertex) and we call $S$ the \emph{branch-vertices} of $Q^{\star}_{S}$. The edges of 
$Q^{\star}_{S}$ are $E(Q^{\star}_{S}) = \{(v_{S},v)\mid v\in S\}$. Then we construct $Q$ as
\[
Q:=\bigcup_{1\leq k\leq \bar{k},S\in{\cal N}_{\vertex,k}}Q^{\star}_{S}.
\]
Note that, when taking the union of $Q^{\star}_{S}$, we use the global vertex identifiers to identify vertices. We emphasize that $Q$ is an unweighted graph and each edge in $E(Q)$ has unit length.

\begin{lemma}
\label{lemma:EmulatorCorrectness}
The graph $Q$ is a $(\alpha_{\low},\alpha_{\up},h)$-length-reducing emulator of $G$ with
\[
\alpha_{\up} = 4\bar{k},\ \alpha_{\low} = \lambda_{\rt,h}(t)\cdot\lambda_{\init,h}\cdot\lambda_{\dynED,h}\cdot 2^{O(\bar{k})}\cdot (\lambda_{H,\up}\cdot\lambda_{H,\low})^{\bar{k}}.
\]
\end{lemma}
\begin{proof}

Let $k_{0}$ be the minimum level s.t. $C_{k_{0}}$ and $L_{k_{0}}$ are empty in the up-to-date hierarchy. Note that $k_{0}$ must exist at all time (but may be different depending on the current time) by \Cref{claim:EmptyLevel}.

\

\noindent\textbf{Property (1).} The first property $V(G)\subseteq V(Q)$ is trivial, because at the first level, the union of star graphs of ${\cal N}_{\vertex,1}$ will contain $V(G)$.

\

\noindent\textbf{Property (2).} Consider two vertices $u,v\in V(G)$ with $\dist_{G}(u,v)\leq h$. We will prove $\dist_{Q}(u,v)\leq 4\bar{k}$ by induction. Let the induction hypothesis be, for each level $1\leq k\leq k_{0}$ and each $u,v\in T_{k}$ with $\dist_{H_{k}}(u,v)\leq h_{k}$, there is $\dist_{Q}(\tau_{k}(u),\tau_{k}(v))\leq 4(k_{0}-k+1)$.

For the base case $k=k_{0}$, note that $C_{k}$ is empty, so ${\cal N}_{\tmn,k}$ is a pairwise cover of $T_{k}$ on $H_{k}$ with cover radius $h_{\cov,k}= 4\cdot h_{k}$. Therefore, $u,v\in T_{k}$ will be contained by some cluster $S\in{\cal N}_{\tmn,k}$, and by our construction, the star graph $Q^{\star}_{S}$ has a path connecting $\tau_{k}(u)$ and $\tau_{k}(v)$ with length $2\leq 4(\bar{k}-k+1)$.

For an inductive step at level $k$ s.t. $1\leq k\leq k_{0}-1$, ${\cal N}_{\tmn,k}$ is a pairwise cover of $T_{k}$ on $H_{k}-C_{k}$ with cover radius $h_{\cov,k}\geq 4\cdot h_{k}$. If $\dist_{H_{k}-C_{k}}(u,v)\leq 4\cdot h_{k}$, then some cluster $S\in{\cal N}_{\tmn,k}$ will contain $u$ and $v$, and $\dist_{Q}(\tau_{k}(u),\tau_{k}(v))\leq \dist_{Q^{\star}_{S}}(\tau_{k}(u),\tau_{k}(v))\leq 2\leq 4(\bar{k}-k+1)$. 

Otherwise, we have $\dist_{H_{k}-C_{k}}(u,v)>4h_{k}$. Consider the shortest $u$-$v$ path $P$ in $H_{k}$. By \Cref{claim:HREProperty2}, we can find a vertex $u'\in P$ with $\dist_{H_{k}-C_{k}}(u,u')\leq 2h_{k}$ and also a landmark $w_{u}\in L_{k}$ with $\dist_{H_{k}-C_{k}}(u',w_{u})\leq h_{k}/100$. Similarly, there are such $v'\in P$ and $w_{v}\in L_{k}$ with $\dist_{H_{k}-C_{k}}(v,v')\leq 2h_{k}$ and $\dist_{H_{k}-C_{k}}(v',w_{v})\leq h_{k}/100$.

\begin{claim}
There is a vertex $u'\in P$ s.t. $\dist_{H_{k}-C_{k}}(u,u')\leq 2h_{k}$ and there is a landmark $w_{u}\in L_{k}$ with $\dist_{H_{k}-C_{k}}(w_{u},u')\leq h_{k}/100$.
\label{claim:HREProperty2}
\end{claim}
\begin{proof}

Let $P_{u}$ be the longest prefix of $P$ from $u$ with length at most $2h_{k}$ in $H_{k}-C_{k}$, and let $P^{+}_{u}$ be the prefix from extending $P_{u}$ along $P$ by one more edge. Observe that $C_{k}(P^{+}_{u}) \geq \ell_{H_{k}-C_{k}}(P^{+}_{u}) - \ell_{H_{k}}(P^{+}_{u})\geq 2h_{k}-h_{k}\geq h_{k}$. We can simply pick $u'\in P_{u}$ incident to some $e\in \supp(C_{k})\cap P^{+}_{u}$, and let $w_{u}\in L_{k}$ be the landmark of $u'$ with $\dist_{H_{k}-C_{k}}(w_{u},u')\leq \sigma_{k}\leq h_{k}/100$.

\end{proof}

Now, because $\dist_{H_{k}-C_{k}}(u,w_{u})\leq 4h_{k} \leq h_{\cov,k}$, there is a cluster $S\in{\cal N}_{\vertex,k}$ containing $u$ and $w_{u}$, which implies 
\[
\dist_{Q}(u,w_{u})\leq 2\text{ and similarly }\dist_{Q}(v,w_{v})\leq 2.
\]
Also, $\dist_{H_{k}}(w_{u},w_{v})\leq \dist_{H_{k}-C_{k}}(w_{u},u') + \dist_{H_{k}}(u',v') + \dist_{H_{k}-C_{k}}(v',w_{v})\leq 1.02 h_{k}\leq 2h_{k} = \bar{h}_{k}$. Since $H_{k+1}$ is always a $(\lambda_{H,\low},\lambda_{H,\up},\bar{h}_{k})$-sparsifier of $L_{k}$ on $H_{k}$, we have $\dist_{H_{k+1}}(w_{u},w_{v})\leq \lambda_{H,\up}\cdot \bar{h}_{k} = 2\cdot\lambda_{H,\up}\cdot h_{k} =  h_{k+1}$, so 
\[
\dist_{Q}(w_{u},w_{v})\leq 4(k_{0}-k)
\]
by the induction. Finally, we can conclude that $\dist_{Q}(u,v)\leq \dist_{Q}(u,w_{u}) + \dist_{Q}(w_{u},w_{v}) + \dist_{Q}(w_{v},v)\leq 2 + 4(k_{0}-k) + 2 = 4(k_{0}-k+1)$ as desired.

\

\noindent\textbf{Property (3).} Consider an arbitrary pair of vertices $u,v\in V(G)$. We will show $\dist_{G}(u,v)\leq \alpha_{\low}\cdot h\cdot \dist_{Q}(u,v)$ by induction.

We begin with some notations. For each level $k$, let
\[
Q_{k} = \bigcup_{S\in{\cal N}_{k}} Q^{\star}_{S},\text{ and }Q_{\geq k} = \bigcup_{k'\geq k,S\in{\cal N}_{k'}} Q^{\star}_{S}.
\]
Let the induction hypothesis be, for each level $1\leq k\leq\bar{k}$ and each pair of vertices $u,v\in V(H_{k})$, we have
\[
\dist_{H_{k}}(u,v)\leq \beta\cdot (2\cdot \lambda_{H,\up}\cdot\lambda_{H,\low})^{\bar{k}-k}\cdot h_{k}\cdot \dist_{Q_{\geq k}}(u,v).
\]
where $\beta = O(\lambda_{\rt,h}(t)\cdot\lambda_{\init,h}\cdot\lambda_{\dynED,h})$.

\underline{The Base Case.} For the base case $k = \bar{k}$, consider the $u$-$v$ shortest path on $Q_{\geq k}$, denoted by $P_{Q_{\geq k}}(u,v)$. Note that $Q_{\geq k} = Q_{k}$ in this case, and $Q_{k}$ is a union of stars whose branch-vertices are in $V(H_{k})$. If we partition $P_{Q_{\geq k}}(u,v)$ into subpaths with endpoints in $V(H_{k})$ but no internal vertex in $V(H_{k})$, then each subpath $P'$ (with endpoints $u',v'$) has length exactly $2$ and it is contained by some $Q^{\star}_{S}\subseteq Q_{k}$ for some $S\in {\cal N}_{\vertex,k}$ s.t. $u',v'\in S$. Thus, by \Cref{lemma:CertifiedEDToDist},
\begin{align*}
\dist_{H_{k}}(u',v')&\leq \lambda_{\rt,h}(f_{k})\cdot h_{\emb,k}\leq O(\lambda_{\rt,h}(t)\cdot\lambda_{\init,h}\cdot\lambda_{\dynED,h}\cdot h_{k}).
\end{align*}
and because the number of such subpaths $P'$ is $\dist_{Q_{\geq k}}(u,v)/2$, we have 
\[
\dist_{H_{k}}(u,v)\leq O(\lambda_{\rt,h}(t)\cdot\lambda_{\init,h}\cdot\lambda_{\dynED,h}\cdot h_{k}\cdot \dist_{Q_{\geq k}}(u,v))
\]
as desired.

\underline{The Inductive Step.} For an inductive step $1\leq k \leq \bar{k}-1$, again we consider the $u$-$v$ shortest path $P_{Q_{\geq k}}(u,v)$ on $Q_{\geq k}$. We partition $P_{Q_{\geq k}}(u,v)$ into (maximal) subpaths totally on $Q_{k}$ or totally on $Q_{\geq k+1}$.

\begin{claim}
$V(Q_{k})\cap V(Q_{\geq k+1}) = V(H_{k})\cap V(H_{k+1})$.
\label{claim:QcapQ}
\end{claim}
\begin{proof}
By the construction of $Q$, $V(Q_{k})$ is the union of $V(H_{k})$ and artificial centers of stars, and similarly $V(Q_{\geq k+1})$ is the union of $\bigcup_{k+1\leq k'\leq\bar{k}}V(H_{k'})$ and artificial centers. Because the artificial centers of $V(Q_{k})$ and $V(Q_{\geq k+1})$ are disjoint, we have $V(Q_{k})\cap V(Q_{\geq k+1}) = V(H_{k})\cap \bigcup_{k+1\leq k'\leq \bar{k}} V(H_{k'}) = V(H_{k})\cap V(H_{k+1})$ by the guarantee of sparsifiers.
\end{proof}

For each subpath $P'_{\geq k+1}$ (with endpoints $u',v'$) on $Q_{\geq k+1}$, we have $u',v'\in V(H_{k})\cap V(H_{k+1})$ by \Cref{claim:QcapQ}. Then, the induction hypothesis at level $k+1$ tells
\[
\dist_{H_{k+1}}(u',v')\leq \beta\cdot (2\cdot \lambda_{H,\up}\cdot\lambda_{H,\low})^{\bar{k}-k-1}\cdot h_{k+1}\cdot \dist_{Q_{\geq k+1}}(u',v')
\]
By property (2) in \Cref{def:Sparsifier} of the sparsifier $H_{k+1}$ and the fact that $u',v'\in V(H_{k})\cap V(H_{k+1})$, we have 
\begin{align*}
\dist_{H_{k}}(u',v')&\leq \lambda_{H,\low}\cdot\dist_{H_{k+1}}(u',v')\\
&\leq \beta\cdot\lambda_{H,\low}\cdot(2\cdot\lambda_{H,\up}\cdot\lambda_{H,\low})^{\bar{k}-k-1}\cdot h_{k+1}\cdot \dist_{Q_{\geq k+1}}(u',v')\\
&\leq \beta\cdot\lambda_{H,\low}\cdot(2\cdot\lambda_{H,\up}\cdot\lambda_{H,\low})^{\bar{k}-k-1}\cdot (2\cdot\lambda_{H,\up}\cdot h_{k})\cdot \dist_{Q_{\geq k+1}}(u',v')\\
&\leq \beta\cdot(2\cdot\lambda_{H,\up}\cdot\lambda_{H,\low})^{\bar{k}-k}\cdot h_{k}\cdot \dist_{Q_{\geq k}}(u',v'),
\end{align*}
where the last inequality is by $\dist_{Q_{\geq k+1}}(u',v') = \dist_{Q_{\geq k+1}}(P'_{\geq k+1}) = \dist_{Q_{\geq k}}(P'_{\geq k+1}) = \dist_{Q_{\geq k}}(u',v')$. 

Next, for each subpath $P'_{k}$ (with endpoints $u',v'$) totally on $Q_{k}$, we have $u',v'\in V(H_{k})$ by \Cref{claim:QcapQ} and the fact that $u,v\in V(H_{k})$. An argument similar to the base case can show 
\[
\dist_{H_{k}}(u',v')\leq \beta\cdot h_{k}\cdot\dist_{Q_{k}}(u',v')= \beta\cdot h_{k}\cdot \dist_{Q_{\geq k}}(u',v').
\]
Lastly, we can conclude
\[
\dist_{H_{k}}(u,v)\leq \beta\cdot(2\cdot\lambda_{H,\up}\cdot\lambda_{H,\low})^{\bar{k}-k}\cdot h_{k}\cdot \dist_{Q_{\geq k}}(u,v)
\]
by concatenating the subpaths. 

\end{proof}

\paragraph{The Size of $Q$.} By the construction, we have
\[
|Q|\leq \sum_{1\leq k\leq \bar{k}}\size({\cal N}_{\vertex,k})\leq \sum_{k}\omega_{k}\cdot|V(H_{k})|\leq O(\lambda_{k}\cdot t\cdot\kappa_{\PC,\omega}\cdot |G|),
\]
where the last inequality is because $|V(H_{k})|\leq |H_{k}|\leq |G|$ (since for each $k$, $|H_{k+1}|\leq O(t\cdot\kappa_{H,\size}\cdot\phi\cdot H_{k})\leq |H_{k}|$).

Besides, from $|V(H_{k+1})|\leq n^{O(\epsilon^{4})}\cdot |V(H_{k})|$, we have $|V(H_{k})|\leq n^{O(\lambda_{k}/\epsilon^{4})}\cdot|V(G)| = n^{O(\epsilon^{2})}\cdot |V(G)|$. Therefore, we can claim another bound $|Q|\leq n^{O(\epsilon^{2})}\cdot |V(G)|$.

\paragraph{The Recourse from $Q^{(i-1)}$ to $Q^{(i)}$.} By our construction of $Q$, each vertex insertion/deletion from ${\cal N}^{(i-1)}_{\vertex,k}$ to ${\cal N}^{(i)}_{\vertex,k}$ will cause at most constant number of unit updates from $Q^{(i-1)}$ to $Q^{(i)}$. Therefore,
\begin{align*}
\recourse(Q^{(i-1)}\to Q^{(i)}) &\leq \sum_{1\leq k\leq \bar{k}}\recourse({\cal N}^{(i-1)}_{\vertex,k}\to {\cal N}^{(i)}_{\vertex,k})\leq \sum_{1\leq k\leq \bar{k}}\recourse({\cal N}^{(i-1)}_{k}\to {\cal N}^{(i)}_{k})\\
&\leq O(\kappa_{\EH,\NC,\rcs}\cdot|\pi^{(i)}|/\phi) = O(\kappa_{\EH,\NC,\rcs}\cdot\kappa_{\phi}\cdot|\pi^{(i)}|).
\end{align*}

\paragraph{The Initialization and Update time.} The running time is dominated by the initialization and maintenance of the length-constrained expander hierarchy. Hence by \Cref{thm:ExpanderHierarchy}, the initialization time is 
$|E(G^{(0)})|\cdot\poly(h)\cdot n^{O(\epsilon)}$
and the update time is
$|\pi^{(i)}|\cdot \poly(h)\cdot n^{O(\epsilon)}$.

\paragraph{Path Unfolding.} Given a path $P_{Q}$ connecting $u,v\in V(G)$ in $Q$, the $u$-$v$ path $P_{G}$ in $G$ can be generated by the following algorithm, which basically follows the analysis of property (3) in \Cref{lemma:EmulatorCorrectness}. We will use the same notations in the proof of property (3) in \Cref{lemma:EmulatorCorrectness}.

Precisely, we consider the levels from top to bottom (i.e. from level $\bar{k}$ to level $1$). At each level $k$, we implement a subroutine that given a path $P_{Q_{\geq k}}$ on $Q_{\geq k}$ connecting two vertices $u,v\in V(H_{k})$, returns a $u$-$v$ path $P_{H_{k}}$ on $H_{k}$ with $\ell_{H_{k}}(P_{H_{k}})\leq \beta\cdot (2\cdot\lambda_{H,\up}\cdot\lambda_{H,\low})^{\bar{k}-k}\cdot h_{k}\cdot \ell_{Q_{\geq k}}(P_{Q_{\geq k}})$ and $|P_{H_{k}}|\geq |P_{Q\geq k}|/2^{\bar{k}-k+1}$ in time $|P_{H_{k}}|\cdot 2^{O(1/\epsilon^{4})}\cdot (2^{\bar{k}-k+1}-1)$.

For the top level $k = \bar{k}$, we have $Q_{\geq k} = Q_{k}$ and we decompose the path $P_{Q_{\geq k}}$ into \emph{star subpaths} s.t. each subpath $P'_{Q_{k}}$ connects two vertices $u',v'\in V(H_{k})$ and is contained by some star in $Q_{k}$. Each star subpath $P'_{Q_{k}}$ can be unfolded to a $u'$-$v'$ path $P'_{H_{k}}$ on $H_{k}$ with $\ell_{H_{k}}(P'_{H_{k}})\leq \lambda_{\rt}(f_{k})\cdot h_{\emb,k} \leq \beta\cdot h_{k}\leq \beta\cdot h_{k}\cdot\ell_{Q_{k}}(P'_{Q_{k}})$ (the last inequality is because each star subpath has length exactly 2) and $|P'_{H_{k}}|\geq |P'_{Q_{k}}|/2$ (which needs to take care of the case $u'=v'$). This can be done by exploiting the router of the cluster of this star and the embedding in time $2^{O(1/\epsilon^{4})} + O(|P'_{H_{k}}|)$. A similar argument has appeared several times in e.g. the path-unfolding paragraph in the proof of \Cref{thm:NonHopReducingEmulator}, so we omit it here. By substituting each star subpath $P'_{Q_{k}}$ on $P_{Q_{\geq k}}$ with $P'_{H_{k}}$, we can obtain the desired $P_{H_{k}}$. The total running time is $|P_{Q_{k}}|\cdot 2^{O(1/\epsilon^{4})} + \sum |P'_{H_{k}}| = |P_{H_{k}}|\cdot 2^{O(1/\epsilon^{4})}$ because we invoke the path-reporting algorithm of routers $O(|P_{Q_{k}}|)$ times.

For an arbitrary level $1\leq k\leq \bar{k}-1$, we decompose the path $P_{Q_{\geq k}}$ into subpaths totally on $Q_{k}$ or totally on $Q_{\geq k+1}$. 
\begin{enumerate}
\item For each subpath $P'_{Q_{\geq k+1}}$ totally on $Q_{\geq k+1}$, note that it has endpoints $u',v'\in V(Q_{\geq k+1})\cap V(Q_{k}) = V(H_{k+1})\cap V(H_{k})$ by \Cref{claim:QcapQ}. The subroutine at level $k+1$ will give a $u'$-$v'$ path $P'_{H_{k+1}}$ on $H_{k+1}$, and the path unfolding algorithm for sparsifiers will further give a $u'$-$v'$ path $P'_{H_{k}}$ on $H_{k}$ s.t. $\ell_{H_{k}}(P'_{H_{k}})\leq \lambda_{H,\low}\cdot \ell_{H_{k+1}}(P'_{H_{k+1}})$ and $|P'_{H_{k}}|\geq |P'_{H_{k+1}}|/2$.
\item For each subpath $P'_{Q_{k}}$ totally on $Q_{k}$, it has endpoints $u',v'\in V(H_{k})$. Similar to the algorithm for the top level $\bar{k}$, we can get a $u'$-$v'$ path $P'_{H_{k}}$ on $H_{k}$ with $\ell_{H_{k}}(P'_{H_{k}}) \leq \beta\cdot h_{k}\cdot \ell_{Q_{k}}(P'_{Q_{k}})$ and $|P'_{H_{k}}|\geq |P'_{Q_{k}}|/2$.
\end{enumerate}
By substituting each subpath of $P_{Q_{\geq k}}$ with its unfolded path on $H_{k}$, we can obtain the desired $P_{H_{k}}$. The running time for this subroutine is 
\[
\sum_{\text{each }P'_{Q_{\geq k+1}}} (|P'_{H_{k+1}}|\cdot 2^{O(1/\epsilon^{4})}\cdot (2^{\bar{k}-k}-1) + |P'_{H_{k}}|\cdot 2^{O(1/\epsilon^{4})}) + \sum_{\text{each }P'_{Q_{k}}} |P'_{H_{k}}|\cdot 2^{O(1/\epsilon^{4})},
\]
which is at most $|P_{H_{k}}|\cdot 2^{O(1/\epsilon^{4})}\cdot (2^{\bar{k}-k+1}-1)$ because each $P'_{H_{k+1}}$ in case 1 has $|P'_{H_{k+1}}|\leq 2|P'_{H_{k}}|$

Finally, the subroutine for level $1$ can return the $u$-$v$ path $P_{G}$ in $G$ with guarantee $\ell_{G}(P_{G})\leq \kappa_{Q,\low}\cdot\ell_{Q}(P_{Q})$ and $|P_{G}|\geq |P_{Q}|/2^{O(\lambda_{k})}$. The path-unfolding time is $|P_{G}|\cdot 2^{O(1/\epsilon^{4})}\cdot 2^{\lambda_{k}} = |P_{G}|\cdot 2^{O(1/\epsilon^{4})}$ because $\lambda_{k} = O(1/\epsilon^{2})$.

\end{proof}

%% file: 11-distance_oracle.tex
\section{Fully Dynamic Distance Oracles}
\label{sect:Oracle}

In this section, we finally prove our main result:

\begin{theorem}
\label{thm:MainDetailed}
Let $G$ be a fully dynamic undirected graph under edge insertions/deletions and isolated vertex insertions/deletions s.t. edge lengths on $G$ are all polynomially bounded positive integers. Let $m_{0}$ denote the initial graph size and $n$ be an upper bound of number of vertices. 

For some sufficiently small constant $c>0$, given a parameter $1/\log^{c}n\leq \epsilon\leq 1$, there is an approximate deterministic fully dynamic distance oracle supports the following query. At any time, given two vertices $u,v$ in the current $G$, if $u$ and $v$ are disconnected in $G$, the oracle answers $u$ and $v$ are disconnected. Otherwise, the oracle answers an approximate distance $\wtilde{d}$ s.t. $\dist_{G}(u,v)\leq \wtilde{d}\leq 2^{\poly(1/\epsilon)}\cdot\dist_{G}(u,v)$ in $O(\log\log n/\epsilon^{4})$ query time. The oracle can further output
\begin{itemize}
\item a (possibly non-simple) path $P$ connecting $u$ and $v$ on $G$ with $\ell_{G}(P)\leq 2^{\poly(1/\epsilon)}\cdot\dist_{G}(u,v)$ in additional $O(|P|)$ time, and
\item a simple path $P_{\sp}$ connecting $u$ and $v$ on $G$ with $\ell_{G}(P_{\sp})\leq 2^{\poly(1/\epsilon)}\cdot\dist_{G}(u,v)$ in additional $O(|P_{\sp}|\cdot n^{\epsilon})$ time.
\end{itemize}
The oracle has initialization time $O(m_{0}\cdot n^{\epsilon})$ and worst-case update time $O(n^{\epsilon})$.
\end{theorem}

Our approach is to ``stack'' length-reducing emulators on top of each other. These emulators allow us to assume that the distances are small, so we can apply the distance oracle for the short-distance regime from \Cref{sect:OracleShort}.
The approach is standard and was used in many previous works
\cite{chechik,chuzhoy2021decremental,bernstein2022deterministic,lkacki2022near,chuzhoy2023new}.

\subsection{One Level of the Stacking}
\label{sect:Stacking}

\begin{lemma}
Let $G$ be a graph under $t$ batches of updates $\pi^{(1)},\pi^{(2)},...,\pi^{(t)}$ of edge insertions and edge deletions. Given a parameter $h$, there is an algorithm that maintains a \text{distance-reduced} graph $G_{\stk}$ s.t. $V(G)\subseteq V(G_{\stk})$ and for each $u,v\in V(G)$, 
\begin{align*}
\dist_{G}(u,v)&\leq \lambda_{\stk,\low}\cdot h\cdot \dist_{G_{\stk}}(u,v),\text{ and }\\
\dist_{G_{\stk}}(u,v)&\leq \lambda_{\stk,\up}\cdot\max\{ \dist_{G}(u,v)/h,1\},
\end{align*}
where $\lambda_{\stk,\low} = \lambda_{Q,\low} = 2^{\poly(1/\epsilon)}$ and $\lambda_{\stk,\up} = O(\lambda_{Q,\up}) = O(1/\epsilon^{2})$. The size of $G_{\stk}$ is $|G_{\stk}|\leq \kappa_{\stk,\size}\cdot |G|$, where $\kappa_{\stk,\size} = \kappa_{Q,\size} + 1$. The algorithm further guarantee the following.
\begin{itemize}
\item The recourse from $G_{\stk}^{(i-1)}$ to $G_{\stk}^{(i)}$ is at most $\kappa_{\stk,\rcs}\cdot|\pi^{(i)}|$, where $\kappa_{\stk,\rcs} = \kappa_{Q,\rcs} + 1$. 
\item The initialization time is
$|G^{(0)}|\cdot\poly(h)\cdot n^{O(\epsilon)}$
and the update time is
$|\pi^{(i)}|\cdot \poly(h)\cdot n^{O(\epsilon)}$.
\item Given a path $P_{G_{\stk}}$ on $G_{\stk}$ connecting $u,v\in V(G)$, the algorithm can compute a $u$-$v$ path $P_{G}$ on $G$ with $\ell_{G}(P_{G})\leq \lambda_{\stk,\low}\cdot h\cdot \ell_{G_{\stk}}(P_{G_{\stk}})$ and $|P_{G}|\geq |P_{G_{\stk}}|/2^{O(\lambda_{k})}$ in time $2^{O(1/\epsilon^{4})}\cdot |P_{G}|$.
\end{itemize}
\label{thm:Stacking}
\end{lemma}
\begin{proof}

We maintain a $(\lambda_{Q,\low},\lambda_{Q,\up},h)$-length-reducing emulator $Q$ of $G$ by \Cref{thm:HopReducingEmulator}. Note that each edge in $Q$ has unit length. Let $E_{\heavy} = \{(u,v)\mid (u,v)\in E(G), \ell_{G}(u,v)\geq h/3\}$, and we define the length of each edge $e\in E_{\heavy}$ to be $\lceil\ell_{G}(u,v)/h\rceil$. Then we define $G_{\stk} = Q\cup E_{\heavy}$.

\paragraph{Quality of $G_{\stk}$.} 

We first show $\dist_{G}(u,v)\leq \lambda_{\stk,\low}\cdot h\cdot \dist_{G_{\stk}}(u,v)$. Let $P_{\stk}$ be a shortest $u$-$v$ path in $G_{\stk}$. We decompose $P_{\stk}$ into subpaths with endpoints in $V(G)$ but no intermediate vertex in $V(G)$. Note that each subpath is either a single edge $e\in E_{\heavy}$ or a path totally inside $Q$. For each subpath which is an edge $e=(u,v)\in E_{\heavy}$, we have $\dist_{G}(u,v)\leq \ell_{G_{\stk}}(e)\cdot h$ by the definition of $E_{\heavy}$. For each subpath $P'$ inside $Q$, we have $\dist_{G}(u,v)\leq \lambda_{Q,\low}\cdot h\cdot \dist_{Q}(u,v) = \lambda_{Q,\low}\cdot h\cdot \ell_{Q}(P')$. Therefore, we can conclude that $\dist_{G}(u,v)\leq \lambda_{\stk,\low}\cdot h\cdot\dist_{G_{\stk}}(u,v)$.

Next we show $\dist_{G_{\stk}}(u,v)\leq \lambda_{\stk,\up}\cdot\lceil \dist_{G}(u,v)/h\rceil$. We first consider the case when $\dist_{G}(u,v)\leq h$, in which we have $\dist_{G_{\stk}}(u,v)\leq \dist_{Q}(u,v)\leq \lambda_{Q,\up}$.

For a general case with $\dist_{G}(u,v)>h$. Let $P$ be a shortest $u$-$v$ path in $G$. We decompose $P$ in the following way. First, we view each edge $e\in P$ with $\ell_{G}(e)\geq h/3$ as a \textit{singleton subpath}, and then remove those edges, which breaks $P$ into subpaths called \textit{segments}. For each segment $P'$ with $\ell_{G}(P')<h$, we call it a \textit{light segment}. For each segment $P'$ with $\ell_{G}(P')\geq h$, we further decompose it into sub-segments $P''$ with $h/3\leq \ell_{G}(P'')\leq h$ (this is doable since each edge on such segments has length at most $h/3$), and we call these sub-segments \textit{heavy sub-segments}. 

For each singleton subpath $P'$ connecting $u',v'$, we have $\dist_{G_{\stk}}(u',v') \leq \lceil\ell_{G}(P')/h\rceil\leq 3\ell_{G}(P')/h$, because the edge $e\in P'$ is inside $E_{\heavy}$. For each light segment $P'$ connecting $u',v'$, we have $\dist_{G_{\stk}}(u',v')\leq \dist_{Q}(u',v')\leq \lambda_{Q,\up}$. For each heavy sub-segment $P''$ connecting $u''$ and $v''$, we have $\dist_{G_{\stk}}(u'',v'')\leq \dist_{Q}(u',v')\leq\lambda_{Q,\up}\leq 3\cdot\lambda_{Q,\up}\cdot\ell_{G}(P'')/h$. Also, note that the number of light segments is at most the number of singleton subpaths plus 1. Summing over all things,
\begin{align*}
\dist_{G_{\stk}}(u,v) &\leq \sum_{\text{singleton subpath }P'} 3\ell_{G}(P')/h + \sum_{\text{light segments }P'} \lambda_{Q,\up} + \sum_{\text{heavy sub-segments }P''} 3\cdot\lambda_{Q,\up}\cdot\ell_{G}(P'')/h\\
&\leq \lambda_{Q,\up} + \sum_{\text{singleton subpath }P'} (1+\lambda_{Q,\up})\cdot 3\cdot \ell_{G}(P')/h + \sum_{\text{heavy sub-segments }P''} 3\cdot\lambda_{Q,\up}\cdot\ell_{G}(P'')/h\\
&\leq \lambda_{Q,\up} + 3(1+\lambda_{Q,\up})\ell_{G}(P)/h\\
&\leq O(\lambda_{Q,\up})\cdot \dist_{G}(u,v)/h,
\end{align*}
where the second inequality is because we distribute the term $\lambda_{Q,\up}$ of light segments to singleton subpaths (and there may be one left), and the last inequality is because $\ell_{G}(P) = \dist_{G}(u,v)\geq h$.

\paragraph{The Size of $G_{\stk}$.} By \Cref{thm:HopReducingEmulator}, $|Q|\leq \kappa_{Q,\size}\cdot |G|$, so $|G_{\stk}|\leq |Q| + |G| \leq (\kappa_{Q,\size} + 1)\cdot |G|$.

\paragraph{The Recourse from $G_{\stk}^{(i-1)}$ to $G_{\stk}^{(i)}$.} The total recourse is simply the sum of recourse from $Q^{(i-1)}$ to $Q^{(i)}$ and the changes from $E^{(i-1)}_{\heavy}$ to $E^{(i)}_{\heavy}$. The former is $\kappa_{Q,\rcs}\cdot |\pi^{(i)}|$ by \Cref{thm:HopReducingEmulator} and the latter is trivially at most $|\pi^{(i)}|$, so $\recourse(G^{(i-1)}_{\stk}\to G^{(i)}_{\stk})\leq (\kappa_{Q,\rcs} + 1)\cdot|\pi^{(i)}|$.

\paragraph{The Initialization and Update Time.} They have the same bounds as those in \Cref{thm:HopReducingEmulator}.

\paragraph{Path Unfolding.} Note that $G_{\stk}$ is the union of the emulator $Q$ and some original edges $E_{\heavy}$. Therefore, given a path $P_{G_{\stk}}$ connecting some $u,v\in V(G)$, we just decompose it into original edges and subpaths totally on $Q$, and then unfold each $Q$-subpath to a path on $G$ using \Cref{thm:HopReducingEmulator}.

\end{proof}

\subsection{Online-Batch Dynamic Distance Oracles from Stacking}
\begin{theorem}
Let $G$ be a graph under $t$ batched updates $\pi^{(1)},...,\pi^{(t)}$ of edge deletions and edge insertions. There is an oracle supporting the following query. Given two vertices $u,v\in V(G)$, if $u$ and $v$ are disconnected in $G$, the oracle answers they are disconnected. Otherwise, the oracle answers an approximation $\wtilde{d}$ of $\dist_{G}(u,v)$ s.t. $\dist_{G}(u,v)\leq \wtilde{d}\leq \lambda_{\DO,\alpha}\cdot\dist_{G}(u,v)$.
The approximation ratio 
\[
\lambda_{\DO,\alpha} = O(\lambda_{\query,\alpha})\cdot (\lambda_{\stk,\low}\cdot\lambda_{\stk,\up})^{\lambda_{x}}\cdot 2^{O(1/\epsilon^{4})},
\]
where $\lambda_{x} = O(1/\epsilon)$ is a sufficiently large global parameter. Furthermore, the algorithm can output a path $P$ connecting $u$ and $v$ on $G$ with $\ell_{G}(P)\leq \lambda_{\DO,\alpha}\cdot \dist_{G}(u,v)$.

The initialization time is $|G^{(0)}|\cdot n^{O(\epsilon)}$, the update time for batch $\pi^{(i)}$ is $|\pi^{(i)}|\cdot n^{O(\epsilon)}$, the distance query time is $O(\log\log n/\epsilon^{4})$, and the additional path-reporting time is $O(|P|)$.
\label{thm:BatchDynDistanceOracle}
\end{theorem}
\begin{proof}

\ 

\paragraph{The Oracle.} We pick a parameter $h =  n^{\epsilon}$ and construct a chain ${\cal G}$ of $\bar{x} = \lambda_{x}$ many graphs $G_{1}, G_{2}, ..., G_{\bar{x}}$ with the following properties, where $\lambda_{x} = O(1/\epsilon)$ is a sufficiently large global parameter.
\begin{itemize}
\item[(1)] $G_{1} = G$ is the original graph.
\item[(2)] For each $2\leq x\leq \bar{x}$, $G_{x}$ is the distance-reduced graph from applying \Cref{thm:Stacking} on $G_{x-1}$ with length parameter $h$.
\end{itemize}

\begin{claim}
For the last graph $G_{\bar{x}}$, we have $\dist_{G_{\bar{x}}}(u,v)\leq h$ for all connected $u,v\in V(G)$.
\label{claim:LastGraphMaxDis}
\end{claim}
\begin{proof}
Consider an arbitrary pair of connected vertices $u$ and $v$. Note that $\dist_{G}(u,v) = \poly(n)$ because we assume the edge lengths are all polynomial and the graph $G$ is connected. For each $1\leq x\leq \bar{x}-1$, we have the following by \Cref{thm:Stacking}. 
\begin{itemize}
\item If $\dist_{G_{x}}(u,v)> h$, then $\dist_{G_{x+1}}(u,v)\leq \lambda_{\stk,\up}\cdot \dist_{G_{x}}(u,v)/h$.
\item If $\dist_{G_{x}}(u,v)\leq h$, then $\dist_{G_{x+1}}(u,v)\leq \lambda_{\stk,\up}\leq h$ because $\lambda_{\stk,\up} = 2^{\poly(1/\epsilon)}$ and $h = n^{\epsilon}$.
\end{itemize}
Therefore, if there is some $1\leq x\leq\bar{x}-1$ s.t. $\dist_{G_{x}}(u,v)\leq h$, we have $\dist_{G_{\bar{x}}}(u,v) \leq h$. Otherwise, we have $\dist_{G_{\bar{x}}}(u,v) \leq (\lambda_{\stk,\up}/h)^{\lambda_{x}-1}\cdot \dist_{G}(x,y)\leq \poly(n)/(n^{\epsilon})^{\lambda_{x}-1}\leq h$ because $\lambda_{x} = O(1/\epsilon)$ is sufficiently large. 

\end{proof}

Next, for each $G_{x}$, we construct several low-distance oracles with different length thresholds. Precisely, for each $1\leq y\leq \bar{y} = \lceil\log h\rceil + 1$, let $h_{y} = 2^{y}$ and we construct a low-distance oracle ${\cal O}_{x,y}$ by applying \Cref{thm:LowDistanceOracle} on $G_{x}$ with length parameter $h_{y}$.

Note that ${\cal G}$ and its low-distance oracles $\{{\cal O}_{x,y}\}$ can be maintain under batched updates by \Cref{thm:Stacking} and \Cref{thm:LowDistanceOracle}. Precisely, suppose $G_{x}$ is under batched updates $\pi_{x}^{(1)},...,\pi_{x}^{(t)}$ (the first graph $G_{1}$ has $(\pi_{1}^{(1)} ,...,\pi_{1}^{(t)}) = (\pi^{(1)},...,\pi^{(t)})$), we maintain ${\cal O}_{x,y}$ for each $1\leq y\leq \bar{y}$ under $\pi^{(1)}_{x},...,\pi^{(t)}_{x}$ by \Cref{thm:LowDistanceOracle}. Also, we maintain $G_{x+1}$ under $\pi^{(1)}_{x},...,\pi^{(t)}_{x}$ by \Cref{thm:Stacking}, which will generate the batched updates $\pi^{(1)}_{x+1},...,\pi^{(t)}_{x+1}$ for $G_{x+1}$ s.t. $|\pi^{(i)}_{x+1}|\leq \kappa_{\stk,\rcs}\cdot|\pi^{(i)}_{x}|$ for each $1\leq i\leq t$. From this, we have $|\pi^{(i)}_{x}|\leq (\kappa_{\stk,\rcs})^{x-1}\cdot |\pi^{(i)}|$ for each $1\leq x\leq \bar{x}$.

\paragraph{The Initialization and Update Time.} First we bound the size of graphs $G_{x}$. At any time, we have $|G_{x+1}|\leq \kappa_{\stk,\size}\cdot |G_{x}|$ for each $1\leq x\leq \bar{x}-1$. Simply we have $|G_{x}| \leq \kappa_{\stk,\size}^{\lambda_{x}}\cdot |G|\leq n^{O(\epsilon)}\cdot |G|$ for all $x$, because $\kappa_{\stk,\size} = n^{O(\epsilon^{4})}$. 

For the initialization, the construction all graphs in ${\cal G}$ takes total time $\sum_{1\leq x\leq\bar{x}} |G_{x}^{(0)}|\cdot \poly(h)\cdot n^{O(\epsilon)}$, 
and the construction of all low-distance oracles $\{{\cal O}_{x,y}\}$ takes total time $\sum_{1\leq x\leq \bar{x}}\sum_{1\leq y\leq \bar{y}}|G^{(0)}_{x}|\cdot \poly(h)\cdot n^{O(\epsilon)}$.
Hence the total initialization time is $|G^{(0)}|\cdot n^{O(\epsilon)}$ because $h = n^{\Theta(\epsilon)}, \bar{x}\leq \lambda_{h} = O(1/\epsilon)$ and $\bar{y} = O(\log h) = O(\log n)$.

The update time to handle each $\pi^{(i)}$ is bounded by $\sum_{1\leq x\leq \bar{x}}|\pi^{(i)}_{x}|\cdot\poly(h)\cdot n^{O(\epsilon)} = |\pi^{(i)}|\cdot n^{O(\epsilon)}$, because $|\pi^{(i)}_{x}|\leq (\kappa_{\stk,\rcs})^{\lambda_{x}}\cdot |\pi^{(i)}| = |\pi^{(i)}|\cdot n^{O(\epsilon)}$ (this is from $\kappa_{\stk,\rcs} = n^{O(\epsilon^{2})}$ and $\lambda_{x} = O(1/\epsilon)$)

\paragraph{The Query Algorithm.} 

Suppose a query $(u,v)$ arrives right after the update $\pi^{(i)}$. We will query ${\cal G}^{(i)}$ and $\{{\cal O}^{(i)}_{x,y}\}$ at that moment, and for simplicity, we omit the superscript $(i)$ in what follows. 

The algorithm is as foolows.
\begin{enumerate}
\item We first find the minimum $x^{\star}$ s.t. ${\cal O}_{x^{\star},\bar{y}}$ declares $\dist_{G_{x^{\star}}}(u,v)\leq \lambda_{\query,\alpha}\cdot h_{\bar{y}}$ but for each $1\leq x\leq x^{\star}-1$, ${\cal O}_{x^{\star},\bar{y}}$ declares $\dist_{G_{x}}(u,v)>h_{\bar{y}}$. Note that if $u$ and $v$ are connected, such $x^{\star}$ must exist because $\dist_{G_{\bar{x}}}(u,v)\leq h\leq h_{\bar{y}}$ by \Cref{claim:LastGraphMaxDis}. If we cannot find such $x^{\star}$, we can safely declare $u$ and $v$ are disconnected.

\item Next, we use a binary search to find a $y^{\star}\in[1,\bar{y}]$ s.t. ${\cal O}_{x^{\star},y^{\star}}$ declares $\dist_{G_{x^{\star}}}(u,v)\leq \lambda_{\query,\alpha}\cdot h_{y^{\star}}$ and if $y^{\star}\geq 2$, ${\cal O}_{x^{\star},y^{\star}-1}$ declares $\dist_{G_{x^{\star}}}(u,v)>h_{y^{\star}-1}$.

\item The output is $\wtilde{d} = h^{x^{\star}-1}\cdot h_{y^{\star}}\cdot (\lambda_{\stk,\low}^{x^{\star}-1}\cdot \lambda_{\query,\alpha})$.

\end{enumerate}

\paragraph{The Correctness and the Query Time.} The correctness of the query algorithm is given by \Cref{lemma:FinalQueryCorrectness}. The total query time is $O((\lambda_{x} + \log\bar{y})\cdot \lambda_{k}\cdot t) = O(\log\log n/\eps^{4})$ by the following reasons. Step 1 takes $O(\lambda_{x}\cdot \lambda_{k}\cdot t) = O(1/\epsilon^{4})$ because we make at most $\bar{x} = \lambda_{x} = O(1/\epsilon)$ queries to low-distance oracles, each of which takes $O(\lambda_{k}\cdot t) = O(1/\epsilon^{3})$ time by \Cref{thm:LowDistanceOracle}. Step 2 takes $O(\log \bar{y}\cdot\bar{k}\cdot t)$ time because the binary search query the low-distance oracles at most $O(\log\bar{y}) = O(\log\log n/\epsilon^{4})$ times and $\bar{y} = O(\log h) = O(\epsilon\cdot\log n)$.

\begin{lemma}
\label{lemma:FinalQueryCorrectness}
If $u$ and $v$ are connected, the approximate distance $\wtilde{d}$ satisfies
\[
\dist_{G}(u,v)\leq \wtilde{d}\leq 2\cdot\lambda_{\stk,\up}^{x^{\star}-1}\cdot \lambda_{\stk,\low}^{x^{\star}-1}\cdot \lambda_{\query,\alpha}\cdot\dist_{G}(u,v).
\]
\end{lemma}
\begin{proof}
First, we bound $\dist_{G^{\star}}(u,v)$ using $\dist_{G}(u,v)$. For each $1\leq x\leq x^{\star}-1$, we know ${\cal O}_{x,\bar{y}}$ declared $\dist_{G_{x}}(u,v)>h_{\bar{y}}\geq h$. Therefore, from \Cref{thm:Stacking}, we have 
$\dist_{G_{x+1}}(u,v)\leq \lambda_{\stk,\up}\cdot\max\{ \dist_{G_{x}}(u,v)/h,1\} = \lambda_{\stk,\up}\cdot  \dist_{G_{x}}(u,v)/h$,
which implies
\[
\dist_{G_{x^{\star}}}(u,v)\cdot h^{x^{\star}-1}\leq \lambda_{\stk,\up}^{x^{\star}-1}\cdot\dist_{G}(u,v).
\]
On the lower bound side, we have for each $1\leq x\leq x^{\star}-1$,
$\dist_{G_{x}}(u,v)\leq \lambda_{\stk,\low}\cdot h\cdot \dist_{G_{x+1}}(u,v)$,
so
\[
\dist_{G_{x^{\star}}}(u,v)\cdot h^{x^{\star}-1}\geq \dist_{G}(u,v)/\lambda_{\stk,\low}^{x^{\star}-1}.
\]

Next, we bound $\dist_{G^{\star}}(u,v)$ using $h_{y^{\star}}$. On the lower bound side, we have
\[
\dist_{G_{x^{\star}}}(u,v)\geq h_{y^{\star}}/2
\]
because if $y^{\star}\geq 2$, ${\cal O}_{x^{\star},y^{\star}-1}$ will declare $\dist_{G_{x^{\star}}}(u,v)>h_{y^{\star}-1}=h_{y^{\star}}/2$, and if $y^{\star} = 1$, we still have $\dist_{G_{x^{\star}}}(u,v)\geq 1 = h_{y^{\star}}/2$, since all edges have positive weights and $h_{1} = 2$. On the upper bound side, ${\cal O}_{x^{\star},y^{\star}}$ declares 
\[
\dist_{G_{x^\star}}(u,v)\leq \lambda_{\query,\alpha}\cdot h_{y^{\star}}.
\]

Combining the above bounds of $\dist_{G_{x^{\star}}}(u,v)$, we have
\[
\dist_{G}(u,v)/(\lambda_{\stk,\low}^{x^{\star}-1}\cdot \lambda_{\query,\alpha})\leq h_{y^{\star}}\cdot h^{x^{\star}-1}\leq 2\cdot\lambda_{\stk,\up}^{x^{\star}-1}\cdot\dist_{G}(u,v).
\]
\end{proof}

\paragraph{Path Reporting.} Let $u,v\in V(G)$ be the input vertices. Because ${\cal O}_{x^{\star},y^{\star}}$ declares $\dist_{G_{x^{\star}}}(u,v)\leq \lambda_{\query,\alpha}\cdot h_{y^{\star}}$, by \Cref{thm:LowDistanceOracle}, we can obtain a $u$-$v$ path $P_{G_{x^{\star}}}$ on $G_{x^{\star}}$ with $\ell_{G_{x^{\star}}}(P_{G_{x^{\star}}})\leq \lambda_{\query,\alpha}\cdot h_{y^{\star}}$ in $2^{O(\lambda_{k})}\cdot |P_{G_{x^{\star}}}|$ time. Next, we iterate $x$ from $x^{\star}-1$ to $1$. For each $x$, by the path unfolding algorithm in \Cref{thm:Stacking} (with input $P_{G_{x+1}}$), we obtain a $u$-$v$ path $P_{G_{x}}$ on $G_{x}$ with $\ell_{G_{x}}(P_{G_{x}}) \leq \lambda_{\stk,\low}\cdot h\cdot \ell_{G_{x+1}}(P_{G_{x+1}})$ and $|P_{G_{x}}|\geq |P_{G_{x+1}}|/2^{O(\lambda_{k})}$ in time $2^{O(\lambda_{k})}\cdot |P_{G_{x}}|$.

At last, we obtain a $u$-$v$ path $P_{G} = P_{G_{1}}$ on $G$ with $\ell_{G}(P_{G})\leq \lambda_{\stk,\low}^{x^{\star}-1}\cdot\lambda_{\query,\alpha}\cdot h^{x^{\star}-1}\cdot h_{y^{\star}}$. From the proof of \Cref{lemma:FinalQueryCorrectness}, we have $h_{y^{\star}}\leq 2\cdot\lambda_{\stk,\low}^{x^{\star}-1}/h^{x^{\star}-1}\cdot \dist_{G}(u,v)$. Therefore, $\ell_{G}(P_{G})\leq 2\cdot \lambda_{\query,\alpha}\cdot (\lambda_{\stk,\low}\cdot \lambda_{\stk,\up})^{x^{\star}-1} \cdot \dist_{G}(u,v)$. The path-reporting time can be bounded by $\sum_{1\leq x\leq x^{\star}} 2^{O(1/\epsilon^{4})}\cdot |P_{G_{x}}| \leq 2^{O(1/\epsilon^{4})}\cdot 2^{O(\lambda_{k}\cdot\lambda_{x})}\cdot |P_{G}|$ because each $P_{G_{x}}$ has $|P_{G_{x}}|\leq P_{G}\cdot 2^{O(\lambda_{x}\cdot\lambda_{k})}$ and $\lambda_{x}\cdot\lambda_{k} = O(1/\epsilon^{3})$.

We note that, when $P_{G}$ is not required to be a simple path, it is trivial to move the overhead on path-reporting time to approximation. For example, we can define a new $u$-$v$ path $\wtilde{P}_{G}$ by walking along $P_{G}$ but going back and forth on each edge $2^{O(\lambda_{k}\cdot\lambda_{x})}$ times. Then we have $\ell_{G}(\wtilde{P}_{G})\leq 2^{O(1/\epsilon^{4})}\cdot \ell_{G}(P_{G})\leq \lambda_{\DO,\alpha}\cdot \dist_{G}(u,v)$. Moreover, because $|\wtilde{P}_{G}| = 2^{O(1/\epsilon^{4})}\cdot |P_{G}|$, we can rewrite the path-reporting time to be $O(|\wtilde{P}_{G}|)$.

\end{proof}

\subsection{Fully Dynamic Algorithms from Online-Batch Algorithms}
\label{sect:Reduction}

The last step is to turn our online-batch dynamic distance oracle to a fully dynamic one. This can be done by a standard online-batch-to-fully-dynamic technique, which was implicitly shown in \cite{NanongkaiSW17} and formalized in \cite{JS22} afterwards. 

\begin{lemma}[Theorem 10.1 in \cite{JS22}, Section 5 in \cite{NanongkaiSW17}]
Let $G$ be a graph undergoing batch updates. Assume for two parameters $\zeta$ and $w$, there is a preprocessing algorithm with preprocessing time $t_{\rm preprocess}$ and an update algorithm with amortized update time $t_{\rm amortized}$ for a data structure ${\cal D}$ in the online-batch setting with batch number $\zeta$ and sensitivity $w$, where $t_{\rm preprocess}$ and $t_{\rm amortized}$ are functions the map the upper bounds of some graph measures throughout the update, e.g. maximum number of edges, to non-negative numbers.

Then for any $\xi\leq \zeta$ satisfying $w\geq 2\cdot 6^{\xi}$, there is a fully dynmaic algorithm with preprocessing time $O(2^{\xi}\cdot t_{\rm preprocess})$ and worst case update time $O(4^{\xi}\cdot(t_{\rm preprocess}/w + w^{(1/\xi)}\cdot t_{\rm amortized}))$ to maintain a set of $O(2^{\xi})$ instances of the data structure ${\cal D}$ such that after each update, the update algorithm specifies one of the maintained data structure instances satisfying the following conditions
\begin{itemize}
\item[(1)] The specified data structure instances is for the up-to-date graph.
\item[(2)] The online-batch update algorithm is executed for at most $\xi$ times on the specified data structure instance with each update batch of size at most $w$.
\end{itemize}
\label{lemma:Reduction}
\end{lemma}

The sensitivity $w$ is the maximum size of each update batch allowed by the online-batch dynamic algorithm. For our online-batch dynamic distance oracle in \Cref{thm:BatchDynDistanceOracle}, there is no such restriction and we can simply choose $w = n^{2}$. Recall that we mentioned in \Cref{sect:GlobalParameters} that we will set the batch number $\zeta = t = \Theta(1/\epsilon)$. Furthermore, we pick $\xi = \zeta = \Theta(1/\epsilon)$.

By \Cref{thm:BatchDynDistanceOracle}, our online-batch dynamic distance oracle algorithm has preprocessing time $t_{\rm preprocess} = |G^{(0)}|\cdot n^{O(\epsilon)} = n^{2+O(\epsilon)}$. An online-batch dynamic algorithm has amortized update time $t_{\rm amortized}$ if it can handle each update batch $\pi^{(i)}$ in time $|\pi^{(i)}|\cdot t_{\rm amortized}$. By this definition, \Cref{thm:BatchDynDistanceOracle} has $t_{\rm amortized} = n^{O(\epsilon)}$.

Therefore, by \Cref{lemma:Reduction}, we can obtain a fully dynamic distance oracle algorithm with initialization time $O(2^{\xi}\cdot |G^{(0)}|\cdot n^{O(\epsilon)}) = |G^{(0)}|\cdot n^{O(\epsilon)}$ (we can check in the proof of \Cref{lemma:Reduction} in \cite{JS22} that the initialization time only depends on the initial graph size), worst-case update time $O(4^{\xi}\cdot(t_{\rm preprocess}/w + w^{(1/\xi)}\cdot t_{\rm amortized})) = n^{O(\epsilon)}$, the same query time $O(\log\log n/\epsilon^{4})$ and the same path-reporting time $O(|P|)$. This completes the proof of \Cref{thm:MainDetailed} except the part about reporting simple paths, which will be included in \Cref{sect:ReportSimplePath}.

\subsection{Report Simple Paths}
\label{sect:ReportSimplePath}

In this section, we will discuss how to ensure the path reported by the oracle is simple, while keeping the path-reporting time linear (with an $n^{O(\epsilon)}$ overhead actually) to the number of edges in this path. Note that, after we obtain a non-simple $u$-$v$ path $P$ from the oracle, we cannot just pick an arbitrary $u$-$v$ simple path $P_{\sp}$ on $P$ as the output, because $|P|$ can be much larger than $|P_{\sp}|$, and then the path-reporting time $O(|P|)$ cannot be bounded by $|P_{\sp}|\cdot n^{O(\epsilon)}$. An algorithm for addressing this issue has been proposed by \cite{chuzhoy2021decremental}. We include it below for completeness, and integrate it into the oracle for modularity (the algorithm in \cite{chuzhoy2021decremental} is on the application side for using an oracle which may report non-simple paths).

Let $\alpha^{\star} = \lceil \lambda_{\DO,\alpha} \rceil$ be the approximation factor of our oracle (round it up to an integer). We restate the query interfaces of the oracle here. Let $u,v\in V(G)$ be a given pair of vertices.
\begin{itemize}
\item The oracle can return an integral value $\wtilde{d}$ s.t. $\dist_{G}(u,v)\leq \wtilde{d}\leq \alpha^{\star}\cdot\dist_{G}(u,v)$. The distance query time is $O(\log\log n/\epsilon^{4})$.
\item The oracle can return a $u$-$v$ path $P$ on $G$ s.t. $\dist_{G}(u,v)\leq \ell_{G}(P)\leq \alpha^{\star}\cdot\dist_{G}(u,v)$. The path-reporting time is $O(\log\log n/\epsilon^{4} + |P|)$ (which includes the distance query time).
\end{itemize}

For each $1\leq y\leq \bar{y} = O(\log n)$ (s.t. $2^{\bar{y}}$ upper bounds the diameter of $G$ over time) and $1\leq z\leq \bar{z} = \lceil 1/\epsilon \rceil+1$, let $d_{y} = 2^{y}$ and $L_{z} = \lceil n^{\epsilon}\rceil^{z}$, and we define $G_{y,z}$ to be a dynamic graph exactly the same with $G$ (undergoing the same sequence of updates) except that we increase the length of each edge (including initial edges and new edges) by $d_{y}/L_{z}$ additionally. We maintain an oracle ${\cal O}_{y,z}$ on each graph $G_{y,z}$, and also an oracle ${\cal O}$ on the original graph $G$. Strictly speaking, our oracle required the input dynamic graph has positive integral edge lengths, and we can ensure this by scaling up the edge lengths of $G_{y,z}$ (and then removing this scaling in the answers). 

Let $u,v\in V(G)$ be the input of the path-reporting query, and let $\wtilde{d}$ be the approximate 
$u$-$v$ distance on $G$ returned by ${\cal O}$. The algorithm has at most $\bar{z}$ phases (from phase $1$ to phase $\bar{z}$, but it might terminate at some intermediate phase). Let $z$ be the current phase. At the beginning of phase $z$, we have an invariant that any $u$-$v$ path $P$ on $G$ with $\ell_{G}(P)\leq \wtilde{d}\cdot (10\cdot\alpha^{\star})^{\bar{z} - z + 1}$ has $|P|\geq L_{z-1}$ (we define $L_{0} = 1$). For the first phase, this invariant trivially holds, and for each phase $z\geq 2$, this invariant is guaranteed by phase $z-1$ as we will see. 

Let $y$ be the minimum s.t. $d_{y}\geq \wtilde{d}\cdot (10\cdot\alpha^{\star})^{\bar{z} - z}$, which implies $d_{y}\leq 2\cdot \wtilde{d}\cdot (10\cdot\alpha^{\star})^{\bar{z} - z}$. We query the oracle ${\cal O}_{y,z}$ with input $(u,v)$, and it will return a path $P_{y,z}$ on $G_{y,z}$, which corresponds to a path on $G$. 
\begin{itemize}
\item Suppose $\ell_{G_{y,z}}(P_{y,z})\leq 2\cdot \alpha^{\star}\cdot d_{y}$. Then we know $|P_{y,z}|\leq 2\cdot\alpha^{\star}\cdot L_{z}$ because each edge on $G_{y,z}$ has an addition length increasing $d_{y}/L_{z}$. We take an arbitrary $u$-$v$ simple path $P_{\sp}$ on $P_{y,z}$. Trivially $\ell_{G}(P_{\sp})\leq \ell_{G}(P_{y,z})\leq \ell_{G_{y,z}}(P_{y,z})\leq 2\cdot\alpha^{\star}\cdot d_{y}\leq 4\cdot\alpha^{\star}\cdot\wtilde{d}\cdot (10\cdot\alpha^{\star})^{\bar{z}-z}$. Therefore, by the invariant, we know $|P_{\sp}|\geq L_{z-1}$. We terminate the algorithm with output $P_{\sp}$.

\item Otherwise, we have $\ell_{G_{y,z}}(P_{y,z})> 2\cdot\alpha^{\star}\cdot d_{y}$, which implies $\dist_{G_{y,z}}(u,v)\geq \ell_{G_{y,z}}(P_{y,z})/\alpha^{\star}> 2\cdot d_{y}$ from the approximation guarantee of the oracle. Now we show that the invariant of phase $z+1$ holds, i.e. an arbitrary $u$-$v$ path $P$ with $\ell_{G}(P)\leq \wtilde{d}\cdot (10\cdot\alpha^{\star})^{\bar{z}-z}$ has $|P|\geq L_{z}$. Assume $|P|<L_{z}$ for contradiction. Then we have $\ell_{G_{y,z}}(P)\leq \ell_{G}(P) + d_{y}\cdot |P|/L_{z}\leq \wtilde{d}\cdot (10\cdot\alpha^{\star})^{\bar{z}-z} + d_{y}\leq 2\cdot d_{y}$, a contradiction.
\end{itemize}

The algorithm must finally output some $u$-$v$ simple path $P_{\sp}$ before the last phase $\bar{z}$ because the invariant at the last phase cannot hold, since the shortest $u$-$v$ path $P$ on $G$ has $\ell_{G}(P) = \dist_{G}(u,v)\leq \wtilde{d}\cdot (10\cdot\alpha^{\star})$ and $|P|\leq n-1 < L_{\bar{z}-1}$. 

Let $z^{\star}$ be the phase at which $P_{\sp}$ is obtained. Then $\ell_{G}(P_{\sp})\leq 4\cdot\alpha^{\star}\cdot\wtilde{d}\cdot (10\cdot\alpha^{\star})^{\bar{z}-z^{\star}}\leq 4\cdot (\alpha^{\star})^{2}\cdot (10\cdot \alpha^{\star})^{\bar{z}}\cdot\dist_{G}(u,v)\leq 2^{\poly(1/\epsilon)}\cdot\dist_{G}(u,v)$. The total path-report time is $\sum_{1\leq z\leq z^{\star}}O(|P_{y,z}|)\leq \sum_{1\leq z\leq z^{\star}}O(2\cdot\alpha^{\star}\cdot L_{z})\leq |P_{\sp}|\cdot n^{O(\epsilon)}$ where the last inequality is because for each $1\leq z\leq z^{\star}$, $|L_{z}|\leq \lceil n^{\epsilon}\rceil\cdot L_{z^{\star}-1}\leq |P_{\sp}|\cdot n^{O(\epsilon)}$. Maintaining all the oracles ${\cal O}$ and ${\cal O}_{y,z}$ will only increase the initialization time and the update time by a factor of $O(\bar{y}\cdot\bar{z}) = O(\log n/\epsilon)$, which is negligible.

%% file: 12-spasifier_extension.tex
\section{Extension: Dynamic Vertex Sparsifiers for All Distances}
\label{sect:ExtendedSparsifier}

In this section, we generalize our results about dynamic vertex sparsifiers in \Cref{sect:DynSparsifier}. The theorem below restates \Cref{thm:vertex sparsifier}

\begin{theorem}
Let $G$ be a fully dynamic graph with a terminal set $T\subseteq V(G)$ undergoing a sequence of edge insertions/deletions, isolated vertex insertions/deletions and terminal insertions/deletions. For $\alpha_{\low} = 2^{\poly(1/\epsilon)}$ and $\alpha_{\up} = 2^{O(1/\epsilon^{3})}$, there is a deterministic algorithm that maintains either
\begin{enumerate}
\item\label{SparsifierAmortized} a graph $H$ with size $|H|\leq |T|\cdot n^{O(\epsilon^{4})}$ which is an $(\alpha_{\low},\alpha_{\up})$-sparsifier of $T$ on $G$ at all time, using $n^{O(\epsilon)}$ amortized update time, or
\item\label{SparsifierWorstCase} a collection of $2^{O(1/\epsilon)}$ graphs $H_{1},H_{2},...,H_{2^{O(1/\epsilon)}}$ with sizes at most $|T|\cdot n^{O(\epsilon^{4})}$ using $n^{O(\epsilon)}$ worst-case update time, such that at each moment $i$, one of these graphs is an $(\alpha_{\low},\alpha_{\up})$-sparsifier of $T$ on $G$ and it will be specified by the algorithm.
\end{enumerate}
The initialization time is $|G^{(0)}|\cdot n^{O(\epsilon)}$.

\label{thm:FullyDynamicSparsifier}
\end{theorem}

\begin{lemma}
Let $G$ be an online-batch dynamic graph with a fully dynamic terminal set $T\subseteq V(G)$ under $t$ batches of updates $\pi^{(1)},...,\pi^{(t)}$ of edge insertions/deletions, isolated vertex insertions/deletions and terminal insertions/deletions. Given parameters $h$, there is an algorithm that maintains an $(\alpha_{\low}, \alpha_{\up}, h)$-sparsifier $H$ of $T$ on $G$ with size $|H|\leq O(\kappa_{H,\size}\cdot|T|) = |T|\cdot n^{O(\epsilon^{4})}$,
\[
\alpha_{\low} = \lambda_{H,\low}^{\lambda_{k}} = 2^{\poly(1/\epsilon)}\text{ and }\alpha_{\up} = \lambda_{H,\up}^{\lambda_{k}} = 2^{O(1/\epsilon^{2})}
\]
The initialization time is $|G^{(0)}|\cdot \poly(h)\cdot n^{O(\epsilon)}$ and the update time for $\pi^{(i)}$ is $|\pi^{(i)}|\cdot\poly(h)\cdot n^{O(\epsilon)}$.
\label{lemma:SmallSparsifierLowDist}
\end{lemma}

\begin{proof}
We fix a parameter $\phi = \kappa_{\phi} = n^{\epsilon^{2}}$ and $\bar{k} = \lambda = O(1/\epsilon^{2})$ to be sufficently large. For each $1\leq k\leq \bar{k}$, we define $h_{k} = h\cdot \lambda_{H,\up}^{k-1}$.
Our strategy is to construct and maintain a hierarchy of sparsifiers with $\bar{k}$ levels by applying \Cref{thm:NonHopReducingEmulator} repeatedly. For each level $1\leq k\leq \bar{k}$, we let $H_{k}$ and $T_{k}$ denote the sparsifier and terminal set at this level. We will guarantee the invariant $T_{1}\supseteq T_{2}\supseteq ...\supseteq T_{\bar{k}}\supseteq T$ at any time.

Initially, let $H_{1}^{(0)} = G^{(0)}$ and $T_{1}^{(0)} = T^{(0)}$ be the first-level sparsifier and terminal set. At each level $1\leq k\leq \bar{k}-1$, we initialize a $(\lambda_{H,\low}, \lambda_{H,\up}, h_{k})$-sparsifier $H_{k+1}^{(0)}$ of terminal set $T_{k}^{(0)}$ on $H_{k}^{(0)}$ via \Cref{thm:NonHopReducingEmulator}, and let $T_{k+1}^{(0)} = T^{(0)}$ be the terminal set of the next level. When an update $\pi^{(i)}$ comes, we update the hierarchy from the bottom up. Concretely, at each level $1\leq k\leq \bar{k}-1$, let $\pi^{(i)}_{k}$ be the update at this level (the first-level update is $\pi^{(i)}_{1} = \pi^{(i)}$). Let $(\pi^{(i)}_{k})_{\mid T^{-}}, (\pi^{(i)}_{k})_{\mid T^{+}}\subseteq \pi^{(i)}_{k}$ denote the terminal deletions and insertions in $\pi^{(i)}_{k}$ respectively. Let $\hat{\pi}^{(i)}_{k} = \pi^{(i)}_{k}\setminus (\pi^{(i)}_{k})_{\mid T^{-}}$ represent edge insertions, edge deletions and terminal insertions in $\pi^{(i)}_{k}$.
\begin{itemize}
\item If $|(\pi^{(i)}_{k})_{\mid T^{-}}|> \phi\cdot |H^{(i)}_{k}|$ or $|\hat{\pi}^{(i)}_{k}|>\phi\cdot |H^{(i)}_{k}|$, we reinitialize the sparsifier and terminal set via \Cref{thm:NonHopReducingEmulator} at all level above. Concretely, for each level $k'$ s.t. $k\leq k'\leq \bar{k}-1$, initialize a $(\lambda_{H,\low}, \lambda_{H,\up}, h_{k'})$-sparsifier $H_{k'+1}^{(i)}$ of terminal set $T_{k'}^{(i)} = T^{(i)}$ on $H_{k'}^{(i)}$ via \Cref{thm:NonHopReducingEmulator}. The update algorithm at time $i$ terminates after the reinitialization.
\item Otherwise, we have $|(\pi^{(i)}_{k})_{\mid T^{-}}|\leq \phi\cdot |H^{(i)}_{k}|$ and $|\hat{\pi}^{(i)}_{k}|\leq\phi\cdot |H^{(i)}_{k}|$. We feed $\hat{\pi}^{(i)}_{k}$ to the \Cref{thm:NonHopReducingEmulator} maintaining $H_{k+1}$ ($|\hat{\pi}^{(i)}_{k}|\leq \phi\cdot|H^{(i)}_{k}|$ meets the requirement of \Cref{thm:NonHopReducingEmulator}), it will update $H_{k+1}^{(i-1)}$ to $H_{k+1}^{(i)}$ s.t. $H_{k+1}^{(i)}$ is a $(\lambda_{H,\low}, \lambda_{H,\up}, h_{k})$-sparisifer of $T_{k}^{(i)} = T_{k}^{(i-1)}\cup (\pi^{(i)}_{k})_{\mid T^{+}}$ on $H_{k}^{(i)}$. Furthermore, let $\bar{\pi}_{k}^{(i)}$ be the batched update that updates $H_{k}^{(i-1)}$ to $H_{k}^{(i)}$. We define the batched update at the next level be $\pi_{k+1}^{(i)} = \bar{\pi}_{k}^{(i)}\cup (\pi_{k}^{(i)})_{\mid T^{+}} \cup (\pi_{k}^{(i)})_{\mid T^{-}}$.
\end{itemize}

\paragraph{Quality of the Highest Sparsifier $H = H_{\bar{k}}$.} First, $T\subseteq V(H)$ holds because we indeed have the invariant $T_{1}\supseteq T_{2}\supseteq ... T_{\bar{k}}\supseteq T$ at any time by the algorithm.

Second, we show that for each $u,v\in T$, we have $\dist_{H}(u,v)\cdot\alpha_{\low}\geq \dist_{G}(u,v)$, where $\alpha_{\low} = \lambda_{H,\low}^{\bar{k}-1} = 2^{\poly(1/\epsilon)}$. The reason is as follows. For each level 
$1\leq k\leq \bar{k}-1$, because $T\subseteq T_{k}$ and $H_{k+1}$ is a $(\lambda_{H,\low}, \lambda_{H,\up}, h_{k})$-sparsifier of $T_{k}$ on $H_{k}$, we have $\dist_{H_{k+1}}(u,v)\cdot\lambda_{H,\low}\geq \dist_{H_{k}}(u,v)$. Therefore, $\dist_{H}(u,v)\cdot\alpha_{\low}\geq \dist_{G}(u,v)$.

Third, for each $u,v\in T$ s.t. $\dist_{G}(u,v)\leq h$, we have $\dist_{H}(u,v)\leq\alpha_{\up}\cdot \dist_{G}(u,v)$, where $\alpha_{\up} = \lambda_{H,\up}^{\bar{k}-1} = 2^{O(1/\epsilon^{2})}$, by the following reason. At each level $1\leq k\leq \bar{k}-1$, $\dist_{H_{k+1}}(u,v)\leq \lambda_{H,\up}\cdot \dist_{H_{k}}(u,v)$ for each $u,v\in T\subseteq T_{k}$ s.t. $\dist_{H_{k}}(u,v)\leq h_{k} = h\cdot\lambda_{H,\up}^{k-1}$. This implies for each $u,v\in T$, $\dist_{H_{k}}(u,v)\leq h\cdot \dist_{G}(u,v)\cdot\lambda_{H,\up}^{k-1}$ at each level $1\leq k\leq \bar{k}$, and in particular, $\dist_{H}(u,v) = \dist_{H_{\bar{k}}}(u,v) \leq \alpha_{\up}\cdot\dist_{G}(u,v)$. 

The size of $H$ is $|H|\leq O(\kappa_{H,\size}\cdot|T|)$ by the following reason. For each level $1\leq k\leq \bar{k}-1$, we have 
\begin{align*}
|H_{k+1}|&\leq \kappa_{H,\size}\cdot(|T_{k}|+\phi\cdot |H_{k}|)\\
&\leq \kappa_{H,\size}\cdot(|T|+(\phi+3t\phi\cdot(1+\phi)^{t})\cdot |H_{k}|)\\
&\leq \kappa_{H,\size}\cdot |T| + |H_{k}|/n^{\Theta(\epsilon^{2})},
\end{align*}
where the second inequality is by \Cref{claim:13.3} and the third inequality is by $\phi = n^{\epsilon^{2}}$. This implies $|H_{k+1}|\leq O(\kappa_{H,\size}\cdot |T|) + |G|/(n^{\Theta(\epsilon^{2})})^{k}$. Because $\bar{k} = O(1/\epsilon^{2})$ is sufficiently large, we have $|H_{\bar{k}}|\leq O(\kappa_{H,\size}\cdot |T|)$.

\begin{claim}
For each level $1\leq k\leq \bar{k}-1$ and each moment $0\leq i\leq t$, we have $|T_{k}^{(i)}|\leq |T^{(i)}| + 3t\phi\cdot(1+\phi)^{t}|H_{k}^{(i)}|$.
\label{claim:13.3}
\end{claim}
\begin{proof}
Let $i^{\star}\leq i$ be the latest moment that $H_{k+1}^{(i^{\star})}$ is obtained by (re)initialization. Then by the algorithm, we have $T_{k}^{(i^{\star})} = T^{(i^{\star})}$. Consider any moment $i'$ s.t. $i^{\star}+1\leq i'\leq i$, we have $|(\pi^{(i')}_{k})_{\mid T^{-}}|\leq \phi\cdot |H^{(i')}_{k}|$ and $|\hat{\pi}^{(i')}_{k}|\leq\phi\cdot |H^{(i')}_{k}|$, which implies $|\pi^{(i')}_{k}|\leq 2\phi\cdot |H^{(i')}_{k}|$.
Obviously, the difference between $|T^{(i)}|$ and $|T^{(i^{\star})}|$ is bounded by the total size of batched updates during this period. Formally, we have
\begin{align*}
|T^{(i)}|\geq |T^{(i^{\star})}| - \sum_{i^{\star}+1\leq i'\leq i} |\pi^{(i')}_{k}|\geq |T^{(i^{\star})}| - \sum_{i^{\star}+1\leq i'\leq i} 2\phi\cdot |H^{(i')}_{k}|.
\end{align*}
Similarly, the difference between $|T^{(i)}_{k}|$ and $|T^{(i^{\star})}_{k}|$ is bounded by
\begin{align*}
|T^{(i)}_{k}|\leq |T^{(i^{\star})}_{k}| + \sum_{i^{\star}+1\leq i'\leq i} |\hat{\pi}^{(i')}_{k}|\leq |T^{(i^{\star})}_{k}| + \sum_{i^{\star}+1\leq i'\leq i} \phi\cdot |H^{(i')}_{k}|.
\end{align*}
Next, we bound $|H^{(i')}_{k}|$ in terms of $|H^{(i)}_{k}|$. For each $i^{\star}+1\leq i'\leq i$, we have $|H^{(i'-1)}_{k}|\leq |H^{(i')}_{k}| + |\hat{\pi}^{(i')}_{k}|\leq  (1+1/\phi)|H^{(i')}_{k}|$ because $\hat{\pi}^{(i')}_{k}$ is the batched update from $H^{(i'-1)}_{k}$ to $H^{(i')}_{k}$. This means
\[
|H^{(i')}_{k}|\leq (1+\phi)^{i-i'}\cdot |H^{(i)}_{k}|\leq (1+\phi)^{t}\cdot |H^{(i)}_{k}|.
\]
Putting all together, we have $|T^{(i)}_{k}|\leq |T^{(i)}| + 3t\phi(1+\phi)^{t}\cdot|H^{(i)}_{k}|$.

\end{proof}

\paragraph{The Running Time.} The initialization time is 
\[
\sum_{k}(|H_{k}^{(0)}| + |T^{(0)}|/\phi)\cdot\poly(h_{k})\cdot n^{O(\epsilon)}/\phi = |G^{(0)}|\cdot \poly(h)\cdot n^{O(\epsilon)}.
\]

Now consider the batched update $\pi^{(i)}$ at time $i$, and let $k^{\star}$ be the minimum level that a reinitialization is triggered. For each level $1\leq k\leq k^{\star}-2$, we have $|\pi^{(i)}_{k+1}|
\leq |\bar{\pi}^{(i)}_{k+1}| + |\pi^{(i)}_{k}|\leq O(\kappa_{H,\rcs}\cdot|\pi^{(i)}_{k}|)\leq O(\kappa_{H,\rcs})^{k}\cdot |\pi^{(i)}|\leq n^{O(\epsilon)}\cdot |\pi^{(i)}|$. The total update time summing over all levels $1\leq k\leq k^{\star}-1$ is
\[
\sum_{1\leq k\leq k^{\star}-1}|\pi^{(i)}_{k}|\cdot \poly(h_{k})\cdot n^{O(\epsilon)}/\phi^{2} = |\pi^{(i)}|\cdot \poly(h)\cdot n^{O(\epsilon)}.
\]
The total reinitialization time summing over all levels $k^{\star}\leq k\leq \bar{k}-1$ is
\[
\sum_{k^{\star}\leq k\leq \bar{k}-1} (|H_{k}^{(i)}| + |T_{k}^{(i)}|/\phi)\cdot\poly(h_{k})\cdot n^{O(\epsilon)}/\phi = |\pi^{(i)}|\cdot\poly(h)\cdot n^{O(\epsilon)},
\]
because the reinitialization time is dominated by level 
$k^{\star}$, and we have $|\pi^{(i)}_{k^{\star}}|\geq \phi\cdot |H_{k^{\star}}^{(i)}|$, $|T_{k^{\star}}^{(i)}|\leq |H^{(i)}_{k^{\star}}|$ and $|\pi^{(i)}_{k^{\star}}|\leq n^{O(\epsilon)}\cdot |\pi^{(i)}|$.

\end{proof}

\begin{lemma}
Let $G$ be an online-batch dynamic graph with a fully dynamic terminal set $T\subseteq V(G)$ under $t$ batches of updates $\pi^{(1)},...,\pi^{(t)}$ of edge insertions/deletions, isolated vertex insertions/deletions, and terminal insertions/deletions. There is an algorithm that maintains an $(\alpha_{\low}, \alpha_{\up})$-sparsifier $H$ of $T$ on $G$ with $\alpha_{\low} = 2^{\poly(1/\epsilon)}$, $\alpha_{\up} = 2^{O(1/\epsilon^{3})}$, and size $|H|\leq O(\lambda_{x}\cdot \kappa_{H,\size}\cdot|T|) = |T|\cdot n^{O(\epsilon^{4})}$.
The initialization time is $|G^{(0)}|\cdot n^{O(\epsilon)}$ and the update time for $\pi^{(i)}$ is $|\pi^{(i)}|\cdot n^{O(\epsilon)}$.
\label{lemma:SmallSparsifier}
\end{lemma}
\begin{proof}

Fix $h = n^{\epsilon}$. We maintain the same chain ${\cal G}$ of distance-reduced graphs $G_{1},G_{2},...,G_{\bar{x}}$ as that in the proof of \Cref{thm:BatchDynDistanceOracle}. 
\begin{itemize}
\item[(1)] $G_{1} = G$ is the original graph.
\item[(2)] For each $2\leq x\leq \bar{x}$, $G_{x}$ is the distance-reduced graph from applying \Cref{thm:Stacking} on $G_{x-1}$ with length parameter $h$.
\end{itemize}
By setting $\bar{x} = \lambda_{x} = O(1/\epsilon)$ to be sufficiently large, \Cref{claim:LastGraphMaxDis} states that the last graph $G_{\bar{x}}$ has $\dist_{G_{\bar{x}}}(u,v)\leq h$ for all $u,v\in V(G)$.

For each $1\leq x\leq \bar{x}$, we maintain an $(\alpha'_{\low},\alpha'_{\up},h)$-sparsifier $H'_{x}$ of $T$ on $G_{x}$ by applying \Cref{lemma:SmallSparsifierLowDist}, where $\alpha'_{\low} = 2^{\poly(1/\epsilon)}$, $\alpha'_{\up} = 2^{O(1/\epsilon^{2})}$ and $|H'_{x}| = O(\kappa_{H,\size}\cdot |T|)$. Let $H_{x}$ be the graph from scaling up the length of all edges in $H'_{x}$ by a factor $h^{x-1}$.

We construct the final sparsifier $H$ as follows. For each dynamic sparsifier $H_{x}$, we maintain a \emph{fresh graph} $H_{x,\fresh}$ which is isomorphic to $H_{x}$ over all time, but vertices in $H_{x,\fresh}$ are \emph{fresh vertices}. In other words, $H_{x,\fresh}$ is obtained by assigning each vertex in $H_{x}$ a completely new global vertex identifier. This step is to ensure that, over all time, $V(H_{x_{1},\fresh})$ and $V(H_{x_{2},\fresh})$ are disjoint for different $1\leq x_{1},x_{2}\leq\bar{x}$, and each $V(H_{x,\fresh})$ is disjoint from $T$. For each $H_{x,\fresh}$, let $T_{x,\fresh}\subseteq V(H_{x,\fresh})$ be vertices corresponding to $T\subseteq V(H_{x})$. Then we define $H$ by setting $V(H) = T\cup\bigcup_{x}V(H_{x,\fresh})$ and
\[
E(H) = \{(v,v_{x,\fresh})\mid v\in T,1\leq x\leq \bar{x}\}\cup \bigcup_{x}E(H_{x,\fresh}),
\]
where $v_{x,\fresh}\in T_{x,\fresh}$ is the counterpart of $v\in T$, and each edge $(v,v_{x,\fresh})$ has length $1$ (if zero-length edges are allowed, it is more intuitive to set the length to be $0$). In other words, we take the union of all $H_{x,\fresh}$, add the terminals $T$, and then for each terminal $v\in T$ and each $1\leq x\leq \bar{x}$, add an edge $(v,v_{x,\fresh})$ connecting $v$ to its counterpart $v_{x,\fresh}\in T_{x,\fresh}$.

To see the stretch of $H$, consider an arbitrary pair of different vertices $u,v\in T$. We first show that $\dist_{H}(u,v)\leq \alpha_{\up}\cdot \dist_{G}(u,v)$ for $\alpha_{\up} = \lambda^{x^{\star}-1}_{\stk,\up}\cdot \alpha'_{\up} + 2 = 2^{O(1/\epsilon^{2})}$. If $u$ and $v$ are not connected in $G$, then this statement trivially holds. Now suppose $u$ and $v$ are connected in $G$. We take the minimum $x^{\star}$ s.t. $\dist_{G_{x^{\star}}}(u,v)\leq h$. Note that $x^{\star}$ must exist by \Cref{claim:LastGraphMaxDis}. Then for each $2\leq x\leq x^{\star}$, we have $\dist_{G_{x}}(u,v)\leq \lambda_{\stk,\up}\cdot \dist_{G_{x-1}}(u,v)/h$, which implies $\dist_{G_{x^{\star}}}(u,v)\leq \lambda_{\stk,\up}^{x^{\star}-1}\cdot \dist_{G}(u,v)/h^{x^{\star}-1}$. Because $\dist_{G_{x^{\star}}}(u,v)\leq h$, the properties of $H'_{x^{\star}}$ gives $\dist_{H'_{x^{\star}}}(u,v)\leq \alpha'_{\up}\cdot \dist_{G_{x^{\star}}}(u,v)$. Therefore, we have
\begin{align*}
\dist_{H}(u,v)&\leq 2 + \dist_{H_{x,\fresh}}(u_{x,\fresh},v_{x,\fresh})= 2 + \dist_{H_{x}}(u,v)\\
&= 2 + \dist_{H'_{x^{\star}}}(u,v)\cdot h^{x^{\star}-1}\leq 2 + \lambda^{x^{\star}-1}_{\stk,\up}\cdot \alpha'_{\up}\cdot \dist_{G}(u,v)\\
&\leq (\lambda^{x^{\star}-1}_{\stk,\up}\cdot \alpha'_{\up} + 2)\cdot \dist_{G}(u,v).
\end{align*}

Next, we show that $\dist_{G}(u,v)\leq \alpha_{\low}\cdot \dist_{H}(u,v)$, for $\alpha_{\low} = \lambda^{x^{\star}-1}_{\stk,\low}\cdot \alpha'_{\low} \leq 2^{\poly(1/\epsilon)}$. We consider the shortest $u$-to-$v$ path $P$ on $H$. We decompose $P$ into maximal subpaths $P'$ with both endpoints in $T$ but no internal vertices in $T$. For each such subpath $P'$ (with endpoints $u',v'\in T$), we will show $\dist_{G}(u',v')\leq \alpha_{\low}\cdot \ell_{H}(P')$ in a moment, so $\dist_{G}(u,v)\leq \alpha_{\low}\cdot \dist_{H}(u,v)$ holds by concatenating all such subpaths. Now consider one such subpath $P'$. Note that there exists an $x$ s.t. $P'$ is the concatenation of edge $(u',u'_{x,\fresh})$, subpath $P'_{x,\fresh}$ (connecting $u'_{x,\fresh}$ and $v'_{x,\fresh}$) and edge $(v'_{x,\fresh},v')$, where $P'_{x,\fresh}$ is in $H_{x,\fresh}$. This implies $\ell_{H}(P') = 2 + \ell_{H_{x,\fresh}}(P'_{x,\fresh})$. Then
\begin{align*}
\dist_{G}(u',v')&\leq \lambda_{\stk,\low}^{x-1}\cdot h^{x-1}\cdot \dist_{G_{x}}(u',v')\\
&\leq \lambda_{\stk,\low}^{x-1}\cdot\alpha'_{\low}\cdot h^{x-1}\cdot \dist_{H'_{x}}(u',v')\\
&= \lambda_{\stk,\low}^{x-1}\cdot\alpha'_{\low}\cdot\dist_{H_{x}}(u',v')\\
&=\lambda_{\stk,\low}^{x-1}\cdot\alpha'_{\low}\cdot\dist_{H_{x,\fresh}}(u'_{x,\fresh},v'_{x,\fresh})\\
&\leq \lambda_{\stk,\low}^{x^{\star}-1}\cdot\alpha'_{\low}\cdot\ell_{H_{x}}(P'),
\end{align*}
where the first inequality is by property of distance-reduced graphs (see \Cref{thm:Stacking}), and the second inequality is by the property of sparsifier $H'_{x}$.

The running time is straightforward combining \Cref{lemma:SmallSparsifierLowDist} and the analysis in the proof of \Cref{thm:BatchDynDistanceOracle}.

\end{proof}

\begin{proof}[Proof of \Cref{thm:FullyDynamicSparsifier}]
The result \ref{SparsifierWorstCase} with worst-case update time in \Cref{thm:FullyDynamicSparsifier} is a simple corollary from \Cref{lemma:SmallSparsifier} and \Cref{lemma:Reduction} by an argument similar to that in \Cref{sect:Reduction}. We will maintain $2^{O(1/\epsilon)}$ graphs $H_{1},...,H_{2^{O(1/\epsilon)}}$ as candidates of the sparsifer instead of only one, because the reduction in \Cref{lemma:Reduction} basically obtains the fully dynamic algorithm by running $2^{O(1/\epsilon)}$ instances of the online-batch dynamic algorithm in parallel, and at any moment, only one of the instances is ready for the up-to-date $G$. In other words, if we insist to maintain only one sparsifier $H$ explicitly, then when we switch from the copy $k_{i-1}$ for time $i-1$ to the copy $k_{i}$ for time $i$, the sparsifier $H$ needs to change from $H_{k_{i-1}}$ to $H_{k_{i}}$, and this may take $|H_{k_{i-1}}| + |H_{k_{i}}| = (|T^{(i-1)}| + |T^{(i)}|)\cdot n^{O(\epsilon^{4})}$ (much larger than $n^{O(\epsilon)}$) in the worst case, which is unaffordable.

However, if amortized update time is allowed, then a much simpler transformation from online-batch dynamic algorithm to fully dynamic algorithm enables us to maintain the sparsifier $H$ explicitly. Let $b = \lceil n^{\epsilon}\rceil$ and $\bar{z} = \lceil 2/\epsilon\rceil$. For each integer $0\leq i\leq n^{2}$, we can represent it by a unique polynomial $i = \sum_{0\leq z\leq \bar{z}} a_{z}\cdot b^{\bar{z}-z}$, where $0\leq a_{z}\leq b-1$ are integers. We initialize one instance of the online-batch dynamic algorithm, and keep the invariant that at any time $i$ (if we arrive at a moment $i>n^{2}$, we just restart the algorithm), the algorithm is for batched updates $\pi^{(1)},\pi^{(2)},...,\pi^{(\bar{z})}$ where each $\pi^{(z)}$ has size $|\pi^{(z)}| = a_{z}\cdot b^{\bar{z}-z}$ (drop those $\pi^{(z)}$ with size zero) and they are generated by chopping the current sequence of unit updates. By the invariant, when we move from moment $i-1$ to $i$, we may drop a suffix of the batched updates and perform a new batch update (which is the union of the dropped batched updates and the unit update at moment $i$). To simulate this, we just need to withdraw the suffix of batched updates we are planning to drop (which takes time at most their total update time) and then perform the new batched update on the online-batch dynamic algorithm. Regarding the amortized update time, note that each unit update will be inside at most $\bar{z}\cdot b =  n^{O(\epsilon)}$ batched updates (including those got withdrew). Therefore, the total size of batched updates applied to the online-batch dynamic algorithm is at most the multiplication of $n^{O(\epsilon)}$ and the length of the sequence of unit updates. Since the update time for a batch $\pi$ is $|\pi|\cdot n^{O(\epsilon)}$, the final amortized update time is $n^{O(\epsilon)}\cdot n^{O(\epsilon)} = n^{O(\epsilon)}$.

\end{proof}

%% file: 14-multicommodity_flow.tex
\section{Application: Maximum Multicommodity Flows}
\label{sect:MultiFlow}
In this section, we show a constant-approximation multicommodity maxflows algorithm for vertex-unit-capacitated graphs. In the multicommodity maxflow problem, the input is $k$ pairs of source and sink vertices, denoted by $\{(s_{j},t_{j})\mid 1\leq j\leq k\}$. Let ${\cal P}_{j}$ collect all simple paths in $G$ connecting $s_{j}$ and $t_{j}$. The problem is formulated by the LP on the left below, and its dual is on the right.

\begin{align*}
~~~~~~~~\max &\sum_{1\leq j\leq k}\sum_{P\in {\cal P}_{j}} f(P) ~~~~~~~~~~~~~~~~~~~~~~~~~~~~\min\sum_{v\in V(G)} w(v)\\
\text{s.t. }&\sum_{P\ni v} f(P)\leq 1, \forall v\in V(G)~~~~~~~~~~~~~~~~~~\text{s.t. }\sum_{v\in P}w(v)\geq 1,\forall P\in\bigcup_{j}{\cal P}_{j}\\
&f(P)\geq 0, \forall P\in \bigcup_{j}{\cal P}_{j}~~~~~~~~~~~~~~~~~~~~~~~~~~~~w(v)\geq 0,\forall v\in V(G)
\end{align*}

\begin{theorem}
Let $G$ be an vertex-unit-capacitated undirected graph with $n$ vertices, $m$ edges, and $k$ source sink pairs $\{(s_{j},t_{j})\mid 1\leq j\leq k\}$. For some sufficiently small constant $c>0$, given a parameter $1/\log^{c}n\leq \epsilon\leq 1$, there is a deterministic algorithm computes a $2^{\poly(1/\epsilon)}$-approximate multicommodity maxflow in $O((m+n+k)n^{\epsilon})$ time.
\end{theorem}
\begin{proof}
We first prepare distance oracles (with simple-path-reporting interfaces) for graphs with vertex lengths. Precisely, let $G_{w}$ be a graph with length function $w$ on vertices. Naturally, for any simple path $P$ on $G_{w}$, its length is $w(P) = \sum_{v\in P}w(v)$. For two vertices $u,v\in V(G_{w})$, the distance between them is the length of the shortest $u$-$v$ simple paths. 

To implement an oracle for vertex lengths, we define a graph $G_{\ell}$ (isomorphic to $G_{w}$) with length function $\ell$ on edges, where for each edge $e=(u,v)\in E(G_{\ell})$, define its length $\ell(e):=(w(u) + w(v))/2$. Obviously for each simple path $P$, we have $\ell(P)\leq w(P)\leq 2\ell(P)$, so any $\alpha'$-approximate distance oracle for $G_{\ell}$ is a $(2\alpha')$-approximate distance oracle for $G_{w}$.

Therefore, applying the oracle in \Cref{thm:MainDetailed} on $G_{\ell}$, we can obtain an $\alpha$-approximate dynamic distance oracle ${\cal O}$ for $G_{w}$ with approximate ratio $\alpha = 2^{\poly(1/\epsilon)}$. A small technical detail is that \Cref{thm:MainDetailed} requires $G_{\ell}$ has polynomially-bounded positive integral edge lengths. Fortunately, in our application in \Cref{algo:MultiFlow}, we can easily scale up the vertex lengths $w$ to be real numbers (at least $1$ and at most $\poly(n)$), round them up to integers, and multiply them by $2$. This will ensure edge lengths $\ell$ are polynomially-bounded positive integral, and only increase the approximate ratio $\alpha$ by a constant factor.

\paragraph{The Algorithm.} Given an $\alpha$-approximate deterministic dynamic distance oracle (with \emph{simple}-path-reporting interfaces), it is quite standard to obtain fast algorithm for approximate multicommodity maxflows using a multiplicative weight update styled framework. The algorithm is shown in \Cref{algo:MultiFlow}. We use $G_{w}$ to denote the graph $G$ with function $w$ as the edge length function. In \Cref{algo:MultiFlow}, we will initialized an oracle ${\cal O}$ on $G_{w}$ at line 2, query the oracle at line 7 and update the oracle at line 10.

\begin{algorithm}[H]
\caption{Approximate Multicommodity Maxflows}
\label{algo:MultiFlow}
\begin{algorithmic}[1]
\Require A vertex-unit-capacitated undirected graph $G$ with $k$ source-sink pairs $\{(s_{j},t_{j})\mid 1\leq j\leq k\}$ and a constant $0<\delta<0.01$.
\Ensure A $((1+100\delta)\alpha^{2})$-approximate feasible multicommodity flow.
\State Let $\zeta = 1/\delta$ and $\eta = \delta^{2}/((1+10\delta)\log n)$.
\State Initialize $w(v) = 1/n^{\zeta}$ for all vertices $v\in V(G)$.
\State Initialize $\lambda = 1/n^{\zeta}$.
\State Initialize $f$ to be a zero flow.
\While{$\lambda < 1$}
\For{each source-sink pair $(s_{j},t_{j})$}
\State Query an $\alpha$-approximate distance $\wtilde{d}$ for $(s_{j},t_{j})$ in $G_{w}$.
\State If $\wtilde{d} > (1+\delta)\alpha\lambda$, then continue to the next source-sink pair.
\State Query an $\alpha$-approximate shortest simple $s_{j}$-$t_{j}$ path $P$ in $G_{w}$.
\State Add $P$ into $f$ with $f(P) = 1$.
\State For each vertex $v\in P$, $w(v)\gets (1+\delta)\cdot w(v)$.
\EndFor
\State $\lambda \gets (1+\delta)\lambda$
\EndWhile
\State Return $f\cdot\eta$.
\end{algorithmic}
\end{algorithm}

\begin{claim}
At the moment an $\alpha$-approximate shortest simple $s_{j}$-$t_{j}$ path $P$ is obtained, we have $w(P)\leq (1+\delta)\alpha^{2}\lambda$.
\label{claim:MultiApprox}
\end{claim}
\begin{proof}
This is from 
\[
w(P)\leq \alpha\cdot\dist_{G_{w}}(s_{j},t_{j})\leq \alpha\cdot\wtilde{d}\leq (1+\delta)\alpha^{2}\lambda,
\]
where $\wtilde{d}\leq (1+\delta)\alpha\lambda$ is by line 8.
\end{proof}

\paragraph{The Correctness.} The proof of correctness below follows the presentation of \cite{HaeuplerHT2023length}. Combining \Cref{lemma:MultiFlowFeasible,lemma:MultiFlowApprox} and $\alpha = 2^{\poly(1/\epsilon)}$, we have that $f\cdot\eta$ is a $2^{\poly(1/\epsilon)}$-approximate feasible solution.

\begin{lemma}
At any time, we have $\dist_{G_{w}}(s_{j},t_{j})\geq \lambda$ for each $1\leq j\leq k$.
\label{lemma:MultiflowLambda}
\end{lemma}
\begin{proof}
This statement holds initially. Every time $\lambda$ is updated to $(1+\delta)\lambda$, the exit condition of the inner loop guarantees that $\dist_{G_{w}}(s_{j},t_{j})> (1+\delta)\lambda$ for each pair $(s_{j},t_{j})$, because the $\alpha$-approximate distance $\wtilde{d}$ of $(s_{j},t_{j})$ has $w(P)>(1+\delta)\alpha\lambda$. Therefore, this statement still holds after each update of $\lambda$. 
\end{proof}

\begin{lemma}
The flow $f\cdot \eta$ is feasible.
\label{lemma:MultiFlowFeasible}
\end{lemma}
\begin{proof}
It suffices to show that $f(v)\leq 1/\eta$ for each vertex $v\in V(G)$. By the update rule, the final weight of $v$ is $w_{\final}(v) = (1+\delta)^{f(v)}/n^{\zeta}$. Moreover, $w_{\final}(v)\leq(1+\delta)^{2}\alpha$, because right before the last time $w(v)$ is updated, $w(v)\leq w(P)\leq (1+\delta)\alpha^{2}\lambda$ (by \Cref{claim:MultiApprox}) and $\lambda\leq 1$. Therefore,
\begin{align*}
(1+\delta)^{f(v)}/n^{\zeta}&\leq (1+\delta)^{2}\alpha^{2}\\
(1+\delta)^{f(v)}&\leq n^{\zeta+1}\\
f(v)&\leq \frac{(\zeta+1)\log n}{\log(1+\delta)}\\
f(v)&\leq \frac{(1+\delta)(\zeta+1)\log n}{\delta}\\
f(v)&\leq \frac{(1+\delta)^{2}\log n}{\delta^{2}},
\end{align*}
where we assume $(1+\delta)^{2}\alpha^{2}\leq n$ in the second line, and uses $\delta/(1+\delta)\leq \log(1+\delta)$ for all $\delta>-1$ in the fourth line. Therefore, $f(v)\leq 1/\eta$ by our choice of $\eta$.
\end{proof}

\begin{lemma}
The flow $f\cdot\eta$ is an $((1+100\delta)\cdot\alpha^{2})$-approximate multicommodity maxflow.
\label{lemma:MultiFlowApprox}
\end{lemma}
\begin{proof}
Let $\beta$ denote the value of the maxflow.
Consider an iteration $i$ (we refer to the inner loop as one iteration). Let $\lambda^{(i)},w^{(i)}$ denote the variable $\lambda$ and function $w$ at the beginning of this iteration. Let $D^{(i)} = \sum w^{(i)}(v)$. Note that $w^{(i)}/\lambda^{(i)}$ is a feasible dual solution, so $\beta\leq D^{(i)}/\lambda^{(i)}$. Let $P_{i}$ be the path $P$ in iteration $i$. By our update rule, we have
\begin{align*}
D^{(i+1)} &= \sum_{v\notin P_{i}} w^{(i)}(v) + \sum_{v\in P_{i}} (1+\delta)\cdot w^{(i)}(v)\\
&\leq \sum_{v} w^{(i)}(v) + \sum_{v\in P_{i}}\delta\cdot w^{(i)}(v)\\
&\leq D^{(i)} + \delta (1+\delta)\alpha^{2}\lambda^{(i)}\\
&\leq D^{(i)}\cdot (1 + \delta(1+\delta)\alpha^{2}/\beta)\\
&\leq D^{(i)}\cdot\exp(\delta(1+\delta)\alpha^{2}/\beta).
\end{align*}
From the second line to the third line, we use $w^{(i)}(P_{i})\leq (1+\delta)\alpha^{2}\lambda^{(i)}$ from \Cref{claim:MultiApprox}.

Let $K = \val(f)$ and let $T-1$ be the last iteration, then the final dual value satisfies
\[
D^{(T)}\leq D^{(1)}\cdot \exp(\delta(1+\delta)\alpha^{2} K/\beta).
\]
From the initialization, we have $D^{(1)} = 1/n^{\zeta-1}$. At the end of the algorithm, we have $\min\dist_{G_{w}}(s_{j},t_{j})\geq \lambda\geq 1$ (by \Cref{lemma:MultiflowLambda}), which implies $D^{(T)}\geq 1$. Therefore,
\begin{align*}
n^{\zeta-1}&\leq \exp(\delta(1+\delta)\alpha^{2} K/\beta)\\
\beta\log n\cdot \frac{\zeta -1}{\delta(1+\delta)\alpha^{2}} &\leq K\\
\frac{\beta(1-\delta)}{(1+\delta)(1+10\delta)\alpha^{2}}&\leq K\eta\\
\frac{\beta}{(1+100\delta)\alpha^{2}}&\leq K\eta
\end{align*}
which means $\val(f\cdot\eta) = K\eta \geq \beta/((1+100\delta)\alpha^{2})$.
\end{proof}

\paragraph{The Running Time.} The running time is dominated by the total update time, total distance query time and total path-reporting time of the oracle ${\cal O}$ on $G_{\ell}$. 

Let ${\cal P}_{\report}$ collect all paths obtained at line 9. We have the total number of edges on paths in ${\cal P}_{\report}$ is at most $O(n\log n)$, i.e. $\sum_{P\in{\cal P}_{\report}} |P|\leq O(n\log n)$ by the fact that $f\cdot \eta$ is a feasible flow and $\eta = \Theta(1/\log n)$. Therefore, the total path-reporting time is $O(n^{1+\epsilon})$ by \Cref{thm:MainDetailed}. 

For the total distance query time, note that each distance query followed by a path-reporting query will have its distance query time $O(\log\log n/\epsilon^{4})$ dominated by the path-reporting time $O(|P|\cdot n^{\epsilon})$ because $|P|\geq 1$. Hence, we only need to consider the extra distance queries of those inner-loop iterations ended at line 8. There are at most $O(k\cdot \log_{1+\delta}n^{\zeta}) = O(k\log n)$ such iterations, so the extra distance queries takes $O(k\log n\log\log n/\epsilon^{4})$ time.

Finally, we bound the total update time. For each vertex $v\in V(G)$, its weight $w(v)$ will be updated at most $O(\log_{1+\delta} n^{\delta}) = O(\log n)$ time. Each update to $w(v)$ will change the edge lengths of $\deg_{G}(v)$ edges in $G_{\ell}$. To change the length of some edge in $G_{\ell}$, we can remove the edge and add it back with new edge length. Therefore, the total number of updates to $G_{\ell}$ is $O(m\log n)$, and the total update time is $O(mn^{\epsilon})$.

In conclusion, the total running time is $O(n^{1+\epsilon} + k\log n\log\log n/\epsilon^{4} + mn^{\epsilon}) = O((m+n+k)n^{\epsilon})$.

\end{proof}

%% file: tables.tex
\newpage
\section{Dynamic Distance Oracles in the Literature}

\begin{table}[H]
\footnotesize{

\begin{tabular}{|>{\centering}p{0.1\textwidth}|>{\centering}p{0.1\textwidth}|>{\centering}p{0.08\textwidth}|>{\centering}p{0.15\textwidth}|>{\centering}p{0.15\textwidth}|>{\centering}p{0.08\textwidth}|>{\centering}p{0.12\textwidth}|>{\centering}p{0.1\textwidth}|}
\hline 
 & Fully Dynamic & Approx & Update Time & Query Time & Det ? & Worst-Case & Note\tabularnewline
\hline 
\hline 
\textbf{Exact Distance Matrix} &  &  &  &  &  &  & \tabularnewline
\hline 
\cite{ausiello1991incremental} & incremental & $1$ & $n^{3}/m$ & 1 & det & amortized & unweighted, directed\tabularnewline
\hline 
\cite{BaswanaHS07} & decremental & $1$ & $n^{3}/m$ & $1$ & adaptive & amortized & unweighted, directed\tabularnewline
\hline 
\cite{evald2021decremental} & decremental & $1$ & $n^{3}/m$ & $1$ & det & amortized & unweighted, directed\tabularnewline
\hline 
\cite{king1999fully,demetrescu2001fully} & fully & 1 & $n^{2.5}$ & 1 & det & amortized & unweighted, directed\tabularnewline
\hline 
\cite{demetrescu2004new} & fully & $1$ & $n^{2}$ & $1$ & det & amortized & directed\tabularnewline
\hline 
\cite{thorup2005worst} & fully & $1$ & $n^{3-1/4}$ & 1 & det & worst-case & directed\tabularnewline
\hline 
\cite{abraham2017fully} & fully & 1 & $n^{3-1/3}$ & 1 & adaptive & worst-case & directed\tabularnewline
\hline 
\cite{gutenberg2020fully} & fully & $1$ & $n^{3-2/7}$ & 1 & det & worst-case & directed\tabularnewline
\hline 
\cite{chechik2023faster} & fully & $1$ & $n^{3-20/61}$ & $1$ & det & worst-case & directed\tabularnewline
\hline 
 &  &  &  &  &  &  & \tabularnewline
\hline 
\textbf{Exact Oracle} &  &  &  &  &  &  & \tabularnewline
\hline 
\cite{Sankowski05} & fully & $1$ & $n^{1.897}$ & $n^{1.265}$ & adaptive & worst-case & unweighted, directed\tabularnewline
\hline 
\cite{RodittyZ11} & fully & $1$ & $m\sqrt{n}$ & $n^{3/4}$ & adaptive & amortized & unweighted, directed\tabularnewline
\hline 
\cite{BrandNS19} & fully & $1$ & $n^{1.724}$ & $n^{1.724}$ & adaptive & worst-case & unweighted, directed\tabularnewline
\hline 
\cite{karczmarz2023sensitivity} & fully & $1$ & $n^{2}$ & $n$ & adaptive & worst-case & unweighted, directed\tabularnewline
\hline 
\cite{karczmarz2023fully} & fully & $1$ & $mn^{4/5}$ & $n^{4/5}$ & det & worst-case & directed\tabularnewline
\hline 
\end{tabular}

}

\caption{\label{tab:exact}Exact regime. We assume $\protect\poly(n)$ aspect ratio. We often omit $\protect\poly(\log n)$ factors in the update time.}
\end{table}

\begin{table}[H]
\footnotesize{

\begin{tabular}{|>{\centering}p{0.1\textwidth}|>{\centering}p{0.1\textwidth}|>{\centering}p{0.08\textwidth}|>{\centering}p{0.15\textwidth}|>{\centering}p{0.15\textwidth}|>{\centering}p{0.08\textwidth}|>{\centering}p{0.12\textwidth}|>{\centering}p{0.1\textwidth}|}
\hline 
 & Fully Dynamic & Approx & Update Time & Query Time & Det ? & Worst-Case & Note\tabularnewline
\hline 
\hline 
\textbf{Partially dynamic $(1+\epsilon)$-approx} &  &  &  &  &  &  & \tabularnewline
\hline 
\cite{BaswanaHS07} & decremental & $(1+\epsilon)$ & $n^{2}/\epsilon\sqrt{m}$ & $1$ & oblivious & amortized & unweighted, directed\tabularnewline
\hline 
\cite{rodittyZ2} & decremental & $(1+\epsilon)$ & $n/\epsilon$ & $1$ & oblivious & amortized & unweighted\tabularnewline
\hline 
\cite{henzinger16} & decremental & $(1+\epsilon)$ & $n/\epsilon$ & $1$ & det & amortized & unweighted\tabularnewline
\hline 
\cite{bernstein16} & decremental & $(1+\epsilon)$ & $n/\epsilon$ & $1$ & oblivious & amortized & directed\tabularnewline
\hline 
\cite{karczmarz2019reliable} & incremental & $(1+\epsilon)$ & $n^{4/3}/\epsilon$ & $1$ & det & amortized & directed\tabularnewline
\hline 
\cite{evald2021decremental} & decremental & $(1+\epsilon)$ & $n^{2}/\epsilon\sqrt{m}$ & $1$ & det & amortized & unweighted, directed\tabularnewline
\hline 
\cite{bernstein2022deterministic} & decremental & $(1+\epsilon)$ & $n/\epsilon$ & $1$ & det & amortized & \tabularnewline
\hline 
 &  &  &  &  &  &  & \tabularnewline
\hline 
\textbf{Fully dynamic $(1+\epsilon)$-approx} &  &  &  &  &  &  & \tabularnewline
\hline 
\cite{van2019dynamic} & fully & $(1+\epsilon)$ & $n^{2}/\epsilon^{\omega+1}$ & $1$ & oblivious & worst-case & unweighted\tabularnewline
\hline 
\cite{van2019dynamic} & fully & $(1+\epsilon)$ & $n^{2.045}/\epsilon^{2}$ & $1$ & oblivious & worst-case & directed\tabularnewline
\hline 
\cite{brand2021fast} & fully & $(1+\epsilon)$ & $n^{2}O(\frac{1}{\epsilon})^{\sqrt{2\log_{1/\epsilon}n}}$ & $1$ & det & worst-case & \tabularnewline
\hline 
\cite{van2023deterministic} & fully & $(1+\epsilon)$ & $n^{2.293}/\epsilon$ & 1 & det & worst-case & directed\tabularnewline
\hline 
\cite{van2019dynamic} & fully & $(1+\epsilon)$ & $n^{1.863}\frac{1}{\epsilon^{2}}\log\frac{1}{\epsilon}$ & $n^{0.66}/\epsilon^{2}$ & oblivious & worst-case & directed\tabularnewline
\hline 
\cite{brand2021fast} & fully & $(1+\epsilon)$ & $n^{1.788}O(\frac{1}{\epsilon})^{\sqrt{2\log_{1/\epsilon}n}}$ & $n^{0.45}/\epsilon^{2}$ & det & worst-case & unweighted\tabularnewline
\hline 
\cite{bergamaschi2021new} & fully  & $(1+\epsilon)$ & $n^{1.529}$ & $n^{1+o(1)}$ & oblivious & worst-case & unweighted\tabularnewline
\hline 
 &  &  &  &  &  &  & \tabularnewline
\hline 
\textbf{$(2+\epsilon)$-approx} &  &  &  &  &  &  & \tabularnewline
\hline 
\cite{king1999fully} & fully & $(2+\epsilon)$ & $n^{2}$ & $1$ & det & amortized & directed\tabularnewline
\hline 
\cite{APSPfully2} & fully & $(2+\epsilon)$ & $m\cdot2^{O(\sqrt{\log n})}\cdot(\frac{1}{\epsilon})^{O(\log_{1/\epsilon}n)}$ & $\log\log\log n$ & obl & amortized & \tabularnewline
\hline 
\cite{bernstein2022deterministic} & fully & $(2+\epsilon)$ & $m\cdot2^{O(\sqrt{\log n})}\cdot(\frac{1}{\epsilon})^{O(\log_{1/\epsilon}n)}$ & $\log\log\log n$ & det & amortized & \tabularnewline
\hline 
\cite{dory2022new} & Decremental & $(2+\epsilon)$ & $\max\{n^{2/3},$ $n^{3/2+o(1)}/\sqrt{m}\}$ & $1$ & obl & amortized & \tabularnewline
\hline 
\end{tabular}

}

\caption{\label{tab:small approx}Small-approximation regime. We assume $\protect\poly(n)$ aspect ratio. We often omit $\protect\poly(\log n)$ factors in the update time.}
\end{table}

\begin{table}[H]
\footnotesize{

\begin{tabular}{|>{\centering}p{0.1\textwidth}|>{\centering}p{0.1\textwidth}|>{\centering}p{0.14\textwidth}|>{\centering}p{0.13\textwidth}|>{\centering}p{0.18\textwidth}|>{\centering}p{0.05\textwidth}|>{\centering}p{0.1\textwidth}|>{\centering}p{0.1\textwidth}|}
\hline 
 & Fully Dynamic & Approx & Update Time & Query Time & Det ? & Worst-Case & Note\tabularnewline
\hline 
\hline 
\textbf{Lower bounds} &  &  &  &  &  &  & \tabularnewline
\hline 
\cite{DorHZ00,HenzingerKNS15} & fully & $2-\epsilon$ & $n^{1-o(1)}$ & $n^{2-o(1)}$ & obl & amortized & $m=\Omega(n^{2})$\tabularnewline
\hline 
\cite{abboud2023stronger} & fully & $\alpha\in[2,3)$ & $n^{2/(1+\alpha)-o(1)}$ & $n^{o(1)}$ & obl & amortized & $m=n^{1+o(1)}$\tabularnewline
\hline 
\cite{abboud2023stronger,jin2023removing} & fully & $2k-1$ & $n^{1/k-o(1)}$ & $n^{o(1)}$ & obl & amortized & $m=O(n)$\tabularnewline
\hline 
\cite{abboud2022hardness} & fully & $2k-1$ & $n^{\Omega(1/k)}$ & $n^{\Omega(1/k)}$ & obl & amortized & $m=O(n)$\tabularnewline
\hline 
\cite{BrandNS19} & partially & $5/3-\epsilon$ & $n^{2-o(1)}$ & $n^{o(1)}$ & obl & worst-case & $m=\Omega(n^{2})$\tabularnewline
\hline 
 &  &  &  &  &  &  & \tabularnewline
\hline 
\textbf{Spanners} &  &  &  &  &  &  & \tabularnewline
\hline 
\cite{BaswanaKS12} & fully & $(2k-1)$ & $k^{2}\log^{2}n$ & $kn^{1+1/k}\log n$ & obl & amortized & \tabularnewline
\hline 
\cite{bernstein2021deamortization} & fully & $(2k-1)$ & $O(1)^{k}\log^{3}n$ & $kn^{1+1/k}\polylog n$ & obl & worst-case & \tabularnewline
\hline 
\cite{bernstein2022fully} & fully & $\polylog n$ & $\polylog n$ & $n\cdot\polylog n$ & adapt & amortized & \tabularnewline
\hline 
\cite{bernstein2022fully} & fully & $\polylog n$ & $n^{o(1)}$ & $n^{1+o(1)}$ & adapt & worst-case & \tabularnewline
\hline 
 &  &  &  &  &  &  & \tabularnewline
\hline 
\textbf{Partially Dynamic} &  &  &  &  &  &  & \tabularnewline
\hline 
\cite{rodittyZ2} & decremental & $(2+\epsilon)k-1$ & $n$ & $1$ & obl & amortized & unweighted\tabularnewline
\hline 
\cite{BernsteinR11} & decremental & $(2+\epsilon)k-1$ & $\frac{n^{2}}{m}\cdot n^{1/k}$ & $1$ & obl & amortized & unweighted\tabularnewline
\hline 
\cite{chechik} & decremental & $(2+\epsilon)k-1$ & $n^{1/k+o(1)}$ & $\log\log n$ & obl & amortized & \tabularnewline
\hline 
\cite{lkacki2022near} & decremental & $(2+\epsilon)k-1$ & $n^{1/k}(\frac{k\log n}{\epsilon})^{O(k)}$ & $k$ & obl & amortized & \tabularnewline
\hline 
\cite{chuzhoy2021decremental} & decremental & $(\log n)^{2^{O(1/\epsilon)}}$ & $n^{\epsilon}$ & $\log n\log\log n$ & det & amortized & \tabularnewline
\hline 
\cite{bernstein2022deterministic} & decremental & $n^{o(1)}$ & $n^{o(1)}$ & $n^{o(1)}$ & det & amortized & \tabularnewline
\hline 
\cite{chen2020fast} & incremental & $(1/\epsilon)^{O(1/\epsilon)}$ & $n^{\epsilon}$ & $n^{\epsilon}$ & det & worst-case & \tabularnewline
\hline 
\cite{chen2020fast} & incremental & $(2k-1)$ & $m^{1/2}n^{1/k}$ & $m^{1/2}n^{1/k}$ & det & worst-case & \tabularnewline
\hline 
\cite{forster2023deterministic} & incremental & $\polylog n$ & $\polylog n$ & $\log\log n$ & det & amortized & \tabularnewline
\hline 
 &  &  &  &  &  &  & \tabularnewline
\hline 
\textbf{Fully Dynamic} &  &  &  &  &  &  & \tabularnewline
\hline 
\cite{abraham2014fully} & fully & $2^{O(1/\epsilon^{2})}$ & $m^{1/2+\epsilon}$ & $O(1/\epsilon^{4})$ & obl & amortized & \tabularnewline
\hline 
\cite{forster2021dynamic} & fully & $(\log n)^{O(1/\epsilon)}$ & $n^{\epsilon}$ & $(1/\epsilon)\cdot\log^{2}n$ & obl & amortized & \tabularnewline
\hline 
\cite{forster2023bootstrapping} & fully & $(1/\epsilon)^{O(1/\epsilon)}$ & $n^{\epsilon}$ & $n^{\epsilon/8}$ & obl & amortized & \tabularnewline
\hline 
\cite{chuzhoy2023new} & fully & $(\log\log n)^{2^{O(1/\epsilon^{3})}}$ & $n^{\epsilon}$ & $2^{\poly(1/\epsilon)}\log n\log\log n$ & det & amortized & \tabularnewline
\hline 
\makecell{\cite{KMP23}\\\tiny{Independent}} & fully & $n^{o(1)}$ & $n^{o(1)}$ & $\log^{2} n$ & det & worst-case & \tabularnewline
\hline 
\textbf{Our} & fully & $2^{\poly(1/\epsilon)}$ & $n^{\epsilon}$ & $\poly(1/\epsilon)\cdot\log\log n$ & det & worst-case & \tabularnewline
\hline 
\end{tabular}

}

\caption{\label{tab:large approx}Large-approximation regime. We assume $\protect\poly(n)$ aspect ratio. We often omit $\protect\poly(\log n)$ factors in the update time. The variable $m$ in the lower bounds indicates the number of edges in the hard instances. \cite{KMP23} is an independent work as discussed in \Cref{sect:IndependentWork}.}
\end{table}

%% file: main.bbl
\newcommand{\etalchar}[1]{$^{#1}$}